\documentclass[reqno]{amsart}
\usepackage{amsmath}
\usepackage{amsthm}
\usepackage{amssymb}
\usepackage{mathrsfs}
\usepackage{amsfonts}
\usepackage{newlfont}
\usepackage[sans]{dsfont}
\usepackage{dcolumn}
\usepackage{bm}
\usepackage{stmaryrd}
\usepackage{bbm}
\usepackage{relsize}
\usepackage{euscript}
\usepackage[margin=1.0in]{geometry}
\numberwithin{equation}{section}
\newtheorem{thm}{Theorem}[section]

\newtheorem{cor}[thm]{Corollary}
\newtheorem{lem}[thm]{Lemma}
\newtheorem{prop}[thm]{Proposition}
\theoremstyle{definition}
\newtheorem{defn}[thm]{Definition}
\newtheorem{rem}[thm]{Remark}
\newtheorem{exam}[thm]{Example}
\begin{document}
{\allowdisplaybreaks
\title[]{STOCHASTIC CONTROL OF TOLMAN-OPPENHEIMER-SNYDER COLLAPSE OF ZERO-PRESSURE STARS TO BLACK HOLES: RIGOROUS CRITERIA FOR DENSITY BOUNDS AND SINGULARITY SMOOTHING}
\author{Steven D. Miller}
\address{}
\email{stevendM@ed-alumnus.net}
\date{\today}
\begin{abstract}
The Tolman-Oppenheimer-Snyder description gives exact analytical solutions for an Einstein-matter system describing total gravitational collapse of a zero-pressure perfect-fluid sphere, representing a massive star which has exhausted its nuclear fuel. The star collapses to a point of infinite density within a finite comoving proper time interval $[0,t_{*}]$, and the exterior metric matches the Schwarzchild black hole metric. The description is re-expressed in terms of a 'density function' $u(t)=(\rho(t)/\rho_{o}))^{1/3}=R^{-1}(t)$ for initial density $u_{0}=R^{-1}(0)=1$ and radius $R(0)$, whereby the general-relativistic formulation reduces to an autonomous nonlinear ODE for $u(t)$. The solution blows up or is singular at $t=t_{*}=\pi/2(8\pi G/3\rho_{o})^{1/2}$. The blowup interval $[0,t_{*}]$ is partitioned into domains $[0,t_{\epsilon}]\bigcup[t_{\epsilon},t_{*}]$,with $t_{*}=t_{\epsilon}+|\epsilon|$ and $|\epsilon|\ll 1$, so that $t_{\epsilon}$ can be infinitesimally close to $t_{*}$. Randomness or 'stochastic control' is introduced via the 'switching on' of specific (white-noise) perturbations at $t=t_{\epsilon}$. Hybrid nonlinear ODES-SDES are then 'engineered' over the partition. Within the Ito interpretation, the resulting density function diffusion $\overline{u(t)}$ is proved to be a martingale on $[1,\infty)$ whose supremum, volatility and higher-order moments are finite, bounded and singularity free for all finite $t>t_{\epsilon}$. The collapse is (comovingly) eternal but never becomes singular. Extensive and rigorous boundedness and no-blowup criteria are established via various methods, and blowup probability is always zero. The density singularity is therefore smoothed or 'noise-suppressed'. Within the Stratanovitch interpretation, the singularity formation probability is unity; however, null recurrence ensures the expected comoving time for this to occur is now infinite.
\end{abstract}
\maketitle
\tableofcontents
\raggedbottom
\section{Introduction}
Applied stochastic functional analysis deals with the mathematics of dynamical systems interacting with noise or randomness [56-63,71]. The motivation of this paper is to tentatively promote the potential utilization of techniques and tools from stochastic functional analysis and probability theory within mathematical general relativity. As such, it should be considered as mathematics inspired by a physical scenario. One can tentatively seek to rigorously incorporate the possible effects of classical random fluctuations, perturbations or noise within classical general relativistic problems; or describe an underlying 'stochasticity of spacetime' which might exist on very microscopic scales. This is considered in relation to formulating a rigorous stochastic extension of the Tolman-Oppenheimer-Snyder analysis of black hole formation from a gravitationally collapsing star that has exhausted its nuclear fuel and is collapsing to a state of infinite density and essentially zero size. This classic model essentially applies general relativity to a self-gravitating perfect-fluid sphere with an initially finite density but zero pressure [1,2,3]. This is an exactly  solvable Einstein-Euler-matter problem and the spherical star collapses to a singular state of infinite density within a finite comoving proper time. However, in this controlled stochastic extension, randomness is incorporated or 'switched' on in the final (microsopic) stage of the collapse. It is then analysed and interpreted as a 'nonlinear stochastic or random dynamical system' that is potentially free from density blowups and hard singularities. In particular, well-established stabilisation, blow-up or singularity-suppressing properties of noise for nonlinear differential equations and systems can be applied [45-53].

Singularities arise from the total gravitational collapse of very massive stars and also within cosmology at the birth of the universe at the initial Big Bang [4-10] Although the universe is expanding, black holes can be interpreted as regions of the universe which are locally collapsing. The interior of a collapsing star can be interpreted as the reverse of the Big Bang process of expansion; and indeed, in the TOV description, the interior metric of the collapsing star is of the FRW form [3]. The boundary of this locally collapsing region is a trapped surface or the 'event horizon' [4.5]. It finally became accepted with the work of Hawking and Penrose that they are an inherent and unavoidable feature of pure general relativity. While general relativity is a rigorous classical description of gravitation and space-time structure, the presence of singularities also essentially highlights its 'ultraviolet incompleteness'. A complete or extended theory is therefore assumed to be some form of a quantum gravity theory.
\raggedbottom

A 'fundamental dogma' of quantum gravity is that quantum mechanical effects will somehow resolve the inherent singularity problems of pure Einstein gravity [16-19]. However, the final state of some very massive or supergiant collapsed stars well above the Tolman-Oppenheimer-Volkoff limit [3,11] may not actually be black holes per se but extreme objects of very immense albeit finite density, such as quarks stars, preon stars, electroweak stars or Planck stars [12-15]. Although a trapped surface or horizon will form, and the exterior region is still a black hole, the final 'interior state' may still not be a hard singularity but a microscopic and hyperdense fluctuating entity such as a Planck star--but certainly some kind of extreme final state or 'Planckian core' of finite but very extreme density [13-15]. Within quantum gravity, one expects that quantum effects or fluctuations in space-time have somehow exerted a counterbalancing "pressure" once the collapsing stellar matter reaches a certain critical density and such that the final collapsed state is some kind of non-singular Planckian core. However, no complete QG theory exists which can provide any rigorous or deeper analysis of this problem beyond heuristic arguments.

Even if such a true quantum gravity theory does exist, it is viable that there could still an intermediate regime of scales where a purely classical and rigorous stochastic theory or extension of general relativity may be both physically and mathematically viable. If we take a 'top-down' approach and attempt to extend classical general relativity to very small scales, the smooth matter geodesics of gravitational collapsing matter at its very latter stages, may evolve into 'classical random diffusions'; from the perspective of comoving observers moving with the stellar matter as it collapses towards microscopic or Planckian scales, with some kind of new physics perhaps coming into effect at these scales.

All known physical systems possess random fluctuations or noise on some length scale, either of thermal or quantum origin so that the effect of fluctuations will become a crucial issue at these critical scales, which are microscopic scales [20,21,71]. Also, very many physical processes tend to exhibit nonlinearity: these include turbulence, nonlinear dynamics, chaos and pattern formation and fractals; in effect,randomness and nonlinearity are ubiquitous in nature [20,21,71]. Coupling stochastic noise or random fields to linear and nonlinear ODEs or PDEs is also a useful and powerful methodology for studying a variety of such problems [20-24] in particular, the field of stochastic PDES is now an area of growing interest [25]. A key example is the statistical study of turbulence via the coupling random fields or noise to the nonlinear Navier-Stokes equations [23,24], which like the Einstein equations are also of nonlinear hyperbolic type.

General relativity is also a highly non-linear theory by which matter interacts with spacetime. If spacetime itself is an environment/system that is a fluctuating or a noisy at very small scales--for example, the spacetime 'foam' near the Planck scale--then this could potentially counteract the formation of singularities in gravitational collapse and cosmology [43]. There is scope to develop and apply new and existing stochastic analytical methods in order to rigorously incorporate classical random fluctuations, perturbations, noise or random structure into the theory for various applications. This might also include structure formation in the very early universe.

In this paper, stochastic analysis and probabilistic methods are tentatively applied to the basic Oppenheimer-Snyder gravitational collapse problem for a pressureless perfect-fluid sphere, with stochastic or random perturbations or fluctuations coming into effect or being 'switched on' at microscopic scales close to when the matter density normally blows up or becomes singular. These random perturbations are taken to represent a manifestation of some deeper underlying (but unknown) random properties of microscopic spacetime scales, that are modelled classically but which may have a deeper quantum mechanical origin.

A physical analogy of this classical stochastic interpretation would be thermal Brownian motions of particles suspended in a liquid, which has its origins in atomic/molecular physics, but which can be modelled as a classical stochastic process or random walk at the length scales comparable to the particle size [60,71. Statistical hydrodynamical models of turbulence in fluids also extend the classical hydrodynamical theory by coupling random fields/noises to the nonlinear Navier-Stokes equations [23-25]. Radiative transfer and photon transport in random or turbid medias can be modelled at the length scale of the scattering particles as 'photon diffusions' or classical random walks, even though the underlying light scattering is fundamentally quantum mechanical in origin [26].
\section{The equilibrium of relativistic and Newtonian stars}
We briefly establish the basic machinery from general relativity that describes the hydrostatic equilibrium and gravitational field of a massive but static, spherically symmetric perfect-fluid star of uniform density and central pressure. In ordinary stars, the central pressure due to gravitational compression is balanced by thermal and radiation pressure from the thermonuclear burning of fuel in the core; and for cold collapsed stars in equilibrium by electron degeneracy pressure. But from the conditions of relativistic stellar equilibrium alone it can be proved that a star of sufficiently high density, beyond a critical density, has an infinite central pressure and so cannot be prevented from completely collapsing beyond its Schwartzchild radius/horizon to a point. The exterior region is then a black hole. There is extensive literature on the equilibrium and collapse of relativistic and Newtonian fluid spheres or stars [1,2,3,27-37].

Given the Einstein equations formulated for a spherically symmetric perfect fluid star, one can derive the Tolman-Oppenheimer-Volkof (TOV) equation for relativistic hydrostatic equilibrium. Using a polytropic equation of state, the Lane-Emden equations also follow in the Newtonian limit [3,30-32]. The full TOV equation with a polytropic state equation gives the relativistic Lane-Emden equations. Integrating the TOV for an incompressible fluid yields a universal formula for the central pressure of the star. This also establishes the Buckdahl Theorem $(GM/R)<4/9$, which actually holds for all stars in the universe regardless of the equation of state [3].
\subsection{Relativistic hydrodynamics and the Einstein equations for a static perfect-fluid star}
\begin{defn}
The energy momentum tensor for a relativistic perfect fluid in the presence of gravitation is of the form
\begin{equation}
\mathbf{T}^{\mu\nu}=p\mathbf{g}^{\mu\nu}+(p+\rho)U^{\mu}U^{\nu}
\end{equation}
where p is the pressure and $\rho$ is the density of the fluid and $U^{\mu}$ is the fluid 4-velocity, which is the local velocity $dx^{\mu}/d\tau$ of a comoving fluid element. The relation $\mathbf{g}_{\mu\nu}U^{\mu}U^{\nu}=-1$ always holds.
\end{defn}
\begin{defn}
For a static fluid $U^{0}=(-\mathbf{g}_{00})^{1/2}$ and $U^{i}=0$, for $i=1,2,3$, and $\partial_{0}P=\partial_{0}\rho=\partial_{0}\mathbf{g}_{\mu\nu}=0$. Also $\mathbf{\Gamma}_{00}^{\mu}=-\tfrac{1}{2}\mathbf{g}^{\mu\nu}\partial_{\nu}\mathbf{g}_{00}$ and
\begin{equation}
\partial_{\nu}[(p+\rho)U^{\mu}U^{\nu}]=0
\end{equation}
\end{defn}
The conditions for conservation of energy and momentum are
\begin{equation}
\mathlarger{\mathlarger{\mathcal{U}}}(\mathbf{X}_{II}\cup\mathbf{X}_{III})\mathlarger{\mathbf{D}}_{\mu}\mathbf{T}^{\mu\nu}={\mathbf{D}}_{\mu}[p\mathbf{g}^{\mu\nu}+(p+\rho)U^{\mu}U^{\nu}]=0
\end{equation}
which gives the hydrodynamic equations
\begin{equation}
\mathlarger{\mathbf{D}}_{\mu}\mathbf{T}^{\mu\nu}=\partial_{\mu}p\mathbf{g}^{\kappa\mu}
+\mathbf{g}^{-1/2}\partial_{\nu}\mathbf{g}^{1/2}(p+\rho)U^{\mu}U^{\nu}+\mathbf{\Gamma}_{\nu\lambda}^{\mu}
(P+\rho)U^{\mu}U^{\nu}
\end{equation}
where $\mathbf{D}_{\mu}$ is the covariant derivative. Multiplying (2.4) by $\mathbf{g}_{\mu\lambda}$ gives
\begin{equation}
-\partial_{\lambda}p-(p+\rho)\partial_{\lambda}\log(-\mathbf{g}_{00})^{1/2}
\end{equation}
\begin{defn}
The standard form of the spherically symmetric static metric is
\begin{equation}
ds^{2}=Y(r)dt^{2}-X(r)dr^{2}-r^{2}(d\theta^{2}+\sin^{2}\theta d\varphi^{2})
\end{equation}
where $\mathbf{g}_{tt}=-Y(r),\mathbf{g}_{rr}=X(r),\mathbf{g}_{\theta\theta}=r^{2},\mathbf{g}_{\varphi\varphi}=r^{2}\sin^{2}\theta$ and $\mathbf{g}_{\mu\nu}=0$ for $\mu\ne \nu$. For a fluid at rest or equilibrium $U^{r}=U^{\theta}=U^{\varphi}=0$ and $U^{t}=-(-\mathbf{g}^{tt})^{-1/2}=-\sqrt{Y(r)}$, which  follows from $\mathbf{g}^{\mu\nu}U_{\mu}U_{\nu}=-1$
\end{defn}
The Einstein equations coupled to a perfect fluid source term are
\begin{equation}
\bm{\mathrm{G}}_{\mu\nu}\equiv\bm{\mathrm{Ric}}_{\mu\nu}-\frac{1}{2}\bm{g}_{\mu\nu}\bm{\mathrm{R}}=
-8\pi G\bm{\bm{\mathrm{T}}}_{\mu\nu}
\end{equation}
where $\bm{\mathrm{T}}_{\mu\nu}$ is the energy-momentum tensor of a perfect fluid defined by (2.1). Conservation of energy requires
\begin{equation}
\mathbf{D}^{\mu}\bm{\mathrm{G}}_{\mu\nu}\equiv \mathbf{D}^{\mu}[\bm{\mathrm{Ric}}_{\mu\nu}-\frac{1}{2}\bm{g}_{\mu\nu}\bm{\mathrm{R}}]=
-8\pi G\mathbf{D}^{\mu}\mathrm{T}_{\mu\nu}=0
\end{equation}
which is guaranteed on the lhs by the Bianchi identities.

The Einstein equations describing a static spherically symmetric star comprised of a fluid described by (2.1), and the metric (2.6) are [3,9]
\begin{align}
&\mathbf{R ic}_{rr}=\frac{Y^{\prime\prime}}{2Y}-\frac{Y^{\prime}}{4Y}\left(\frac{X^{\prime}}{X}+
\frac{Y^{\prime}}{Y}\right)-\frac{X^\prime}{r X}=-4\pi G(\rho-p)X\\&
\mathbf{Ric}_{tt}=-\frac{Y^{\prime\prime}}{2X}+\frac{Y^{\prime}}{4X}\left(\frac{X^{\prime}}{X}+
\frac{Y^{\prime}}{Y}\right)-\frac{Y^{\prime}}{r X}=-4\pi G (\rho +3p)Y\\&
\mathbf{Ric}_{\theta\theta}=-1+\frac{r}{2X}\left(-\frac{X^{\prime}}{X}+
\frac{Y^{\prime}}{Y}\right)+\frac{1}{X}=-4\pi G(\rho-p)r^{2}
\end{align}
where a prime denotes the radial derivative $d/dr$. Note $\mathbf{R}_{\theta\theta}=\mathbf{R}_{\varphi\varphi}$. Equation (2.5) for hydrostatic equilibrium then gives
\begin{equation}
\frac{Y^{\prime}(r)}{Y(r)}=-\frac{2}{p+\rho}\left(\dfrac{dp(r)}{dr}\right)
\end{equation}
The fundamental Tolman-Oppenheimer-Volkoff equation for hydrostatic equilibrium of relativistic stars can then be established.
\begin{thm}
Given the Einstein equations (2.9) to (2.11), the metric (2.6), the tensor (2.1) and the conditions (2.12), it follows that hydrostatic equilibrium of the star exists if
\begin{equation}
-r^{2}\frac{dp(r)}{dr}=G \mathlarger{\mathlarger{\mathcal{M}}}(r)\rho(r)\left(1+\frac{p(r)}{\rho(r)}\right)
\left( 1+\frac{4\pi r^{3}p(r)}{\mathlarger{\mathlarger{\mathcal{M}}}(r)}\right)\left(1-\frac{2G \mathcal{M}(r)}{r}\right)^{-1}
\end{equation}
which is the Tolman-Oppenheimer-Volkoff equation. Here $\mathcal{M}(r)$ is the mass contained in a shell of radius r. Positive pressure also produces gravitation.
\end{thm}
\begin{proof}
The Einstein equations (2.9), (2.10) and (2.11) can be written as
\begin{equation}
\frac{\mathbf{Ric}_{rr}}{2X}+\frac{\mathbf{Ric}_{\theta\theta}}{r^{2}}
+\frac{\mathbf{Ric}_{tt}}{2Y}=-\frac{X^{\prime}}{rX^{2}}-\frac{1}{r^{2}}+
\frac{1}{Xr^{2}}=-8\pi G \rho
\end{equation}
which is
\begin{equation}
\frac{d}{dr}\left(\frac{r}{X}\right)=1-8\pi G \rho^{2}r^{2}
\end{equation}
Since it is required that $|X(0)|<\infty$ at the centre of the star, the solution for $X(r)$ is
\begin{equation}
X(r)=\left(1-\frac{2G \mathcal{M}(r)}{r}\right)
\end{equation}
where $\mathcal{M}(r)$ is given by
\begin{equation}
\mathcal{M}(r)=\int_{0}^{r}4\pi\bar{r}^{2}\rho(\bar{r})d\bar{r}
\end{equation}
The total mass of the star is then $M=\mathcal{M}(R)$. One can now utilise (2.16) and (2.12) to eliminate the gravitational fields $X(r)$ and $Y(r)$ so that equation (2.11) for $\mathbf{Ric}_{\theta\theta}$ becomes
\begin{equation}
-1+\left(1-\frac{2G \mathcal{M}(r)}{r}\right)\left(1-\frac{r}{p+\rho}\left(\frac{dp(r)}{dr}\right)\right)
+\frac{G \mathcal{M}(r)}{r}-4\pi G \rho r^{2}=-4\pi G(\rho-p)r^{2}
\end{equation}
which can be re-expressed in the TOV form (2.13).
\end{proof}
\subsection{Newtonian and relativistic stars in equilibrium}
In the Newtonian limit, the last three terms on the rhs of the TOV equation reduce to unity and the TOV equation reduces to the basic hydrostatic equilibrium equation of Newtonian astrophysics. If the star is comprised of nucleons of mass $m_{N}$ with number density $n$ then the rest mass density is $\rho\sim m_{n}n$. For Newtonian stars the energy density $\mathscr{E}$ and the pressure $p$ are much less than the rest mass density so
$\mathscr{E}\ll m_{n}n$ and $p\ll m_{n}n$. Also $4\pi r^{2}\rho\ll \mathcal{M}(r)$
and the gravitational potential is everywhere small so that $2G\mathcal{M}(r)/r \ll 1$. The TOV equation then reduces to the fundamental equilibrium equation for Newtonian stars
\begin{equation}
-r^{2}\frac{d p(r)}{dr}=G \mathcal{M}(r)\rho(r)
\end{equation}
Equations (2.19) and (2.17) can be combined into a single Poisson equation by dividing (2.19) by $\rho$ and differentiating giving the 2nd-order ODE
\begin{equation}
\frac{d}{dr}\left(\frac{r^{2}}{\rho(r)}\frac{d p(r)}{dr}\right)=-4\pi Gr^{2}\rho(r)
\end{equation}
These are the key structure equations. For a finite central density $\rho(0)$ it is required that $|dp(r)/dr|_{r=0}=0$. Given an equation of state $p=p(\rho)$ and initial conditions $\rho(0)$, then $\rho$ can be derived by solving (2.20) with $\rho^{\prime}(0)=0$. At the vacuum boundary of the star $p(R)=\rho(R)=0$.

For non-relativistic Newtonian gaseous stars, not in equilibrium, the fundamental equations are the spherically symmetric Poisson-Euler system [33-35].
\begin{align}
&\frac{\partial}{\partial_{t}}\rho+v\frac{\partial}{\partial r}\rho+\rho\frac{\partial}{\partial r}v+\frac{2}{r}(\rho v)=0\\&
\rho\left(\frac{\partial}{\partial t} v+v\frac{\partial}{\partial r}v\right)+\frac{\partial}{\partial r}p=-\rho\frac{\partial}{\partial r}\Phi\\&
\bm{\Delta}\Phi=\frac{1}{r^{2}}\frac{\partial}{\partial r}\left( r^{2}\frac{\partial \Phi}{\partial r}\right)=4\pi G\rho
\end{align}
where $v$ is the fluid velocity and $\Phi$ is the Newtonian gravitational potential. The equations are closed with a  state equation of the form $p=p(\rho)$. If the star is initially contained within a region $\mathbf{B}\subset\mathbf{R}^{3}$ then $p(R,0)=\rho(R,0)=0$ on the boundary $\partial\mathbf{B}$. In static equilibrium, $v=0$ so that $\tfrac{\partial}{\partial r}p=-\rho\tfrac{\partial}{\partial r}\Phi$, then (2.23) reduces to the fundamental equilibrium condition (2.20).
\begin{rem}
For Newtonian stars in equilibrium there exists analytical solutions. For example, if $\rho=\rho_{c}=const.$ then
\begin{align}
&\mathcal{M}(r)=\frac{4\pi}{3}\rho_{c}r^{3}\\&
p(r)=\frac{1}{2}p_{c}\beta\left(1-\left(\frac{r}{R}\right)^{2}\right)
\end{align}
where $\beta=2GM/R$ and with central pressure $p_{c}=\tfrac{1}{2}p_{c}\beta=2GM/(8\pi R^{3}).$
\end{rem}
For Newtonian stars in equilibrium, the following theorem was established by Chandrasekhar [31].
\begin{thm}
If $\langle \rho(r)\rangle=\mathcal{M}(r)/[(4\pi/3)r^{3}]$ is the average density and is non-increasing as r increases radially outward with the star, then so does the function
\begin{equation}
\mathcal{F}(r)=\bigg|p(r)+\frac{3}{8\pi}\frac{(\mathcal{M}(r))^{2}}{r^{4}}\bigg|
\end{equation}
giving the inequality
\begin{equation}
p_{c}\ge p(r)+\frac{3}{8\pi}\frac{(\mathcal{M}(r))^{2}}{r^{4}}\ge \frac{3}{8\pi}\bigg(\frac{M^{2}}{R^{4}}\bigg)
\end{equation}
\end{thm}
so that $p_{c}\rightarrow\infty$ iff $R\rightarrow0$. (Proof in [31].)
\begin{rem}
In equilibrium, the supporting pressure $p(r)$ within the star can arise from several sources:
\begin{itemize}
\item Thermal pressure from gas at high temperature such that $p/\rho\sim k_{B}T/\mu m_{p}$, which arises from the burning of nuclear fuel in the core; such as the PP cycle for stars like the Sun, the CNO cycle, or the triple-alpha cycle in red giants. This is the case for all stars visible in the night sky. When the star exhausts its fuel, the thermal pressure will decrease and the star will leave an equilibrium state and will collapse, either to a new denser equilibrium state or collapse totally creating a black hole if M is large enough.[30,31]
\item Electron degeneracy pressure such that the equation of state is $p\sim |\rho|^{5/3}$ (non-relativistic) or $p\sim |\rho|^{4/3}$ (relativistic). Such stars are white dwarfs [3,12].
\item Degeneracy pressure of neutrons and repulsion due to strong interactions. Such stars are neutron stars [3,11,12]
\item Degeneracy and confinement pressure arising from a 'bag' of u,d and s quarks. Such stars are known as 'quark stars' or 'strange stars' [38].
\end{itemize}
 \end{rem}
 \begin{thm}
The fundamental equilibrium condition for a stable self-gravitating and non-relativistic bounded system is also given by the Virial Theorem
\begin{equation}
2\big\langle\mathlarger{K}\big\rangle + \big\langle\mathlarger{U}\big\rangle=0
\end{equation}
where $\big\langle K\big\rangle $ is the (mean) kinetic energy and $\big\langle U\big\rangle $ is the mean gravitational potential energy. The total energy is then $\langle E\rangle=\langle K\rangle+\langle U\rangle=\tfrac{1}{2}\big\langle U\big\rangle $ where $\langle U\rangle\sim -\tfrac{3}{5}GM^{2}/R$ and $\langle K\rangle \sim \tfrac{3}{2} N k_{B}T$ for a star comprised of N particles with (mean) temperature T, and so pressure p grows with K which grows with T. Then:
\begin{enumerate}
\item If $2\big|\big\langle\mathlarger{K}\big\rangle\big| + \big|\big\langle \mathlarger{U}\big\rangle|\sim 0 $, the star is in equilibrium.
\item If $2|\big\langle\mathlarger{K}\big\rangle\big| \gg |\big\langle\mathlarger{U}\big\rangle| $, then thermal pressure dominates and the star expands.
\item If $2|\langle\mathlarger{K}|\rangle \ll |\langle\mathlarger{U}\rangle| $, then the central pressure decreases and gravity dominates and so the star collapses.
\end{enumerate}
\end{thm}
\begin{defn}
A polytropic fluid or gas has the equation of state [3,12,30,31]
\begin{equation}
p=\kappa|\rho|^{\gamma}\equiv \kappa|\rho|^{1+\tfrac{1}{n}}
\end{equation}
where $\kappa$ is a constant depending on chemical composition and entropy per nucleon, and n is a 'polytropic index'.
\end{defn}
\begin{defn}
Any star having a polytropic gas equation of state of the form (2.29)is called a polytropic gas star or a polytrope. For a relativistic degenerate (quantum) gas comprising a cold white dwarf core for example, the equation of state has the form [3]
\begin{equation}
p=\frac{\hbar}{12\pi^{2}}\left(\frac{3\pi ^{2} \rho}{m_{n}\mu}\right)^{4/3}=\kappa|\rho|^{4/3}
\end{equation}
and so is a polytrope with index $\gamma=4/3$.
\end{defn}
\begin{defn}
Defining the polytropic variables
\begin{align}
&r=\alpha\xi=\left(\frac{\kappa\gamma}{4\pi
G(\gamma-1)}\right)^{1/2}|\rho(0)|^{\frac{\gamma-2}{2}}\xi=\alpha \xi\\&
\rho=\rho(0)|\theta|^{\frac{1}{\gamma-1}}\\&
p=\kappa\rho(0)|\theta|^{\frac{1}{\gamma-1}}
\end{align}
transforms (2.20) to the well-known dimensionless Lane-Emden equation
\begin{equation}
\frac{1}{\xi^{2}}\frac{d}{d\xi}\left(\xi^{2}\frac{d\theta(\xi)}{d\xi}\right)
+\theta^{\frac{1}{\gamma-1}}=0
\end{equation}
with boundary conditions $\theta(0)=1$ and $\theta^{\prime}(0)=0$ at the centre of the star for the LE function $\theta(\xi)$, which has physical meaning only for $\theta(\xi)>0$. For $\gamma>6/5$, $\theta(\xi_{1})=0$ at the surface of a polytrope. There are analytical solutions only for $n=0,1,5$. The case $n=0$ is for an incompressible fluid with $\rho=\rho_{c}=const.$, $p=p_{c}\theta$ so that pressure vanishes at the surface but density remains the same throughout. For $n\ge 5$, the surface is infinite. The interesting cases $n=1.5$ and $n=3$ have no analytical solutions and correspond to stars supported by interior non-relativistic gas pressure and ultra-relativistic gas pressure respectively.
\end{defn}
An important quantity is the integral
\begin{equation}
\mu_{1}(n)=-\xi_{1}^{2}\bigg|\frac{d\theta(\xi}{d\xi}\bigg|_{\xi=\xi_{1}}
=\int_{0}^{\xi_{1}}\theta^{n}\xi^{2}d\xi
\end{equation}
The radius of a polytropic gas star is
\begin{equation}
R=\alpha\xi=\left|\frac{\kappa\gamma}{4\pi
G(\gamma-1)}\right|^{1/2}|\rho(0)|^{\frac{\gamma-2}{2}}\xi_{1}\equiv\alpha \xi_{1}
\end{equation}
and the mass is estimated as
\begin{align}
M=\int_{0}^{R}4\pi r^{2}\rho(r)&=4\pi |\rho(0)|^{\frac{3\gamma-4}{2}}\left|\frac{\kappa\gamma}{4\pi
G(\gamma-1)}\right|^{3/2}\int_{0}^{\xi_{1}}\xi^{2}|\theta(\xi)|^{\frac{1}{\gamma-1}}d\xi\nonumber\\&
=4\pi |\rho(0)|^{\frac{3\gamma-4}{2}}\left|\frac{\kappa\gamma}{4\pi
G(\gamma-1)}\right|^{3/2}\xi_{1}^{2}\bigg|\frac{d\theta(\xi}{d\xi}\bigg|_{\xi=\xi_{1}}
\end{align}
By eliminating $\rho(0)$ between (2.36) and (2.37) the mass-radius relation for a stable polytropic gas star in equilibrium is
\begin{equation}
M=4\pi R^{(3\gamma-4)/(\gamma-2)}\left|\frac{\kappa\gamma}{4\pi
G(\gamma-1)}\right|^{-1/(\gamma-2)}\xi_{1}^{-(3\gamma-4)/(\gamma-2)} \left|\xi_{1}^{2}\bigg|\frac{d\theta(\xi}{d\xi}\bigg|_{\xi=\xi_{1}}\right|
\end{equation}
Values of $\xi_{1}$ and $(\xi_{1}^{2}\big|\frac{d\theta(\xi}{d\xi}\big|_{\xi=\xi_{1}})$ are known for various polytropic indices $\gamma$ or $n$. For the largest mass white dwarfs, $\gamma=4/3$ with $\xi_{1}=6.89685$ and $(-\xi_{1}^{2}\big|\frac{d\theta(\xi}{d\xi}\big|_{\xi=\xi_{1}})=2.01824$, and the upper mass limit $M_{CH}$ for stable white dwarfs, the Chandrasakhar limit, can be computed from (2.37). The gravitational potential at the surface is also small so general relativity plays no role in their structure.

For stars with $M>M_{CH}$ the final state cannot be a white dwarf but will be a neutron star [3]. And for $M\gg M_{OV}$, where $M_{OV}$ is the Oppenheimer-Volkoff upper limit for a neutron star mass. When $M\gg M_{OV}$, and the star finally exhausts its nuclear fuel then its collapse to a density singularity with $\rho\rightarrow =\infty$ is totally unavoidable, regardless of the equation of state--the star collapses to a black hole. For these cases, general relativity now plays a crucial role and any final stellar equilibrium state is impossible.

Returning to the TOV equation, one can prove a well-known and crucial result that there is a critical density or radius for any self-gravitating fluid sphere or star of mass M for which the central pressure is infinite. Hence, if it collapses or is compressed beyond a certain density, there is no physical process that can support it against its own gravity and it will collapse to a point.
\begin{thm}
(Buckdahl Thm [3]). Given an incompressible static perfect-fluid spherical star of radius R and mass M and uniform density $\rho_{o}$, then it will take an infinite pressure to keep the star in equilibrium if $M/R=4/9$. Hence:
\begin{enumerate}
\item All stable static stars in equilibrium must satisfy the sharp inequality
\begin{equation}
\frac{GM}{R}<\frac{4}{9}
\end{equation}
or when $R\le 9/8R_{s}$, if $R_{s}=2GM$ is the Schwarzchild radius.
\item The central pressure within the star is
\begin{equation}
p_{c}\equiv p(0)=\rho_{0}[1-(1-2M/R)^{1/2}][3(1-2M/R)^{1/2}-1]^{-1}
\end{equation}
\item The radius R of the static star in terms of the density and central pressure is
\begin{equation}
R=\frac{3}{8\pi \rho_{o}}\bigg(1-(\rho_{o}+p_{c})^{2}(\rho_{o}+3p_{c})^{-2}\bigg)
\end{equation}
\end{enumerate}
\end{thm}
\begin{proof}
Briefly, the TOV equation can also be written as the ODE
\begin{equation}
-r^{2}\frac{dp(r)}{dr}=G \mathcal{M}(r)\rho(r)\left(1+\frac{p(r)}{\rho(r)}\right)
\left( 1+\frac{4\pi r^{3}p(r)}{\mathcal{M}(r)}\right)\left(1-\frac{2G \mathcal{M}(r)}{r}\right)^{-1}
\end{equation}
With $\rho_{o}$ constant throughout the star the TOV equation can
be written as
\begin{equation}
-p^{\prime}(r)[(\rho+\bar{p}(r))((\rho/3)+\bar{p}(r))]^{-1}=
4\pi G r\bigg(1-\frac{8\pi G\rho r^{3}}{3}\bigg)^{-1}
\end{equation}
then integrated inwards from the surface where $p(R)=0$
\begin{equation}
\int_{0}^{p(r)}-d\bar{p}(r)[(\rho+\bar{p}(r))((\rho/3)+\bar{p}(r))]^{-1}=
\int_{0}^{R}4\pi G r\bigg(1-\frac{8\pi G\rho r^{3}}{3}\bigg)^{-1}dr
\end{equation}
This gives an expression for the pressure so that
\begin{equation}
(p(r)+\rho)(3p(r)+\rho)^{-1}= \bigg((1-8\pi\rho R^{2}/3)(1-8\pi\rho r^{2}/3)^{-1}\bigg)^{1/2}
\end{equation}
Using $\rho=3M/4\pi R^{3}$ for all $r<R$ (and with $G=1$)
\begin{equation}
p(r)=\rho_{o}\bigg( \frac{(1-2Mr^{2}/R^{3}))^{1/2}-(1-2M/R)^{1/2}}{    1-3(1-2M/R)^{1/3}-(2Mr^{2}/R^{3})^{1/2}}\bigg)
\end{equation}
so that there is a pressure gradient when the density is constant. The maximal central pressure within the star for $r=0$ is
\begin{equation}
p_{c}\equiv {p}(0)=\rho_{0}[1-(1-2M/R)^{1/2}][3(1-2M/R)^{1/2}-1]^{-1}
\end{equation}
so that
\begin{equation}
\lim_{\frac{M}{R}\uparrow 4/9}p_{c}=\lim_{\frac{M}{R}\uparrow 4/9}\rho_{0}[1-(1-2M/R)^{1/2}][3(1-2M/R)^{1/2}-1]^{-1}=\infty
\end{equation}
The radius R of the static star in terms of the density and central pressure is
\begin{equation}
R=\frac{3}{8\pi \rho_{o}}\bigg(1-(\rho_{o}+p_{c})^{2}(\rho_{o}+3p_{c})^{-2}\bigg)
=C(1-f(\rho_{c},p_{c}))
\end{equation}
\end{proof}
Therefore, all static stars beyond a given critical density corresponding to $(M/R)=4/9$ cannot be supported by any finite pressure--should a star reach such sufficiently high density compactness with $(M/R)<4/9$ it must collapse to a point of infinite density--essentially a black hole. There are problems with this description with regard to the weak/strong energy conditions, and also that the speed of sound in the star can begin to exceed the speed of light. However, the significance of this result is to demonstrate that gravity always wins out. The inequality (2.35) holds for $\underline{all}$ stars in the universe, compressible or incompressible, Newtonian and relativistic, regardless of any equation of state--the inequality is essentially forced on all self-gravitating systems in the universe by the structure of the Einstein equations.
\section{The Tolman-Oppenheimer-Synder description of a pressureless collapsing fluid star}
Consider now the dynamical problem of the total spherically symmetric collapse of a heavy spherical gaseous star comprised of a perfect fluid or "dust", which now has zero or very negligible pressure. Once its nuclear fuel has been exhausted (and ignoring possible supernova of the outer layers) there is a depletion in the counter-balancing thermal or radiation pressure from the core. Depending on the initial mass, the star collapses to a white dwarf or neutron star, or (if heavy enough) to infinite density giving an exterior black hole region. The collapse for the non-relativistic or Newtonian case is first considered then the full general-relativistic treatment is given. Both descriptions lead o the same basic nonlinear ODE, and the conclusion that the pressureless sphere collapses to a point within the same free-fall time.
\subsection{Non-relativistic case:'free-fall' gravitational collapse of a star or 'cloud' of radius R}
For the non-relativistic case, the Virial Theorem gives $ 2\langle K\rangle \ll \langle U\rangle $ so that gravity dominates.
\begin{defn}
An element of matter in a self-gravitating homogenous sphere of density $\rho_{o}$ will free fall a distance of order $R$ in a "hydrodynamic time" or "free fall "time $t_{ff}$ of the order $t_{ff}\sim \sqrt{1/G\rho_{o}}$. It is the time it takes a self-gravitating pressureless sphere to contract to a point or a radius much less than the initial radius. All parts of the spherical cloud or star will take the same length of time to collapse and the density will increase at the same rate everywhere within the cloud. This behaviour is known as homologous collapse [30,39].
\end{defn}
Heuristically, the estimate is derived from $t_{ff}=\sqrt{2R/g}$ and $r=\tfrac{1}{2}at^{2}$ for an acceleration $a$. Since $M=\tfrac{4}{3}\pi R^{3}/\rho_{o}$ then
\begin{equation}
t_{ff}=\bigg(\frac{2R^{3}}{GM}\bigg)=\bigg(\frac{3R^{3}}{2\pi G R^{3}\rho}\bigg)^{1/2}= \bigg(\frac{3}{32\pi}\bigg)^{1/2}\bigg(\frac{1}{G\rho_{o}}\bigg)^{1/2}
\end{equation}
A more rigorous description gives a first-order nonlinear ODE for the free fall of a shell of mass $\mathcal{M}(r)$. In astrophysics, this description applies to local regions of unstable gas clouds or nebulae gravitationally collapsing to form protostars if they exceed the Jeans mass [30,39,40].
\begin{thm}
For a Newtonian star or cloud of mass M and initial radius R comprised of pressureless matter (point particles) undergoing free fall with $r(t=0)=R$ and $\mathcal{M}(R)=M$, the free fall obeys the nonlinear parametric cycloid equation [39]
\begin{align}
&\frac{dr(t)}{dt}=\pm\bigg[2GM\bigg(\frac{1}{r(t)}-\frac{1}{R}\bigg)\bigg]^{1/2}\nonumber\\&
\equiv\pm\bigg[2GMR^{-1}\bigg(\frac{R}{r(t)}-1\bigg)\bigg]^{1/2}\nonumber\\&
\equiv\pm\bigg[\frac{8\pi}{3}G\rho_{o}R^{2}\bigg(\frac{R}{r(t)}-1\bigg)\bigg]^{1/2}\nonumber\\&
\equiv\pm\chi^{1/2}\bigg(\frac{R}{r(t)}-1\bigg)^{1/2}
\end{align}
The radius can normalised as $R=1$. The star collapses to $r(t)=0$ at
\begin{equation}
t_{ff}=t_{*}=\pi/2\chi^{1/2}= \left(\frac{3\pi}{32}\frac{1}{G\rho_{o}}\right)^{1/2}
\end{equation}
which agrees with (3.1).
\end{thm}
\begin{proof}
The unperturbed flow is described by the Lagrangian equation of motion for a spherically symmetric distribution of gas or fluid. A free-falling spherical shell of mass $\mathcal{M}(r)$ has the equation of motion
\begin{equation}
\frac{d^{2}r(t)}{dt^{2}}=-\frac{G \mathcal{M}(r)}{r^{2}}
\end{equation}
Let $\mathcal{V}(r(t))=\mathcal{V}(t)$ be a velocity such that $\mathcal{V}(r(t))=dr(t)/dt$. Using
\begin{align}
&\frac{d^{2}r(t)}{dt^{2}}=\frac{d}{dt}\left(\frac{dr(t)}{dt}\right)=\frac{d}{dt}\mathcal{V}(r(t))
=\frac{dr(t)}{dt}\left(\frac{d\mathcal{V}(r(t))}{dr}\right)=\mathcal{V}(r(t))\frac{d\mathcal{V}(r(t))}{dr}
=\frac{1}{2}\frac{d}{dr}|\mathcal{V}(r(t)|^{2}
\end{align}
then gives
\begin{equation}
\frac{1}{2}\left|\frac{d\mathcal{V}(r(t))}{dr}\right|^{2}=-\frac{G \mathcal{M}(r)}{r^{2}}
\end{equation}
Integrating with respect to the boundary conditions and with $\mathcal{M}(r)=M$
\begin{equation}
|\mathcal{V}(r(t))|^{2}\equiv \bigg|\frac{dr(t)}{dt}\bigg|^{2}=-2\int_{R}^{r(t)}\frac{GM d\bar{r}}{\bar{r}^{2}}
=2GM\left(\frac{1}{r(t)}-\frac{1}{R}\right)
\end{equation}
and so (3.2) follows using $M=\frac{4}{3}\pi R^{3}\rho_{o}$ and taking the negative sign. It is straightforward to integrate the ODE giving the solution with $r(0)=R=1$ as
\begin{align}
& -\chi^{1/2}t=\left[ \bar{r}(t)\left[-\left(\frac{1}{\bar{r}(t)}-1\right)^{1/2}\right]-\frac{1}{2}\left[\tan^{-1} \frac{\left(\frac{1}{\bar{r}(t)}-1\right)^{1/2}(1-2\bar{r}(t))}{2(1-\bar{r}(t))}\right]\right]_{1}^{r(t)}
\nonumber\\& =r(t)\left[-\left(\frac{1}{r(t)}-1\right)^{1/2}\right]-\frac{1}{2}\left[\tan^{-1} \frac{\left(\frac{1}{r(t)}-1\right)^{1/2}(1-2r(t))}{2(1-r(t))}\right]-\frac{\pi}{4}
\end{align}
Now $r(t)=0$ occurs at a time
\begin{equation}
-\chi^{1/2}t_{ff}=-\frac{\pi}{4}+\frac{\pi}{4}=\frac{\pi}{2}
\end{equation}
so solving for $t_{ff}$ gives (3.3) as required since $\chi=(8\pi G\rho/3)$.
\end{proof}
\begin{cor}
There are also parametric cycloid solutions to (3.2) of the form (with R=1)
\begin{equation}
r=\frac{1}{2}(1+\cos\zeta)
\end{equation}
\begin{equation}
t=\frac{1}{2}\chi^{-1/2}(\zeta+\sin\zeta)
\end{equation}
so that $r=0$ at $\cos\zeta=-1$ or $\zeta=\pi$ at time $t=t_{ff}=\tfrac{1}{2}\chi^{-1/2}\pi$, which again gives (3.3). An alternative and equivalent solution is found by setting $r=R\cos^{2}\zeta$ so that
\begin{align}
&2R\zeta(\cos\zeta\sin\zeta)=-\left(\frac{2GM}{R}\right)^{1/2}\left(\frac{1}{\cos^{2}\zeta}-1\right)^{1/2}
=-\left(\frac{\chi}{R}\right)^{1/2}\left(\frac{1}{\cos^{2}\zeta}-1\right)^{1/2}\\&
2\zeta (\cos\zeta\sin\zeta)=\left(\frac{2GM}{R^{3}}\right)^{1/2}\tan\zeta
=\left(\frac{\chi}{R^{3}}\right)^{1/2}\tan\zeta\\&
2\cos^{2}\zeta d\zeta =\left(\frac{2GM}{R^{3}}\right)^{1/2}dt=\left(\frac{\chi}{R^{3}}\right)^{1/2}dt\\&
\zeta+\frac{1}{2}\sin2\zeta=t\left(\frac{2GM}{R^{3}}\right)^{1/2}=t\left(\frac{\chi}{R^{3}}\right)^{1/2}
\end{align}
Then $r=0$ when $\zeta=\pi/2$. Equation (3.15) then gives $\pi/2=t(2GM/R^{3})^{1/2}$ and so $t_{eff}$ or (3.3) follows again upon using $M=\tfrac{4}{3}\pi R^{3}\rho$.
\end{cor}
\subsection{General relativistic case: total collapse to zero size and formation of exterior black hole region}
The gravitational collapse problem is now analysed using general relativity. Even from the basic equation $\rho=3M/4\pi R^{3}$, it is clear that the density will blow-up as $R\rightarrow 0$. Nothing can hold the star up against its own gravity if its mass is well above the Oppenheimer-Volkoff limit such that $M\gg M_{OV}$, and the Buckdahl Theorem implies that beyond a given density there is no physical pressure which can support the star up.  From (2.49), there can be no equilibrium state if $p_{c} \rightarrow 0$ so that
\begin{equation}
\lim_{{p}_{c}\rightarrow 0}R=\lim_{p_{c}\rightarrow 0}\frac{3}{8\pi \rho_{o}}\bigg(1-(\rho_{o}+p_{c})^{2}(\rho_{o}+3p_{c})^{-2}\bigg)=0
\end{equation}
If the central pressure is gradually reduced then equilibrium will be re-established at a smaller radius and the star will reduce to zero size if there is zero pressure. If it is very heavy with £$M\gg M_{TO}$ then the star will collapse to a singular or black-hole state of infinite density. This then requires the full general-relativistic description.

It is difficult to find exact analytical solutions for an Einstein-Euler-matter system describing spherically symmetric gravitational collapse. The classic papers by Tolman, Oppenheimer, Snyder and Bondi give the first rigorous and exactly solvable mathematical models of gravitational collapse, utilising general relativity [1,2,3,36,37]. They provided a spherically symmetric solution of the Einstein equations for an Einstein-Euler-matter system describing a ball of perfect fluid or 'dust' with zero pressure and uniform density. This can represent a star which has exhausted its thermonuclear fuel. Although this is idealised, it nevertheless captures the crucial features of total gravitational collapse and is exactly solvable. The conclusion of this work was ahead of its time--stars that are sufficiently heavy will collapse to zero size or to infinite density, or a singular state, with creation of an exterior 'black hole' region.

Since the particles comprising the collapsing star are acted upon by purely gravitational forces, they can form the basis of a comoving coordinate system. The derivation closely follows that given by Weinberg [3].
\begin{defn}
The comoving metric for time-dependent spherically symmetric gravitational fields is
\begin{equation}
ds^{2}=dt^{2}-\mathrm{X}(r,t)dt^{2}-\mathrm{Y}(r,t)(d\theta^{2}+\sin^{2}\theta d\varphi^{2})
\end{equation}
\end{defn}
The Ricci tensor components for this metric are
\begin{align}
&\mathbf{Ric}_{rr}=\frac{1}{X}\left[\frac{Y^{\prime\prime}}{Y}-\frac{\dot{Y}}{2Y^{2}}-
\frac{X^{\prime}Y^{\prime}}{2XY}\right]-\frac{\ddot{X}}{2X}+
\frac{(\dot{X})^{2}}{4X^{2}}-\frac{\dot{X}\dot{Y}}{2XY}\\&
\mathbf{Ric}_{\theta\theta}=-\frac{1}{Y}+\frac{1}{X}\left[\frac{Y^{\prime\prime}}{2Y}-
\frac{X^{\prime}Y^{\prime}}{4XY}\right]-\frac{\ddot{Y}}{2Y}-\frac{\dot{X}\dot{Y}}{4XY}\\&
\mathbf{Ric}_{tt}=\frac{\dot{X}}{2X}+ \frac{\ddot{Y}}{Y} -\frac{\dot{X}}{4X^{2}}
-\frac{\dot{Y}^{2}}{2Y^{2}}
\\&\mathbf{Ric}_{tr}=\frac{\dot{Y}^{\prime}}{Y}-\frac{Y^{\prime}\dot{Y}}{2Y^{2}}-\frac{\dot{X}Y^{\prime}}{2XY}=0
\end{align}
The energy-momentum tensor of a pressureless perfect fluid with proper energy density $\rho(r,t)$ and $p(r,t)=0$ is
is
\begin{equation}
\bm{\mathrm{T}}_{\mu\nu}=\rho \mathrm{U}^{\mu}\mathrm{U}^{\nu}
\end{equation}
where $\mathrm{U}^{\alpha}$ is the velocity 4-vector. For a comoving system
${\mathrm{U}}^{r}={\mathrm{U}}^{\theta}={\mathrm{U}}^{\phi}=0$ and ${\mathrm{U}}^{t}=1$. The Einstein equations coupled to a pressureless fluid are then
\begin{equation}
\bm{\mathrm{Ric}}_{\mu\nu}-\frac{1}{2}\bm{\mathrm{R}}\bm{g}_{\mu\nu}=-8\pi G\rho
{\mathrm{U}}^{\mu}{\mathrm{U}}^{\nu}
\end{equation}
with the energy conservation condition $\mathbf{D}_{\mu}\mathbf{T}^{\mu\nu}=0$. The energy
conservation condition now gives
\begin{equation}
-\partial_{t}\rho-
\rho\bigg(\frac{\dot{X}}{2X}+\frac{\dot{Y}}{Y}\bigg)=0
\end{equation}
or
\begin{equation}
\partial_{t}(\rho YX^{1/2})=0
\end{equation}
The Einstein equations can be written in "cosmological form" so that
\begin{equation}
\bm{\mathrm{Ric}}_{\mu\nu}=-8\pi G\bm{\Theta}_{\lambda}^{\lambda}=
\rho(\tfrac{1}{2}\bm{g}_{\mu\nu}+{\mathrm{U}}_{\mu}{\mathrm{U}}_{\nu})
\end{equation}
with $\bm{\Theta}_{rr}=\tfrac{1}{2}\rho X$, $\bm{\Theta}_{\theta\theta}=\tfrac{1}{2}\rho Y$,
$\bm{\Theta}_{\varphi\varphi}=\bm{\Theta}_{\theta\theta}\sin^{2}\theta$, $\bm{\Theta}_{tt}=
\tfrac{1}{2}\rho$. Then
\begin{align}
&\bm{\mathrm{Ric}}_{rr}=-4\pi G\rho X\\&
\bm{\mathrm{Ric}}_{tt}=-4\pi G\rho X\\&
\bm{\mathrm{Ric}}_{\theta\theta}=-4\pi G\rho Y
\end{align}
In full, the Einstein equations for the collapsing star are now
\begin{align}
&\frac{1}{X}\left[\frac{Y^{\prime\prime}}{Y}-\frac{\dot{Y}}{2Y^{2}}-
\frac{X^{\prime}Y^{\prime}}{2XY}\right]-\frac{\ddot{X}}{2X}+
\frac{(\dot{X})^{2}}{4X^{2}}-\frac{\dot{X}\dot{Y}}{2XY}=-4\pi G\rho\\&-\frac{1}{Y}+\frac{1}{X}\left[\frac{Y^{\prime\prime}}{2Y}-
\frac{X^{\prime}Y^{\prime}}{4XY}\right]-\frac{\ddot{Y}}{2Y}-\frac{\dot{X}\dot{Y}}{4XY}=-4\pi G\rho\\&
\frac{\dot{X}}{2X}+ \frac{\ddot{Y}}{Y} -\frac{\dot{X}}{4X^{2}}
-\frac{\dot{Y}^{2}}{2Y^{2}}=-4\pi G\rho\\&\frac{\dot{Y}^{\prime}}{Y}-\frac{Y^{\prime}\dot{Y}}{2Y^{2}}-\frac{\dot{X}Y^{\prime}}{2XY}=0
\end{align}
\begin{thm}
For the comoving metric, the system of Einstein equations can be reduced to the first-order nonlinear ODE for a 'scale factor' $R(t)$
\begin{equation}
\frac{dR(t)}{dt}=-\kappa^{1/2}(R^{-1}(t)-1)^{1/2}
\end{equation}
where $k=8\pi G/3 \rho_{o}$ and with initial condition $R(0)=1$ with respect to comoving proper time $t$.The explicit solution is given by the parametric equations of a cycloid
\begin{align}
&R(t)=\frac{1}{2}R(0)(1+\cos(\zeta(t))=\frac{1}{2}(1+\cos(\zeta(t))\\&
t=\tfrac{1}{2}(\zeta(t)+\sin(\zeta(t))k^{-1/2}
\end{align}
so that $R(t)\rightarrow 0 $ as $t\rightarrow \pi/2\kappa^{1/2}=(3/32\pi)^{1/2}(1/G\rho_{o})^{1/2}$.
\end{thm}
\begin{proof}
Setting $\rho(r,t)=-\rho$ for a homogenous matter distribution, one seeks a
separable solution of the form $X=R^{2}(t)f(r)$ and $Y=S^{2}(r)g(r)$. Substituting gives
$\ddot{R}(t)/R(t)=\ddot{S}(t)/S(t)$ so that $R(t)=S(t)$. The radial coordinate $r$
can be rescaled so that $g(r)=r^{2}$ giving
\begin{equation}
X=R^{2}(t)f(r)
\end{equation}
\begin{equation}
Y=R^{2}(t)r^{2}
\end{equation}
The first two Einstein equations (3.30) and (3.31) then reduce to
\begin{equation}
P(r)-\ddot{R}(t)R(t)-2(\dot{R}(t))^{2}=-4\pi G R^{2}(t)\rho(t)
\end{equation}
\begin{equation}
Q(r)-\ddot{R}(t)R(t)-2(\dot{R}(t))^{2}=-4\pi G R^{2}(t)\rho(t)
\end{equation}
where
\begin{align}
&P(r)=\dot{f}(r) r^{-2}f^{-2}(r)\nonumber\\&
Q(r)=-r^{-2}+r^{-1}f^{-2}(r)-\frac{1}{2}f^{\prime}(r)f^{-2}(r)
\end{align}
One can set $P(r)=Q(r)=-2k$ and the unique solution is
$f(r)=(1-kr^{2})^{-1}$. The interior metric of the star then takes the isotropic and homogenous FRW form
\begin{equation}
ds^{2}=dt^{2}-\frac{R^{2}(t)}{[(1-\kappa r^{2})}dr^{2}+r^{2}d\theta^{2}+r^{2}d\theta^{2}+
r^{2}\sin^{2}\theta d\varphi^{2}]
\end{equation}
The next step is to compute the functions $\rho(t)$ and $R(t)$. Using $X=R^{2}(t)f(r)$ and $Y=R^{2}(t)r^{2}$ in the conservation equation gives
\begin{equation}
r^{2}f(r)\partial_{t}(\rho(t)R^{3}(t))=0
\end{equation}
which is simply $\rho(t)R^{3}(t)=const.$. Normalising $R(0)=1$ then $\rho(t)=\rho(0)R^{-3}(t)$. The Einstein equations now reduce to the following ODEs in terms of a scale factor $R(t)$ so that the collapse is much like a reversed FRW Big Bang cosmology \begin{align}
&-2\kappa-\ddot{R}(t)R(t)-2\dot{R}^{2}(t)=-4\pi G\rho(0)R^{-1}(t)\\&
\ddot{R}(t)R(t)=-\frac{4\pi G}{3}\rho(0)R^{-1}(t)
\end{align}
Adding these, the general relativistic machinery finally reduces down to the first-order nonlinear ODE
\begin{equation}
\dot{R}(t)=\frac{dR(t)}{dt}=-\kappa^{1/2}(R^{-1}(t)-1)^{1/2}
\end{equation}
where the negative root is taken and $\kappa=8\pi G/3\rho(0)$. In differential form
\begin{equation}
dR(t)=-\kappa^{1/2}(R^{-1}(t)-1)^{1/2}dt
\end{equation}
The initial data are that the star is initially static so that $\dot{R}(0)=0$ and
$R(0)=1$. The formal solution statement is then
\begin{equation}
R(t)=1-\kappa^{1/2}\int_{0}^{t}(R^{-1}(s)-1)^{1/2}ds
\end{equation}
which is zero at some comoving time $t=t_{*}$ so that
\begin{equation}
R(t_{*})=1-\kappa^{1/2}\int_{0}^{t_{*}}(R^{-1}(s)-1)^{1/2}ds=0
\end{equation}
The explicit solution is given by the parametric equations of a cycloid
\begin{align}
&R(t)=\frac{1}{2}R(0)(1+\cos(\zeta(t))=\frac{1}{2}(1+\cos(\zeta(t))\\&
t=\tfrac{1}{2}(\zeta(t)+\sin(\zeta(t))\kappa^{-1/2}
\end{align}
with $\zeta(t)\in[0,2\pi]$. The star collapses to zero size or infinite energy density when $R(t)=0$ or for $\zeta(t)=\pi$, and at
\begin{equation}
t_{*}=\pi/2\chi^{1/2}=\left(\frac{3\pi}{32}\frac{1}{G\rho_{o}}\right)^{1/2}
\end{equation}
A equivalent solution with $R(0)=R$ is
\begin{align}
&2R\zeta(\cos\zeta\sin\zeta)=-\left(\frac{k}{R}\right)^{1/2}\left(\frac{1}{\cos^{2}\zeta}-1\right)^{1/2}\\&
2\zeta (\cos\zeta\sin\zeta)=\left(\frac{k}{R^{3}}\right)^{1/2}\tan\zeta\\&
2\cos^{2}\zeta d\zeta =\left(\frac{k}{R^{3}}\right)^{1/2}dt\\&
\zeta+\frac{1}{2}\sin2\zeta=t\left(\frac{k}{R^{3}}\right)^{1/2}
\end{align}
so that $R(r)=0$ at $\zeta=\pi/2$ which gives $t=t_{*}\equiv t_{eff}=\pi/2k^{1/2}$ as before.
\end{proof}
Hence a pressureless supermassive spherical star of initial density $\rho(0)$ collapses from rest to to a point within a finite proper time $t^{*}$. In other words, the collapse time to zero size from the perspective of a comoving observer falling with the stellar matter, is finite. For an exterior observer, the exterior of the star can be matched to the usual Schwarzchild metric by the Birkhoff Theorem. The differential equation (3.30) can also be integrated directly so that
\begin{equation}
\int_{R(0)=1}^{R(t)}\frac{d\bar{R}(t)}
{(\bar{R}^{-1}(t)-1)^{1/2}}=-\chi^{1/2}\int_{0}^{t}ds
\end{equation}
\begin{align}
&t=k^{-1/2}R(t)(R^{-1}(t)-1)^{1/2}+C(R(t))-k^{-1/2}R(0)(R^{-1}(0)-1)^{1/2}+C(R(0))
\end{align}
where
\begin{equation}
C(R(t))=\frac{1}{k^{-1/2}}\tan^{-1}\bigg[\frac{1}{2}(2R(t)-1)(R^{-1}(t)-1)^{1/2}(2(R(t)-1))^{-1}\bigg]
\end{equation}
The second term in (3.58) vanishes since $R(0)=1$. The term $\frac{1}{\kappa^{1/2}}R(t)(R^{-1}(t)-1)^{1/2}$ is a concave function that vanishes at $R(t)=0$ and $R(0)=1$. When the star has collapsed to zero size then $R(t)=R(t_{*})=0$ and $\tan^{-1}(\infty)=\pi/2$ so that (3.58) gives $t_{*}=\pi/2 k^{1/2}\equiv t_{ff}$ as before.
\begin{rem}
Note that the nonlinear ODE (3.46) is of the same form as (3.2) and that comoving collapse time or 'free-fall' or 'hydrodynamic' time (3.3) computed from the Newtonian model is identical to $t_{*}$ of (3.52)
\end{rem}
\begin{defn}
The Kretschmann curvature scalar $\mathbf{K}$ is defined as $\mathbf{K}=\mathbf{Ric}_{\mu\nu\gamma\delta} \mathbf{Ric}^{\mu\nu\gamma\delta}$ and a blowup in $\mathbf{K}$ is indicative of a true curvature singularity. Since the interior metric of the collapsing star is of the FRW form (3.37), the associated Kretschmann scalar
is
\begin{equation}
\mathbf{K}=\mathbf{Ric}_{\mu\nu\gamma\delta}\mathbf{Ric}^{\mu\nu\gamma\delta}=[12|R(t)|^{2}|\ddot{R}(t)|^{2}+
(k-|\dot{R}(t)|^{2})^{2}] |R(t)|^{-4}
\end{equation}
Since $\dot{R}(t)=-k^{1/2}=-k^{1/2}(R^{-1}-1))^{1/2}$ then
\begin{align}
&\ddot{R}(t)=-\tfrac{1}{2}k^{1/2}(R^{-1}(t)-1)^{1/2}(-R^{-2}(t))\dot{R}(t)\nonumber\\&=-\frac{1}{2} k^{1/2}(R^{-1}(t)-1)^{-1/2}+R^{-2}(t)
(R^{-1}(t)-1)^{1/2}=-\frac{1}{2}k R^{-2}(t)
\end{align}
then $\dot{R}(t)$ and $\ddot{R}(t)$ can be eliminated to give
\begin{equation}
\mathbf{K}=\frac{12k^{2}|R(t)|^{2}(R^{-1}(t)-1)+(k-k(R^{-1}(t)-1))^{2}}{|R(t)|^{4}}
\end{equation}
so that $\lim_{t\uparrow t_{*}}\mathbf{K}\equiv \lim_{R(t)\uparrow 0}\mathbf{K}=\infty$. This is therefore a true singularity.
\end{defn}
\begin{rem}
The star collapses to zero size or infinite proper energy density within a finite comoving proper time $t=t_{*}$. By the Birkhoff Theorem, the exterior is a black hole described by the Schwarzchild metric. Denoting the exterior coordinates as $(\overline{r},\overline{t},\overline{\theta},\overline{\varphi})$ then
\begin{equation}
Exterior:d\bar{s}^{2}=\bigg(1-\frac{2GM}{\bar{r}}\bigg)d\bar{t}^{2}-\bigg(1-\frac{2GM}{\bar{r}}\bigg)^{-1}d\bar{r}^{2}+\bar{r}^{2}d\Omega^{2}
\end{equation}
The interior metric is essentially the FRW metric
\begin{equation}
Interior:ds^{2}=dt^{2}-\frac{|R(t)|^{2}}{((1-kr^{2})}dr^{2}+r^{2}d\Omega^{2}))
\end{equation}
For consistency, these must match at the surface of the star. There is a map that matches the exterior vacuum and interior matter solutions such that $f_{1}:rR(t)\rightarrow\bar{r};f_{2}:t\rightarrow\bar{t},f_{3}:\theta\rightarrow\bar{\theta},
f_{4}:\varphi\rightarrow\varphi$. The interior solution can be transformed to a standard form
\begin{equation}
ds^{2}=B(\bar{r},\bar{t})d\bar{r}^{2}-A(\bar{r},\bar{t})d\bar{r}^{2}-\bar{r}^{2}d\Omega^{2}
\end{equation}
via an integrating factor method [3]. If $r=a$ the surface of the star then $\bar{r}=\bar{a}(t)R(t)$ and the standard metric becomes
\begin{equation}
ds^{2}=B(\bar{a},\bar{t})d\bar{a}^{2}-A(\bar{a},\bar{t})d\bar{a}^{2}-\bar{a}^{2}d\Omega^{2}
\end{equation}
where $\bar{r}=\bar{a}(t)=aR(t)$ and
\begin{align}
\bar{t}&=(1-k a^{2})^{1/2}k^{-1/2}\int_{R(t)}^{1}dR(t)(1-\chi a^{2}/R(t))^{-1}(1-R(t))^{-1/2}R(t)\\&
B(\bar{a},t)=(1-k a^{2}R^{-1}(t))\\&
A(\bar{a},t)=(1-k a^{}R^{-1}(t)^{-1}
\end{align}
The metrics match at $r=aR(t)$ if $\chi=2GM/a^{3}$ so that $M=\tfrac{4}{3}\pi\rho(0)|a|^{3}$ as expected [3].
\end{rem}
\subsection{Density function and density function singularity or blowup}
The differential equation (3.46) can be reformulated in terms of a matter density function
$u(t)$ which grows monotonically within the finite comoving proper time $[0,t_{*}]$ as the star collapses, and becomes singular at $t=t_{*}$.
\begin{prop}
Let a perfect-fluid star have normalised initial scale factor or radius $R(0)=1$ at $t=0$ and initial density $\rho(0)$, and the collapse for $t>0$ within the Oppenheimer model is
given by the parametric cycloid differential equation $\dot{R}(t)
=-k^{1/2}(R^{-1}(t)-1)^{1/2}$. Suppose for the metric (3.17) and the Einstein
equations (3.30) to (3.33) a 'density function' separable solutions are now sought such that
$X=u^{-2}(t)f(r)$ and $Y(t)=v^{-2}(t)g(r)$ with $u:\mathbf{R}^{+}\rightarrow\mathbf{R}^{+}\bigcup{\infty}$ and
$v:\mathbf{R}^{+}\rightarrow\mathbf{R}^{+}\bigcup{\infty}$. Then
\begin{equation}
u(t)=(\rho(t)/\rho(0))^{1/3}=R^{-1}(t)
\end{equation}
The total collapse of the star is now described by a nonlinear ODE so that $R(t_{*})=0$ is equivalent to the blowup in the matter density function with $u(t_{*})=\infty$ at $t=t_{*}=\pi/2 \kappa^{1/2}$. The nonlinear ODE is
\begin{equation}
\dot{u}(t)=\frac{du(t)}{dt}=k^{1/2}((u(t))^{2}(u(t)-1)^{1/2}
\equiv {\psi}(u(t))
\end{equation}
with the formal solution
\begin{equation}
u(t)=u(0)+k^{1/2}\int_{0}^{t}|(u(s))^{2}(u(s)-1)^{1/2}|ds
\end{equation}
The collapse can now be studied within this dual description in terms of a non-linear differential equation, with a singularity or blowup in the matter density function at $t=t_{*}$.
\end{prop}
\begin{proof}
A separable solution of the Einstein equations is sought with $X=u^{-2}(t)f(r)$ and
$Y=v^{-2}g(r)$. The derivatives are
\begin{align}
&\dot{X}=-2v^{-3}(t)\dot{u}(t)\dot{u}(t)f(t)\\&
\dot{Y}=-2v^{-3}(t)\dot{v}(t)\dot{v}(t)f(t)\\&
X^{\prime}=u^{-2}(t)f^{\prime}(r)\\&
Y^{\prime}=v^{-2}(t)g^{\prime}(r)\\&
\dot{Y}^{\prime}=-2v^{-3}(t)\dot{v(t)}g'(r)
\end{align}
Substituting into gives $\dot{u}(t)/u(t)=\dot{v}(t)/v(t)$ so one can set
$u(t)\equiv v(t)$. As before, the radial coordinate $r$ can be
rescaled so that now $X=u^{-2}(t)f(r)$ and
$Y=v^{- 2}(t)r^{2}$. The Einstein equations (3.26) and (3.27) are now
\begin{align}
&-r^{-1}f^{-1}(r)f^{\prime}(r)+3u^{-5}(t)(\dot{u}(t))^{3}-
u^{-4}(t)(\dot{u}(t))(\ddot{u}(t))-2u(t)^{-4}
\dot{u}(t))^{2}=-4\pi G u^{-2}\rho(t)\\&
r^{-2}+r^{-1}f^{-2}(r)-\frac{1}{2}f^{-2}(r)f^{\prime}(r)\nonumber\\&=
+3u^{-5}(t)\dot{u}(t)^{3}-u^{-4}(t)\dot{u}(t))\ddot{u}(t)-2u(t)^{-4}\
\dot{u}(t))^{2}=-4\pi G u^{-2}\rho(t)
\end{align}
The first two terms are equal and as before must equal $-2k$ so that
\begin{equation}
-r^{-1}f^{-2}f^{\prime}(r)=-r^{-2}+r^{-2}f^{-1}(r)-\frac{1}{2}r^{-1}f^{-2}(r)
f^{\prime}(r)=-2k
\end{equation}
with the unique solution $f(r)=(1-kr^{-2})^{-1}$. Now from the energy
conservation equation $\partial_{t}(\rho(t)YX^{1/2}$ and applying $X=u^{-2}(t)f(r)$ and $
Y=u^{-2}(t)r^{2})$
\begin{equation}
\partial_{t}(\rho(t)YX^{1/2})=r^{2}|f(r)|^{2}\partial_{t}(\rho(t)u^{-3}(t))=0
\end{equation}
or $\rho(t)u^{-3}(t)=const.$ Normalising so that $u(0)=1$ gives
$\rho(t)u^{-3}=\rho(0)$ or
\begin{equation}
u(t)=(\rho(t)/\rho(0)^{1/3}\equiv R^{-1}(t)
\end{equation}
which an be interpreted as a 'matter density function'. The Einstein equations become
\begin{align}
&-2k+2u^{5}(t)(\dot{u}(t))^{3}-u^{4}(t)\dot{u}(t)\dot{u}^{2}(t)-
2u^{-4}(\dot{u}(t))^{2}=-4\pi Gu\rho(0)\\&-2u^{5}(t)(\dot{u}(t))^{3}+u^{4}(t)\dot{u}(t)\ddot{u}(t)
=-\frac{4}{3}\pi Gu\rho(0)
\end{align}
Adding these
\begin{equation}
-2k-2u^{-4}(\dot{u}(t))^{2}=-\frac{16}{3}\pi G\rho(0)u(t)
\end{equation}
which is $\dot{u}(t)^{2}=-k u^{4}(u(t)-1)$ so that
\begin{equation}
\dot{u}(t)=\frac{du(t)}{dt}=\kappa^{1/2}u(t)(u(t)-1)^{1/2}\equiv {\psi}(u(t)dt
\end{equation}
or
\begin{equation}
du(t)=\kappa^{1/2}u(t)(u(t)-1)^{1/2}dt=\psi(u(t))dt
\end{equation}
\end{proof}
\begin{rem}
The nonlinear ODE (3.71) also follows directly from the nonlinear ODE (3.34). Given the density function $u(t)=R^{-1}(t)$ then
$\dot{u}(t)=du(t)/dt=-R^{-2}(t)(dR(t)/dt)$ so that $\dot{R}(t)=dR(t)/dt=-u^{-2}(t)(du(t)/dt)$. Equation (3.34) then becomes
\begin{equation}
-\frac{1}{(u(t))^{2}}\frac{du(t)}{dt}=-k^{1/2}(u(t)-1)^{1/2}
\end{equation}
or
\begin{equation}
\frac{du(t)}{dt}=k^{1/2}(u^{2}(t))(u(t)-1)^{1/2}
\end{equation}
\end{rem}
\begin{lem}
If ${\psi}(u(t))$ is globally Lipschitz then for all $t\in[0,T]\subset\mathbf{R}^{+}$ then $\exists$ $K>0$ such that
\begin{align}
&|{\psi}(u(t))-{\psi}(\widehat{u}(t)|=
|k^{1/2}u^{4}(t)(u(t)-1)^{1/2}-k^{1/2}|\widehat{u}(t)|^{4}(\widehat{u}(t)-1)^{1/2}\\\nonumber&
|\le K|u(t)-\widehat{u}(t)|
\end{align}
Then for all $u(0)\in[1,\infty)$ and $\widehat{u}(0)\in[1,\infty)$ there is a unique solution that is either global or local of the form
\begin{equation}
u(t)=u(0)+\int_{t_{\epsilon}}^{t}\psi(u(s))ds=
k^{1/2}\int_{t_{\epsilon}}^{t}u^{4}(s)(u(s)-1)^{1/2}ds
\end{equation}
\begin{proof}
Given the solutions
\begin{align}
&u(t)=u(0)+\int_{t_{\epsilon}}^{t}\psi(u(s))ds\\&
\widehat{u}(t)
=\widehat{u}(0)+\int_{t_{\epsilon}}^{t}\psi(\widehat{u}(s))ds
\end{align}
\begin{align}
&|u(t)-\widehat{u}(t)| \le |u(0)-\widehat{u}_{2}(0)|+\int_{t_{\epsilon}}^{t}
[\psi(u(t)-{\psi}(\bm{u}(t))]\nonumber \\&\le
|u(0)-\widehat{u}(0)|+\int_{t_{\epsilon}}^{t}K|u(s)-\widehat{u}(s)|ds
\end{align}
The Gronwall Lemma (Appendix C) then gives
\begin{equation}
|u(t)-\widehat{u}(t)|\le |u(0)-\widehat{u}(0)|\exp(K|t-t_{\epsilon}|)
\end{equation}
Hence, if $u(0)=\widehat{u}(0)$ it follows that $u(t)=\widehat{u}(t)$
for all t and so the solution is unique.
\end{proof}
\end{lem}
Note that if the solution is unique then (3.95) is always satisfied even if $u(t)$ becomes
singular at $t=t_{*}$ since $\lim{t\uparrow t_{*}}|u(t)-\widehat{u}(t)|=0$
\begin{cor}
The explicit solution of (3.89) can again be expressed as  a parametric cycloid
\begin{align}
&u(t)=2u(0)[1+\cos\zeta]^{-1}= 2[1+\cos(\zeta)]^{-1}\equiv R^{-1}(t)\nonumber\\&
t=\frac{1}{2\chi^{1/2}}(\chi+\sin(\zeta))
\end{align}
since $u(0)=R^{-1}(0)=1$ and a blowup or density singularity again occurs when $\cos(\zeta(t)=-1$ or $\zeta(t)=\pi$. Since the solution blows
up at $t=t_{*}$, the solution is not global.
\end{cor}
As a consistency check, the ODE (3.89) is also easily integrated so that
\begin{equation}
\int_{u(0)}^{u(t)}\frac{d\bar{u}(t)}{(\bar{u}(t))^{-2}(\bar{u}(t)-1)^{1/2}}=
\kappa^{1/2}\int_{0}^{t}ds
\end{equation}
which gives
\begin{eqnarray}
\kappa^{1/2}t=(u(t)-1))^{1/2}(u(t))^{-1}+\tan^{-1}(u(t)-1))^{1/2})\\\nonumber
-(u(0)-1))^{1/2}(u(0))^{-1}+\tan^{-1)(u(0)-1}(u(0)-1))^{1/2}
\end{eqnarray}
Since $u(0)=1$ and $\tan^{-1}(0)=0$ then
\begin{equation}
\chi^{1/2}t=(u(t)-1))^{1/2}(u(t))^{-1}+\tan^{-1}(u(t)-1))^{1/2}
\end{equation}
When $u(t)=\infty$ then $\tan^{-1}(\infty)=\pi/2$ and $(u(t)-1))^{1/2}u(t))^{-1}=0$
so that this blowup still occurs when $t=t_{*}=\pi/2\kappa^{1/2}$ as required. In terms of the density $\rho(t)$
\begin{equation}
\kappa^{1/2}t=((\rho(t)/\rho_{o})^{1/3}-1))^{1/2}
(\rho(t)/\rho_{o})^{-1/3}+\tan^{-1}((\rho(t)/\rho_{o})^{1/3}
-1))^{1/2}
\end{equation}
so that again the stellar matter density blowup $\rho(t)=\infty$ occurs when $t=t_{*}=\pi/2\kappa^{1/2}$. The global minimum is at $u(t)=u(0)=1$, so the density function grows from its minimum value to infinity.
\begin{prop}
Since $R^{-1}(t)\equiv u(t)$, the Kretschmann scalar can be expressed as
\begin{equation}
\mathbf{K}=[12k^{2}|u(t)|(u(t)-1)+(k-k(u(t)-1))^{2}]|u(t)|^{4}
\end{equation}
It is tedious but straightforward to reduce $\mathbf{K}$ to a polynomial in $u(t)$ such that
\begin{align}
&\mathbf{K}=k^{2}[|u(t)|^{6}-4|u(t)|^{5}+|u(t)|^{4}+12|u(t)|^{3}-12|u(t)|^{2}+3]\nonumber\\&
\equiv[\alpha|u(t)|^{6}+\beta|u(t)|^{5}+\gamma|u(t)|^{4}+\delta|u(t)|^{3}+\epsilon|u(t)|^{2}+3]
\end{align}
where $(\alpha,\beta,\gamma,\delta,\epsilon)$ are constants. This then diverges or blows up as $u(t)\rightarrow\infty$, again indicating a singular gravitational collapse to infinite energy density.
\end{prop}
There are various definitions of a 'black hole' but for this paper, it will be defined as follows.
\begin{defn}
A singular black hole is a region formed from the total gravitational collapse of a spherically symmetric fluid sphere(star) if the following hold:
\begin{enumerate}
\item The density function $u(t)=(\rho(t)/\rho_{o})^{1/3}$ becomes singular at a finite comoving time $t=t_{*}=\pi/2\kappa^{1/2}$ such that
\begin{equation}
u(t_{*})=(\rho(t_{*})/\rho_{0})^{1/3}=R^{-1}(t_{*})=\infty
\end{equation}
\item The exterior metric of the collapsing sphere in vacuum (in standard coordinates) is the Schwarzchild metric with horizon radius $2GM$. The horizon has the standard definition of being the casual past of future null infinity [4,9,10]
\end{enumerate}
\end{defn}
\subsection{Motivation and brief outline}
To extend the Tolman-Oppenheimer-type gravitational collapse model via the nonlinear ODE (3.71) for the density function $u(t)$, the basic strategy will be to tentatively incorporate stochasticity or 'stochastic control'; that is, randomness or 'white-noise' random perturbations at small scales close to the comoving blowup time, and then consider viable forms of nonlinear stochastic differential equations. The theory of SDEs (and SPDEs) is essentially a formalism for describing dynamical systems that incorporate both non-random and random properties. Essentially, multiplicative random perturbations with specific properties are "switched on" very close to the (comoving) blowup time so that the collapsing fluid sphere is now interpreted as a 'controlled' stochastic nonlinear dynamical system. This may or may not affect formation of density blowup or singularities. However, is well established within stochastic analysis that there are circumstances in which noise or random perturbations can suppress blowups or singularities in occurring in deterministic dynamical systems.
\begin{exam}
(Self-gravitating Brownian motion). Although, the tentative stochastic extension is for total gravitational collapse within general relativity that is solvable, it can also be motivated by considerations of stochastic N-body problems and self-gravitating Brownian motions within Newtonian gravitation [41,42]. For example, N classical particles interacting gravitationally with coordinates $\bm{r}_{1}(t),...,\bm{r}_{N}(t)$ and velocities $\bm{v}_{1}(t),...,\bm{v}_{N}(t)$ are described by N coupled equations
\begin{align}
&\frac{d\bm{r}_{i}(t)}{dt}=\bm{v}_{i}(t)\\&
\frac{d\bm{v}_{i}(t)}{dt}=-Gm\sum_{i\ne j}\frac{\bm{r}_{i}(t)-\bm{r}_{j}(t)}{|\bm{r}_{i}(t)-\bm{r}_{j}(t)|^{3}}\equiv
-m\bm{\nabla} U(\bm{r}_{1}...\bm{r}_{n})
\end{align}
where $ U(\bm{r}_{1}(t)...\bm{r}_{n}(t))=\sum_{i<j}\Phi(\bm{r}_{i}(t)-\bm{r}_{j}(t))$
and where $ \Phi(\bm{r}_{i}(t)-\bm{r}_{j}(t))=-\frac{G}{|\bm{r}_{i}(t)-\bm{r}_{j}(t)|}$ is the Newtonian potential between pairs of particles. If the evolution of the system is considered over some interval $\mathbf{R}=[0,T]$ from some initial data, then suppose stochastic interactions or noise is switched on at some $t$ such that $\bm{\mathbf{R}}=[0,t']\bigcup[t',T]\equiv\mathbf{X}_{II}\bigcup\mathbf{X}_{III}$. Then for $t>t'$ or $t\in\mathbf{X}_{II}$ the self-gravitating system is now described by an n-body stochastic system, essentially a system of self-gravitating Brownian particles such that $ \frac{d\overline{r_{i}(t)}}{dt}=\bm{v}_{i}(t)$ as before but now
\begin{equation}
\frac{d\bm{v}_{i}(t)}{dt}=-m\bm{\nabla} U(\bm{r}_{1}...\bm{r}_{n})-\xi\bm{v}_{i}(t)+(2D)^{1/2}
\mathlarger{\mathlarger{\mathscr{W}}}_{i}(t)
\end{equation}
where $\xi v_{i}(t)$ is a "friction force", D is a diffusion coefficient and
$\mathlarger{\mathlarger{\mathscr{W}}}_{i}(t)$ is a Gaussian white noise perturbation with $\mathlarger{\mathlarger{\bm{\mathcal{E}}}}\llbracket \mathlarger{\mathlarger{\mathscr{W}}}_{i}(t)\rrbracket =0$ and $\mathlarger{\mathlarger{\bm{\mathcal{E}}}}\llbracket \mathlarger{\mathlarger{\mathscr{W}}}_{i}(t)
\mathlarger{\mathlarger{\mathscr{W}}}_{j}(s)\rrbracket=\delta_{ij}\delta(t-s)$. The inverse temperature is $\beta=1/k_{B}T$ and $\xi=D\beta m$ is the Einstein relation. For $\xi=D=0$, the system reduces to the original Newtonian-Hamiltonian system with Hamiltonian
\begin{equation}
\mathbf{H}=\sum_{i=1}^{N}\frac{1}{2}|m\bm{v}_{i}(t)|^{2}+m^{2}U(\bm{r}_{1}...\bm{r}_{n})
\end{equation}
An Ito interpretation is then
\begin{align}
&d\overline{v_{i}(t)}=-[m\bm{\nabla} U(\bm{r}_{1}(t)...\bm{r}_{n}(t))-\xi\bm{v}_{i}(t)]dt+(2D)^{1/2}
d\mathlarger{\mathlarger{\mathscr{B}}}_{i}(t)\nonumber\\&\equiv F[r_{i}(t)]dt+(2D)^{1/2}
d\mathlarger{\mathlarger{\mathscr{B}}}_{i}(t)
\end{align}
\end{exam}
One can formulate analogous SDEs for (3.89) by also switching on stochastic mixing or noise. This is an interesting strategy also suggested [43] for problems in gravitation and cosmology, which was tentatively explored at a heuristic level but not rigorously. Here, we wish to analyse in detail the growth, finiteness, upper bounds or possible blowups of stochastic integral solutions describing a "density diffusion". In particular, whether the blowup or singularity of the original formalism is now smoothed, suppressed or dissipated by the noise. The general aims of the paper are as follows:
\begin{enumerate}
\item Analysing the effect of 'switching on' suitable stochastic fields or control just prior to the blowup time $t_{*}$ and then constructing viable nonlinear stochastic SDEs from the nonlinear ODE for a density function diffusion $\widehat{u}(t)$ just before the comoving time at which the blowup occurs. This is studied
within both the Ito and Stratanovich interpretations.
\item Provide proofs that the density function blowup or singularity is smoothed out or 'dissipated' within the Ito interpretation for the formal density function diffusion solution of the SDE and for all comoving proper times; in particular, establishing that the solution is a bounded martingale whose supremum remains finite and that the blowup probability is always zero. To derive estimates for the high-order moments of the expectations of $|\overline{\widehat{u}(t)}|^{p}$ and show that these are always finite and bounded.
\item Develop and apply other rigorous blowup no criteria such as Lyaponov functions
and a Feller Test for blowup, showing that the blow-up probability is always effectively zero.
\item Prove 'null recurrence' within the alternative Stratanovich interpretation by solving the SDE: that the density function blowup still occurs for the SDE but that the expected comoving proper time for this to occur is now infinite.
\end{enumerate}
\section{Stochastic control and extension of the Tolman-Oppenheimer description}
In this section, a stochastic extension or control of the the Tolman-Oppenheimer description in terms of the nonlinear ODE of the form (3.89) is proposed. For a comoving observer, the ball of matter or star collapses to zero size or infinite density on the comoving collapse interval $[0,t_{*}]=[0,\pi/2\kappa^{1/2}]$ so that $R(t_{*})=0$ or $u(t_{*})=\infty$. The stochastic control consists of 'switching on' multiplicative stochastic 'mixing' or perturbations with specific properties, very (or infinitesimally) close to the comoving proper time $t=t_{*}$ at which the density singularity or blowup occurs. Within this stochastically controlled extension, it is possible for the singularity to be 'smoothed out' and for solutions to be globally regular for all t>0.

First, the strategy is outlined utilising simple examples of a nonlinear autonomous ODE which have a blowup or singularity.
\subsection{Noise suppression of blowup for singular nonlinear ODEs:examples}
Under specific conditions, noise or random perturbations can suppress blowups or singularities, and can also stabilize a dynamical system that was unstable
in the absence of the random perturbations [44-53,55]. Let $u(t)$ be a function such that $u:[0,T]\rightarrow\mathbf{R}^{+}\cup\infty =[0,\infty]$ and let $\psi(u(t))$
be a nonlinear functional of $u(t)$ such that $\psi:\mathbf{R}^{+}\times\mathbf{R}^{=}\rightarrow\mathbf{R}^{+}$. it satisfies the nonlinear autonomous ODE
\begin{equation}
\frac{du(t)}{dt}= \psi(u(t))
\end{equation}
with initial data $u(0)=u_{0}$, and $\psi(u(t))$ satisfying
some standard uniqueness and existence conditions. The formal (unique) solution is
\begin{equation}
u(t)=u(0)+\int_{0}^{t}\psi(u(s))ds
\end{equation}
and $u(t)$ generally grows monotonically with t. If the solution has a singularity or blowup at some $t=t_{*}$ then $|u(t_{*})|=\infty$ or $ u(t_{*})=u(0)+\int_{0}^{t_{*}}\psi(u(s))ds=\infty $. Hence, there is no global solution and the solution only exists on the subinterval $[0,t_{*}]$. However, random perturbations or white noise can suppress this blowup and give a singularity free global solution for all $t\in[0,\infty)$.
\begin{exam}
As an example, consider a logistic ODE
\begin{equation}
du(t)=\psi(u(t)dt=u(t)[\beta+\alpha u(t)]
\end{equation}
for all $t\ge 0$ and initial data $u(0)=u_{O}$ with $(\alpha,\beta)$. This ODE occurs in population dynamics [44,55] whereby $u(t)$ describes a biological population at time $t$. For $\alpha,\beta>0$, the solution is
\begin{equation}
u(t)={\beta}[(-\alpha+\exp(-\beta t)(\beta+\alpha u_{0})/u_{o})]^{-1}
\end{equation}
so that a blowup or singularity $u(t_{*})=\infty$ occurs at
\begin{equation}
t_{*}=-\frac{1}{\beta}\log\bigg(\frac{\alpha u_{o}}{\beta+\alpha u_{o}}\bigg)
\end{equation}
Let the ODE be perturbed by a white noise $\mathlarger{\mathlarger{\mathscr{W}}}(t)$ such that
\begin{equation}
d\overline{u(t)}=\psi(u(t)dt+\xi d\mathlarger{\mathlarger{\mathscr{B}}}(t)=u(t)(\beta+\alpha u(t))+\xi d\mathlarger{\mathlarger{\mathscr{B}}}(t)
\end{equation}
where $d\mathlarger{\mathlarger{\mathscr{B}}}(t)(t)=\mathlarger{\mathlarger{\mathscr{W}}}(t)dt$ is the standard Brownian or Weiner process.(Appendix A). Then the solution will not blowup or become singular for any finite time, for any finite $\xi>0$. The deterministic singularity is therefore 'suppressed' by the noise. It could also be shown that $\mathlarger{\mathrm{I\!P}}(|\overline{u(t)}|=\infty)=0$ so that there is stability in probability, and that the moments are finite such that
$\mathlarger{\mathlarger{\mathcal{E}}}\bigg\llbracket|\overline{u(t)}|^{p}
\bigg\rrbracket<\infty$.[].
\end{exam}
\begin{exam}
A second example is the ODE $du(t)=\alpha \psi(u(t)dt=\alpha|u(t)|^{\beta}$ for $\beta\ge 1$. For $\beta=2$ this is a Riccati equation and the solution is not globally regular, having a solution $u(t)=|u^{-1}_{o}-\alpha t|^{-1}$ with a blowup. However, the diffusion
\begin{equation}
d\overline{u(t)}=\alpha |u(t)|^{\beta}d\mathlarger{\mathlarger{\mathscr{B}}}(t)=
\alpha |u(t)|^{\beta}\mathlarger{\mathscr{W}}(t)dt
\end{equation}
does not blowup for any $\beta>0$. Equations of this form occur in 'financial engineering' as 'variance of elasticity' models [54]. For the singular Riccati equation with $\beta=2$, the following randomly perturbed equation or SDE
\begin{equation}
d\overline{u(t)}=\alpha |u(t)|^{\beta}dt+|u(t)|^{q}d\mathlarger{\mathlarger{\mathscr{B}}}(t)
\end{equation}
is non-exploding or nonsingular if $(3-2q)<0$ and explodes if $(3-2q)>0$. (See [61].)
\end{exam}
In this paper, we will focus on multiplicative white-noise perturbations of the singular nonlinear ODE (3.71) for the density function. It is well known within the theory of SDEs that pure diffusions without drift of the generic form
\begin{equation}
d\overline{u(t)}=\alpha \psi(u(t))d\mathlarger{\mathlarger{\mathscr{B}}}(t)
\end{equation}
never explode or become singular within any finite time, irrespective of the  growth or continuity properties of the coefficient and even if the underlying deterministic ODE $du(t)=\psi(u(t)dt$ does explode.
\begin{prop}
Given a generic nonlinear ODE $du(t)=\psi(u(t))dt$ with blowup at $t=t_{*}$, suppose now the semi-infinite real is partitioned as
\begin{equation}
\widehat{\mathrm{I\!R}}^{+}=[0,\infty)\bigcup\infty ={\mathbf{X}}_{I}\bigcup{\mathrm{X}}_{II}\bigcup{\mathlarger{\mathrm{X}}}_{III}
=[0,t_{\epsilon}]\bigcup (t_{\epsilon},t_{*}]\bigcup[t_{*},\infty]\nonumber
\end{equation}
where $t_{*}=t_{\epsilon}+|\epsilon|$ with $|\epsilon|\ll 1$, and where $\epsilon|$ can be arbitrarily small. Then $[0,t_{*}]={\mathbf{X}}_{I}\bigcup{\mathbf{X}}_{II}$. At time $t=t_{\epsilon}$ stochastic noise or control $\mathlarger{\mathlarger{\mathscr{W}}}(t)$ is "switched on" for all $t>t_{\epsilon}$. Here $\mathlarger{\mathlarger{\mathscr{W}}}(t)$ is a Gaussian white noise or Wiener process such that $\mathlarger{\mathlarger{\mathcal{E}}}\big\llbracket\mathlarger{\mathlarger{\mathscr{W}}}\big\rrbracket=0$ and $\mathlarger{\mathlarger{\mathcal{E}}}\big\llbracket(\mathlarger{\mathlarger{\mathscr{W}}}(t)
\mathlarger{\mathlarger{\mathscr{W}}}(s)\big\rrbracket
=\beta\delta(t-s)$. The 'stochastically controlled' system, then consists of the 'hybrid' ODE-SDE system
\begin{align}
&\frac{d\overline{u(t)}}{dt}=\mathlarger{\mathcal{C}}[\mathbf{X}_{I}]\psi(u(t))
+\mathlarger{\mathcal{C}}[\mathlarger{\mathbf{X}}_{II}\cup{\mathbf{X}}_{III}] \psi(u(t))\mathlarger{\mathlarger{\mathscr{W}}}(t)\nonumber\\&
\equiv\mathlarger{\mathcal{C}}[\mathlarger{\mathbf{X}}_{I}]\frac{du(t)}{dt}
+\mathlarger{\mathcal{C}}[\mathlarger{\mathbf{X}}_{II}\cup\mathlarger{\mathbf{X}}_{III}] \psi(u(t))\mathlarger{\mathlarger{\mathscr{W}}}(t)
\end{align}
where $\mathlarger{\mathcal{C}}$ is a controlling "indicator function" such that
\begin{align}
&\mathcal{C}(\mathbf{X}_{I})=1~for~t<t_{\epsilon}~or~ t\in\mathbf{X}_{I};~~
\mathlarger{\mathcal{C}}(\mathbf{X}_{I})=0~for~t\ge t_{\epsilon} or~t\notin\mathbf{X}_{I}\nonumber\\&
\mathlarger{\mathcal{C}}(\mathbf{X}_{II}\cup\mathbf{X}_{III})=1~for~t\ge t_{\epsilon}~or~ t\in\mathbf{X}_{II}\cup\mathbf{X}_{III}\nonumber\\&
\mathlarger{\mathcal{C}}(\mathbf{X}_{II}\cup\mathbf{X}_{III})=0~for~t< t_{\epsilon} or~t\notin\mathbf{X}_{I}\cup\mathbf{X}_{III}\nonumber
\end{align}
Hence,$\mathcal{C}(\mathbf{X}_{I})=1$ when $\mathlarger{\mathcal{C}}(\mathbf{X}_{II}\cup\mathbf{X}_{III})=0$ and vice versa. The solution for all $t\ge t_{\epsilon}$  is then
\begin{align}
&\overline{u(t)}=\underbrace{\mathcal{C}[\mathbf{X}_{I}]\int_{0}^{t_{\epsilon}}\psi(u(t))dt}_{deterministic}+
\underbrace{\mathlarger{\mathcal{C}}[\mathbf{X}_{II}\cup{\mathbf{X}}_{III}]]
\int_{t_{\epsilon}}^{t}\psi(u(s)) d\mathlarger{\mathlarger{\mathscr{B}}}(t)}_{stochastic}\nonumber\\&=u_{\epsilon}+\mathlarger{\mathcal{C}}[{\mathbf{X}}_{II}\cup
{\mathbf{X}}_{III}]\int_{t_{\epsilon}}^{t}\psi(u(s))d\mathlarger{\mathlarger{\mathscr{B}}}(s)\equiv u_{\epsilon}+\int_{t_{\epsilon}}^{t}\psi(u(s))d\mathlarger{\mathlarger{\mathscr{B}}}(s)
\end{align}
where $d\mathlarger{\mathlarger{\mathscr{B}}}(t)=\mathlarger{\mathlarger{\mathscr{W}}}(t)dt$ is again a standard Brownian process.
\end{prop}
The integral is taken in the Ito sense--and will later be rigorously shown to be a martingale) which is now free of any singularity. The stochastically averaged equations recover the original ODE so that
\begin{align}
&\mathlarger{\mathlarger{\bm{\mathcal{E}}}}\bigg\llbracket(\frac{d}{dt}\overline{u(t)})\bigg\rrbracket=\frac{du(t)}{dt}=\alpha \mathlarger{\mathcal{C}}(\mathsf{X}_{I})\psi(u(t))dt+
\mathlarger{\mathcal{C}}[{\mathbf{X}}_{II}\cup{\mathbf{X}}(t)]
\mathlarger{\mathlarger{\mathcal{E}}}\bigg\llbracket\mathlarger{\mathlarger{\mathscr{W}}}(t)\bigg\rrbracket= \psi(u(t))
\end{align}
since $\mathlarger{\mathlarger{\mathcal{E}}}\big\llbracket\mathlarger{\mathlarger{\mathscr{W}}}(t)\big\rrbracket=0$
\begin{prop}
We will say that the solution (4.11) is non-singular or has global regularity on ${\mathbf{X}}_{I}\cup{\mathbf{X}}_{II}
\cup{\mathbf{X}}_{III}$ if at least the following apply.
\begin{enumerate}
\item Estimates for the expectation is bounded for all $p\ge 1$ so that for any finite
$T\in{\mathbf{X}}_{II}\cup{\mathbf{X}}_{III}$
there is a $C>0$ such that the stochastic average is
\begin{equation}
\mathlarger{\mathlarger{\mathcal{E}}}\big\llbracket\sup_{t_{\epsilon}\le T}|\overline{u(t)}|^{p}\big\rrbracket < C
\end{equation}
\item The probability of a blowup or singularity in any finite time $t$ is zero so that for $p\ge 1$
\begin{equation}
\mathlarger{\mathrm{I\!P}}\big\llbracket \sup_{t_{\epsilon}\le T}|\overline{u(t)}|^{p}=\infty\big\rrbracket=0
\end{equation}
\end{enumerate}
\end{prop}
It would be efficacious to establish that the process $\overline{u(t)}$ is actually a true martingale since it will satisfy the above criteria on at least $t\in{\mathbf{X}}_{II}$ or globally for all $t\ge t_{\epsilon}$ since martingales have the following desirable and useful properties [56-61
\begin{itemize}
\item Even under mild conditions and general initial data, the suprema of martingales are bounded over finite or semi-infinite intervals, making them useful and powerful tools for studying the growth, bounds or explosions of stochastic processes and diffusions. Powerful and well-established theorems then apply.
\item They are closely related to Markov and Wiener processes.
\item They are incorporated within a well-defined and rigorous Ito calculas.
\item Bounds on martingales are of two types: bounds on the probabilities that the martingale $\overline{u(t)}$, or its absolute value $|\overline{u(t)}|$, exceeds some given value $C$ for $t$ within some finite or semi-infinite interval; and bounds or estimates on the stochastic expectations of moments $\mathlarger{\mathlarger{\mathcal{E}}}\llbracket |\overline{u(t)}|^{p}\rrbracket $ for $p\ge 1$ and $p\in\mathbf{Z}$ such that $\mathlarger{\mathlarger{\mathcal{E}}}\llbracket |\overline{u(t)}|^{p}\rrbracket <\infty$ for all $t>t_{\epsilon}$
\item The expectation is finite so that $\mathlarger{\mathlarger{\mathcal{E}}}\llbracket\overline{u(t)}\rrbracket=u_{\epsilon}$ on the entire real line or at least over the interval where the underlying ODE had a blowup.
\item The probability of a blowup or singularity is zero so that $\mathlarger{\mathrm{I\!P}}[\overline{u(t)}=\infty]=0$.
\item The martingale or diffusion process $\overline{u(t)}$ has a linear generator
\begin{equation}
\mathlarger{\mathrm{I\!H}}=\frac{1}{2}(\psi(u(t))^{2}\mathlarger{\mathrm{D}}_{u}
\mathlarger{\mathrm{D}}_{u}
\end{equation}
where $\mathrm{D}_{u}=d/du(t)$.
\end{itemize}
\subsection{Stochastic control of the Tolman-Oppenheimer collapse}
This extension and partition is now applied to the nonlinear singular ODE for the density function $u(t)=(\rho(t)/\rho(0))^{1/3}$, on the comoving collapse interval $[0,t_{*}=\pi/2k^{1/2}]$.
\begin{prop}
Let the finite comoving collapse interval for the collapsing ball of matter or spherical fluid star be $\mathbf{C}=[0,t_{*}]\subset\mathbf{R}^{+}=[0,\infty)$.
If $t_{\epsilon}$ is a comoving proper time infinitesimally close to $t_{*}$ then $t_{*}=t_{\epsilon}+\epsilon$ with $|\epsilon|\ll 1$. Here $\epsilon$ is very
close to zero and could be fixed such such that $\epsilon\in[0,\tfrac{1}{n}]$ with $n\ge 10000$ for example. Let $\mathbf{X}_{II}=[t_{\epsilon},t_{*}]$ then partition ${\mathbf{C}}$ as ${\mathbf{X}}=[0,t_{\epsilon})\bigcup[t_{\epsilon},t_{*}]$. The complete half-real line${\mathbf{X}}^{+}=[0,\infty)$ can be partitioned into the subintervals or domains ${\mathbf{X}}_{I}\subset {\mathbf{X}}^{+},\mathbf{X}_{II}
\subset\mathbf{X}^{+},{\mathbf{X}}_{III}^{\infty}\subset\mathbf{X}^{+}$ such that
\begin{align}
&\widehat{\mathbf{R}}^{+}=\mathbf{R}^{+}\bigcup{\infty}={\mathbf{X}}_{I}\bigcup\bigg({\mathbf{X}}_{II}
\bigcup{\mathbf{X}}_{III}\bigg)=[0,t_{\epsilon}]\bigcup\bigg((t_{\epsilon},t_{*}]\bigcup(t_{*},
\infty]\bigg)\nonumber\\&\equiv [0,t_{\epsilon}]\bigcup(t_{\epsilon},\infty]=[0,t)
\end{align}
When specific stochastic perturbations are introduced at $t=t_{\epsilon}$, the blowup may still occur at either $t_{*}$, some random proper time $t\in
[t_{\epsilon},\infty)=\mathbf{X}_{II} \bigcup\mathbf{X}_{III}$ or may never occur. Using $t_{*}=t_{\epsilon}+\epsilon$, then (4.20) can be written as
\begin{align}
&\widehat{\mathbf{R}}^{+}=\mathbf{X}_{I}\bigcup\bigg(\mathbf{X}_{II}\bigcup\mathbf{X}_{III}\bigg)
=\underbrace{[0,t_{\epsilon})}_{deterministic~regime}\bigcup\underbrace{\bigg([t_{\epsilon},t_{\epsilon}+\epsilon]\bigcup [t_{\epsilon}+\epsilon,~\infty)\bigg)}_{stochastic~regime}\equiv[0,t_{\epsilon}]\bigcup[t_{\epsilon},\infty)=[0,\infty)
\end{align}
\end{prop}
It is then possible to 'engineer' a controlled nonlinear hybrid ODE-SDE on this partition, which constitutes a stochastic extension of the original (deterministic)collapse description. This is made more precise in the following proposition.
\begin{prop}
The following hold:
\begin{enumerate}
\item The fluid star collapses from $R(0)=1$ to zero size $R(t_{*})=0$ within the finite comoving proper time $t_{*}=\pi/2 k^{1/2}$. This is equivalent to a blowup in
 the density function $u(t)=(\rho(t)/\rho(0))^{1/3}$ such that $u(t_{*})=\infty$ from initial data $u(0)=R^{-1}(0)=1$.
\item Let $t_{\epsilon}$ be a proper time infinitesimally close to the blowup proper time $t_{*}$ such that $t_{*}=t_{\epsilon}+\epsilon$ with $|\epsilon|\ll 1$.
Define a partition such that
\begin{align}
&\widehat{\mathbf{R}}^{+}=\mathbf{X}^{+}\bigcup(\mathbf{X}^{+}\bigcup\mathbf{X}_{III})=[0,t_{\epsilon}]
\bigcup([t_{\epsilon},t_{*}]\bigcup[t_{*},\infty]
\equiv[0,t_{\epsilon}\bigcup([t_{\epsilon},t_{\epsilon}+\epsilon]\bigcup [t_{\epsilon}+\epsilon,~\infty]
\end{align}
where the usual collapse comoving proper-time interval is $\mathbf{C}=[0,t_{*}]
\subset\mathbf{R}^{+}$
\item For $t\in\mathbf{X}_{II}\cup \mathbf{X}_{III}$ and given a probability triplet $(\Omega,\mathfrak{F},\mathlarger{\mathrm{I\!P}})$ define a stochastic process $\overline{u}:\mathbf{X}_{II}\cup\mathbf{X}_{III}\times\Omega \rightarrow [u_{\epsilon},\infty)$. Or for all $t\in\mathbf{X}_{II}\cup\mathbf{X}_{III}$ 'switch on' a Gaussian(white) stochastic 'noise'or random perturbations $\mathlarger{\mathlarger{\mathscr{W}}}((t))$ with expectation $\mathlarger{\mathlarger{\mathcal{E}}}\llbracket\mathlarger{\mathlarger{\mathscr{W}}}(t)\rrbracket=0$ and 2-point correlation $\mathlarger{\mathlarger{\mathcal{E}}}\llbracket\mathlarger{\mathlarger{\mathscr{W}}}(t)
\mathlarger{\mathlarger{\mathscr{W}}}(s))=\alpha\delta(t-s)$ defined for $(t,s)\in \mathbf{X}_{II}\bigcup\mathbf{X}_{III}$ and $\alpha$ a constant.
\end{enumerate}
Then one can 'engineer' the following controlled system of deterministic and stochastic differential equations over the entire partition $\mathbf{R}^{+}
=\mathbf{X}_{I}~\cup~(\mathbf{X}_{II}~\cup~\mathbf{X}_{III}^{\infty})$ such that
\begin{align}
&\frac{d\overline{u(t)}}{dt}=\psi(u(t))=\kappa^{1/2}(u(t))^{2} (u(t)-1)^{1/2}~
if~t\in{\mathbf{X}}_{I}=[0,t_{\epsilon}]\\&
\frac{d}{dt}\overline{u(t)}=\psi(u(t))\mathlarger{\mathlarger{\mathscr{W}}}(t)=\kappa^{1/2}|u(t)^{2} (u(t)-1)^{1/2}|\mathlarger{\mathlarger{\mathscr{W}}}(t)~if~t\in{\mathbf{X}}_{II}~
\bigcup~{\mathbf{X}}_{III}
\end{align}
or
\begin{align}
&d\overline{\overline{u(t)}}=\psi(u(t))dt=\kappa^{1/2}u^{2}(t)(u(t)-1)^{1/2}dt~if~
t\in \bm{\mathbf{X}}^{+}_{I}=[0,t_{\epsilon}]\\&
d\overline{u(t)}=\psi(u(t))\mathlarger{\mathlarger{\mathscr{W}}}dt=\kappa^{1/2}|u^{2}(t) (u(t)-1)^{1/2}|d\mathlarger{\mathlarger{\mathscr{B}}}(t)~
if~t\in{\mathbf{X}}_{II}~\cup~{\mathbf{X}}_{III}
\end{align}
since $d\mathlarger{\mathlarger{\mathscr{B}}}(t)=\mathlarger{\mathscr{W}}(t)dt$, where $\mathlarger{\mathlarger{\mathscr{B}}}(t)$ is the standard Weiner process or 'Brownian motion' which now exists for $t\ge t_{\epsilon}$ or $t\in\mathbf{X}_{II}\cup\mathbf{X}^{+}_{III}$.(Note that all stochastic quantities will be emphasised in overline and the stochastic expectation is denoted $\mathlarger{\mathlarger{\mathcal{E}}}\big\llbracket...\big\rrbracket$.
\end{prop}
Using 'indicator functions' this can be expressed on the partition $\mathbf{R}^{+}=\mathbf{X}_{I}\bigcup(\mathbf{X}_{II}\bigcup\mathbf{X}_{III})$ as a 'hybrid' ODE-SDE within the Ito interpretation
\begin{align}
&\frac{d u(t)}{dt}=\mathlarger{\mathcal{C}}(\mathbf{X}_{I})\psi(u(t))
+\mathlarger{\mathcal{C}}[\mathbf{X}_{II}\cup\mathbf{X}_{III}]\psi(u(t))
\mathlarger{\mathscr{W}}(t)\nonumber\\&
d\overline{u(t)}=\mathlarger{\mathcal{C}}[\mathbf{X}_{I}]du(t)+
\mathlarger{\mathcal{C}}[\mathbf{X}_{II}\cup\mathbf{X}_{III}]d\overline{u(t)}
\end{align}
In full
\begin{align}
&\frac{d\overline{u(t)}}{dt}=\mathlarger{\mathcal{C}}(\mathbf{X}_{I})
\kappa^{1/2}(u(t))^{2}(u(t)-1)^{1/2}+\mathlarger{\mathcal{C}}(\mathbf{X}_{II}
\cup\mathbf{X}_{III})\kappa^{1/2}|u^{2}(t))(u(t)-1)^{1/2}|\mathlarger{\mathlarger{\mathscr{W}}}(t)
\end{align}
where $\mathlarger{\mathcal{C}}(\mathbf{X}_{I})=1$ if $t\in\mathbf{X}_{I}$ and zero otherwise, while $\mathlarger{\mathcal{C}}(\mathbf{X}_{II}\cup\mathbf{X})=1$ if $t\ge t_{\epsilon}$ or equivalently
\begin{align}
&d\overline{u(t)}=\mathlarger{\mathcal{C}}[\bm{\mathbf{X}}_{I}]\kappa^{1/2}u(t))^{2}(u(t)-1)^{1/2}dt
+\mathlarger{\mathcal{C}}[\mathbf{X}_{I}\cup\mathbf{X}_{II}]|
\kappa^{1/2}u(t))^{2}(u(t)-1)^{1/2}|
d\mathlarger{\mathlarger{\mathscr{B}}}(t)
\end{align}
\begin{rem}
The stochastic noise $\mathlarger{\mathlarger{\mathscr{W}}}(t)$ can represent some new (albeit unknown)physics coming into effect at the very final stages of the collapse for
$t\ge t_{\epsilon}$ before the density singularity occurs; for example 'quantum effects'
but which are represented or manifested here as a (classical) stochastic process or random field that is mathematically compatible with aspects of classical general relativity. The analysis therefore interprets the collapsing system as a 'noisy' stochastically controlled nonlinear dynamical system, and the differential
$d\mathlarger{\mathlarger{\mathscr{B}}}(t)=\mathlarger{\mathlarger{\mathscr{W}}}(t)dt$ ensures a well-defined stochastic calculas and SDE.
\end{rem}
\begin{rem}
One can interpret this as a problem in 'stochastic control' of a dynamical system beginning at $t=t_{\epsilon}$. Equations (4.26) can be expressed as
\begin{equation}
d\overline{u(t)}=\psi(u(t))dZ(t)=\kappa^{1/2}(u(t))^{2}(u(t)-1)^{1/2}dZ(t)
\end{equation}
where
\begin{align}
&dZ(t)=dt~for~t\in\mathbf{X}_{I}~or~t\le t_{\epsilon}\\&
d\overline{Z}(t)=\alpha d\mathlarger{\mathlarger{\mathscr{B}}}(t)~for~
t\in\mathbf{X}_{II}\bigcup\mathbf{X}_{III}~or~t>t_{\epsilon}
\end{align}
and where $\alpha>0$. Or
\begin{equation}
d\overline{u(t)}=\psi(u(t))dZ(t)=\kappa^{1/2}(u(t))^{2}(u(t)-1)^{1/2}dZ(t)
\end{equation}
and where \begin{align}
&dZ(t)=dt~for~t\in\mathbf{X}_{I}~or~t\le t_{\epsilon}\\&
dZ(t)=\alpha \mathcal{F}(u(t),\epsilon)d\mathlarger{\mathlarger{\mathscr{B}}}(t)~for~t\in\mathbf{X}_{II}\bigcup\mathbf{X}_{III}~or~t>t_{\epsilon}
\end{align}
Here $\mathcal{F}(u(t),\epsilon)$ is an adjustable 'control functional'. The noise or perturbations are 'switched on' at $t=t_{\epsilon}$ and the diffusion is 'adjusted' or controlled to remove the blowup or singularity at $t=t_{*}=t_{\epsilon}+\epsilon$. For this work, we take $\mathcal{F}(u(t),\epsilon)$=1.
\end{rem}
It is also possible to engineer a similar set of hybrid differential equations
in terms of the original radial scale function $R(t)$ using equation (3.34).
\begin{prop}
The same conditions hold as in Proposition (4.6), then the controlled hybrid DE-SDE over the partition $\mathbf{X}_{a}\bigcup\mathbf{X}_{II}\cup \mathbf{X}_{III}$ in terms of the radial function $R(t)$ are then
\begin{align}
&\frac{d}{dt}{\overline{R(t)}}={\Psi}(R(t))=-k^{1/2}(R^{-1}(t)-1)^{1/2}~if~t\in
\mathbf{X}^{+}_{a}=[0,t_{\epsilon}]\\&
\frac{d}{dt}{\overline{R(t)}}={\Psi}(R(t))\mathlarger{\mathlarger{\mathscr{W}}}(t)=
-|k^{1/2} (R(t)-1)^{1/2}|\mathlarger{\mathlarger{\mathscr{W}}}(t)~if~t\in\mathbf{X}_{II}~\cup~
{\mathbf{X}}_{III}
\end{align}
or
\begin{align}
&d\overline{\mathrm{R}(t)}=\Psi(R(t))=-\kappa(R(t)-1)^{1/2}dt~if~t\in \mathbf{X}_{I}=[0,t_{\epsilon}]\\&
d\overline{\mathrm{R}(t)}=\Psi(R(t))\mathlarger{\mathlarger{\mathscr{W}}}(t)dt=
-\kappa^{1/2}(R^{-1}(t)-1)^{1/2}
\circ d\mathlarger{\mathlarger{\mathscr{B}}}(t)~if~t\in\mathbf{X}_{II}~\cup~\mathbf{X}_{III}
\end{align}
\end{prop}
The noise would then suppress the event that $\overline{R(t)}=0$. However, in this work, the equations for the density function $u(t)$ and its diffusion $\overline{u(t)}$ will generally prove to be more compatible with a stochastic analysis since boundedness and blowup criteria can then be studied. The aim is to prove rigorously that the stochastic integrals are now always finite and bounded over the partition $\mathbf{X}_{II}\bigcup\mathbf{X}_{III}$ when the stochastic control is switched on or come into effect at $t=t_{\epsilon}=t_{*}-{\epsilon}$.

The $\mathcal{L}_{p}$ norm of $\overline{u(t)}$ is defined as follows
\begin{defn}
The $\mathcal{L}_{p}$ norm of $\overline{u(t)}$ for all $p\in[1,\infty)$ is
\begin{equation}
\|\overline{u(t)}\|_{\mathcal{L}_{p}}=\left\|\int_{t_{\epsilon}}^{t}
\psi(u(s)d\mathlarger{\mathlarger{\mathscr{B}}}(t)\right\|_{\mathcal{L}_{p}}
=\left(\mathlarger{\mathlarger{\mathcal{E}}}\left\llbracket\left|\int_{t_{\epsilon}}^{t}
\psi(u(s)d\mathlarger{\mathlarger{\mathscr{B}}}(t)\right|^{p}\right\rrbracket\right)^{1/p}
\end{equation}
or
\begin{equation}
\|\overline{u(t)}\|^{p}_{\mathcal{L}_{p}}=\left\|\int_{t_{\epsilon}}^{t}
\psi(u(s)d\mathlarger{\mathlarger{\mathscr{B}}}(t)\right\|^{p}_{\mathcal{L}_{p}}
=\mathlarger{\mathlarger{\mathcal{E}}}\left\llbracket\left|\int_{t_{\epsilon}}^{t}
\psi(u(s)d\mathlarger{\mathlarger{\mathscr{B}}}(t)\right|^{p}\right\rrbracket
\end{equation}
so that the expectation is the $\mathcal{L}_{1}$ norm
\begin{equation}
\|\overline{u(t)}\|_{\mathcal{L}_{p}}=\left\|\int_{t_{\epsilon}}^{t}
\psi(u(s)d\mathlarger{\mathlarger{\mathscr{B}}}(t)\right\|_{\mathcal{L}_{p}}\equiv
\mathlarger{\mathlarger{\mathcal{E}}}\left\llbracket\int_{t_{\epsilon}}^{t}
\psi(u(s)d\mathlarger{\mathlarger{\mathscr{B}}}(t)\right\rrbracket
\end{equation}
The (Banach) space $\mathcal{L}_{p}=\mathcal{L}_{p}(\Omega,\Sigma,\mathlarger{\mathrm{I\!P}})$ then consists of all random variables $\overline{u(t)}$ on $\Omega$ for all $t>t_{\epsilon}$ with finite $\mathcal{L}_{P}$ norm so that
\begin{equation}
\mathcal{L}_{p}(\Omega,\Sigma,\mathlarger{\mathrm{I\!P}})=\left\lbrace\overline{u(t)}:
\|\overline{u(t)}\|_{\mathcal{L}_{p}}=\left\|k^{1/2}\int_{t_{\epsilon}}^{t}
\psi(u(s)d\mathlarger{\mathlarger{\mathscr{B}}}(t)\right\|_{\mathcal{L}_{p}}<\infty\right\rbrace
\end{equation}
The $\mathcal{L}_{P}$ normed space satisfies the criteria:
\begin{enumerate}
\item $\big\|\overline{u(t)}\|_{\mathcal{L}_{p}}\ge 0$.
\item For any $\alpha >0$ and $\overline{u(t)}\in\mathcal{L}_{p}$ then $\big\|\alpha \overline{u(t)}\big\|=|\alpha|\|\overline{u(t)}\|_{\mathcal{L}_{p}}$.
\item For any $\alpha>0$ and $\overline{u(t)},\overline{v(t)}\in\mathcal{L}_{p}$ then
$\|\overline{u(t)}-\overline{v(t)}\|_{\mathcal{L}_{p}}\le
\|\overline{u(t)}\|-\|\overline{v(t)}\|$
\item The 'metric distance' is $d(\overline{u(t)},\overline{v(t)})=\|\overline{u(t)}-\overline{v(t)}\|_{\mathcal{L}_{p}}$
\item For any $1\le p<q<r<\infty$ then $\mathcal{L}_{p}\subset\mathcal{L}_{q}\subset\mathcal{L}_{r}$ and $\|\overline{u(t)}\|_{\mathcal{L}_{p}}\le \|\overline{u(t)}\|_{\mathcal{L}_{q}}\le\| \overline{u(t)}\|_{\mathcal{L}_{r}}$
\end{enumerate}
\end{defn}
\begin{prop}
The Ito Integral solution for the density function diffusion over the entire collapse
interval ${\mathbf{R}}^{+}={\mathbf{X}}_{I}\bigcup(\mathbf{X}_{II}
\bigcup{\mathbf{X}}_{III})$ is now formally
\begin{align}
&\overline{u(t)}=\mathcal{C}[{\mathbf{X}}_{I}]\kappa^{1/2}\int_{0}^{t_{\epsilon}}(u(s))^{2}(u(s)-1)^{1/2}ds
+\mathlarger{\mathcal{C}}[{\mathbf{X}}_{II}\cup{\mathbf{X}}_{III}]\kappa^{1/2}
\int_{t_{\epsilon}}^{t}(u(s))^{2}(u(s)-1)^{1/2}d\mathlarger{\mathlarger{\mathscr{B}}}(s)
\end{align}
which is $ \widehat{u}(t)=u_{\epsilon}+\kappa^{1/2}\int_{t_{\epsilon}}^{t}
(u(s))^{2}(u(s)-1)^{1/2}d\mathlarger{\mathlarger{\mathscr{B}}}(s).$. If the random perturbations suppress the blowup or singularity in $|\overline{u(t)}|$ then the system is still unstable in that it will collapse eternally and never come to a new equilibrium, but there will be blowup or singularity such that $|\overline{u(t)}|=\infty$ for any finite $t>t_{\epsilon}$
\begin{enumerate}
\item The integral in the first term in (-) is now cut off at $t=t_{\epsilon}<t_{*}=\pi/2\kappa^{1/2}$  so that
\begin{equation}
u(t_{\epsilon})=u_{\epsilon}=u_{0}+\kappa^{1/2}\int_{0}^{t_{\epsilon}}|(u(s))^{2}
(u(s)-1)^{1/2}|ds<\infty
\end{equation}
whereas at $t=t_{*}=\pi/2\kappa^{1/2}$, the original density function explodes:
\begin{equation}
\overline{u(t_{*})}=u_{\epsilon}=\kappa^{1/2}\int_{0}^{t_{*}}(u(s))^{2}(u(s)-1)^{1/2}ds=\infty
\end{equation}
\item At $t=t_{\epsilon}$, the density of the ball is
very large or 'hyperdense' with $u_{\epsilon}\gg 1$ but still finite.
\begin{equation}
\overline{u(t)}=u_{\epsilon}+\kappa^{1/2}\int_{t_{\epsilon}}^{t}(u(s))^{2}(u(s)-1)^{1/2}
d\mathlarger{\mathlarger{\mathscr{B}}}(s)<\infty
\end{equation}
or at least over ${\mathbf{X}}_{II}=[t_{\epsilon},t_{*}]$ so that $
\overline{u(t)}<\infty$
\item There is a blowup or singularity if $ \mathlarger{\mathrm{I\!P}}(|\overline{u(t)}|
    =\infty)=1 $ or $\mathlarger{\mathrm{I\!P}}(|\overline{u(t)}-u)_{\epsilon}|=\infty)=1$ for some finite $t>t_{\epsilon}$.
\item There is no blowup or singularity for any finite $t>t_{\epsilon}$ if
\begin{align}
\mathlarger{\mathrm{I\!P}}(|\overline{u(t)}|)=\infty)=0,~~~~\lim_{t\uparrow\infty}\mathlarger{\mathrm{I\!P}}(|\overline{u(t)}|)=\infty)=0\nonumber
\end{align}
\item The system is unstable (collapsing) but nonsingular if there is zero probability that $|u(t)-u_{\epsilon}|$ cannot be contained with any finite Euclidean ball ${\mathbf{B}}(L)$ of radius $L$ in the limit as $t\rightarrow\infty$ so that
\begin{equation}
\lim_{t\uparrow\infty}\mathlarger{\mathrm{I\!P}}(|\overline{u(t)}-u^{\epsilon}|\in
{\mathbf{B}}(L))\equiv
\lim_{t\uparrow\infty}\mathlarger{\mathrm{I\!P}}(|\overline{u(t)}-u^{\epsilon}|\le L)=0
\end{equation}
\begin{equation}
\lim_{t\uparrow\infty}\mathlarger{\mathrm{I\!P}}(|\overline{u(t)}-u^{\epsilon}|
\notin{\mathbf{B}}(L))\equiv
\lim_{t\uparrow\infty}\mathlarger{\mathrm{I\!P}}(|\overline{u(t)}-u^{\epsilon}|> L)=1
\end{equation}
\item There is no blowup or singularity for any finite $t>t_{\epsilon}$ and $p\ge 1$ if
\begin{equation}
\mathlarger{\mathlarger{\mathcal{E}}}\big \|\overline{u(t)}|\big\|^{p}_{\mathcal{L}_{p}} \equiv\mathlarger{\mathlarger{\mathcal{E}}}\bigg\llbracket|\overline{u(t)}|^{p}\bigg\rrbracket <Q
\end{equation}
and $\lim_{t\uparrow\infty}\mathlarger{\mathlarger{\mathcal{E}}}
\big\llbracket|u(t)|^{p}\big\rrbracket=\infty$. A singularity occurs for some finite $t$ if $\mathlarger{\mathlarger{\mathcal{E}}}\big\llbracket|\overline{u(t)}|^{p}\big\rrbracket=\infty$
\item The collapse is 'p-moment exponentially unstable' but nonsingular if $\exists (C_{1},C_{2})>0$ such that
\begin{equation}
\mathlarger{\mathlarger{\mathcal{E}}}\big \|\overline{u(t)}|\big\|^{p}_{\mathcal{L}_{p}}\equiv\mathlarger{\mathlarger{\mathcal{E}}}\bigg\llbracket|
\overline{u(t)}|^{p}\bigg\rrbracket\le C_{1}|u_{\epsilon}|\exp(C_{2}|t-t_{\epsilon}|)
\end{equation}
$\lim_{t\uparrow\infty}\mathlarger{\mathlarger{\mathcal{E}}}
\big\llbracket|u(t)|^{p}\big\rrbracket=\infty$.
The moments of the stochastic density function then grow exponentially for eternity but remain finite and bounded for all finite $t>t_{\epsilon}$.
\end{enumerate}
\end{prop}
The density blow up or singularity, which previously occurred at $t=t_{*}=\pi/2\kappa^{1/2}$ in the original analysis of the collapse, might then be 'smoothed out' or 'dissipated' within a stochastic extension of the theory. In practice, the stochastic expectations and moments $\mathlarger{\mathlarger{\mathcal{E}}}(|\overline{u(t)}|^{p})$ of these integrals are calculated and one attempts to prove they are finite and have upper bounds for  $t>t_{\epsilon}$. This will be the case if the stochastic integral $\overline{u(t)}$ is a martingale for $t\in {\mathbf{X}}_{II}\bigcup{\mathbf{X}}_{III}$.
\begin{defn}
Within the controlled stochastic extension of the Tolman-Oppenhimer description, a nonsingular black hole formed from the collapse of a spherically symmetric fluid sphere (star) satisfies the following:
\begin{enumerate}
\item The conditions of Proposition (4.6) hold for all comoving times $t\in\mathbf{X}_{II}\cup\mathbf{X}_{II}$ or $t\in[t_{\epsilon},T]$, crucially $\mathlarger{\mathrm{I\!P}}[\overline{u(t)}=\infty]=0$ and $\mathlarger{\mathlarger{\mathcal{E}}}\big\llbracket|\overline{u(t)}|^{p}\big\rrbracket<\infty$ for all $p\ge 1$ so that $\overline{u(t)}$ is globally regular and never becomes singular for any comoving time $t\in\mathbf{X}_{II}\cup\mathbf{X}_{II}$.
\item The expectation of the stochastic Kretschmann invariant is bounded so that $\mathlarger{\mathlarger{\mathcal{E}}}\big\llbracket\mathbf{K}(t)\big\rrbracket<\infty$.
\item The exterior (static) vacuum metric for an external observer is still the Schwarzchild metric with horizon radius $2GM$, and the fluid sphere has collapsed through the horizon radius.
\end{enumerate}
\end{defn}
Here, the Ito interpretation of a stochastic integral is utilised (Appendix A) in that the Riemann-Stieltjes sum approximation of the integral takes the value of $u(t)$ at $ t_{i}^{n}$. (The Stratanovich interpretation is considered in Section IV.) Ito stochastic integrals have a well-defined Ito calculas and mathematical formalism but the Stratanovich interpretation has the advantage that the rules of ordinary calculas apply.

There is also the possibility of interpreting the random fluctuations as arising from random fluctuations in the Newton constant $G$ itself and treating it as a random field or 'control parameter' for $t>t_{\epsilon}$.
\begin{prop}
Since the factor $\kappa^{1/2}=(8\pi G/3\rho(0))^{1/2}\equiv (8\pi/3\rho(0))^{1/2}G^{1/2}=\gamma^{1/2} G^{1/2}$, the stochastic perturbations or noise for $t>t_{\epsilon}$ can also be interpreted as random fluctuations in $G^{1/2}$ so that $\mathscr{G}^{1/2}\rightarrow G^{1/2}\mathlarger{\mathlarger{\mathscr{W}}}(t)$. If one writes $d\mathscr{G}(t)=G^{1/2}\mathlarger{\mathlarger{\mathscr{W}}}(t)dt\equiv G^{1/2}d\mathlarger{\mathlarger{\mathscr{B}}}(t)$ then this is a new Brownian differential. The hybrid ODE-SDE defined over the entire comoving interval $\mathbf{X}^{+}_{I}\bigcup\mathbf{X}^{+}_{II}
\bigcup{\mathbf{X}}^{+}_{III}$.
is equivalently
\begin{align}
&\frac{d\overline{u(t)}}{dt}=\mathcal{C}({\mathbf{X}}_{I})
{\gamma}^{1/2}G^{1/2}u^{2}(t)(u(t)-1)^{1/2}+\mathlarger{\mathcal{C}}(\mathbf{X}_{II}
\cup\mathbf{X}_{III})
{\gamma}^{1/2}u^{2}(t)(u(t)-1)^{1/2}(G^{1/2}\mathlarger{\mathlarger{\mathscr{W}}}(t))
\end{align}
Then
\begin{align}
&d\overline{u(t)}=\mathcal{C}(\mathbf{X}_{I}){\gamma}^{1/2}u^{2}(t)(u(t)-1)^{1/2}G^{1/2}dt+
\mathcal{C}(\mathbf{X}_{II})\cup\mathbf{X}_{III}){\gamma}^{1/2}
u^{2}(t)(u(t)-1)^{1/2}G^{1/2}d\mathlarger{\mathlarger{\mathscr{B}}}(t)
\end{align}
or equivalently in terms of the controlled hybrid ODE-SDE
\begin{equation}
d\overline{u(t)}=\mathcal{C}({\mathbf{X}}_{I})
\gamma^{1/2}u^{2}(t)(u(t)-1)^{1/2}dt+\mathcal{C}(\mathbf{X}_{II}\cup\mathbf{X}_{III})
\gamma^{1/2}u^{2}(t)(u(t)-1)^{1/2}d\bm{\mathscr{G}}(t))
\end{equation}
Although $\mathlarger{\mathlarger{\mathcal{E}}}(\bm{\mathscr{G}}^{1/2}(t))=G^{1/2}\mathlarger{\mathlarger{\mathcal{E}}}
\big\llbracket\mathlarger{\mathlarger{\mathscr{W}}}(t)\rrbracket=0$, the 2-point correlation is
\begin{equation}
\mathlarger{\mathlarger{\mathcal{E}}}\bigg\llbracket\mathlarger{\mathlarger{\mathscr{G}}}(t)
\otimes \mathlarger{\mathlarger{\mathscr{G}}}(s)\bigg\rrbracket
=G^{1/2}\alpha\delta(t-s)
\end{equation}
If the delta function is smeared out into a very narrow but highly-peaked but finite and normalised function $f(|t-s|;\tau)$ with correlation length $\tau$ then $
\mathlarger{\mathlarger{\mathcal{E}}}\big\llbracket\bm{\mathcal{G}}(t)\bm{\mathcal{G}}(s))
\bigg\rrbracket=G\mathlarger{\mathlarger{\mathcal{E}}}(\mathlarger{\mathlarger{\mathscr{W}}}(t)
\mathlarger{\mathlarger{\mathscr{W}}}(s))=\Pi(|t-s|;\varsigma)$. At $t=s$ this is now regulated so that if $\Pi(0,\varsigma)=1$ then $
\mathlarger{\mathlarger{\mathcal{E}}}({|\mathscr{G}}|^{2})=G|\Pi(0,\varsigma)|=G.$. (Note that one could choose $R_{\eta}\sim L_{p}$ and $\tau\sim T_{P}$, where $L_{p}$ and $T_{p}$ are the Planck scale and Planck time.) The solution for $t>t_{\epsilon}$ is then the stochastic integral
\begin{equation}
\overline{u(t)}=u_{\epsilon}+\gamma^{1/2}\int_{t_{\epsilon}}^{t}|u^{4}(s)(u(s)-1)^{1/2}|d{\mathscr{G}}(s)
\end{equation}
which is equivalent to (2.35).
\end{prop}
\raggedbottom
For the sake of  completeness the functional derivative is also defined as it will be utilised frequently in subsequent sections. Given the density function $u(t)$ for $t\in\mathbf{X}_{I}^{+}$ or $t<t_{\epsilon}$ and the density function diffusion $\overline{u(t)}$ for $t\in{\mathbf{X}}_{II}^{+}\cup{\mathbf{X}}_{III}$ or $t>t_{\epsilon}$ then one can construct functionals $\Phi(u(t))$ and $\Phi(\overline{u(t)})$. A functional derivative can then operate on these functionals.
\begin{defn}
Given the density function $u(t)$ for $t\in{\mathbf{X}}_{I}^{+}$ or $t<0$ define $\mathrm{D}_{u}$ as a differential operator or functional derivative such that
$\mathlarger{\mathlarger{\mathrm{D}}}_{u}{\psi}(u(t))=\frac{d}{du(t)}{\psi}(u(t))$
with differential $ d{\psi}(u(t))= \mathlarger{\mathlarger{\mathrm{D}}}_{u}{\psi}(u(t))du(t)\equiv
\frac{\delta}{\delta u(t)}{\psi}(u(t))du(t) $ so that ${\psi}(u(T))={\psi}(u(t_{i}))+\int_{t_{i}}^{T}d{\psi}(u(t))$.
\begin{equation}
\Phi(\overline{u(T)})=\Phi_{t_{\epsilon}}+\int_{u(t_{i})}^{u(T)}
\mathlarger{\mathlarger{\mathrm{D}}}_{u}{\psi}(u(t))
du(t)\equiv \int_{t_{i}}^{T}\frac{\delta}{\delta u(t)}\psi(u(t))du(t)
\end{equation}
For example, if ${\psi}(u(t))=|u(t)|^{p}$ then $\mathlarger{\mathlarger{\mathrm{D}}}_{u}{\psi}(u(t))
=p|u(t)|^{p-1}$ and $d{\psi}(u(t))=p|u(t)|^{p-1}du(t)$. Integration gives $ \int_{u(t_{i})}^{T}{x}(u(t))du(t)=\int_{u(t_{i})}^{u(T)}|u(t)|^{p}du(t)
=\frac{1}{p+1}|u(T)|^{p+1}-\frac{1}{p+1}|u(t_{i}|^{p+1} $
\end{defn}
\begin{lem}
Given $d\overline{u(t)}={\psi}(u(t))
d\mathlarger{\mathlarger{\mathscr{B}}}(t)$ for $t\in{\mathbf{X}}_{I}^{+}\cup{\mathbf{X}}_{III}$ or $t>t_{\epsilon}$ the differential of any stochastic functional $\Psi(\overline{u(t)}$ is related to the underlying $u(t)$ and is given by the Ito Lemma so that
\begin{align}
&d\Phi(\overline{u(t)})=|\mathlarger{\mathlarger{\mathrm{D}}}_{u}\Phi(u(t))|
d\overline{u(t)}+\frac{1}{2}|\mathlarger{\mathlarger{\mathrm{D}}}_{u}
\mathlarger{\mathlarger{\mathrm{D}}}_{u}d\bigg\langle\overline{u},\overline{u}
\bigg\rangle(t)\nonumber\\
&=\mathlarger{\mathlarger{\mathrm{D}}}_{u}\Phi(u(t))d\overline{u(t)}+\frac{1}{2}|
\mathlarger{\mathrm{D}}_{u}\mathlarger{\mathrm{D}}_{u}
\Phi(u(t))|{\psi}(u(t))|^{2}dt\nonumber\\&
=|\mathlarger{\mathlarger{\mathrm{D}}}_{u}\Phi(u(t))|\psi(u(t))d\mathlarger{\mathlarger{\mathscr{B}}}(t)
+\frac{1}{2}|\mathlarger{\mathlarger{\mathrm{D}}}_{u}\mathlarger{\mathlarger{\mathrm{D}}}_{u}
\Phi(u(t))||\psi(u(t))|^{2}dt\nonumber\\&=\kappa^{1/2}|
\mathlarger{\mathlarger{\mathrm{D}}}_{u}
\Phi(u(t))|u^{2}(t)(u(t)-1)^{1/2}|
d\mathlarger{\mathlarger{\mathscr{B}}}(t)
+\frac{1}{2}\kappa|\mathlarger{\mathlarger{\mathrm{D}}}_{u}\mathlarger{\mathlarger{\mathrm{D}}}_{u}
||(u(t))^{4}(u(t)-1)|dt
\end{align}
since $d\overline{u(t)}=\psi(u(t))d\mathlarger{\mathlarger{\mathscr{B}}}(t)$. Here,
$d\bigg\langle\overline{u},\overline{u}\bigg\rangle(t)$ is the differential of the quadratic variation $\bigg\langle \overline{u},\overline{u}\bigg\rangle(t)=|\psi(u(t))|^{2}dt$. (Appendix A).Then integrating
\begin{align}
& \Phi(\overline{u(T}))=\Phi(u(t_{\epsilon})+
\int_{t_{\epsilon}}^{t}|\mathlarger{\mathlarger{\mathrm{D}}}_{u}\Phi(u(t)|{\psi}(u(t))
d\mathlarger{\mathlarger{\mathscr{B}}}(t)
+\frac{1}{2}\int_{t_{\epsilon}^{T}}|
\mathlarger{\mathlarger{\mathrm{D}}}_{u}\mathlarger{\mathlarger{\mathrm{D}}}_{u}
\Phi(u(t))||{\psi}(u(\tau))|^{2}d\tau\nonumber\\&
=\kappa^{1/2}\int_{t_{\epsilon}}^{T}|\mathlarger{\mathlarger{\mathrm{D}}}_{u}
\Phi(u(s))|u^{2}(s)(u(s)-1)^{1/2}| d\mathlarger{\mathlarger{\mathscr{B}}}(s)+\frac{1}{2}\kappa\int_{t_{\epsilon}}^{t}|
\mathlarger{\mathlarger{\mathrm{D}}}_{u}
\mathlarger{\mathlarger{\mathrm{D}}}_{u}{\Phi}(u(s))||(u(s))^{4}(u(s)-1)|ds
\end{align}
Taking the stochastic expectation gives the useful result
\begin{align}
&{\mathlarger{\mathcal{E}}}\bigg\llbracket\Phi(\overline{u(T}))\bigg\rrbracket=\Phi(u(t_{\epsilon})+
\mathlarger{\mathlarger{\mathcal{E}}}\bigg\llbracket\int_{t_{\epsilon}}^{t}|
\mathlarger{\mathlarger{\mathrm{D}}}_{u}\Phi(u(t)|{\psi}(u(t))
d\mathlarger{\mathlarger{\mathscr{B}}}(t)(t)\bigg\rrbracket+\frac{1}{2}
\mathlarger{\mathlarger{\mathcal{E}}}\bigg\llbracket\int_{t_{\epsilon}^{T}}|
\mathlarger{\mathlarger{\mathrm{D}}}_{u}\mathlarger{\mathrm{D}}_{u}
\Phi(u(t))||{\psi}(u(\tau))|^{2}d\tau\bigg\rrbracket\nonumber\\&
=\kappa^{1/2}\mathlarger{\mathlarger{\mathcal{E}}}\bigg\llbracket\int_{t_{\epsilon}}^{t}|
\mathlarger{\mathlarger{\mathrm{D}}}_{u}\Phi(u(s))|u^{2}(s)(u(s)-1)^{1/2}|
d\mathlarger{\mathlarger{\mathscr{B}}}(s)\bigg\rrbracket
+\frac{1}{2}k\mathlarger{\mathlarger{\mathcal{E}}}\bigg\llbracket\int_{t_{\epsilon}}^{t}|
\mathlarger{\mathlarger{\mathrm{D}}}_{u}\mathlarger{\mathlarger{\mathrm{D}}}_{u}
\Phi(u(s))||(u(s))^{4}(u(s)-1)|ds\bigg\rrbracket\nonumber\\&
=\Phi(u(t_{\epsilon})+\frac{1}{2}\kappa\int_{t_{\epsilon}}^{t}
\mathlarger{\mathlarger{\mathcal{E}}}
\bigg\llbracket\mathlarger{\mathlarger{\mathrm{D}}}_{u}
\mathlarger{\mathlarger{\mathrm{D}}}_{u}\Phi(u(s))
|(u(s))^{4}(u(s)-1)|\bigg\rrbracket ds
\end{align}
or equivalently
\begin{align}
&\big\|\Phi(\overline{u(T)})\big\|_{\mathcal{L}_{1}}=\Phi(u(t_{\epsilon})+
\bigg\|\int_{t_{\epsilon}}^{t}|\mathlarger{\mathrm{D}}_{u}\Phi(u(t)|{\psi}(u(t))
d\mathlarger{\mathlarger{\mathscr{B}}}(t)\bigg\|_{\mathcal{L}_{1}}+\frac{1}{2}
\bigg\|\int_{t_{\epsilon}^{T}}|
\mathlarger{\mathrm{D}}_{u}\mathlarger{\mathrm{D}}_{u}
\Phi(u(t))||{\psi}(u(\tau))|^{2}d\tau\bigg\|_{\mathcal{L}_{1}}
\nonumber\\&=\kappa^{1/2}\bigg\|\int_{t_{\epsilon}}^{t}|\mathlarger{\mathrm{D}}_{u}
\Phi(u(s))|u^{2}(s)(u(s)-1)^{1/2}|d\mathlarger{\mathlarger{\mathscr{B}}}(s)\bigg\|_{\mathcal{L}_{1}}
+\frac{1}{2}\kappa\bigg\|\int_{t_{\epsilon}}^{t}|\mathlarger{\mathrm{D}}_{u}
\mathlarger{\mathrm{D}}_{u}\Phi(u(s))||(u(s))^{4}(u(s)-1)|ds\bigg\|_{\mathcal{L}_{1}}\nonumber\\&
=\Phi(u(t_{\epsilon})+\frac{1}{2}\kappa\int_{t_{\epsilon}}^{t}\bigg\|
\mathlarger{\mathrm{D}}_{u}\mathlarger{\mathrm{D}}_{u}\Phi(u(s))\bigg\||(u(s))^{4}(u(s)-1)|ds
\end{align}
\end{lem}
\begin{prop}
Since $\mathlarger{\mathlarger{\mathcal{E}}}[\mathlarger{\mathlarger{\mathscr{W}}}(t))dt]\equiv d \mathlarger{\mathlarger{\mathcal{E}}}\big\llbracket\mathlarger{\mathlarger{\mathscr{B}}}(t)\big\rrbracket=0$. The stochastically averaged equations over the partition $\mathbf{X}_{II}\cup\mathbf{X}_{III}$ are
\begin{align}
&\frac{d}{dt}\mathlarger{\mathlarger{\mathcal{E}}}\bigg\llbracket{\overline{u(t)}}\bigg\rrbracket
=\mathlarger{\mathcal{C}}(\mathbf{X}_{I})\kappa^{1/2}|(u(t))^{2}(u(t)-1)^{1/2}|+\mathlarger{\mathcal{C}}(\mathbf{X}_{II}\cup\mathbf{X}_{III})k^{1/2}
|(u(t))^{2}(u(t)-1)^{1/2}|\mathlarger{\mathlarger{\mathcal{E}}}
\bigg\llbracket\mathlarger{\mathlarger{\mathscr{W}}}(t)\bigg\rrbracket
\end{align}
where $t_{\epsilon}\gg t_{*}-t_{\epsilon}$ so that
\begin{equation}
\mathlarger{\mathlarger{\mathcal{E}}}\bigg\llbracket\frac{d\overline{u(t)}}{dt}\bigg\rrbracket
=\dot{u}(t)=k^{1/2}(u(t))^{2}(u(t)-1)^{1/2}
\end{equation}
which is the original nonlinear ODE describing the evolution of the density function as the pressureless sphere or star collapses.
\end{prop}
\subsection{Existence and stability and uniqueness of the SDE}
The existence and uniqueness of the SDE is considered using the H$\ddot{o}$lder and Lipschitz conditions [56-61].
\begin{lem}
Given the diffusion or SDE or diffusion $d\widehat{u}(t)={\psi}(u(t))
d\mathlarger{\mathlarger{\mathscr{B}}}(t)=\kappa^{1/2}|u^{2}(t)(u-1)^{1/2}|
d\mathlarger{\mathlarger{\mathscr{B}}}(t)$, defined for $t\in\mathbf{X}_{II}\bigcup\mathbf{X}_{III}=[t_{\epsilon},\infty]$ with initial data $u(t_{\epsilon})=u_{\epsilon}>1$, then let the following hold:
\begin{enumerate}
\item The Lipzsitch condition holds, such that for finite $(u(t),v(t))<n$ and for
some constant $L>0$
\begin{equation}
|{\psi}(u(t)-{\psi}(v(t)|~\le~L|u(t)-v(t)|
\end{equation}
for any solutions $u(t)$ and $v(t)$.
\item There is bounded polynomial growth for the coefficient $\psi(u(t))$. For $p\ge 2$,
there is a ${C}>0$ such that the coefficient $|\psi(u(t))|^{2}=\kappa u^{4}(u(t)-1)$
is bounded by any of the following polynomial growth conditions such
that $\exists$ constants $({J},{K}>0)$ so that for $p\ge 1$ for any $[t_{\epsilon},T]\subset\mathbf{X}_{II}\cup\mathbf{X}_{III}$.
\begin{align}
&|{\psi}(u(t)|^{2}=k|u^{4}(t)(u(t)-1)|~\le~ K(1+|u(t)|^{p}
\le K(1+|u(t)|)^{p}\\&
|{\psi}(u(t)|^{2}=k|u^{4}(t)(u(t)-1)|~\le~ J+K|u(t)|^{p}\\&
|{\psi}(u(t)|^{2}=k|u^{4}(t)(u(t)-1)|~\le~ K|u(t)|^{p}
\end{align}
Note that even when $u(t_{*})=\infty$ the equality condition still holds in that
\begin{equation}
\kappa u(t_{*})^{4}(u(t_{*})-1)= K(1+|u(t_{*})|^{p}=\infty
\end{equation}
or $ \infty=\infty$.
\end{enumerate}
\end{lem}
If $u_{\epsilon}$ is the initial condition with $\mathlarger{\mathlarger{\mathcal{E}}}(|u_{\epsilon}|^{2})
<\infty$ then $\exists$ a unique t-continuous strong solution $\overline{u(t)}$ of the
SDE $d\overline{u(t)}={\psi}(u(t))d\mathlarger{\mathlarger{\mathscr{B}}}(t)$ adapted to a filtration $\mathfrak{F}_{t}$.
\begin{proof}
To prove that (1) holds, it is sufficient to prove that $\mathlarger{\mathrm{D}}_{u}\psi(u)\equiv \tfrac{d}{du}\psi(u)<\infty$ for finite $|u|< n$ where $\mathlarger{\mathrm{D}}_{u}=d/du(t)$ since (1) essentially describes a derivative
\begin{equation}
\mathlarger{\mathrm{D}}_{u}|{\psi}[u(t)]| \sim \frac{|{\psi}(u(t)-{\psi}(\widehat{u}(t)|}
{|u(t)-\widehat{u}(t)|}
~\le~N
\end{equation}
This gives for $t\in\mathbf{X}_{II}\cup\mathbf{X}_{III}$
\begin{equation}
\mathlarger{\mathrm{D}}_{u}|{\psi}(u(t)|=2\kappa^{1/2}u(t)(u(t)-1)^{1/2}+\frac{1}{2}\kappa^{1/2}
u^{2}(t)(u(t)-1)^{-1/2}<\infty
\end{equation}
If $u(t)\le N$ then $\mathlarger{\mathrm{D}}_{u}{\psi}(u)\le 2k^{1/2}N(N-1)^{1/2}
+\frac{1}{2}\kappa^{1/2}N^{2}(N-1)^{-1/2}<\infty $. However,$\mathlarger{\mathrm{D}}_{u}\psi(u(t))=\infty$ if $u(t)=1$ but $u(t)=1$ only if $t=t_{o}=0$. However, we consider the SDE only for $t\ge t_{\epsilon}$ for which $u\ne 1$ so the condition always holds. To prove(2),there  should be a $\gamma\ge 2$ such that $ k^{1/2}u(t)|^{2}(u(t)-1)^{1/2}\le {L}(1+|u(t)|)^{p}$ so that setting $u(t)\rightarrow u(t)+1$ implies $k^{1/2}(u(t))^{2}(u(t)-1)^{1/2} < \kappa^{1/2}u(t)+1)^{2}u(t))^{1/2} \le {L}(1+|u(t)|)^{p/2}$, giving
\begin{equation}
\psi(u(t))=k^{1/2}|u^{2}(t)(u(t)-1)^{1/2}| < L(1+|u(t)|)^{p}
\end{equation}
\end{proof}
\begin{lem}
Given any finite subinterval $[t_{\epsilon},T]\subset[t_{\epsilon},\infty]$, where $T\ge t_{*}$  then there always $\exists L>0$ such that over this interval the linear growth conditions hold
\begin{align}
&k|u(t)^{4}(u(t)-1)|\le L(1+|u(t)|^{2})\le L|u(t)|^{2})\le L(1+|u(t)|)^{2}
\\&
k|u(T)^{4}(u(T)-1)|\le L(1+|u(T)|^{2})\le L|u(T)|^{2})\le L(1+|u(T)|)^{2}
\end{align}
Note that even when $u(T)=u(t_{*})=\infty$ the equality condition still holds in that $k|u(t_{*})^{4}(u(t_{*})-1)|= K(1+|u(t_{*})|^{2})$ or $\infty=\infty$.
\end{lem}
The following theorem establishes stability and uniqueness of the solution
$\overline{u(t)}$ based on the Holder and Lipschitz conditions.
\begin{thm}(Stability and uniqueness). A solution $\overline{u(t)}$ of the SDE is unique if any other solution $\overline{v(t)}$ is indistinguishable from $\overline{u(t)}$ for the same initial data $u(t_{\epsilon})=v(t_{\epsilon})$ if
$\mathlarger{\mathrm{I\!P}}[\overline{u(t)}-\overline{v(t)}]=0 $ or equivalently $\mathlarger{\mathrm{I\!P}}[\overline{u(t)}=\overline{v}(t)]=1$ for $t\in\mathbf{X}^{+}_{II}\cup\mathbf{X}^{+}_{III} =[t_{\epsilon},T]$. Given the SDE $d\overline{u(t)}=\psi(u(t)d\mathlarger{\mathlarger{\mathscr{B}}}(t)=\kappa^{1/2}u^{2}(t)(u(t)-1)^{1/2}
d\mathlarger{\mathlarger{\mathscr{B}}}(t)$ for $t>t_{\epsilon}$ let the following hold:
\begin{enumerate}
\item The global Lipschitz condition such that for two solutions $(\overline{u(t)},\overline{v}(t))\in[1,\infty]$, for $t>t_{\epsilon}$,$\exists$ $L>0$ such that
    \begin{equation}
    |{\psi}(u(t))- {\psi}(v(t)|^{2}\le L^{2}|u(t)-v(t)|^{2}
    \end{equation}
    \item The linear growth Holder condition on a finite interval $[t_{\epsilon},T]$ such
that $\exists$ $K>0$ such that $|{\psi}(u(t)|^{2}\le
K|u(t)|^{2}$ \item The maximal inequality
\begin{align}
&\mathlarger{\mathlarger{\mathcal{E}}}\left\llbracket\int_{t_{\epsilon}}^{T}|{\psi}(u(s))|^{2}ds\right\rrbracket\le
\mathlarger{\mathlarger{\mathcal{E}}}\left\llbracket\sup_{t_{\epsilon}\le t\le T}\left(\int_{t_{\epsilon}}^{T}{\psi}(u(s))d\mathlarger{\mathlarger{\mathscr{B}}}(s)\right)\right\rrbracket\le
4\mathlarger{\mathlarger{\mathcal{E}}}\left\llbracket\int_{t_{\epsilon}}^{T}|{\psi}(u(s)|^{2}ds\right\rrbracket
\end{align}
\end{enumerate}
Then for any two solutions $\overline{u(t)}$ and $\overline{v(t)}$ with initial values
$u_{\epsilon}$ and $v_{\epsilon}$ at $t=t_{\epsilon}$
\begin{equation}
\mathlarger{\mathlarger{\mathcal{E}}}\bigg\llbracket\sup_{t\le T}|\overline{u(t)}-\overline{v(t)}|^{2}\bigg\rrbracket\le 2|u_{\epsilon}-v_{\epsilon}|^{2}\exp(8L^{2}
|T-t_{\epsilon}||))
\end{equation}
Hence if the solutions have the same initial data then $ u_{\epsilon}=
v_{\epsilon}$ it follows that $\overline{u(t)}=\overline{v(t)}$ for
all $t>t_{\epsilon}$ and the solution is unique.
\end{thm}
\begin{proof}
Consider two solutions $\overline{u(t)}$ and $\overline{v(t)}$
\begin{align}
&\overline{u(t)}=u_{\epsilon}+k^{1/2}\int_{t_{\epsilon}}^{T}u^{2}(t)(u(t)-1)^{1/2}
d\mathlarger{\mathlarger{\mathscr{B}}}(t)\\&
\overline{v(t)}=v_{\epsilon}+k^{1/2}\int_{t_{\epsilon}}^{T}v^{2}(t)(v(t)-1)^{1/2}
d\mathlarger{\mathlarger{\mathscr{B}}}(t)
\end{align}
Then
\begin{equation}
|\overline{u(t)}-\overline{v(t)}|^{2}=\bigg|u_{\epsilon}-v_{\epsilon}+k^{1/2}\int_{t_{\epsilon}}^{T}[u^{2}(t)(u(t)-1)^{1/2}
-k^{1/2}v^{2}(t)(v(t)-1)^{1/2}]d\mathlarger{\mathlarger{\mathscr{B}}}(t)\bigg|^{2}
\end{equation}
with expectation
\begin{equation}
\mathlarger{\mathlarger{\mathcal{E}}}\bigg\llbracket|\overline{u(t)}-\overline{v(t)}|^{2}\bigg\rrbracket
=\mathlarger{\mathlarger{\mathcal{E}}}\bigg\llbracket\bigg|u_{\epsilon}-v_{\epsilon}+k^{1/2}\int_{t_{\epsilon}}^{T}[u^{2}(t)(u(t)-1)^{1/2}
-k^{1/2}v^{2}(t)(v(t)-1)^{1/2}]d\mathlarger{\mathlarger{\mathscr{B}}}(t)\bigg|^{2}\bigg\rrbracket
\end{equation}
Using the basic inequality $(a+b)^{2}\le 2a^{2}+2b^{2}$, then the Lipschitz condition, the maximal inequality and the Ito isometry gives
\begin{align}
\mathlarger{\mathlarger{\mathcal{E}}}\llbracket|\overline{u(t)}-\overline{v(t)}|^{2}\bigg\rrbracket&
=\mathlarger{\mathlarger{\mathcal{E}}}\bigg\llbracket\bigg|u_{\epsilon}-v_{\epsilon}+k^{1/2}\int_{t_{\epsilon}}^{T}[u^{2}(t)(u(t)-1)^{1/2}
-k^{1/2}v^{2}(t)(v(t)-1)^{1/2}]d\mathlarger{\mathlarger{\mathscr{B}}}(t)\bigg|^{2}\bigg\rrbracket\nonumber\\&
\le 2|u_{\epsilon}-v_{\epsilon}|^{2}+\underbrace{2
\mathlarger{\mathlarger{\mathcal{E}}}\bigg\llbracket\bigg|k\int_{t_{\epsilon}}^{T}[u^{2}(t)(u(t)-1)^{1/2}-v^{2}(t)(v(t)-1)^{1/2}]
d\mathlarger{\mathlarger{\mathscr{B}}}(t)
\bigg|^{2}}_{use~Ito~isometry}
\bigg\rrbracket\nonumber\\&\le2|u_{\epsilon}-v_{\epsilon}|^{2}+\underbrace{\mathlarger{\mathlarger{\mathcal{E}}}\bigg\llbracket\bigg|k\int_{t_{\epsilon}}^{T}
[u^{4}(t)(u(t)-1)-v^{4}(t)(v(t)-1)]ds
\bigg|^{2}\bigg\rrbracket}_{use~Lipshitz~condition~and~maximal~ineq.}\nonumber\\&
2|u_{\epsilon}-v_{\epsilon}|^{2}+8L^{2}\int_{t_{\epsilon}}^{T}\mathlarger{\mathlarger{\mathcal{E}}}\bigg\llbracket|u(t)|^{2}-|v(t)|^{2}\bigg\rrbracket dt
\end{align}
Using the Gronwall Lemma then gives the estimate
\begin{align}
&d(\overline{u(t)},\overline{v(t)})\equiv\big\|\overline{u(t)}-\overline{v(t)}|\big\|^{2}_{\mathcal{L}_{2}}
\equiv\mathlarger{\mathlarger{\mathcal{E}}}\bigg\llbracket|\overline{u(t)}-\overline{v(t)}|^{2}\bigg\rrbracket\le
2|u_{\epsilon}-u_{\epsilon}|^{2}
\exp\left(8L^{2}\int_{t_{\epsilon}}^{T}dt\right)\nonumber\\&
=2|u_{\epsilon}-u_{\epsilon}|^{2}\exp(8L^{2}|T-t_{\epsilon}|\big\rbrace
\end{align}
so that $\mathlarger{\mathlarger{\mathcal{E}}}[\sup_{t\le T}|\overline{u(t)}-\overline{u(t)}|^{2}]=0$ if
$u_{\epsilon}=v_{\epsilon}$ and the solution is unique with $\overline{u(t)}=\overline{v(t)}$.
\end{proof}
\begin{thm}(Existence and strong solution). Let the condition of Thm (4.19) hold. Then there exists a unique solution that satisfies
\begin{equation}
\sup_{t\le T}\big\|\overline{u(t)}\big\|^{2}_{\mathcal{L}_{2}}
\equiv\mathlarger{\mathlarger{\mathcal{E}}}\bigg\llbracket\sup_{t\le T}|\overline{u(t)}|^{2}\bigg\rrbracket\le C_{T}\mathlarger{\mathlarger{\mathcal{E}}}\bigg\llbracket(1+|u_{\epsilon}|^{2})\bigg\rrbracket<\infty
\end{equation}
with constant $C_{T}>0$, and which is the condition for a non-exploding strong solution on any $[t_{\epsilon},T]$. In particular
\begin{equation}
\sup_{t\le t_{*}}\big\|\overline{u(t)}\big\|^{2}_{\mathcal{L}_{2}}
\equiv\mathlarger{\mathlarger{\mathcal{E}}}\bigg\llbracket\sup_{t\le t_{*}}|\overline{u(t)}|^{2}\bigg\rrbracket\le C_{_{*}}\mathlarger{\mathlarger{\mathcal{E}}}\bigg\llbracket(1+|u_{\epsilon}|^{2})\bigg\rrbracket<\infty
\end{equation}
so there is no blowup or singularity in $\overline{u(t)}$ at $t=t_{*}$.
\end{thm}
\begin{proof}
Using a Picard iteration $\overline{u_{n}(t)}$ for $n=0,1,2,3...$ set
\begin{equation}
\overline{u}_{o}=\bm{\Lambda}
\end{equation}
\begin{equation}
\overline{u(t)_{n+1}(t)}=\bm{\Lambda}+
\int_{t_{\epsilon}}^{t}{\psi}_{n}(u(t))d\mathlarger{\mathlarger{\mathscr{B}}}_{n}(t)
=\kappa^{1/2}\int_{t_{\epsilon}}^{t}(u_{n}(t))^{2}(u_{n}(t)-1)^{1/2}d\mathlarger{\mathlarger{\mathscr{B}}}(t)_{n}(t)
\end{equation}
so that
\begin{equation}
|\overline{u(t)_{n+1}(t)}-\bm{\Lambda}|^{2}=\left|\int_{t_{\epsilon}}^{t}\psi_{n}
(u(t))d\mathlarger{\mathlarger{\mathscr{B}}}_{n}(t)\right|^{2}=\left|\kappa^{1/2}
\int_{t_{\epsilon}}^{t}(u_{n}(t))^{2}(u_{n}(t)-1)^{1/2}
d\mathlarger{\mathlarger{\mathscr{B}}}_{n}(t)\right|^{2}
\end{equation}
Using $|\psi(u(t)|^{2}=ku^{4}(t)(u(t)-1)\le K(1+|u(t)|)^{2}$ and $(a+b)^{2}\le 2a^{2}+2b^{2}$.
\begin{align}
\mathlarger{\mathlarger{\mathcal{E}}}\bigg\llbracket\sup_{t\le T}|\overline{u_{n+1}(t)}-\bm{\Lambda}|^{2}\bigg\rrbracket&\le
\mathlarger{\mathlarger{\mathcal{E}}}\bigg\llbracket\left|k^{1/2}\int_{t_{\epsilon}}^{t}(u_{n}(t))^{2}(u_{n}(t)-1)^{1/2}
d\mathlarger{\mathlarger{\mathscr{B}}}(t)_{n}(t)\right|^{2}
\bigg\rrbracket\nonumber\\&
\le \underbrace{\mathlarger{\mathlarger{\mathcal{E}}}\bigg\llbracket\left|k\int_{t_{\epsilon}}^{t}(u_{n}(t))^{4}(u_{n}(t)-1)
ds\right|^{2}
\bigg\rrbracket}_{via~Ito~isometry}\nonumber\\&
\le 2K\int_{t_{\epsilon}}^{t}\mathlarger{\mathlarger{\mathcal{E}}}\bigg\llbracket[1+|u_{n}(t)|^{2})\bigg\rrbracket dt\le 2KT\mathlarger{\mathlarger{\mathcal{E}}}\bigg\llbracket\sup_{t\le T}(1+|u_{n}(t)|)^{2}\bigg\rrbracket
\end{align}
Starting at $n=0$, the recursive processes $\overline{u_{n}(t)}$ are well defined for all $n>0$.
\begin{equation}
|\overline{u_{n+1}(t)}-\overline{u_{n}(t)}|^{2} \le 2
\left|\int_{t_{\epsilon}}^{t}[u_{n}(t))^{2}(u_{n}(t)-1)^{1/2}-u_{n-1}(t))^{2}(u_{n-1}(t)-1)^{1/2}]d\mathlarger{\mathlarger{\mathscr{B}}}(t)
\right|^{2}
\end{equation}
Using the Lipschitz condition and the same steps as in (Thm.4.19)
\begin{equation}
\mathlarger{\mathlarger{\mathcal{E}}}\bigg\llbracket\sup_{t\le T}|\overline{u_{n+1}(t)}-\overline{u_{n}(t)}|^{2}\bigg\rrbracket \le 2 L\int_{t_{\epsilon}}^{t}\mathlarger{\mathlarger{\mathcal{E}}}\bigg\llbracket\bigg|u_{n}(t)-u_{n-1}(t))\bigg|^{2}\bigg\rrbracket dt
\end{equation}
Now setting $\phi_{n}(t)=\mathlarger{\mathlarger{\mathcal{E}}}[\sup_{t\le T}|\overline{u_{n+1}(t)}-\overline{u_{n}(t)}|^{2}]$ and $Q_{T}=2L$, successive iterations with each integral over the subinterval $[t_{\epsilon},T]$ gives
\begin{align}
&\phi_{n}(T)\le Q_{T}^{n}\underbrace{\int...\int}_{t_{\epsilon}=t_{1}\le...t_{n-1}...\le T}\phi_{0}(t_{0})dt_{0}...dt_{n-1}\\\nonumber&
\le Q_{T}^{n}\int...\int dt_{0}...dt_{n-1}\phi_{0}(T)=Q_{T}^{n}
\frac{|T-t_{\epsilon}|^{n}}{n!}\psi_{0}(T)
\end{align}
The estimate for $n=0$ is
\begin{equation}
\phi_{0}(T)=\mathlarger{\mathlarger{\mathcal{E}}}\bigg\llbracket\overline{u_{1}(t)}-\bm{\Lambda}|^{2}\bigg\rrbracket\le
2L|T-t_{\epsilon}|\mathlarger{\mathlarger{\mathcal{E}}}\bigg\llbracket 1+|\bm{\Lambda}^{2}|\bigg\rrbracket
\end{equation}
Hence
\begin{align}
\bigg(&\sum_{n=0}^{\infty}\mathlarger{\mathlarger{\mathcal{E}}}\bigg\llbracket\sup_{t\le T}|\overline{u_{n+1}(t)}-\overline{u_{n}(t)}|^{2}\bigg\rrbracket
\bigg)^{1/2} \le \sum_{n=0}^{\infty}\left(Q_{T}^{n}\frac{|T-t_{\epsilon}|}{n!}\phi_{0}(T)
\right)^{1/2}\nonumber\\&\le D_{T}\lbrace(1+|\Lambda|^{2})\rbrace^{1/2}
\end{align}
where $D_{T}=D(L,K,T)$. For $m>n$
\begin{equation}
\big(\mathlarger{\mathlarger{\mathcal{E}}}\bigg\llbracket\sup_{t\le T}|\overline{u_{n}(t)}-\overline{u_{m}(t)}|^{2}\bigg\rrbracket\big)^{1/2}
\le \sum_{j=m+1}^{\infty}\big(\mathlarger{\mathlarger{\mathcal{E}}}\bigg\llbracket\sup_{t\le T}|\overline{u_{j}(t)}-\overline{u_{j-1}(t)}|^{2}\bigg\rrbracket\big)^{1/2}
\end{equation}
Then $\overline{u_{n}(t)}$ is then a Cauchy sequence with limit $\overline{u(t)}$.
\begin{align}
&(\mathlarger{\mathlarger{\mathcal{E}}}\llbracket\sup|\overline{u(t          )}-\bm{\Lambda}^{2}|\bigg\rrbracket^{1/2}\le \sum_{n=0}^{\infty}
\lbrace\mathlarger{\mathlarger{\mathcal{E}}}\bigg\llbracket\sup_{t\le T}|u_{n+1}(t)-u_{n}(t)|^{2}\bigg\rrbracket\big)^{1/2}\nonumber\\&\le
Q_{T}\bigg(\mathlarger{\mathlarger{\mathcal{E}}}\bigg\llbracket 1+|\bm{\Lambda}^{2})\bigg\rrbracket\bigg)^{1/2}\equiv \bigg(Q_{T}\mathlarger{\mathlarger{\mathcal{E}}}\bigg\llbracket\bigg(1+|u_{\epsilon}|\bigg)\bigg\rrbracket^{1/2}\bigg)
\end{align}
so that (4.75) follows with $D_{T}=C{T}+1$. Hence the solution is strong and non-exploding
\end{proof}
\begin{rem}
Given the nonlinear SDE $d\overline{u(t)}={\psi}(u(t))d\mathlarger{\mathlarger{\mathscr{B}}}(t)
=k^{1/2}u^{2}(t)(u(t)-1)^{1/2}d\mathlarger{\mathlarger{\mathscr{B}}}(t)$ with initial
data $u(t_{\epsilon})=u_{\epsilon}$ and $\mathlarger{\mathlarger{\mathscr{B}}}(t_{\epsilon})=0$
for $t\ge t_{\epsilon}$, it is tempting to suppose that the solution follows as
\begin{align}
&k^{-1/2}(\overline{u(t)}-1)^{1/2}\overline{u}^{-1}(t)+\tan((\overline{u(t)}-1)^{1/2})\nonumber\\&
=k^{-1/2}Y(\overline{u(t)})-Y(u_{\epsilon}))=\mathlarger{\mathlarger{\mathscr{B}}}(t)
\end{align}
and where $Y(u_{\epsilon}=\kappa^{-1/2}(u_{\epsilon}-1)^{1/2}
u^{-1}(t)+\tan((u_{\epsilon}-1)^{1/2})$. Then there is a blowup at a
'hitting time' T when $\mathlarger{\mathlarger{\mathscr{B}}}(T)=\tfrac{1}{2}k^{1/2}(\pi/-Y_{\epsilon})$.But within the Ito interpretion
\begin{equation}
k^{1/2}dy(\overline{u(t)})\ne\kappa^{1/2}Y(\overline{u(t)})
d\overline{u}(t)=d\mathlarger{\mathlarger{\mathscr{B}}}(t)
\end{equation}
since $dy(\overline{u(t)})$ is given by the Ito Lemma and not the rules of
ordinary calculas. However, in Section 12, it will be shown that such a solution can be found within the Stratanovich interpretation.
\end{rem}
\section{Criteria for the driftless Ito diffusion $\overline{u(t)}$ to be a true martingale}
It is to be proved that the Ito integral solution
\begin{equation}
\overline{u(t)}=u_{\epsilon}+\int_{t_{\epsilon}}^{t}\psi(u(s))d\mathlarger{\mathlarger{\mathscr{B}}}(s)
=u_{\epsilon}+k^{1/2}\int_{t_{\epsilon}}^{t}(u(s))^{2}(u(s)-1)^{1/2}d\mathlarger{\mathlarger{\mathscr{B}}}(t)
\end{equation}
is actually a true bounded martingale (or at least a local martingale) that always remains finite, for $t\in\mathbf{X}_{II}\bigcup\mathbf{X}_{III}$,or at least locally for $t\in\mathbf{X}_{II}$. If the Ito diffusion $\overline{u(t)}$ is a true martingale $\overline{m(t)}$ then its supremum is bounded for all finite $t>t_{\epsilon}$ with $u_{\epsilon}\in[1,u(t_{*}-\epsilon)]$ and so the usual singularity or blowup at $t=t_{*}=\pi/2 k^{1/2}$ is then essentially "noise-suppressed", when white-noise random perturbations come into effect at $t_{\epsilon}$, close to the blowup comoving proper time at $t=t_{*}$. As stated in Section 2, martingales have the following desirable and useful properties:
\begin{itemize}
\item Even under mild conditions and general initial data, suprema of martingales arebounded over finite or semi-infinite intervals, making them useful and powerful tools for studying the growth, bounds or explosions of stochastic processes and diffusions. Many powerful and well-established theorems then apply.
\item They are closely related to Markov and Wiener processes.
\item They are incorporated within a well-defined and rigorous Ito calculas.
\item Bounds on martingales are of two types: bounds on the probabilities that the martingale $\overline{m(t)}$, or its absolute value $|\overline{m(t)}|$, exceeds some given value for $t$ within some finite or semi-infinite interval; and bounds or
    estimates on the stochastic expectations of the moments $\mathlarger{\mathlarger{\mathcal{E}}}\big\llbracket|\overline{m(t)}|^{p}\big\rrbracket)$ for $p\ge 1$ and $p\in\mathbf{Z}$ such that $\mathlarger{\mathlarger{\mathcal{E}}}\big\llbracket|\overline{m(t)}|^{p}\big\rrbracket <\infty$ for all $t>t_{\epsilon}$
\item The expectation is finite so that $\mathlarger{\mathlarger{\mathcal{E}}}\llbracket\overline{m(t)}\big\rrbracket=const.$ on the entire real line or at least over the sub interval where the underlying ODE had a blowup.
\item The probability of a blowup or singularity is zero for all finite $t$ so that $\mathlarger{\mathrm{I\!P}}[\overline{m(t)}=\infty]=0$.
\end{itemize}
\begin{defn}
A generic martingale is essentially a stochastic process $m:
\mathbf{R}^{+}\times\Omega\rightarrow\mathbf{R}^{+}$ or $\overline{u(t,\omega)}$,
with $\omega\in\Omega$ defined with respect to an underlying probability
triplet $(\Omega,\mathscr{F},\mathlarger{\mathrm{I\!P}})$, whose expected value $\mathlarger{\mathlarger{\mathcal{E}}}\big\llbracket \overline{m(t';\omega)}\big\rrbracket \equiv\int_{\Omega}d\mathlarger{\mathrm{I\!P}}(\omega)
\overline{m(t;\omega)}$ for some future time $t'>t$ is the same as the present value $t$. The proper time parameter $t$ is continuous with $t\in\mathbf{R}^{+}$ or some subset. There is also an increasing family of filtrations $\lbrace\mathscr{F}_{t}\rbrace\subset\mathscr{F}$ of $\sigma$-algebras. Then for any $t'<t$ one has a martingale if
\begin{equation}
(\int_{\Omega}\overline{m(t,\omega)}
d\mathlarger{\mathrm{I\!P}}(\omega)\|\mathscr{F}_{t'})\equiv
\mathlarger{\mathlarger{\mathcal{E}}}\bigg\llbracket(\overline{m(t)}\|\bm{\mathscr{F}}_{t'})\bigg\rrbracket=\overline{m(t')}
\end{equation}
If the process is a submartingale it describes expected growth on average so that $(\int_{\Omega}\overline{m(t,\omega)}d\mathlarger{\mathrm{I\!P}}(\omega)\|\mathscr{F}_{t'})
\equiv\mathlarger{\mathlarger{\mathcal{E}}}(\overline{m(t)}\|\mathscr{F}_{t'})>
\overline{m(t')}$. The process may also blow up. The process is a supermartingale if it diminishes on average $(\int_{\Omega}\overline{m(t,\epsilon)}d\mathlarger{\mathrm{I\!P}}(\omega)\|\mathscr{F}_{t'})
\equiv\mathlarger{\mathlarger{\mathcal{E}}}\big\llbracket\overline{m(t)}\|\mathscr{F}_{t'})\big\rrbracket<\overline{m(t')}$. The diffusion $\overline{u(t)}$ is a martingale if for $t>t'\in\mathbf{X}_{II}\cup\mathbf{X}_{III}$ one has $\mathlarger{\mathlarger{\mathcal{E}}}\big\llbracket(\overline{u(t)}\|\bm{\mathscr{F}}_{t'}\big\rrbracket)
=\overline{u(t')} $
\end{defn}
\begin{defn}
A generic diffusion process $(\overline{u(t)}$ on $(\Omega,\mathscr{F},\mathlarger{\mathrm{I\!P}})$ is a local martingale if it is adapted to a filtration $\mathscr{F}_{t}$ and $\exists$ sequence of stopping times $t_{n}$ for $n\in\mathbf{Z}$ (the localisation sequence)such that $\tau_{n}<\infty$ and $\tau_{n}\uparrow \infty$ such that $\overline{u_{n}(t)}=\overline{u(t\wedge \tau_{n})}$ is an $\mathscr{F}_{t}$-martingale for each $n=1,2,3...$
\end{defn}
However, a generic driftless Ito diffusion is not necessarily a true martingale. Criteria for a driftless diffusion to be a true martingale are discussed in [63] who demonstrated that one can verify whether or not the martingale property of the solution of (5.1) fails. They found necessary and sufficient conditions on the diffusion coefficient that can ascertain whether or not the diffusion is truly a martingale. It is noted also that the solution $\overline{u(t)}$ satisfies both the Markov property as well as weak uniqueness.

It is well known that a generic driftless Ito diffusion $\overline{u(t)}$ is at least a local martingale, and $\exists$ a non-exploding weak solution if $\psi(u(t)^{-2}$ is integrable. If a blowup time is defined as
\begin{equation}
T_{\infty}=\inf(t>t_{\epsilon}:\overline{u(t)}=\infty)
\end{equation}
then $\mathlarger{\mathrm{I\!P}}[T_{\infty}<\infty]=0$ or $\mathlarger{\mathrm{I\!P}}(T_{\infty}=\infty)=1$ for any martingale. The strategy is then to establish that the driftless Ito diffusion for the density function is actually a true martingale for $t>t_{\epsilon}$ and $u(t)\in[1,\infty)$ and $\overline{u(t)}\in[u_{\epsilon},\infty)$. Given
\begin{align}
\overline{u(t)}=u_{\epsilon}+\int_{t_{\epsilon}}^{t}\psi(u(s))d\mathlarger{\mathlarger{\mathscr{B}}}(s)&
=u_{\epsilon}+k^{1/2}\int_{t_{\epsilon}}^{t}(u(s))^{2}(u(s)-1)^{1/2}d\mathlarger{\mathlarger{\mathscr{B}}}(s)
\end{align}
\begin{enumerate}
\item Prove that the driftless Ito diffusion $\overline{u(t)}$ is a true martingale.
\item The solution is then globally regular with no singularities or blowups for all $t>t_{\epsilon}$ or $t\in\mathbf{X}_{II}\cup\mathbf{X}_{III}$.
\item Apply established martingale theorems.
\end{enumerate}
A diffusion $\overline{u(t)}$ will be taken to be a true martingale if it satisfies the criteria given in the following two theorems.
\begin{thm}
The driftless Ito diffusion
\begin{equation}
\overline{m(t)}=\overline{u(t)}=u_{i}+\int_{t_{i}}^{t}\psi(u(s))d\mathlarger{\mathlarger{\mathscr{B}}}(t)
\end{equation}
is a true martingale $\overline{m(t)}$ if it is square integrable in that its second moments are bounded such that from the Ito isometry
\begin{equation}
\mathbf{E}\bigg\llbracket\|u(t)|^{2}\bigg\rrbracket=\int_{t_{i}}^{t}
\mathbf{E}\llbracket|\psi(u(s)|^{2}\rrbracket ds=\int_{t_{i}}^{t}|\psi(u(s)|^{2}ds<\infty
\end{equation}
which is equivalent to the quadratic variation
\begin{equation}
\bigg\langle \overline{u},\overline{u}\bigg\rangle(t)=\bigg\langle \int_{t_{i}}^{t}\psi(u(s))d\mathlarger{\mathlarger{\mathscr{B}}}(t),\int_{t_{i}}^{t}\psi(u(s))d\mathlarger{\mathlarger{\mathscr{B}}}(t)
\bigg\rangle(t)\equiv\int_{t_{i}}^{t}|\psi(u(s)|^{2}ds<\infty
\end{equation}
If $\mathlarger{\mathlarger{\mathcal{E}}}\big\lbrace\|u(t)|^{2}\big\rbrace=\int_{t_{i}}^{t}\mathlarger{\mathlarger{\mathcal{E}}}\lbrace
|\psi(u(s)|^{2}ds=\int_{t_{i}}^{t}|\psi(u(s)|^{2}ds=\infty$ then the diffusion may or may not be a martingale but it is still a local martingale.
\end{thm}
Details can be found in various works [56-61]. In addition to this condition, one can impose the martingale condition of Dalbaen and H. Shirakawa [63] which is based on necessary and sufficient conditions for the diffusion coefficient $\psi(u(t))$ of a driftless Ito diffusion.
\begin{thm}
Given the driftless Ito diffusion $d\overline{u(t)}=\psi(u(t)d\mathlarger{\mathlarger{\mathscr{B}}}(t)$ for $u(t)\in[1,\infty)$ in the underlying ODE $d\overline{u(t)}=\psi(u(t))dt$, and $\overline{u(t)}\in[u_{i},\infty)$ then $\overline{u(t)}=u_{i}+\int_{t_{i}}^{t}\psi(u(s)d\mathlarger{\mathlarger{\mathscr{B}}}(s)$ is a true martingale if
\begin{equation}
\int_{1}^{\infty}{u(t)du(t)}{|\psi(u(t))|^{-2}}\equiv
\int_{1}^{\infty}{xdx}{|\psi(x))|^{-2}}=\infty
\end{equation}
\end{thm}
See [63 for proofs.
\begin{exam}
An an example of application of these theorems consider the driftless Ito diffusions $(\overline{u(t)}\in[\ell,\infty)$ of the form
\begin{equation}
\overline{u(t)}=\int_{0}^{t}\psi(u(t))d\mathlarger{\mathlarger{\mathscr{B}}}(t)+\ell=
\int_{0}^{t}\exp(-u(t))d\mathlarger{\mathlarger{\mathscr{B}}}(t)+\ell
\end{equation}
and
\begin{equation}
\overline{u(t)}=\int_{0}^{t}\psi(u(t))d\mathlarger{\mathlarger{\mathscr{B}}}(t)+\ell=
\int_{0}^{t}\exp(u(t))d\mathlarger{\mathlarger{\mathscr{B}}}(t)+\ell
\end{equation}
Then for $(\overline{u(t)},u(t))\in[\ell,\infty)$, (5.9) is a true martingale but (5.10) is not. From Thm 5.3
\begin{equation}
\int_{\ell}^{\infty}u(s)\exp(2u(s))du(s)=
\frac{1}{4}\exp(2u(t))(2u(t)-1)\big|_{\ell}^{\infty}=\infty
\end{equation}
To establish square integrability as required by Thm (3.3) it is necessary that
\begin{equation}
\int_{0}^{\infty}|\exp(-2u(t))|ds < \infty
\end{equation}
However,
\begin{equation}
y(t)=\int_{0}^{t}\exp(-2y(s)ds < \infty
\end{equation}
is a formal solution of the ODE
\begin{equation}
\frac{d\mathcal{Y}(t)}{dt}=\exp(-2\mathcal{Y}(t))
\end{equation}
Square integrability is then established by solving the ODE explicitly, which is the case if $y(t)$ is infinite only for infinite $t$. The solution is easily seen to be
\begin{equation}
\frac{1}{2}\exp(2y(t))-\frac{1}{2}\exp(2\ell)=t
\end{equation}
so that $y(t)=\infty$ at $t=\infty$. Hence, the integral (5.13) is finite for all finite $t>0$. Hence, (5.9) is a true square-integrable martingale.
\end{exam}
\begin{exam}
For the diffusion(5.10), Thm (5.3) gives
\begin{equation}
\int_{\ell}^{\infty}u(s)\exp(-2u(s))du(s)=\frac{1}{4}\exp(-u(t))(2u(t)+1)\bigg|_{\ell}^{\infty}<\infty
\end{equation}
which suggest it is not a true martingale. Thm (5.3) requires square integrability such that for all finite t one has
\begin{equation}
\int_{\ell}^{t}exp(2u(s))ds<\infty
\end{equation}
Again, this can be interpreted as the formal solution of a simple ODE so that
\begin{equation}
y(t)=\int_{\ell}^{y(t)}\exp(2y(s))ds+\ell
\end{equation}
is the solution of
\begin{equation}
\frac{dy(t)}{dt}=\exp(2y(s))
\end{equation}
The solution is
\begin{equation}
\tfrac{1}{2}\ell\exp(-2\ell)-\tfrac{1}{2}y(t)\exp(-2y(t))=t
\end{equation}
Then there is a finite time $t=\tfrac{1}{2}\ell\exp(-2\ell)$ at which $y(t)=\infty$ so that the diffusion is not square integrable and is not a  true martingale.
\end{exam}
The third example considers the logistic equation which blows up for a finite t.
\begin{exam}
The nonlinear ODE $du(t)=\psi(u(t))dt=u(t)(\beta+\alpha u(t))dt$ has a singular solution for $\alpha,\beta>0$ and initial value $u(0)$ explodes at $t_{*}=-\tfrac{1}{\beta}\log(\alpha u(0)/3+\alpha u(t))$. The driftless diffusion $d\overline{u}(t)=u(t)(\beta+\alpha u(t))d\mathlarger{\mathlarger{\mathscr{B}}}(t)$ has the Ito integral solution
\begin{equation}
\overline{u(t)}=u(0)+\int_{0}^{t}u(s)(\beta+\alpha u(s))d\mathlarger{\mathlarger{\mathscr{B}}}(t)
\end{equation}
It is a true square-integrable martingale if it satisfies the criteria of Thm(3.3) and Thm(3.4). From Thm(3.3)
\begin{equation}
\int_{0}^{\infty}\frac{du(t)}{u(s)(\beta+\alpha u(t))^{2}}=\bigg|\frac{\frac{\beta}{(\alpha u(t)+\beta)}-\log(\alpha u(t)+\beta)+\log(u(t))}{\beta^{2}}\bigg|_{0}^{\infty}=\infty
\end{equation}
From Thm (3.3), the diffusion is square integrable if
\begin{equation}
\int_{0}^{t}(u(s))^{2}(\beta+\alpha u(s))^{2}ds<\infty
\end{equation}
This is equivalent to $y(t)<\infty$ for all finite t where
\begin{equation}
y(t)=\int_{0}^{t}(y(s))^{2}(\beta+\alpha y(s))^{2}ds
\end{equation}
which is a formal solution to the ODE
\begin{equation}
\frac{dy(t)}{dt}=(y(t))^{2}(\beta+\alpha y(t))^{2}
\end{equation}
The equation is readily integrated to give
\begin{equation}
\left|{-\beta^{-2}\bigg(\frac{\alpha}{\alpha y(t)+\beta}+\frac{1}{y(t)}\bigg)-2\alpha\beta^{-3}\log(\alpha y(t)+\beta)+2\alpha\log(y(t))}{\beta^{3}}\right|_{0}^{y(t)}=t
\end{equation}
However
\begin{equation}
\left|\frac{-\beta\bigg(\frac{\alpha}{\alpha y(t)+\beta}+\frac{1}{y(t)}\bigg)-2\alpha\log(\alpha y(t)+\beta)+2\alpha\log(y(t))}{\beta^{3}}\right|_{0}^{\infty}=\infty=t
\end{equation}
so that $y(t)$ is infinite only for infinite $t$ and so (5.24) is finite for all finite t and therefore (5.23) is finite for all finite t. Hence, (5.21) is a square-integrable martingale
\end{exam}
We now prove that the driftless Ito density function diffusion (5.1) satisfies these theorems and is therefore a true martingale
\begin{thm}
The diffusion $\overline{u(t)}=u_{\epsilon}+\int_{t_{\epsilon}}^{t}\psi(u(s))ds=
u_{\epsilon}+\int_{t_{\epsilon}}^{t}(u(s))^{2}(u(s)-1)^{1/2}d\mathlarger{\mathlarger{\mathscr{B}}}(s)$, with $u(t)\in[1,\infty)$ for the underlying singular nonlinear ODE $du(t)=k^{1/2}u^{2}(t)(u(t)-1)^{1/2}$, and $\overline{u(t)}\in[u_{\epsilon},\infty)$
is a true martingale since
\begin{equation}
\int_{1}^{\infty}\frac{u(t)du(t}{|\psi(u(t)|^{2}}
=\int_{1}^{\infty}\frac{du(t)}{(u(t))^{3}(u(t)-1)}\equiv \int_{1}^{\infty}\frac{dx}{x^{3}(x-1)}
=\infty
\end{equation}
\end{thm}
\begin{proof}
The integral is easily evaluated and does not converge so that
\begin{equation}
\int_{1}^{\infty}\frac{dx}{x^{3}(x-1)}=\frac{1}{2x^{2}}+\frac{1}{x}+\log(1-x)-log(x)\bigg|_{1}^{\infty}
=\infty
\end{equation}
Hence the diffusion $\overline{u(t)}$ is a true martingale.
\end{proof}
\begin{thm}
The diffusion $\overline{u(t)}=u_{\epsilon}+\int_{t_{\epsilon}}^{t}\psi(u(s))ds=
u_{\epsilon}+\int_{t_{\epsilon}}^{t}(u(s))^{2}(u(s)-1)^{1/2}d\mathlarger{\mathlarger{\mathscr{B}}}(s)$, with $u(t)\in[1,\infty)$ for the underlying singular nonlinear ODE $du(t)=k^{1/2}u^{2}(t)(u(t)-1)^{1/2}dt$, and $\overline{u(t)}\in[u_{\epsilon},\infty)$
is a true martingale since it is square integrable with
\begin{equation}
\int_{t)_{\epsilon}}^{t}|\psi(u(s)|^{2}ds=k\int_{t_{\epsilon}}^{t}(u(s))^{4}(u(s)-1)ds<\infty
\end{equation}
and for any $p\ge 2$
\begin{equation}
\int_{t_{\epsilon}}^{t}|{\psi}(u(s))|^{p}ds\equiv
k\int_{t_{\epsilon}}^{t}(u(s))^{2p}(u(s)-1)^{p/2}ds<\infty
\end{equation}
\end{thm}
\begin{proof}
If (5.30) is interpreted as the solution of a nonlinear ODE for $t\ge t_{\epsilon}$ then we can write
\begin{equation}
y(t)=u_{\epsilon}+k^{1/2}\int_{t_{\epsilon}}^{t}y^{4}(s)(y(s)-1)ds
\end{equation}
as the formal solution of an ODE
\begin{equation}
\frac{d}{dt}y(t)=k^{1/2}y^{4}(t)(y(t)-1)
\end{equation}
with initial data $y(t_{\epsilon})=y_{\epsilon}$. If $y(t)<\infty$ for
all $t>t_{\epsilon}$ then (5.30) holds. Integrating
\begin{equation}
\int_{y_{\epsilon}}^{y(t)}d\bar{y}(t)[y^{-4}(y(t)-1))^{-1}=
k^{1/2}\int dy=k^{1/2}|t-t_{\epsilon}|
\end{equation}
which gives
\begin{equation}
(\frac{1}{3y(t)^{3}}+\frac{1}{2y^{2}(t)}+\frac{1}{y(t)}+
2\tanh^{-1}(1-2y(t))-A(y_{\epsilon})=k|t-t_{\epsilon}|
\end{equation}
where $A(y_{\epsilon})=(\frac{1}{3y_{\epsilon}^{3}}+\frac{1}{2y_{\epsilon}^{2}}+
\frac{1}{y_{\epsilon}}+2\tanh^{-1}(1-2y_{\epsilon})$
But
\begin{align}
&\lim_{y(t)\rightarrow\infty}\frac{1}{3y(t)^{3}}+\frac{1}{2y^{2}(t)}+\frac{1}{y(t)}\nonumber\\&
+2\tanh^{-1}(1-2y(t))-A(y_{\epsilon})-A(y_{\epsilon})=-i\frac{1}{2}\pi=k|t-t_{\epsilon}|
\end{align}
since $\tanh^{-1}(-\infty)=+\tfrac{1}{2}\pi i$, so there is no finite real time $t>t_{\epsilon}$ for which one can have $y(t)=\infty$. Hence the rhs is always finite and bounded for any $t>t_{\epsilon}$ so (5.30) is established. Next, equation (5.31) can be interpreted as the solution of the ODE
\begin{equation}
\frac{dy(t)}{dt}=k^{p}y(t)^{2p}(y(t)-1)^{p/2}
\end{equation}
If $\exists$ any $t>t_{\epsilon}$ such that $y(t)=\infty$ then the rhs of (5.30) is also
infinite; otherwise it is finite and hence (5.30) holds. The nonlinear ODE $dy(t)/dt=(y^{2}(t)(y(t)-1)^{1/2})^{p}$ has the formal solution
\begin{equation}
y(t)=y_{\epsilon}+\int_{t_{\epsilon}}^{t}(y^{2}(t)(y(s)-1)^{1/2})^{p}ds
\end{equation}
The explicit solution is
\begin{equation}
\frac{(y(t)-1)^{(1-\frac{p}{2})}(y(t))^{1-2p}~_{2}F_{1}(1,2,-\frac{5p}{2};2(1-p),
y(t))}{2p-1}-A(t_{\epsilon})=k|t-t_{\epsilon}|
\end{equation}
where $_{2}F_{1}$ is the hypergeometric function and where
\begin{equation}
A(t_{\epsilon})=\frac{(y_{\epsilon}-1)^{(1-\frac{p}{2})}(u_{\epsilon})
^{1-2p}~_{2}F_{1}(1,2,-\frac{5p}{2};2(1-p),
u_{\epsilon}}{2p-1}
\end{equation}
But
\begin{equation}
\lim_{y(t)\rightarrow\infty}\frac{(y(t)-1)^{(1-\frac{p}{2})}(y(t))^{1-2p}~_{2}F_{1}(1,2,-\frac{5p}{2};2(1-p),
y(t))}{2p-1}=0
\end{equation}
since $(y(t)-1)^{(1-\frac{p}{2})}(y(t))^{1-2p}=0$ for $y(t)=\infty$. Using the Euler
integral representation of the hypergeometric function
\begin{equation}
_{2}F_{1}(a,b;c;z)=\frac{\Gamma(c)}{\Gamma(b)\Gamma(c-b)}\int_{0}^{1}
\frac{\xi^{b-1}(1-\xi)^{c-b-1}}{(1-\xi z)}d\xi
\end{equation}
where $\Gamma(a)$ us the gamma function. It follows that
\begin{equation}
\lim_{y(t)_\rightarrow\infty}
\frac{(y(t)-1)^{(1-\frac{p}{2})}(y(t))^{1-2p}}{2p-1}
\frac{\Gamma(c)}{\Gamma(b)\Gamma(c-b)}\int_{0}^{1}
\frac{\xi^{1-\frac{5n}{2}}(1-\xi)^{\frac{n}{2}-1}}
{(1-\xi y(t))}d\xi=0
\end{equation}
Hence, $y(t)=\infty$ for a value of t such that $k|t-t_{\epsilon}|=-A_{\epsilon}/k.$ Setting $t=t_{\epsilon}+\epsilon$ with $\epsilon|\ll 1$ and $|\epsilon|>o$, and since $\mathcal{Q}_{\epsilon}>0$ for all $y(t)>1$ and $p\ge 2$ we must have $|\epsilon|=-\mathcal{Q}(t_{\epsilon})/k<0$. This is a contradiction so $y(t)$ cannot be infinite for any finite positive $t\in\mathbf{X}_{II}
\cup\mathbf{X}_{III}$. Hence $y(t)$ is always finite and the rhs is always finite.
\end{proof}
The driftless Ito diffusion $\overline{u(t)}$ satisfies the basic martingale property
\begin{lem}
If $d\overline{u(t)}=\psi(u(t))d\mathlarger{\mathlarger{\mathscr{B}}}(t)$ and $\int_{t_{\epsilon}}^{t}|\psi((s))|^{2}ds<\infty$ and $\int_{t_{\epsilon}}^{t}
\mathlarger{\mathlarger{\mathcal{E}}}\big\llbracket|\psi(u(s))|^{2}\big\rrbracket ds<\infty$, and the latter implies the former via Fubini's theorem, then for a filtration $\mathscr{F}_{t'}$, with $t'<t$, the diffusion $\overline{u(t)}$ is a martingale
\begin{equation}
\mathlarger{\mathlarger{\mathcal{E}}}\bigg\llbracket k\int_{t_{\epsilon}}^{t}(u(s))^{2}(u(s)-1)^{1/2}d\mathlarger{\mathlarger{\mathscr{B}}}(s)\bigg|\mathscr{F}_{t'}\bigg\rrbracket
=k\int_{t_{\epsilon}}^{t'}(u(s))^{2}(u(s)-1)^{1/2}d\mathlarger{\mathlarger{\mathscr{B}}}(s)
\end{equation}
\end{lem}
\begin{proof}
Let $s<t$ then for all $(s,t)\in\mathbf{X}_{II}\cup\mathbf{X}_{III}$ and filtration $\mathscr{F}_{s}$
\begin{equation}
\mathlarger{\mathlarger{\mathcal{E}}}\bigg\llbracket k\int_{s}^{t}|u^{4}(s)(u(s)-1)|d\mathlarger{\mathlarger{\mathscr{B}}}(s)
\bigg\|\mathfrak{F}_{t'}\bigg\rrbracket=0
\end{equation}
Then
\begin{align}
&\mathlarger{\mathlarger{\mathcal{E}}}\big\llbracket(\overline{u(t)}|\mathscr{F}_{t'})\big\rrbracket=
\mathlarger{\mathlarger{\mathcal{E}}}\bigg\llbracket\bigg(k\int_{t_{\epsilon}^{t}}|u^{4}(s)(u(s)-1)|
d\mathlarger{\mathlarger{\mathscr{B}}}(t)\bigg\|\mathscr{F}_{s}\bigg)\bigg\rrbracket\nonumber\\&
=\mathlarger{\mathlarger{\mathcal{E}}}\bigg\llbracket k\int_{t_{\epsilon}^{s}}u^{4}(s)(u(s)-1)d\mathlarger{\mathlarger{\mathscr{B}}}(s)\bigg\rrbracket
+\mathlarger{\mathlarger{\mathcal{E}}}\bigg\llbracket\bigg(k\int_{t'}^{t}u^{4}(s)(u(s)-1)
d\mathlarger{\mathlarger{\mathscr{B}}}(s)\bigg\|\mathscr{F}_{t'}\bigg)\bigg\rrbracket \nonumber\\&
=\mathlarger{\mathlarger{\mathcal{E}}}\bigg\llbracket k\int_{t_{\epsilon}^{t'}}|u^{4}(s)(u(s)-1)|d\mathlarger{\mathlarger{\mathscr{B}}}(\bigg\rrbracket
=\overline{u(s)}
\end{align}
or equivalently
\begin{align}
&\bigg\|(\overline{u(t)}|\mathscr{F}_{t'})\bigg\|_{\mathcal{L}_{1}}=
\bigg\|\bigg(k\int_{t_{\epsilon}^{t}}|u^{4}(s)(u(s)-1)|
d\mathlarger{\mathlarger{\mathscr{B}}}(s)\bigg\|\mathscr{F}_{s}\bigg)\bigg\|_{\mathcal{L}_{1}}\nonumber\\&
=\bigg\|\kappa\int_{t_{\epsilon}^{s}}u^{4}(s)(u(s)-1)d\mathlarger{\mathlarger{\mathscr{B}}}(s)\bigg\|
+\bigg\|\bigg(\kappa\int_{t'}^{t}u^{4}(s)(u(s)-1)
d\mathlarger{\mathlarger{\mathscr{B}}}(s)\bigg\|\mathscr{F}_{t'}\bigg)\bigg\|_{\mathcal{L}_{1}} \nonumber\\&
=\bigg\|\kappa\int_{t_{\epsilon}^{t'}}|u^{4}(s)(u(s)-1)|
d\mathlarger{\mathlarger{\mathscr{B}}}(s)\bigg\|_{\mathcal{L}_{1}}
=\widehat{u}(s)
\end{align}
Hence $\overline{u(t)}$ is a martingale
\end{proof}
The next result utilises the Ito isometry (Appendix) to establish the finiteness of the first moment and the second moment or volatility of $\widehat{u}(t)$.
\begin{lem}
The driftless Ito diffusion $\overline{u(t)}$ is a martingale for all $t>t_{\epsilon}$ and $\overline{u(t_{\epsilon})}=u_{\epsilon}$. The first moment is
\begin{equation}
\mathlarger{\mathlarger{\mathcal{E}}}\bigg\llbracket|\overline{u(t)}|\bigg\rrbracket=u_{\epsilon}
\end{equation}
\begin{enumerate}
\item Since the diffusion $\overline{u(t)}$ is a square-integrable martingale for all
$t\in\mathbf{X}_{II}\cup\mathbf{X}_{III}$ then the volatility or second moment is always finite and bounded
\begin{equation}
\mathlarger{\mathcal{V}}(t)=\mathlarger{\mathlarger{\mathcal{E}}}\bigg\llbracket|\overline{u(t)}|^{2}\bigg\rrbracket
=\mathlarger{\mathlarger{\mathcal{E}}}\big\lbrace|\overline{u(t)}|^{2}\big\rbrace=
\bigg|\int_{t_{\epsilon}}^{t}|\psi(u(s))|^{2}ds\bigg|<\infty
\end{equation}
\end{enumerate}
\end{lem}
\begin{proof}
Using the Riemann-Stieljes sum definition of an Ito integral
\begin{align}
&\mathlarger{\mathlarger{\mathcal{E}}}\bigg\llbracket\overline{u(t)}\bigg\rrbracket =u_{\epsilon}+\int_{t_{\epsilon}}^{t}(u(s))^{2}(u(s)-1)^{1/2}d\mathlarger{\mathlarger{\mathscr{B}}}(t)\nonumber \\&=u_{\epsilon}+\lim_{n\uparrow\infty}\sum_{i=1}^{n}(u(t_{i}^{n}))^{2}(u(t_{i}^{n})-1)^{1/2}
\mathlarger{\mathlarger{\mathcal{E}}}\bigg\llbracket[\mathlarger{\mathlarger{\mathscr{B}}}(T_{i+1}^{n})-\mathlarger{\mathlarger{\mathscr{B}}}(T_{i}^{n})]\bigg\rrbracket=u_{\epsilon}
\end{align}
Next
\begin{equation}
|\overline{u(t)}|^{2}=u_{\epsilon}|^{2}+2u_{\epsilon}\bigg|\int_{t_{\epsilon}}^{t}
{\psi}(u(s))d\mathlarger{\mathlarger{\mathscr{B}}}(s)\bigg|+\bigg|\int_{[t_{\epsilon},t]}
\int_{[t_{\epsilon},t]}|{\psi}(u(s))|^{2}d\mathlarger{\mathlarger{\mathscr{B}}}(s)
d\mathlarger{\mathlarger{\mathscr{B}}}(s)\bigg|
\end{equation}
Taking the expectation and using the Ito isometry
\begin{align}
&\mathlarger{\mathlarger{\mathcal{E}}}\bigg\llbracket|\overline{{u}(t)}|^{2}\bigg\rrbracket=|u_{\epsilon}|^{2}+
\mathlarger{\mathlarger{\mathcal{E}}}\bigg\llbracket \bigg|\int_{[t_{\epsilon},t]}\int_{[t_{\epsilon},t]}|{\psi}(u(s))|^{2}
d\mathlarger{\mathlarger{\mathscr{B}}}(s)
d\mathlarger{\mathlarger{\mathscr{B}}}(s)\bigg|\bigg\rrbracket \nonumber\\&=|u_{\epsilon}|^{2}+\mathlarger{\mathlarger{\mathcal{E}}}\bigg\rrbracket \bigg|\int_{t_{\epsilon}}^{t}|\psi
(u(s))|^{2}dt\bigg|\bigg\rrbracket<\infty
\end{align}
since from Lemma $\int_{t_{\epsilon}}^{t}|{\psi}(u(s))|^{2}ds<\infty$. Note that the result also follows from the stochastic white-noise perturbations $\mathscr{W}(t)$ in that
\begin{align}
&\int_{t_{\epsilon}}^{t}\int_{t_{\epsilon}}^{t}{\psi}(u(s)|\psi(u(s))|
d\mathlarger{\mathlarger{\mathscr{B}}}(s)d\mathlarger{\mathlarger{\mathscr{B}}}(s)=\int_{[t_{\epsilon},t]}
\int_{[t_{\epsilon},t]}|\psi(u(s))||{\psi}(u(s))|\psi(\mathlarger{\mathlarger{\mathscr{W}}}(u)
\mathlarger{\mathlarger{\mathscr{W}}}(v))ds\nonumber\\&=\int_{t_{\epsilon}}^{t}\int_{t_{\epsilon}}^{t}|\psi(u(s))||\psi(u(s))|
\delta(u-v)ds dv\nonumber=\int_{t_{\epsilon}}^{t}\int_{t_{\epsilon}}^{t}|
\psi(u(s))|^{2}\delta(u-v)du dv=\int_{t_{\epsilon}}^{t}|\psi(u(s))|^{2}dt
<\infty
\end{align}
since $\mathcal{E}\llbracket \mathscr{W}(t)\mathscr{W}(s)\rrbracket=\delta(t-s)$
\end{proof}
\subsection{Additional criteria for generic driftless diffusions to be non-exploding or singularity free}
To further strengthen claims that the driftless SDE is globally non-exploding, we briefly consider some additional non-explosion criteria for general driftless Ito diffusions. The following theorem is based on [64] for a general time-dependent driftless Ito diffusion.
\begin{thm}
Given a generic diffusion
\begin{equation}
d\overline{X(t)}=\psi(X(t),t)d\mathlarger{\mathlarger{\mathscr{B}}}(t)
\end{equation}
for $\overline{X(t)}\in[X_{\epsilon},\infty]$ and with underlying probability space $(\Omega,\mathfrak{F},\mathlarger{\mathrm{I\!P}})$ with filtration $\mathfrak{F}_{t}$,and coefficient $\psi:[t_{\epsilon},t)\times\mathbf{R}^{+}\rightarrow^{+}$. We have
$\mathlarger{\mathrm{I\!P}}(X(t_{\epsilon}=X_{\epsilon})=1$. The solution is
\begin{equation}
\overline{X(t)}=X_{\epsilon}+\int_{t_{\epsilon}}^{t}\psi(X(t),t)
d\mathlarger{\mathlarger{\mathscr{B}}}(t)
\end{equation}
Define a possible blowup or explosion time
\begin{equation}
T_{\infty}(X)=\inf\lbrace t>t_{\epsilon}:\overline{X(t)}=\infty\rbrace
\end{equation}
then the solution does not explode if $\mathlarger{\mathrm{I\!P}}(T_{\infty}(X)=\infty)=0$ or $\mathlarger{\mathrm{I\!P}}(T_{\infty}(X)<\infty)=1$. In terms of the time-dependent diffusion coefficient, the criteria for non-explosion are the integrability bounds
\begin{align}
&\int_{-m}^{m}\int_{t_{\epsilon}}^{t}|\psi(X(s),s)|^{2}dsdX(s)<\infty,~\forall t>t_{\epsilon},\forall m\in\mathbf{N}\\&
\int_{-m}^{m}\int_{t_{\epsilon}}^{t}|\psi(X(s),s)|^{-2}dsdX(s)<\infty,~\forall t>t_{\epsilon},\forall m\in\mathbf{N}
\end{align}
\end{thm}
Extensive proof is given in [64]. The main tools utilised within the proof are Krylov-type integrals. Defining a 'stopping time' $T_{n}(X)=\inf\lbrace t\ge 0:|\overline{X(t)}|\ge n\rbrace,~\forall~m\in\mathbf{N}$ and a locally square-integrable function $f:[0,\infty)\rightarrow[0,\infty)=\mathbf{R}^{+}$, $\exists~C>0$ such that for all~$m\in\mathbf{N}$ and $t\ge 0$, one has the Krylov-type bound
\begin{equation}
\mathlarger{\mathlarger{\mathcal{E}}}\left\llbracket\int_{0}^{t\wedge T_{n}(X)}f(\overline{X(s)},s)\psi(X(s),s)ds\right\rrbracket\le C\left(\int_{0}^{t}\int_{-m}^{m}|f(y,s)|^{2}dyds\right)
\end{equation}
where $t\wedge T_{n}(X)=\min\lbrace t,T_{n}(X)\rbrace$.
\begin{rem}
The SDE (5.54) is a specific case of the Levy driven diffusion
\begin{equation}
d\overline{X(t)}=\psi(X(t),t)d\mathlarger{\mathlarger{\mathscr{Z}}}(t)
\end{equation}
with $\overline{X(0)}=X_{0}$. Here,$\mathlarger{\mathlarger{\mathscr{Z}}}(t)$ is a symmetric $\alpha$-stable process with respect to a space $(\Omega,\mathfrak{F},\mathlarger{\mathrm{I\!P}})$ and filtration $\mathfrak{F}_{t}$ starting at $\mathscr{Z}(0)=0$, with index $\alpha\in(0,2]$ and generating function
\begin{equation}
\exp(i\beta(\mathscr{Z}(t)-\mathscr{Z}(s)|\mathfrak(F)_{t})=\exp(-|t-s||\beta|^{\alpha})
\end{equation}
and $\alpha=2$ gives the standard Weiner-Brownian process. Hence, a symmetric process has stationary $\alpha$-stable symmetric increments that are independent of the past for the given filtration. A process $\overline{X(t)}$ that is solution of (5.54) is a weak solution if there is a symmetric stable $\mathscr{Z}(t)$ such that
\begin{equation}
\overline{X(t)}=X_{\epsilon}+\int_{t_{\epsilon}}^{t}\psi(X(t),t)
d\mathlarger{\mathlarger{\mathscr{Z}}}(t)
\end{equation}
\end{rem}
The case for time-homogenous driftless diffusion is given by Englebert and Schmidt [65].
\begin{thm}
Given a time-homogenous generic diffusion
\begin{equation}
d\overline{X(t)}=\psi(X(t))d\mathlarger{\mathlarger{\mathscr{B}}}(t)
\end{equation}
for $\overline{X(t)}\in[X_{\epsilon},\infty]$ and with underlying probability space $(\Omega,\mathfrak{F},\mathlarger{\mathrm{I\!P}})$ with filtration $\mathfrak{F}_{t}$,and coefficient $\psi:\mathbf{R}^{+}\rightarrow\mathbf{R}^{+}$. We have
$\mathlarger{\mathrm{I\!P}}(X(t_{\epsilon}=X_{\epsilon})=1$. The solution is
\begin{equation}
\overline{X(t)}=X_{o}+\int_{0}^{t}\psi(X(t))
d\mathlarger{\mathlarger{\mathscr{B}}}(t)
\end{equation}
Define the sets
\begin{align}
&\mathcal{S}=\left\lbrace x\in\mathbf{R}:\int_{x-m}^{x+m}\psi^{-2}(y)dy=\infty,~\forall m\in\mathbf{N}\right\rbrace\\&
\mathcal{N}=\lbrace x\in\mathbf{R}:\psi(x)=0\rbrace
\end{align}
Then if $\mathcal{S}\subset\mathcal{N}$,$\exists$ a weak non-exploding solution for all
$t>0$.
\end{thm}
This is really tantamount to the local integrability condition
\begin{equation}
\int_{x-m}^{x+m}|\psi(y)|^{-2}dy<\infty,\forall m\in\mathbf{N}
\end{equation}
Such driftless SDEs never explode for any coefficient irrespective of growth or continuity conditions.
\begin{prop}
Given a time-homogenous diffusion $d\overline{X(t)}=\psi(X(t))d\mathlarger{\mathlarger{\mathscr{B}}}(t)$
for $\overline{X(t)}\in[X_{\epsilon},\infty]$ and with underlying probability space $(\Omega,\mathfrak{F},\mathlarger{\mathrm{I\!P}})$ with filtration $\mathfrak{F}_{t}$,and coefficient $\psi:[t_{\epsilon},t)\times\mathbf{R}^{+}\rightarrow\mathbf{R}^{+}$. We have
$\mathlarger{\mathrm{I\!P}}(X(t_{\epsilon}=X_{\epsilon})=1$. Define a possible blowup or explosion time $T_{E}(X)=\inf\lbrace t>t_{\epsilon}:X(t)=\infty\rbrace $ then the solution does not explode if $\mathlarger{\mathrm{I\!P}}(T_{\infty}(X)=\infty)=0$ or $\mathlarger{\mathrm{I\!P}}(T_{\infty}(X)<\infty)=1$. In terms of the coefficient, the criteria for a weak solution with non-explosion are
\begin{align}
&\int_{-m}^{m}|\psi(X(s)|^{2}<\infty,~\forall m\in\mathbf{N}\\&
\int_{-m}^{m}|\psi(X(s)|^{-2}<\infty,~\forall m\in\mathbf{N}
\end{align}
which is equivalent to the E-S result and Theorem 5.12.
\end{prop}
\begin{cor}
The SDE $d\overline{u(t)}=\psi(u(t))d\mathlarger{\mathlarger{\mathscr{B}}}(t)=\kappa^{1/2}(u(t))^{2}(u(t)-1)^{1/2}
d\mathlarger{\mathlarger{\mathscr{B}}}(t)$ does not explode for any $t>t_{\epsilon}$ since
\begin{align}
&\int_{-m}^{m}|\psi(u(s)|^{2}=\kappa\int_{-m}^{m}|u(s)^{4}(u(s)-1)|<\infty,~\forall m\in\mathbf{N}\\&
\int_{-m}^{m}|\psi(u(s)|^{-2}=\kappa\int_{-m}^{m}|u(s)^{-4}(u(s)-1)^{-1}|
<\infty,~\forall m\in\mathbf{N}
\end{align}
\end{cor}
\begin{proof}
The integral in (5.68) is
\begin{equation}
\int_{-m}^{m}u^{4}(u-1)du=\big(\tfrac{1}{6}m^{6}-\tfrac{1}{5}m^{m}\big)-\big(\tfrac{1}{6}m^{6}
+\tfrac{1}{5}m^{5}\big)=-\tfrac{2}{5}m^{5}<\infty
\end{equation}
The integral in (5.69) is
\begin{equation}
\kappa\int_{-m}^{m}|u(s)^{-4}(u(s)-1)^{-1}|=\frac{2}{3m^{3}}+\frac{2}{m}+\log\left(\frac{1-m}{1+m}\right)+\log(-1)<\infty
\end{equation}
and so both integrability criteria are satisfied.
\end{proof}
\subsection{Weak non-exploding or non-singular solutions via the Dubins-Dumis-Schwarz Theorem}
Any continuous-time martingale can be interpreted as a Brownian motion with 'a change of time' [61]
\begin{thm}(Dubins-Dumis-Schwarz) Let $\overline{m(t)}$ be a continuous martingale with $\overline{m(0)}=0$ with quadratic variation $\big\langle \overline{m(t)},\overline{m(t)}\big\rangle$. Define
\begin{equation}
\mathcal{T}_{t}=\inf\bigg\lbrace s:\bigg\langle \overline{m},\overline{m}\bigg\rangle(t)\bigg\rbrace
\end{equation}
then $\mathlarger{\mathscr{B}}(t)=\overline{m(\mathcal{T}_{t})}$ is a Brownian motion with respect to the filtration $\mathfrak{F}_{\mathcal{T}_{t}}$ and
\begin{equation}
\overline{m(t)}=\mathlarger{\mathscr{B}}\left(\bigg\langle\overline{m},\overline{m}\bigg\rangle(t)\right)
\end{equation}
\end{thm}
Proofs can be found in [56,57,61]
\begin{thm}
Let $\Psi(X(t))$ be an adapted, positive, continuous differentiable functional of $X(t)$
with $\Psi(X(0))=0$ and $\Psi:\mathbf{R}^{+}\times\mathbf{R}^{+}\rightarrow\mathbf{R}^{+}$. Then the driftless SDE
\begin{equation}
d\overline{X(t)}=\sqrt{\frac{d}{dt}\Psi(X(t))}d\mathscr{B}(t)
\end{equation}
has a weak non-exploding solution of the form
\begin{equation}
\overline{X(t)}=\mathlarger{\mathlarger{\mathscr{B}}}(\Psi(t))\equiv \mathlarger{\mathlarger{\mathscr{B}}}\bigg(\bigg\langle \overline{X},\overline{X}\bigg\rangle(t)\bigg)
\end{equation}
Hence, the SDE $d\overline{u(t)}=\psi(u(t)d\mathlarger{\mathscr{B}}(t)=
\kappa^{1/2}(u(t))^{2}(u(t)-1)^{1/2}d\mathlarger{\mathlarger{\mathscr{B}}}$ has a solution
\begin{equation}
\overline{u(t)}=\mathlarger{\mathlarger{\mathscr{B}}}
\bigg(\bigg\langle \overline{u},\overline{u}\bigg\rangle(t)\bigg)
\end{equation}
\end{thm}
\begin{proof}
The solution of (5.76) is [61]
\begin{equation}
\overline{X(t)}=\int_{0}^{t}\sqrt{\frac{d}{dt}\Psi(X(s))}d\mathlarger{\mathlarger{\mathscr{B}}}(s)
\end{equation}
so that
\begin{equation}
\int_{0}^{t}\frac{d}{dt}\Psi(X(t))ds=\bigg\langle\overline{X},\overline{X}\bigg\rangle(t)
\end{equation}
Then $X(\Psi^{-1}(X(t))=\mathlarger{\mathscr{B}}(t)$ is a Brownian motion and a weak non-exploding solution is
\begin{equation}
\overline{X(t)}=\mathscr{B}(\Psi(X(t))=\mathscr{B}\left(\bigg\langle\overline{X},\overline{X}\bigg
\rangle(t)\right)
\end{equation}
Setting $\sqrt{\frac{d}{dt}\Psi(X(t))}=\sqrt{\frac{d}{dt}\Psi(u(t))}
=\kappa^{1/2}(u(t))^{2}(u(t)-1)^{1/2}$ then gives (5.78)
\end{proof}
\begin{thm}
Consider a weak solution of
\begin{equation}
d\overline{X(t)}=\psi(X(t))d\mathlarger{\mathlarger{\mathscr{B}}}(t)
\end{equation}
with $\psi:\mathbf{R}^{+}\times\mathbf{R}^{+}\rightarrow\mathbf{R}^{+}$ such that
\begin{align}
&\mathlarger{\mathlarger{\mathscr{F}}}(t)=\int_{0}^{t} \frac{ds}{|\psi(\mathlarger{\mathlarger{\mathscr{B}}}(s))|^{2}}<\infty,~\forall~t<\infty
\\&
\mathlarger{\mathlarger{\mathscr{F}}}(\infty)
=\int_{0}^{\infty} \frac{ds}{|\psi(\mathlarger{\mathlarger{\mathscr{B}}}(s))|^{2}}<\infty,~\forall~t=\infty
\end{align}
so that $\mathlarger{\mathlarger{\mathscr{F}}}(t)$ is continuous and strictly increasing to
$\mathlarger{\mathlarger{\mathscr{F}}}(t)=\infty$. Then $\mathcal{T}_{t}=\mathlarger{\mathlarger{\mathscr{F}}}^{-1}(t)$ and a weak non-exploding solution of (5.82) is
\begin{equation}
\overline{X(t)}=\mathlarger{\mathlarger{\mathscr{B}}}(T_{t}))
\end{equation}
and so a weak non-exploding solution of $d\overline{u(t)}=\psi(u(t))\mathlarger{\mathscr{B}}(t)$
is
\begin{equation}
\overline{u(t)}=\mathlarger{\mathlarger{\mathscr{B}}}(T_{t}))
\end{equation}
\end{thm}
\begin{proof}
If $\overline{X(t)}=\mathlarger{\mathlarger{\mathscr{B}}}(\mathcal{T} _{t}))=
\mathlarger{\mathlarger{\mathscr{B}}}(\mathscr{F}^{-1}(t))$ then
\begin{align}
&d\overline{X(t)}\equiv d\mathlarger{\mathlarger{\mathscr{B}}}(\mathcal{T} _{t}))\equiv d\mathlarger{\mathlarger{\mathscr{B}}}(\mathscr{F}^{-1}(t))=
\frac{1}{\sqrt{\frac{d}{dt}\int_{0}^{\mathcal{T}_{t}}\frac{ds}{|\psi\mathscr{B}(s)|^{2}}}}d \mathlarger{\mathlarger{\mathscr{B}}}(t)\nonumber\\&
=\sqrt{\frac{1}{\psi^{-2}(\mathlarger{\mathlarger{\mathscr{B}}}(\mathcal{T}_{t}))}}d \mathlarger{\mathlarger{\mathscr{B}}}(t)=\psi(\mathlarger{\mathlarger{\mathscr{B}}}
(\mathcal{T}_{t}))d\mathlarger{\mathlarger{\mathscr{B}}}(t)
\end{align}
\end{proof}
\subsection{Conservation of energy}
Finally, it follows that the energy conservation condition still holds.
\begin{prop}
The energy conservation equations for the matter comprising the collapsing fluid star, are satisfied over the comoving collapse interval $\mathbf{C}=[0,t_{*}]$. For the stochastic extension, if the matter density diffusion $\overline{u(t)}$ is established as a martingale for $t>t_{\epsilon}$, then the energy conservation equations are still satisfied on average over both partitions $\mathbf{X}_{I}~\cup~\mathbf{X}_{II}=[0,t_{\epsilon}]~\cup~[t_{\epsilon},t_{*}]$ and $\mathbf{R}^{+}=\mathbf{X}_{I}\cup(\mathbf{X}_{II}\cup\mathbf{X}_{III}) =[0,t_{\epsilon}]\cup([t_{\epsilon},t_{*}]\cup[t_{*},\infty))$ where $t_{*}=
t_{\epsilon}+\epsilon=\pi/2\kappa^{1/2}$ is the comoving blowup proper time
and $|\epsilon|\ll 1$
\end{prop}
\begin{proof}
Given the energy-momentum tensor $T^{\mu\nu}$ for the pressureless matter
comprising the star, the conservation of energy requires that
\begin{equation}
0=\bm{\mathrm{D}}_{\mu}\bm{\mathrm{T}}^{\mu}_{t}
=-\partial_{t}\rho(t)-\rho(t)(\dot{{X}}(2{X})^{-1}
+\dot{{Y}}{Y}^{-1})
\end{equation}
The energy conservation condition is then $\partial_{t}(\rho(t)YX^{1/2})=0$. Letting $X(r,t)=u^{-2}(t)f(r), Y(r,t)=u^{-2}(t)r^{2}$
\begin{equation}
\partial_{t}(\rho(t)YX^{1/2})=\partial_{t}(\rho(t)u^{-3}(t)|f(r)|^{1/2}r^{2})=0
\end{equation}
or
\begin{equation}
\partial_{t}(\rho(t)|u(t)|^{-3})=0
\end{equation}
Over the partition $\mathbf{X}_{I}\cup\mathbf{X}_{II}\cup\mathbf{X}_{III}$, the conservation equation is
\begin{equation}
\mathcal{C}(\mathbf{X}_{I})\partial_{t}(\rho(t)|u(t)|^{-3})
+\mathcal{C}(\mathbf{X}_{II}\cup\mathbf{X}_{III}\partial_{t}(\overline{\rho(t)}|\overline{u(t)}|^{-3})=0
=0
\end{equation}
Now since $u(t)=(\rho(t)/\rho_{o})^{1/3}$ and $\partial_{t}(\overline{\rho(t)}|\overline{u(t)}|^{-3})$, (5.91) becomes
\begin{align}
&\mathcal{C}(\mathbf{X}_{I})\partial_{t}(\rho_{o}|u(t)|^{3}|u(t)|^{-3})
+\mathcal{C}(\mathbf{X}_{II}\cup\mathbf{X}_{III})
\partial_{t}(\rho_{o}|\overline{u(t)}|^{3}|\overline{u(t)}|^{-3})\nonumber\\&
=\mathcal{C}(\mathbf{X}_{I})\partial_{t}(\rho_{o})
+\mathcal{C}(\mathbf{X}_{II}\cup\mathbf{X}_{III})
\partial_{t}(\rho_{o})
=0
\end{align}
or simply
\begin{equation}
\mathcal{C}(\mathbf{X}_{I})\rho_{o}
+\mathcal{C}(\mathbf{X}_{II}\cup\mathbf{X}_{III})
\rho_{o}\equiv \rho=C=const.
\end{equation}
\end{proof}
\section{Inequalities and boundedness in probability}
If the stochastic density function $\overline{{u}(t)}$ is established as a martingale on $\mathbf{X}_{II}\bigcup\mathbf{X}_{III}$ , then the Doob inequalities and upcrossing theorem are immediately applicable[56-61]. One can then establish that the probability of blowup or singularity in $\overline{u}(t)$ is zero or $\mathlarger{\mathrm{I\!P}}[\sup_{t<T}|u(t)|
=\infty]=0$ for all $T\in \mathbf{X}_{II}\cup\mathbf{X}_{III}$. We begin
with the following definitions
\begin{defn}
Give the process $\overline{{u}(t)}$ with $\mathlarger{\mathlarger{\mathcal{E}}}\big\llbracket (\overline{u(t)}|\mathscr{F}_{s})\big\rrbracket=\overline{u(s)}$, the random time $\tau:\Omega\rightarrow[0,\infty)$ is a stopping time with respect to the filtration $\mathfrak{F}_{t\ge t_{\epsilon}}$ if $\lbrace \tau\le t\rbrace\in\mathfrak{F}_{t}$ for all $t>t+_{\epsilon}$; that is, a random time such that any event for $\tau<t$ should be within $Y_{t}$. For example, the first entry times of a process into a 'cell' $\mathcal{C}=[u_{\alpha},u_{\beta}] \subset[u_{\epsilon},\infty]$ with $u_{\beta}>u_{\alpha}$,$\tau_{\mathcal{C}}^{0}=\inf\lbrace t\ge 0:\overline{u(t)}\in\mathcal{C}\rbrace $
\end{defn}
\begin{defn}
The "hitting time" for the density function diffusion $|\overline{u}(t)|$ to hit a
given value or level $|u|>0$ and $u\in\mathcal{C}$ from some initial data$u_{\epsilon}$ $ \tau_{\mathcal{C}}^{a}=\inf\big\lbrace t>t_{\epsilon}:|\overline{u(t)}|=u\in\mathcal{C}\big\rbrace $with an expected value $ \mathlarger{\mathlarger{\mathcal{E}}}[\tau_{(\alpha)}]=\mathlarger{\mathlarger{\mathcal{E}}}[\inf\lbrace t>t_{\epsilon}:|\overline{u
(t)}|=\alpha\big]$. There are associated probabilities: $\mathlarger{\mathrm{I\!P}}[\tau_{\alpha}<u]$ for any $u>0$, $\mathlarger{\mathrm{I\!P}}[T_{\alpha}|<\infty]$,
$\mathlarger{\mathrm{I\!P}}[T<\infty]$,
$\mathlarger{\mathrm{I\!P}}[|\widehat{u}(t)|\le |u|]$ and $\mathlarger{\mathrm{I\!P}}[|\widehat{T}(t)|\le \infty]$.
\end{defn}
\begin{lem}(Optimal Stopping Theorem)
Let $\tau$ be a stopping time or Markov time. If $\overline{u(t)}$
is a martingale for $t\ge t_{\epsilon}$ then $\overline{u(\tau \wedge t)}$ is also a martingale, where $t\wedge\tau\equiv \min\lbrace t,\rbrace$.
\end{lem}
\begin{proof}
\begin{align}
&\mathlarger{\mathlarger{\mathcal{E}}}\bigg\llbracket u\left(t+1)\bigwedge\tau\right)|\mathscr{F}_{t}\bigg\rrbracket
=\mathlarger{\mathlarger{\mathcal{E}}}\bigg\llbracket[\mathcal{C}(\tau\le t)
\overline{u(\tau)}+\mathcal{C}(\tau>t)\bigg(\overline{u(t+1)}\bigg|\mathscr{F}_{t}\bigg)\bigg\rrbracket\nonumber\\&
=\mathlarger{\mathcal{C}}(\tau\le t)\overline{u(\tau)}+\mathlarger{\mathcal{C}}(\tau>t)\mathlarger{\mathrm{I\!P}}\bigg\llbracket
\bigg(\overline{u(t+1)}\bigg|\mathscr{F}_{t}\bigg)\bigg\rrbracket\nonumber\\&=\mathlarger{\mathcal{C}}(\tau\le t)\overline{u(\tau)}+\mathlarger{\mathcal{C}}(\tau>t)\overline{u(t)}=u\left(t\bigwedge \tau\right)
\end{align}
\begin{defn}
Given a random variable or stochastic process $\overline{X}$ then the Markov inequality states that for all $\zeta>0$
\begin{equation}
\mathlarger{\mathrm{I\!P}}\bigg[\bigg|X\bigg|\ge \zeta\bigg]\le \frac{1}{\zeta}\mathlarger{\mathlarger{\mathcal{E}}}\bigg\llbracket\bigg|X\bigg|\bigg\rrbracket
\end{equation}
\end{defn}
\end{proof}
The following main theorem gives the probability of blow-up or singularity formation as zero
\begin{thm}
Given the martingale or driftless Ito diffusion
\begin{equation}
\overline{u(t)}=u_{\epsilon}+k^{1/2}\int_{t_{\epsilon}}^{t}|u(s)|^{2}|u(s)-1|^{1/2}
d\mathlarger{\mathlarger{\mathscr{B}}}(s)
\end{equation}
for all $t\in\mathbf{X}_{II}\cup\mathbf{X}_{III}=[0,\infty)$ or $t\ge t_{\epsilon}$ and initial data $u_{\epsilon}$ then $\overline{u(s)}=(\overline{u(t)}|\mathscr{F}_{s})$. The tower propery for $t^{\prime}>t>s$ is
\begin{equation}
\overline{u(s)}={{\mathcal{E}}}\bigg\llbracket(\overline{u(t)}|\mathscr{F}_{s})\bigg\rrbracket
={{\mathcal{E}}}\bigg\llbracket\bigg({{\mathcal{E}}}\bigg\llbracket\bigg(|\overline{u(t')}
\bigg|\mathscr{F}_{t}\bigg)\bigg\rrbracket\bigg|\mathscr{F}_{s}\bigg)\bigg\rrbracket
={{\mathcal{E}}}\bigg\llbracket\bigg(\overline{u(t')}\bigg|\mathscr{F}_{s}\bigg)\bigg\rrbracket
\end{equation}
and let
\begin{equation}
\overline{u^{*}(t)}=\sup_{t_{\epsilon}\le s\le t}(|\overline{u(s)}|
\end{equation}
Then for $t\in\mathbf{X}_{II}\cup\mathbf{X}_{III}$ or $t>t_{\epsilon}$:
\begin{enumerate}
\item For any $Q\ge0$ and $Q\ge u_{\epsilon}$ the maximal inequality states that
\begin{align}
\mathlarger{\mathrm{I\!P}}\bigg[|u^{*}(t)|)\ge \zeta\bigg]&\le \frac{1}{\zeta}\mathlarger{\mathlarger{\mathcal{E}}}\bigg\llbracket\bigg(\overline{|u(t)|}:|u^{*}(t)|\ge \zeta\bigg)\bigg\rrbracket\\&
\le \frac{1}{\zeta}\mathlarger{\mathlarger{\mathcal{E}}}\bigg\llbracket\overline{|u(t)|}\bigg\rrbracket\equiv \frac{u_{\epsilon}}{\zeta}
\end{align}
or equivalently as $\mathcal{L}_{1}$ norms
\begin{align}
\mathlarger{\mathrm{I\!P}}\bigg[|u^{*}(t)|)\ge \zeta\bigg]&\le \frac{1}{\zeta}\bigg\|\bigg(\overline{|u(t)|}:|u^{*}(t)|\ge \zeta\bigg)\bigg\|\\&
\le \frac{1}{\zeta}\bigg\|\overline{|u(t)|}\bigg\|\equiv \frac{u_{\epsilon}}{\zeta}
\end{align}
\item For $|\overline{u(t)}|^{p}$
\begin{equation}
\big\|\overline{u(t)}\big\|_{\mathcal{L}_{p}}\equiv\bigg(\mathlarger{\mathlarger{\mathcal{E}}}\bigg\llbracket|
\overline{u(t)}|^{p}\bigg\rrbracket\bigg)^{1/p}\le
 p(p-1)^{-1}\mathlarger{\mathlarger{\mathcal{E}}}(|\overline{u(t)}|^{p})^{1/p}
\end{equation}
\item For eternity from $t=t_{\epsilon}$ , the probability of blow up or density function singularity formation is then zero
\begin{equation}
\lim_{\zeta\uparrow\infty}\mathlarger{\mathrm{I\!P}}\big(|u^{*}(t)|\ge \zeta\big)=\mathlarger{\mathrm{I\!P}}\big(|u^{*}(t)|\big)=\infty\big)=0
\end{equation}
\end{enumerate}
\end{thm}
\begin{proof}
Fix $n\ge 1$ with $n\in\mathbf{Z}$ and let $0\le \xi\le n$. Using the martingale diffusion
\begin{equation}
\overline{u(t+\epsilon)}=u_{\epsilon}+k^{1/2}\int_{t_{\epsilon}}^{t+t_{\epsilon}}|u(s)|^{2}|u(s)-1|^{1/2}
d\mathlarger{\mathlarger{\mathscr{B}}}(s)
\end{equation}
\begin{align}
\mathlarger{\mathrm{I\!P}}\big(\max_{0\le \xi\le n}|u(\tfrac{\xi}{n}t+t_{\epsilon})&\ge \zeta\big)=\sum_{\ell=0}^{n}\mathlarger{\mathrm{I\!P}}\bigg[\bigg|\overline{u(\tfrac{\ell}{n}t+t_{\epsilon})}|\ge \zeta:\max_{0\le\xi\le \ell}\big|\overline{u(\tfrac{\ell}{n}t+t_{\epsilon})}\big|<\zeta\bigg]\nonumber\\&\le\frac{1}{\zeta}
\sum_{\ell=0}^{n}\mathlarger{\mathlarger{\mathcal{E}}}\bigg\llbracket\bigg(\big|\overline{u(\tfrac{\ell}{n}t+t_{\epsilon})}\big|:\big|\overline{u(\tfrac{\ell}{n}t+t_{\epsilon})}\big|\ge \zeta,\max_{0\le\xi\le\ell}\big|\overline{u(\tfrac{\ell}{n}t+t_{\epsilon})}\big|<\zeta\bigg)\bigg\rrbracket\nonumber
\\&\le\frac{1}{\zeta}\sum_{\ell=0}^{n}\mathlarger{\mathlarger{\mathcal{E}}}\bigg\llbracket
\mathlarger{\mathlarger{\mathcal{E}}}\bigg\llbracket\big|\overline{u(t+t_{\epsilon})}
\big|\mathscr{F}_{\tfrac{\ell}{n}t+t_{\epsilon}}\big)\bigg\rrbracket:\big|\overline{u(\tfrac{\ell}{n}t+t_{\epsilon})}\big|\ge \zeta,\max_{0\le\xi\le\ell}\big|\overline{u(\tfrac{\ell}{n}t+t_{\epsilon})}\big|<\zeta\bigg\rrbracket
\nonumber\\&\le\frac{1}{\zeta}\sum_{\ell=0}^{n}\mathlarger{\mathlarger{\mathcal{E}}}\bigg\llbracket\bigg(\big|\overline{u(t+t_{\epsilon})}\big)
:\underbrace{\big|\overline{u(\tfrac{\ell}{n}t+t_{\epsilon})}\big|\ge \zeta,\max_{0\le\xi\le \ell}\big|\overline{u(\tfrac{\ell}{n}t+t_{\epsilon})}\big|<\zeta\bigg)}_{\in\mathscr{F}_{(\ell/n)t+t_{\epsilon}}}\bigg\rrbracket\nonumber\\&
\le\frac{1}{\zeta}\mathlarger{\mathlarger{\mathcal{E}}}\bigg\llbracket\bigg(\big|\overline{u(t+t_{\epsilon})}\big)
:\max_{0\le\xi\le \ell}\big|\overline{u(\tfrac{\ell}{n}t+t_{\epsilon})}\bigg|\ge \zeta\bigg)\bigg\rrbracket
\nonumber\\&\le\frac{1}{\zeta}\mathlarger{\mathlarger{\mathcal{E}}}\bigg\llbracket\bigg(\big|\overline{u(t+t_{\epsilon})}\big)
:\big|\overline{u^{*}(t+t_{\epsilon})}\big|\ge \zeta\bigg)\bigg\rrbracket
\end{align}
Right-continuity of $\overline{u(t)}$ ensures that
\begin{equation}
\lim_{n\uparrow\infty}\max_{t_{\epsilon}\le \xi\le n}\big|\overline{u(\tfrac{\xi}{n}t+t_{\epsilon})}|)=u^{*}(t+t_{\epsilon}=\sup u(s+t_{\epsilon})
\end{equation}
Then for all $\delta>0$
\begin{align}
&\mathlarger{\mathrm{I\!P}}\big(|u^{*}(t+t_{\epsilon})|\ge Q) \le \lim_{n\uparrow\infty}\mathlarger{\mathrm{I\!P}}\big(\max_{t_{\epsilon}\le \xi\le n}|u(\tfrac{\xi}{n}t+t_{\epsilon}|\ge Q-\delta\big)\nonumber\\&
\lim_{n\uparrow\infty}\frac{1}{Q-\delta}\mathlarger{\mathlarger{\mathcal{E}}}\big(\big|\overline{u(t-t_{\epsilon})}|:u^{*}(t+t_{\epsilon})\ge Q\big)
\end{align}
Letting $\delta\rightarrow 0$ and resetting $t=t+t_{\epsilon}$ then gives (6.6). To prove (6.8)
\begin{align}
\mathlarger{\mathlarger{\mathcal{E}}}\big(|u^{*}(t)|^{p})&=\mathlarger{\mathlarger{\mathcal{E}}}
\bigg(p\int_{0}^{\overline{u^{*}(t)}}\vartheta^{p-1}d\vartheta\bigg)
\equiv\mathlarger{\mathlarger{\mathcal{E}}}\bigg(p\int_{0}^{\infty}\mathlarger{\mathcal{C}}(\vartheta\le |u^{*}(t)|)\vartheta^{p-1}d\vartheta\bigg)\nonumber\\&\equiv\bigg(p\int_{0}^{\infty}\mathlarger{\mathlarger{\mathcal{E}}}\mathscr{U}(\vartheta\le |u^{*}(t)|)\vartheta^{p-1}d\vartheta\bigg)\le p\int_{0}^{\infty}\frac{1}{\vartheta}\mathlarger{\mathlarger{\mathcal{E}}}\big(\overline{|u(t)|}; |\overline{u^{*}(t)}|\ge \vartheta)\vartheta^{p-1}d\vartheta\nonumber\\&
p\mathlarger{\mathlarger{\mathcal{E}}}\bigg(\int_{0}^{\overline{u^{*}(t)}}\vartheta^{p-2}d\vartheta
|\overline{u(t)}|\bigg)\le p(p-1)^{-1}\mathlarger{\mathlarger{\mathcal{E}}}\big(\overline{u^{*}(t)}|^{p-1}\nonumber\\&
=p(p-1)^{-1}\mathlarger{\mathlarger{\mathcal{E}}}\big(|\overline{u(t)}|^{p})^{1/p}\mathlarger{\mathlarger{\mathcal{E}}}\big(|\overline{u^{*}(t)}|^{p}\big)^{(p-1)/p}
\end{align}
using Holders inequality and so (6.8) is established. From (6.6) it follows that
\begin{equation}
\mathlarger{\mathrm{I\!P}}\big(|u^{*}(t)|\big)=\infty)=0
\end{equation}
\end{proof}
The following theorem establishes that the expected number of 'upcrossings' of the
density function diffusion $\overline{u(t)}$ across a subinterval or "cell"
$\mathcal{C}=[u_{\alpha},u_{\beta}]\subset[u_{\epsilon},\infty)$ is finite
for all $t\in\mathbf{X}_{II}\cup\mathbf{X}_{III}$, and that the number of
upcrossings on a semi-infinite interval $\mathlarger{\mathbf{I}}_{\alpha\infty}
=[u_{\alpha},\infty]$ is always zero--on other words, $\mathlarger{u(t)}(t)$ cannot reach infinity for any $t>t_{\epsilon}$. This is essentially a Doob upcrossing inequality which will hold if $\overline{u(t)}$ is a martingale.
\begin{defn}
Let $\mathcal{C}=[u_{\alpha},u_{\beta}]\subset[u_{\epsilon},\infty)$ be a cell or slab of thickness $|u_{\beta}-u_{\alpha}|$ with $u_{\alpha}<u_{\beta}$. Let $\mathlarger{\mathbf{J}}\subset\mathlarger{\mathbf{I}}$ with $s_{1}<t_{1}<...<s_{n}<t_{n}\le t$ in $\mathlarger{\mathbf{J}}$. Then
\begin{align}
&\bm{\mathcal{U}}[\alpha,\beta:t,\mathbf{J}]=\sup\big\lbrace n:u(s_{1})<u_{\alpha},u(t_{1})>u_{\beta},...,u(s_{n})<u_{\alpha},u(t_{n})>u_{\beta}\nonumber\\&
~for~s_{1}<t_{1}<...<s_{n}<t_{n}\le t\in \mathlarger{\mathbf{J}}\big\rbrace
\end{align}
Then $\bm{\mathcal{U}}[\alpha,\beta:t,\mathbf{I}]\equiv\bm{\mathcal{U}}[\alpha,\beta:t]$ is the number of upcrossings at time $t$
\end{defn}
\begin{thm}
Let $\overline{u(t)}$ be the density function diffusion or martingale solution for initial
data $u_{\epsilon}=[t=t_{\epsilon},u_{\epsilon}<0,\mathlarger{\mathlarger{\mathscr{B}}}(t)_{\epsilon}=0]$
so that $\overline{u(t)}\in[|u_{\epsilon}|,\infty)$. Let $[|u_{\alpha}|,|u_{\beta}|]
\subset[|u_{\epsilon}|,\infty)\subset\mathbf{R}^{+}$ be a finite interval with $|u_{\alpha}|<|u_{\beta}|$. Now let $\bm{\mathcal{U}}(\alpha,\beta:t)$ denote the number of 'upcrossings' of the interval or cell $\mathcal{C}=[|u_{\alpha}|,|u_{\beta}|]$, which is the number of times that the density function diffusion $|\overline{u(t)}|$ has passed from below $|u_{\alpha}|$ to above $u_{\beta}|$ at some $t\in\mathbf{X}_{II}\bigcup\mathbf{X}_{III}$. Then if $\overline{u(t)}$ is a martingale
\begin{enumerate}
\item The number of upcrossings through the cell $\mathcal{C}
$ is finite for all $t\in\mathbf{X}_{I}\cup\mathbf{X}_{III}$ so that the diffusion randomly 'oscillates' through $\mathlarger{\mathlarger{\mathcal{E}}}$.
\begin{align}
&\mathlarger{\mathlarger{\mathcal{E}}}\bigg\llbracket \bm{\mathcal{U}}(\alpha,\beta:t)\bigg\rrbracket~\le~\frac{1}{|u_{\beta}-u_{\alpha}|}
\mathlarger{\mathlarger{\mathcal{E}}}\bigg\llbracket(|\overline{u(t)}|+|u_{a}|)\bigg\rrbracket\nonumber\\&=
\frac{1}{|u_{\beta}-u_{\alpha}}|\mathlarger{\mathlarger{\mathcal{E}}}\bigg\llbracket|u_{\epsilon}|\bigg\rrbracket+\frac{k^{1/2}}
{|u_{\beta}-u_{\alpha}|}|\mathlarger{\mathlarger{\mathcal{E}}}\bigg\llbracket\bigg(\int_{t_{\epsilon}}^{t}
{\psi}(u(s))d\mathlarger{\mathlarger{\mathscr{B}}}(s)+u_{\alpha}\bigg)\bigg\rrbracket<\infty
\end{align}
or equivalently
\begin{align}
&\bigg\|\bm{\mathcal{U}}(\alpha,\beta:t)\bigg\|_{\mathcal{L}_{1}}~\le~\frac{1}{|u_{\beta}-u_{\alpha}|}
\bigg\|(|\overline{u(t)}|+|u_{a}|)\bigg\|_{\mathcal{L}_{1}}\nonumber\\&=
\frac{1}{|u_{\beta}-u_{\alpha}}|\bigg\|u_{\epsilon}|\bigg\|_{\mathcal{L}_{1}}+\frac{k^{1/2}}
{|u_{\beta}-u_{\alpha}|}|\bigg\|\int_{t_{\epsilon}}^{t}
{\psi}(u(s))d\mathlarger{\mathlarger{\mathscr{B}}}(s)+u_{\alpha}\bigg\|_{\mathcal{L}_{1}}<\infty
\end{align}
\item The limit $|u_{\beta}|\rightarrow\infty$ gives a semi-infinite
    interval $[|u_{\alpha}|,\infty)$ so that the number of upcrossings from below
    $u_{\alpha}$ to infinity at any $t\in\mathbf{X}_{II}
    \cup\mathbf{X}_{III}$, is zero giving
\begin{equation}
\mathlarger{\mathlarger{\mathcal{E}}}(s)\bigg\llbracket\bm{\mathcal{U}}(\alpha\rightarrow \infty:t)\bigg\rrbracket~\le~
\lim_{|u_{\beta}|\rightarrow\infty}
\frac{1}{|u_{\beta}-u_{\alpha}|}(\mathlarger{\mathlarger{\mathcal{E}}}\bigg\llbracket|
\overline{u(t)}|\bigg\rrbracket|=0
\end{equation}
or equivalently
\begin{equation}
\bigg\|\bm{\mathcal{U}}(\alpha\rightarrow \infty:t)\bigg\|_{\mathcal{L}_{1}}~\le~
\lim_{|u_{\beta}|\rightarrow\infty}
\frac{1}{|u_{\beta}-u_{\alpha}|}\bigg\|\overline{u(t)}|\bigg\|_{\mathcal{L}_{1}}=0
\end{equation}
so that for the initial data, the density diffusion $\overline{u(t)}$ never blows up for
any finite $t\in\mathbf{X}_{II}
\cup\mathbf{X}_{III}$.
\end{enumerate}
\end{thm}
\begin{proof}
Define the 'hitting times' $\mathfrak{T}_{n}=\inf\lbrace t>{T}_{n-1}:|\overline{u(t)}|
\le |u_{\alpha}|\rbrace$ and ${T}_{n}=\inf\lbrace t>{S}_{n}:|u(t)|
\ge |u_{\beta}|\rbrace$ with ${T}_{0}=0$ and $n\in\mathbf{N}$. Construct a stochastic process \begin{equation}
{S}(t)=\sum_{n}\big(\big|\overline{u(t\bigwedge{T}_{n})}\big|-
\big|\overline{u(t\bigwedge {S}_{n})}\big|\big)
\end{equation}
Beginning with the first time ${T}_{1}$ that
$|u(t)|\le|u_{\alpha}|$, the process $\overline{u(t)}$ evolves
via increments of $\overline{u(s)}$ until the the first time ${T}_{1}\ge|u_{\beta}|$. The
process repeats if again the density function diffusion $|\overline{u(t)}|$ falls
below $u_{\alpha}|$ or $|\overline{u(t)}|\le |u_{\alpha}|$ at time
$\mathfrak{T}_{2}$ and so on. All terms in the summation vanish for sufficiently large $n$ that $\mathfrak{T}_{n}>t$ so that there are at most $\bm{\mathcal{U}}(\alpha,\beta; t)$ terms that are non-zero. Then
\begin{eqnarray}
\overline{u(t)}=\sum_{n}\big(\overline{u(t\bigwedge{T}_{n})}-\overline{u
(t\bigwedge{\mathcal{S}}_{n})}\big)\ge\big(u_{\beta}-u_{\alpha}\big)
\bm{\mathcal{U}}(\alpha; t)\big(\overline{u(t)}-\big|u_{\alpha}\big)
\end{eqnarray}
so that
\begin{align}
&\mathlarger{\mathlarger{\mathcal{E}}}\bigg\llbracket\overline{u(t)}\bigg\rrbracket\ge
\big(u_{\beta}-u_{\alpha}\big)\mathlarger{\mathcal{E}}\big(\bm{\mathcal{U}}(\alpha,\beta;t)
\big)+\mathlarger{\mathlarger{\mathcal{E}}}\bigg\llbracket(u(s)-u_{\alpha}\bigg\rrbracket\nonumber\\&\ge\big(u_{\beta}-u_{\alpha}\big)
\mathlarger{\mathlarger{\mathcal{E}}}\bigg\llbracket
\bm{\mathcal{U}}(\alpha,\beta;t)-\mathlarger{\mathlarger{\mathcal{E}}}(u(s)-u_{\alpha}\bigg\rrbracket\nonumber\\&\ge(u_{\beta}-u_{\alpha})
\mathlarger{\mathlarger{\mathcal{E}}}\bm{\mathcal{U}}(\alpha,\beta:t)-\mathlarger{\mathlarger{\mathcal{E}}}
(\overline{u(t)}|-u_{\alpha})
\end{align}
Since $\overline{u(t)}$ is also a martingale then $\mathlarger{\mathlarger{\mathcal{E}}}\big\llbracket (\overline{u(t)}\big\rrbracket=\mathlarger{\mathlarger{\mathcal{E}}}(u)_{o})$. It follows that
\begin{equation}
\mathlarger{\mathlarger{\mathcal{E}}}\bigg\llbracket \bm{\mathcal{U}}(\alpha, \beta:t)\bigg\rrbracket~\le~\frac{1}{u_{\beta}-u_{\alpha}}
\big(\mathlarger{\mathlarger{\mathcal{E}}}\bigg\llbracket|\overline{u(t)}\big|\bigg\rrbracket+|u_{\alpha}
\big|)\big)
\end{equation}
Taking the limit as $|u_{\beta}|\rightarrow\infty$.
\begin{eqnarray}
\bm{\mathcal{U}}(\alpha,\infty:t))=\lim_{u_{\beta},\infty}\bm{\mathcal{E}}(\alpha,\beta; t))~\le~\lim_{u_{\beta}\rightarrow\infty}
\frac{1}{u_{\beta}-u_{\alpha}}\big(\mathlarger{\mathlarger{\mathcal{E}}}\bigg\llbracket
\big|\overline{u(t)}\big|\bigg\rrbracket+\big|u_{\alpha}\big|\big)=0
\end{eqnarray}
This then gives $\mathlarger{\mathlarger{\mathcal{E}}}\big\llbracket\bm{\mathcal{U}}(\alpha,\infty:t)
\big\rrbracket=0$ for any finite $ t\in\mathbf{X}_{II}\bigcup\mathbf{X}_{III}$ so a blowup never occurs.
\end{proof}
{\allowdisplaybreaks
\subsection{Asymptotic behavior}
An important result is that if the  process $\overline{u(t)}$ is a uniform-integrable martingale it will converge to a random but finite value as $t\rightarrow\infty$. This result is again originally due Doob (1953).
\begin{thm}
If $\overline{u(t)}$ is a martingale for all $t\in\mathbf{X}_{II}\cup\mathbf{X}_{III}$
then or $t\ge t_{\epsilon}$ then $\exists$ random variable$\overline{u}_{\infty}\in L^{p}
(\Omega,\mathlarger{\mathrm{I\!P}},\mathbf{X}_{II}\cup\mathbf{X}_{III})$ such that $\overline{u(t)}\rightarrow\overline{u}_{\infty}$ a.s as $t\rightarrow\infty$ if $\mathlarger{\mathlarger{\mathcal{E}}}[|\overline{u(t)}|^{p}]<\infty$ for $p>1$.
\end{thm}
\begin{proof}
Since $\overline{u(t)}=\mathlarger{\mathlarger{\mathcal{E}}}(u_{\infty}|Y_{t})$ for $\overline{u}_{\infty}
\in L^{1}(\mathlarger{\mathrm{I\!P}},\Omega)$ then $|\overline{u(t)}|=\mathlarger{\mathlarger{\mathcal{E}}}(|u_{\infty}||Y_{t})$
and so $\mathlarger{\mathlarger{\mathcal{E}}}[\overline{u}(t)]\le \mathlarger{\mathlarger{\mathcal{E}}}[|\overline{u}_{\infty}|$. Then for all
$t\ge t_{\epsilon},u>0$
\begin{align}
&\int_{\lbrace|\overline{u(t)}>u_{\epsilon}\rbrace}|\overline{u(t)}d
\mathlarger{\mathrm{I\!P}}\le \int_{\lbrace|\overline{u(t)}>u_{\epsilon}\rbrace}\mathlarger{\mathlarger{\mathcal{E}}}\bigg\llbracket|\overline{u(t)}|
Y_{t})\bigg\rrbracket d\mathlarger{\mathrm{I\!P}}=\int_{|\overline{u(t)}>u_{\epsilon}}|\overline{u}_{\infty}|
d\mathlarger{\mathrm{I\!P}}\nonumber\\&=\int_{\lbrace|u(t)\ge u_{\epsilon}\rbrace
\cap\lbrace|\overline{u_{\infty}}|>u\rbrace}|\overline{u_{\infty}}
d\mathlarger{\mathrm{I\!P}}+\int_{\lbrace|u(t)\ge u_{\epsilon}\rbrace\cap\lbrace|\overline{u}_{\infty}
|\le u\rbrace}|\overline{u}_{\infty}d\mathlarger{\mathrm{I\!P}}
\nonumber\\& \le
\int_{\lbrace|u(t)|\ge u\rbrace}|\overline{u}_{\infty}|d\mathlarger{\mathrm{I\!P}}+
u\mathlarger{\mathrm{I\!P}}[|\overline{u}(t)|>u]\le\int_{\lbrace|u(t)|\ge u\rbrace}|\overline{u}_{\infty}|d\mathlarger{\mathrm{I\!P}}+
\frac{u}{u_{\epsilon}}\mathlarger{\mathlarger{\mathcal{E}}}\bigg\llbracket|\overline{u}_{t}|\bigg\rrbracket\nonumber\\&
\le\int_{\lbrace|u(t)|\ge u\rbrace}|\overline{u}_{\infty}|d\mathlarger{\mathrm{I\!P}}+
\frac{u}{u_{\epsilon}}\mathlarger{\mathlarger{\mathcal{E}}}\bigg\llbracket|\overline{u}_{\infty}|\bigg\rrbracket
\end{align}
Taking the limit as $u\rightarrow\infty$
\begin{equation}
\mathlarger{\mathlarger{\mathcal{E}}}\bigg\llbracket|u(t)|\bigg\rrbracket=\int_{\lbrace |\overline{u(t)}>u_{\epsilon\rbrace}}|\overline{u(t)}
d\mathlarger{\mathrm{I\!P}} \le \lim_{u\uparrow\infty}\frac{u}{u_{\epsilon}}\mathlarger{\mathlarger{\mathcal{E}}}
\bigg\llbracket|\overline{u}_{\infty}|\bigg\rrbracket=\infty
\end{equation}
so that if $\overline{u(t)}$ converges to $\overline{u}_{\infty} $ it must be uniformly integrable. Hence,$\mathlarger{\mathlarger{\mathcal{E}}}[|u(t)|]<\infty$ and there is no blowup or singularity in $\overline{u(t)}$.
\end{proof}
\subsection{Boundedness in probability:exponential-type martingale inequalities}
Some exponential-type martingale inequalities are given, which also establish boundedness in that the blowup probability for $\overline{u(t)}$ is always zero for any finite $t>t_{\epsilon}$.   \begin{thm}
Given the diffusion $d\overline{u(t)}={\psi}(u(t))d\mathlarger{\mathlarger{\mathscr{B}}}(t)\equiv k^{1/2}u^{2}(t)(u(t)-1)^{1/2}d\mathlarger{\mathlarger{\mathscr{B}}}(t)$ for $t>t_{\epsilon}$ with $u(t_{\epsilon})=u_{\epsilon}$, and the condition $\int_{t_{\epsilon}}^{t}{\psi}(u(s))ds<\infty$ for square integrability, then for any $q>0$ and $\lambda>0$ the following estimate holds
\begin{align}
&\mathlarger{\mathrm{I\!P}}\left[\left|k^{1/2}\int_{t_{\epsilon}}^{t}u^{2}(t)(u(t)-1)
d\mathlarger{\mathlarger{\mathscr{B}}}(s)-\frac{1}{2}q[\overline{u},\overline{u}](t)\right|\ge |u|-|u_{\epsilon}|
\right]\nonumber\\&\le \frac{u_{\epsilon}\exp\left(\frac{1}{2}\lambda^{2}(\lambda-q)
k\int_{t_{\epsilon}}^{t}|u^{4}(t)(u(t)-1)
ds\right)}{\exp(|u|-|u_{\epsilon}|)}
\end{align}
or equivalently
\begin{align}
&\mathlarger{\mathrm{I\!P}}\left[\left|k^{1/2}\int_{t_{\epsilon}}^{t}u^{2}(t)(u(t)-1)
d\mathlarger{\mathlarger{\mathscr{B}}}(s)-
\frac{1}{2}q k\int_{t_{\epsilon}}^{t}u^{4}(t)(u(t)-1)ds\right|\ge |u|-
|u_{\epsilon}|\right]\nonumber\\&\le
\frac{u_{\epsilon}\exp\left(\frac{1}{2}\lambda^{2}(\lambda-q)
k\int_{t_{\epsilon}}^{t}u^{4}(t)(u(t)-1)ds\right)}{\exp(|u|-|u_{\epsilon}|)}
\end{align}
If $\lambda=q$ then
\begin{align}
&\mathlarger{\mathrm{I\!P}}\left[\left|k^{1/2}\int_{t_{\epsilon}}^{t}(u(s))u^{2}(t)(u(t)-1)
d\mathlarger{\mathlarger{\mathscr{B}}}(s)-\frac{1}{2}qk\int_{t_{\epsilon}}^{t}|u^{4}(t)(u(t)-1)ds
\right|\ge|u|-|u_{\epsilon}|\right]\nonumber\\&\le
\frac{u_{\epsilon}}{\exp(|u|-|u_{\epsilon}|)}
\end{align}
Taking the limit as $|u|\rightarrow\infty$ then gives the probability of density function diffusion blowup at a  finite comoving proper time $t$ so that
\begin{align}
&\mathlarger{\mathrm{I\!P}}\left[\left|k^{1/2}\int_{t_{\epsilon}}^{t}u^{2}(t)(u(t)-1)^{1/2}
d\mathlarger{\mathlarger{\mathscr{B}}}(s)-
\frac{1}{2}qk\int_{t_{\epsilon}}^{t}|u^{4}(t)(u(t)-1)
ds\right|=\infty\right]\nonumber\\&\le\lim_{u\rightarrow\infty}
\frac{u_{\epsilon}}{\exp(|u|-|u_{\epsilon}|)}=0
\end{align}
Hence there is zero probability of a blowup or density singularity for any
$t>t_{\epsilon}$ or $t\in\mathbf{X}_{II}
\cup\mathbf{X}_{III}$. In particular, at $t=t_{*}=\pi/2k^{1/2}$
\begin{equation}
\mathlarger{\mathrm{I\!P}}\left[\left|\int_{t_{\epsilon}}^{t_{*}=\pi/2k^{1/2}}\kappa^{1/2}u^{2}(t)(u(t)-1)^{1/2}
d\mathlarger{\mathlarger{\mathscr{B}}}(t)-\frac{1}{2}q\int_{t_{\epsilon}}^{t_{*}=\pi/2\kappa^{1/2}}|
u^{4}(t)(u(t)-1)\right|=
\infty\right]=0
\end{equation}
\end{thm}
\begin{proof}
Since the exponential is a monotone increasing function, it follows that
\begin{equation}
\kappa^{1/2}\int_{t_{\epsilon}}^{t}u^{2}(t)(u(t)-1)^{1/2}d\mathlarger{\mathlarger{\mathscr{B}}}(t)-\frac{1}{2}Q
\int_{t_{\zeta}}^{t}\bigg|u^{4}(t)(u(t)-1)ds\bigg|\ge |u|-|u_{\epsilon}|
\end{equation}
is equivalent to
\begin{align}
&\exp\left(\lambda \left|\kappa^{1/2}\int_{t_{\epsilon}}^{t}u^{2}(t)(u(t)-1)^{1/2}d\mathlarger{\mathlarger{\mathscr{B}}}(t)-
\frac{1}{2}qk\int_{t_{\epsilon}}^{t}|u^{2}(t)(u(t)-1)^{1/2}ds\right|\right)\nonumber\\&
\ge \exp(\lambda|u|-|u_{\epsilon}|)
\end{align}
so that
\begin{align}
&\mathlarger{\mathrm{I\!P}}\bigg[\bigg|k\int_{t_{\epsilon}}^{t}u^{2}(t)(u(t)-1)^{1/2}
d\mathlarger{\mathlarger{\mathscr{B}}}(s)-
\frac{1}{2}q\int_{t_{\epsilon}}^{t}|u^{4}(t)(u(t)-1)d\bigg|\ge |u|-|u_{\epsilon}|
\bigg]\nonumber\\&\equiv\mathlarger{\mathrm{I\!P}}\bigg[\exp\bigg(\lambda \bigg|\kappa^{1/2}\int_{t_{\epsilon}}^{t}
u^{2}(s)(u(s)-1)^{1/2}d\mathlarger{\mathlarger{\mathscr{B}}}(s)
-\frac{1}{2}q\kappa\int_{t_{\epsilon}}^{t}|u^{4}(t)(u(t)-1)|ds\bigg|\bigg)\nonumber\\&
\ge \exp(\lambda|u|-|u_{\epsilon}|)\bigg]
\end{align}
Now using the fundamental Doob inequality
\begin{align}
&\mathlarger{\mathrm{I\!P}}\left[\left|\kappa^{1/2}\int_{t_{\epsilon}}^{t}u^{2}(t)(u(t)-1)^{1/2}
d\mathlarger{\mathlarger{\mathscr{B}}}(s)-\frac{1}{2}q\int_{t_{\epsilon}}^{t}|u^{4}(t)(u(t)-1)ds\right|\ge |u|-|u_{\epsilon}|\right]\nonumber\\&\equiv\mathlarger{\mathrm{I\!P}}\bigg[\exp\bigg(\lambda \bigg|\kappa^{1/2}\int_{t_{\epsilon}}^{t}u^{2}(t)(u(t)-1)^{1/2}
d\mathlarger{\mathlarger{\mathscr{B}}}(s)\nonumber\\&-\frac{1}{2}q\int_{t_{\epsilon}}^{t}
u^{4}(t)(u(s)-1)ds\bigg|\bigg)\ge\exp(\lambda|u|-|u_{\epsilon}|)\bigg]\nonumber\\&\le
\frac{\mathlarger{\mathlarger{\mathcal{E}}}\left\llbracket \exp\left(\lambda\left|\kappa^{1/2}\int_{t_{\epsilon}}^{t}
u^{2}(t)(u(s)-1)^{1/2}d\mathlarger{\mathlarger{\mathscr{B}}}(s)-\frac{1}{2}
q\kappa\int_{t_{\epsilon}}^{t}|u^{4}(t)(u(t)-1)ds
\right|\right)\right\rrbracket}{\exp(\lambda|u|-|u_{\epsilon}|)}
\end{align}
which is
\begin{align}
&\mathlarger{\mathrm{I\!P}}\left[\left|\kappa^{1/2}\int_{t_{\epsilon}}^{t}u^{2}(t)(u(t)-1)^{1/2}
d\mathlarger{\mathlarger{\mathscr{B}}}(t)
-\frac{1}{2}q\kappa^{1/2}\int_{t_{\epsilon}}^{t}|u^{2}(t)(u(t)-1)^{1/2}
ds\right|\ge |u|-|u_{\epsilon}|\right]\nonumber\\&
\le \frac{\mathlarger{\mathlarger{\mathcal{E}}}\left\llbracket\exp\left(\lambda \left|\kappa^{1/2}\int_{t_{\epsilon}}^{t}
u^{2}(t)(u(t)-1)^{1/2}d\mathlarger{\mathlarger{\mathscr{B}}}(t)\right)\exp\left(-\frac{1}{2}\lambda
qk\int_{t_{\epsilon}}^{t}u^{4}(t)(u(t)-1)ds\right|\right)\right\rrbracket}
{\exp(\lambda|u|-|u_{\epsilon}|)}
\end{align}
Now setting
\begin{equation}
\overline{\theta(t)}=(\exp(\lambda |\kappa^{1/2}\int_{t_{\epsilon}}^{t}u^{2}(t)(u(t)-1)^{1/2}
d\mathlarger{\mathlarger{\mathscr{B}}}(t)\nonumber
\end{equation}
the moments of $\overline{\theta(t)}$ are
\begin{equation}
\mathlarger{\mathlarger{\mathcal{E}}}\bigg\llbracket|\overline{\theta(t)}|^{p}]=\mathlarger{\mathlarger{\mathcal{E}}}\bigg(\exp\bigg(p\lambda
\bigg|\kappa^{1/2}\int_{t_{\epsilon}}^{t}u^{2}(t)(u(t)-1)^{1/2}d\mathlarger{\mathlarger{\mathscr{B}}}(t)\bigg\rrbracket
\end{equation}
and the Ito lemma gives
\begin{align}
d|\mathlarger{\mathlarger{\mathcal{E}}}\overline{\theta(t)}|^{p}&=
\mathlarger{\mathrm{D}}_{u}|\phi(t)|^{p}d\overline{\theta(t)}+\frac{1}{2}
\mathlarger{\mathrm{D}}_{u}\mathlarger{\mathrm{D}}_{u}|
\mathlarger{\mathlarger{\mathcal{E}}}(t)|^{p} \lambda^{2}|u^{4}(s)(u(s)-1)ds
\nonumber\\&=p^{2}|\theta(t)|^{p-1}\lambda\kappa^{1/2}u^{2}(t)(u(t)-1)^{1/2}
d\mathlarger{\mathlarger{\mathscr{B}}}(t)
\left(p\lambda |\kappa^{1/2}\int_{t_{\epsilon}}^{t}u^{2}(t)(u(t)-1)^{1/2}
d\mathlarger{\mathlarger{\mathscr{B}}}(s)\right)\nonumber\\&+
\frac{1}{2}p(p-1)|\theta(t)|^{p-2}\lambda^{2}\kappa^{2}|u^{4}(t)(u(t)-1)
dt\nonumber\\& \le p^{2}|\theta(t)|^{p-1}\lambda \kappa^{1/2}u^{2}(t)(u(t)-1)^{1/2}d\mathlarger{\mathlarger{\mathscr{B}}}(t)
\exp\left(p\lambda |\kappa^{1/2}\int_{t_{\epsilon}}^{t}u^{2}(t)(u(t)-1)^{1/2}
d\mathlarger{\mathlarger{\mathscr{B}}}(t)\right)\nonumber\\&
+\frac{1}{2}p^{2}|\theta(t)|^{p}\lambda^{2}\kappa u^{4}(s)(u(s)-1)
ds
\end{align}
since $p(p-1)|\theta(t)|^{p-2}<p^{2}|\theta(t)|^{p}$ and $\theta(t)$ is monotone increasing. Integrating
\begin{align}
&|\overline{\theta(t)}|^{p}=|Z_{\epsilon}|^{p}+p^{2}\int_{t_{\epsilon}}^{t}|\theta(s)|^{p-1}
\lambda \kappa^{1/2}u^{2}(t)(u(t)-1)^{1/2}
d\mathlarger{\mathlarger{\mathscr{B}}}(t)\nonumber\\&\times\exp\bigg(p\lambda
\bigg|\kappa^{1/2}\int_{t_{\epsilon}}^{v}u^{2}(t)(u(t)-1)^{1/2}
d\mathlarger{\mathlarger{\mathscr{B}}}(s)
\bigg)+\frac{1}{2}p^{2}\lambda^{2}\int_{t_{\epsilon}}^{t}|\theta(s)|^{p}
|u^{4}(t)(u(t)-1)ds
\end{align}
then taking the expectation
\begin{eqnarray}
\mathlarger{\mathlarger{\mathcal{E}}}\bigg\llbracket|\overline{\theta(t)}|^{p}\bigg\rrbracket\le |\theta_{\epsilon}|^{p}+\frac{1}{2}p^{2}
\lambda^{2}\int_{t_{\epsilon}}^{t}|\theta(u)|^{p}|u^{4}(t)(u(s)-1)|^{2}ds
\end{eqnarray}
The Gronwall Lemma then gives
\begin{equation}
\mathlarger{\mathlarger{\mathcal{E}}}\bigg\llbracket|\overline{\theta(t)}|^{p}\bigg\rrbracket\le |\theta_{\epsilon}|^{p}\exp\left(+\frac{1}{2}p^{2}
\lambda^{2}\int_{t_{\epsilon}}^{t}|u^{4}(s)(u(s)-1)ds\right)
\end{equation}
and for $p=1$
\begin{align}
&\mathlarger{\mathrm{I\!P}}(|\overline{\theta(t)}|)=\mathlarger{\mathlarger{\mathcal{E}}}\bigg\llbracket\exp\bigg(\lambda
\bigg|\kappa^{1/2}\int_{t_{\epsilon}}^{t}u^{2}(s)(u(s)-1)^{1/2}d\mathlarger{\mathlarger{\mathscr{B}}}(s)\bigg)
\bigg\rrbracket\nonumber\\&
\le |\theta_{\epsilon}|\exp\bigg(+\frac{1}{2}
\lambda^{2}k\int_{t_{\epsilon}}^{t}|u^{4}(s)(u(s)-1)ds\bigg)
\end{align}
Substituting for $\theta(t)$ then gives the required estimate
\begin{align}
&\mathlarger{\mathrm{I\!P}}\left[\left|k^{1/2}\int_{t_{\epsilon}}^{t}u^{2}(t)(u(t)-1)^{1/2}
d\mathlarger{\mathlarger{\mathscr{B}}}(s)-
\frac{1}{2}q\int_{t_{\epsilon}}^{t}|u^{4}(s)(u(s)-1)ds\right|\ge |u|-|u_{\epsilon}|
\right]\nonumber\\&\le\frac{u_{\epsilon}\exp\left(\frac{1}{2}\lambda^{2}(\lambda-q)
k\int_{t_{\epsilon}}^{t}|u^{4}(t)(u(t)-1)ds\right)}{\exp(|u|-|u_{\epsilon}|)}
\end{align}
\end{proof}
The following lemma is similar but is a Bernstein-type inequality (Bower et al 1986) and is proved via the Dambis-Dubins-Schwartz (DDS) Theorem.(Thm 5.17)
\begin{lem}
Given the martingale diffusion $\overline{u(t)}$ for all $t>t_{\epsilon}$ and any $
\lambda>0$ and let $T$ be a random time, then if the martingale is square integrable
\begin{equation}
\mathlarger{\mathrm{I\!P}}\bigg[\sup_{t_{\epsilon}\le t\le T}|\overline{u(t)}|\ge |u|, \bigg\langle\overline{u},\overline{u}\bigg\rangle(T)\le |\lambda|\bigg]\le \exp\bigg(-\frac{1}{2}\frac{|u|^{2}}{\lambda^{2}}\bigg)
\bigg]
\end{equation}
or equivalently
\begin{equation}
\mathlarger{\mathrm{I\!P}}\bigg[\sup_{t_{\epsilon}\le t\le T}|\overline{u(t)}|\ge |u|,
k\int_{t_{\epsilon}}^{t}|u^{4}(t)(u(t)-1))|^{2}ds \le |\lambda|\bigg]\le
\exp\bigg(-\frac{1}{2}\frac{|u|^{2}}{\lambda^{2}}\bigg)
\bigg]
\end{equation}
so that
\begin{equation}
\mathlarger{\mathrm{I\!P}}\bigg[\sup_{t_{\epsilon}\le t\le T}|\overline{u(t)}|=\infty,
k\int_{t_{\epsilon}}^{t}|u^{4}(t)(u(s)-1)ds
\le |\lambda|\bigg]=\lim_{u\rightarrow\infty}
\exp\bigg(-\frac{1}{2}\frac{|u|^{2}}{\lambda^{2}}\bigg)
\bigg]=0
\end{equation}
and there is zero probability of a blowup or singularity in the density function diffusion $\overline{u(t)}|$ for any finite $t>t_{\epsilon}$ ot $t\in\mathbf{X}_{II}\cup\mathbf{X}_{III}$.
\end{lem}
\begin{proof}
By the DDS Theorem, the underlying probability space can be enlarged to accommodate a Brownian
 motion $\mathlarger{\mathlarger{\mathscr{B}}}(t)$ such that $\overline{u(t)}=\mathbf{B}\big(\big\langle\overline{u},\overline{u}\big\rangle(T)\big)$.
 \begin{equation}
 \sup_{t\le T} \overline{u(t)}=\sup_{t\le[\overline{u},\overline{u}](T)}
 \mathlarger{\mathlarger{\mathscr{B}}}(t)
\end{equation}
 Then
\begin{align}
&\mathlarger{\mathrm{I\!P}}\left[\sup_{t_{\epsilon}\le t\le T}|\overline{u(t)}|\ge |u|, k\int_{t_{\epsilon}}^{t}|
u^{4}(t)(u(t)-1)ds \le |\lambda|\right]\le\mathlarger{\mathrm{I\!P}}[\sup_{t\le \lambda}\mathlarger{\mathlarger{\mathscr{B}}}(t)>
|u|]\equiv \int_{|u|}^{\infty}\mathcal{P} (x)dx
\end{align}
where $\mathcal{P}(x)$ is a probability density. This is a Gaussian distribution since BM is Gaussian so that
\begin{equation}
\mathcal{P}(u)=\tfrac{1}{(2\pi \lambda)^{1/2}}\exp(-\tfrac{1}{2}|u|^{2}/\lambda)
\end{equation}
Hence
\begin{align}
&\mathlarger{\mathrm{I\!P}}\left[\sup_{t_{\epsilon}\le t\le T}|\overline{\psi(t)}|\ge |\psi|, \int_{t_{\epsilon}}^{t}|(\psi(s)|^{2}u^{4}(t)(u(t)-1)ds \le |\Lambda|\right]\nonumber\\&\le \mathcal{P}[\sup_{t\le \lambda}\mathlarger{\mathlarger{\mathscr{B}}}(t)>|\psi|]\equiv \frac{2}{(2\pi\lambda)^{1/2}}\int_{|u|}^{\infty}\exp
\left(-\frac{1}{2}\frac{|x|^{2}}{\lambda}\right)dx
\nonumber\\&\equiv \frac{2}{(2\pi\lambda)^{1/2}}\int_{|u|/\lambda^{1/2}}^{\infty}
\exp\left(-\frac{1}{2}|y|^{2}\right)dy\le \exp\left(-\frac{1}{2}
\frac{|u|^{2}}{\lambda^{2}}\right)
\end{align}
using the Gaussian tail bound $ \frac{2}{(2\pi)^{1/2}}
\int_{x}^{\infty}\exp\bigg(-\frac{1}{2}|u|^{2}\bigg)dy\le \exp\bigg(-\frac{1}{2}|x|^{2}\bigg)$
\end{proof}
\begin{thm}
[Mao,44]. Let the conditions of Thm(6.10) hold and let $(\alpha,\beta)>0$ be some constants. Then for all $t>t_{\epsilon}$
\begin{align}
&\mathlarger{\mathrm{I\!P}}(\sup_{t_{\epsilon}\le t\le T}\overline{u(t)}-\frac{1}{2}\beta\bigg\langle\overline{u},\overline{u}\bigg\rangle(t)\ge \beta )\nonumber\\& \equiv\mathlarger{\mathrm{I\!P}}\bigg(\sup_{t_{\epsilon}\le t\le T}\bigg|k^{1/2}\int_{t_{\epsilon}}^{t}|u(s)|^{2}(u(s)-1)^{1/2}ds-\frac{1}{2}B\kappa\int_{t_{\epsilon}}^{t}
|u(s)|^{4}*(u(s)-1)ds\bigg|\ge \beta\bigg)\nonumber\\&\le\exp(-\beta C)
\end{align}
so that there is zero probability of blowup in any finite finite $t>t_{\epsilon}$
\begin{align}
&\mathlarger{\mathrm{I\!P}}(\sup_{t_{\epsilon}\le t\le T}|\overline{u(t)}-\frac{1}{2}\beta\bigg\langle\overline{u},\overline{u}\bigg\rangle(t)|=\infty)\nonumber\\& \equiv\mathlarger{\mathrm{I\!P}}\bigg(\sup_{t_{\epsilon}\le t\le T}\bigg|k^{1/2}\int_{t_{\epsilon}}^{t}|u(s)|^{2}(u(s)-1)^{1/2}ds-\frac{1}{2}\beta
\kappa\int_{t_{\epsilon}}^{t}
|u(s)|^{4}*(u(s)-1)ds\bigg|=\infty\bigg)=0
\end{align}
\end{thm}
\begin{proof}
Define the stopping time
\begin{equation}
\mathcal{T}_{n}=\inf\bigg\lbrace t\ge t_{\epsilon}:\bigg|\int_{t_{\epsilon}}^{t}|u(s)|^{2}(u(s)-1)^{1/2}
d\mathlarger{\mathlarger{\mathscr{B}}}(s)
+k\int_{0}^{t}|u(s)|^{4}(u(s)-1)ds\bigg|\ge n \bigg\rbrace
\end{equation}
and the Ito process
\begin{equation}
\overline{u_{n}(t)}=\beta \kappa^{1/2}\int_{t_{\epsilon}}^{t}\mathlarger{\mathcal{C}}(t_{\epsilon},\mathcal{T}_{n})|u(s)|^{2}(u(s)-1)^{1/2}
d\mathlarger{\mathlarger{\mathscr{B}}}(s)+\frac{1}{2}\kappa B^{2}\int_{t_{\epsilon}}^{t}\mathlarger{\mathcal{C}}(t_{\epsilon},
\mathcal{T}_{n})|u(s)|^{4}(u(s)-1)ds
\end{equation}
where $\mathcal{C}(t_{\epsilon},\mathcal{T}_{n})=1$ if $t\in[t_{\epsilon},\mathcal{T}_{n}]$ and zero otherwise. Let $\mathlarger{\mathrm{D}}_{u}
=d/du(t)$ then the Ito expansion of $d\exp(u_{n}(t))$ is
\begin{align}
d[\exp(\overline{u_{n}(t)})&=\left(\mathlarger{\mathrm{D}}_{u}\exp(u_{n}(t)\right)d\overline{u_{n}(t)}+
\frac{1}{2}\mathlarger{\mathrm{D}}_{u}\mathlarger{\mathrm{D}}_{u}\exp(u_{n}(t))
d\bigg\langle\overline{u},\overline{u}\bigg\rangle(t)\nonumber\\&=
\exp(u_{n}(t)d\overline{u_{n}(t)}+\frac{1}{2}\exp(u_{n}(t)d\bigg\langle\overline{u},\overline{u}\bigg\rangle(t)\nonumber\\&
=\exp(u_{n}(t)(\kappa^{1/2}|u(s)|^{2}(u(s)-1)^{1/2}d\mathlarger{\mathlarger{\mathscr{B}}}(s)
-\frac{1}{2}\beta^{2}\kappa|u(s)|^{4}(u(s)-1)ds)+\frac{1}{2}\exp(u_{n}(t)d[\overline{u},\overline{u}](t)\nonumber\\&
=\exp(u_{n}(t)(\kappa^{1/2}|u(s)|^{2}(u(s)-1)^{1/2}d\mathlarger{\mathlarger{\mathscr{B}}}(s)
-\frac{1}{2}\beta^{2}\kappa|u(s)|^{4}(u(s)-1)ds)\nonumber\\&+\frac{1}{2}k\beta^{2}\exp(u_{n}(t))|u(s)|^{4}(u(s)-1)ds
=\exp(u_{n}(t))\kappa^{1/2}|u(s)|^{2}(u(s)-1)^{1/2}d\mathlarger{\mathlarger{\mathscr{B}}}(s)
\end{align}
so that
\begin{equation}
\exp(u_{n}(t))=1+\int_{t_{\epsilon}}^{t}\exp(u_{n}(t))k^{1/2}|u(s)|^{2}
(u(s)-1)^{1/2}d\mathlarger{\mathlarger{\mathscr{B}}}(s)
\end{equation}
Then $\mathlarger{\mathlarger{\mathcal{E}}}(\exp(\overline{u_{n}(t)})\rbrace=1$ and $
\$\exp(u_{n}(t))$ is a martingale. Using the Doob inequality
\begin{equation}
\mathlarger{\mathrm{I\!P}}(\sup_{t_{\epsilon}\le t\le T}\exp(u_{n}(t))\ge \exp(\beta C))\le\mathlarger{\mathlarger{\mathcal{E}}}\exp(-\beta C)=\exp(-\beta C)
\end{equation}
so that
\begin{align}
&\mathlarger{\mathrm{I\!P}}\bigg(\sup_{t_{\epsilon}\le t\le T}\exp\bigg(\beta \kappa^{1/2}\int_{t_{\epsilon}}^{t}\mathlarger{\mathcal{C}}(t_{\epsilon},\mathcal{T}_{n})|u(s)|^{2}(u(s)-1)^{1/2}
d\mathlarger{\mathlarger{\mathscr{B}}}(s)\nonumber\\&+\frac{1}{2}\kappa\beta^{2}\int_{t_{\epsilon}}^{t}
\mathlarger{\mathbf{I}}(t_{\epsilon},
\mathcal{T}_{n})|u(s)|^{4}(u(s)-1)ds\bigg)\ge \exp(\beta C)
\end{align}
or
\begin{align}
&\mathlarger{\mathrm{I\!P}}\bigg(\sup_{t_{\epsilon}\le t\le T}\kappa^{1/2}\int_{t_{\epsilon}}^{t}\mathlarger{\mathcal{C}}(t_{\epsilon},\mathcal{T}_{n})
|u(s)|^{2}(u(s)-1)^{1/2}d\mathlarger{\mathlarger{\mathscr{B}}}(s)\nonumber\\&+\frac{1}{2}\kappa\beta\int_{t_{\epsilon}}^{t}
\mathlarger{\mathcal{C}}(t_{\epsilon},\mathcal{T}_{n})|u(s)|^{4}(u(s)-1)ds\ge B\bigg)\ge\exp(\beta C)\bigg)
\end{align}
\end{proof}
so that the result follows in the limit as $n\rightarrow\infty$.
\section{Estimates and boundedness of moments to all orders}
The following estimates further establish that within the stochastic control, the pressureless fluid matter sphere or star cannot collapse to zero size or a singular state of infinite density within any finite(comoving) proper time. It can be shown that estimates for the moments ${\mathlarger{\mathcal{E}}}\big\llbracket|\overline{u(t)}|^{p}\big\rrbracket<\infty $ of the density function diffusion are always finite and bounded for all $t\in{\mathbf{X}}_{II}\bigcup{\mathbf{X}}^{+}_{III}$.
\begin{defn}
Given the diffusion $\overline{u(t)}$ for all $t\in\mathbf{X}_{II} \cup\mathbf{X}_{III}$, then the finiteness of the moments to all orders such that ${\mathlarger{\mathcal{E}}}\big\llbracket|\widehat{u}(t)|^{p}\big\rrbracket<K $ or
\begin{equation}
{\mathlarger{\mathfrak{M}}}_{p}(t)\equiv\|\overline{u(t)}\|^{p}_{\mathcal{L}_{p}}
\equiv{\mathlarger{\mathcal{E}}}\bigg\llbracket|\overline{u(t)}|^{p}\bigg\rrbracket
<K(t)<\infty
\end{equation}
defines $p^{th}$-ultimate boundedness for any finite $t>t_{\epsilon}$. This means that the density function diffusion process does not blow up or become singular for any finite $t\in\mathbf{X}_{II}\bigcup\mathbf{X}_{III}$. In particular, at what was the blowup time $t=t_{*}=\tfrac{1}{2}\kappa^{-1/2}\pi$, the estimates are finite so that ${\mathlarger{\mathcal{E}}}\big\llbracket|\overline{u(t_{*})}|^{p}\big\rrbracket<\infty.$
$p^{th}$ ultimate unboundedness is defined as
\begin{equation}
{\mathlarger{\mathfrak{M}}}_{p}(t)\equiv\|\overline{u(t)}\|^{p}_{\mathcal{L}_{p}}\equiv
{\mathlarger{\mathcal{E}}}\bigg\llbracket|\overline{u(t)}|^{p}\bigg\rrbracket
=\infty
\end{equation}
If $K(t)\sim C_{1}\exp(C_{2}|t-t_{\epsilon}|)$ then the system is $p^{th}$ exponentially stable.
\end{defn}
The main theorems for establishing the boundedness of the moments estimates are derived using both the linear and polynomial growth $H\ddot{o}lder$ conditions with the Buckholder-Davis-Gundy inequality and also the Ito Lemma.
\begin{thm}
If $\overline{u(t)}$ is a martingale then the Buckholder-Davis-Gundy inequality applies for $p\ge 2$.(Appendix A.) Given the linear growth condition $\kappa u(t)^{4}(u(t)-1)< K|u(t)|^{2}$, then for all $t\in [t_{\epsilon},T]\subset\mathbf{X}_{II}\cup\mathbf{X}_{III} $, and for initial data $\overline{u(t_{\epsilon})}=u(t_{\epsilon})=u_{\epsilon}$ and for all $p\ge 2$, the following estimate holds
\begin{align}
&\sup_{t_{\epsilon}\le t\le T}\big\|\overline{u(t)}\big\|_{\mathcal{L}_{p}}=
{\mathlarger{\mathcal{E}}}\bigg\llbracket\sup_{t_{\epsilon}\le t\le T}|\overline{u(t)}|^{p}\bigg\rrbracket\le
\alpha |u_{\epsilon}|^{p}\exp\left(\beta Q_{p}(K^{\frac{p}{p-2}}
|t-t_{\epsilon}|)^{\frac{p-2}{2}}\int_{t_{\epsilon}}^{T}ds\right)\nonumber\\&=\alpha |u_{\epsilon }|^{p}\exp\left(\beta Q_{p}(K^{\frac{p}{p-2}}|t-t_{\epsilon}|)^{\frac{p-2}{2}}
|T-t_{\epsilon}|\right)
\end{align}
In particular, the moments are finite and bounded at what was previously the blowup time
$t=t_{*}=\pi/2\kappa^{1/2}$ so that
\begin{eqnarray}
\sup_{t_{\epsilon}\le t\le T}\|\overline{u(t)}\|^{p}_{\mathcal{L}_{p}}\equiv
{\mathlarger{\mathcal{E}}}\bigg\llbracket\sup_{t_{\epsilon}\le t\le t_{*}}|\overline{u(t)}|^{p}\bigg\rrbracket\le
\alpha |u_{\epsilon}|^{p}\exp\left(\beta Q_{p}(K^{\frac{p}{p-2}}|\pi/2k^{1/2}-t_{\epsilon}|)^{\frac{p-2}{2}}
|T-t_{\epsilon}|\right)
\end{eqnarray}
\end{thm}
\begin{proof}
The $p^{th}$ moment is
\begin{equation}
|\overline{u(t)}|^{p}=\bigg|u_{\epsilon}+k^{1/2}\int_{t_{\epsilon}}^{t}u(s)^{2}
(u(s)-1)^{1/2}d\mathlarger{\mathlarger{\mathscr{B}}}(s)\bigg|^{p}
\end{equation}
Using the basic estimate $(a+b)^{p}\le 2^{p-1}(|a|^{p}+|b|^{p})<\alpha|a|^{p}+\beta|b|^{p}$
\begin{equation}
|\overline{u(t)}|^{p}\le \alpha |u_{\epsilon}|^{p}+\beta\bigg|
\int_{t_{\epsilon}}^{t}|u(s)^{2}(u(s)-1)^{1/2}|d\mathlarger{\mathlarger{\mathscr{B}}}(s)\bigg|^{p}
\end{equation}
The expectation is then
\begin{equation}
\|\overline{u(t)}\|^{p}_{\mathcal{L}_{p}}\equiv{\mathlarger{\mathcal{E}}}
\bigg\llbracket|\overline{u(t)}|^{p}\bigg\rrbracket\le \alpha |u_{\epsilon}|^{p}+
\beta{\mathlarger{\mathcal{E}}}\left\llbracket\bigg|\int_{t_{\epsilon}}^{t}\kappa^{1/2} |u(s)^{2}(u(s)-1)^{1/2}|d\mathlarger{\mathlarger{\mathscr{B}}}(s)(s)\bigg|^{p}\right\rrbracket
\end{equation}
Applying the BDG inequality, there is a constant $Q_{p}>0$ such that
\begin{equation}
\|\overline{u(t)}\|^{p}_{\mathcal{L}_{p}}\equiv\mathlarger{\mathcal{E}}\left\llbracket\bigg|k \int_{t_{\epsilon}}^{t}u(s)^{4}(u(s)-1)
d\mathlarger{\mathlarger{\mathscr{B}}}(s)\bigg|^{p}\right\rrbracket\le Q_{p}\bigg\langle \overline{u},\overline{u}\bigg\rangle(t)^{p/2}\equiv
\bigg(k\int_{t_{\epsilon}}^{t}u(s)^{4}(u(s)-1)ds\bigg)^{p/2}
\end{equation}
Hence
\begin{equation}
\|\overline{u(t)}\|^{p}_{\mathcal{L}_{p}}\equiv\mathlarger{\mathcal{E}}\bigg\llbracket|\overline{u(t)}|^{p}\bigg\rrbracket< \alpha |u_{\epsilon}|^{p}+\beta Q_{p}\mathlarger{\mathcal{E}}\left\llbracket \int_{t_{\epsilon}}^{t}k^{2}u(s)^{4}(u(s)-1) ds\right\rrbracket^{p/2}
\end{equation}
or
\begin{equation}
\sup_{t_{\epsilon}\le t\le T}\|\overline{u(t)}\|^{p}_{\mathcal{L}_{p}}\equiv\mathlarger{\mathcal{E}}\bigg
\llbracket\sup_{t_{\epsilon}\le t\le T}|\overline{u(t)}|^{p}\bigg\rrbracket< \alpha |u_{\epsilon}|^{p}+\beta Q_{p}\mathlarger{\mathcal{E}}\bigg\llbracket
k\int_{t_{\epsilon}}^{T}u(s)^{4}(u(s)-1)ds\bigg)^{p/2}\bigg\rrbracket
\end{equation}
Using the linear growth condition $|\psi(u(t))|^{2}=\kappa |u(t)^{4}(u(t)-1)<K|u(t)|^{2}$ on $[t_{\epsilon},T]$ gives
\begin{equation}
\sup_{t_{\epsilon}\le t\le T}\|\overline{u(t)}\|^{p}_{\mathcal{L}_{p}}=\mathlarger{\mathcal{E}}\bigg\llbracket\sup_{t_{\epsilon}\le t\le T}|\overline{u(t)}|^{p}\bigg\rrbracket< \alpha |u_{\epsilon}|^{p}+\beta Q_{p}{\mathcal{E}}\bigg\llbracket(K\int_{t_{\epsilon}}^{T}|u(s)|^{2} ds\bigg\rrbracket^{p/2}
\end{equation}
Now apply the Holder inequality which states that for any two functions $f(t),g(t)$ with
$f:\mathbf{R}^{+}\rightarrow\mathbf{R}^{+}$ and $g:\mathbf{R}^{+}\rightarrow\mathbf{R}^{+}$
\begin{align}
&\int^{T}|f(t)g(t)|dt\le \left(\int^{T}|f(t)|^{n}dt\right)^{\frac{1}{n}}
\left(\int^{T}|g(t)|^{m}\right)^{\frac{1}{m}}\nonumber\\&\equiv
\left(\int^{T}|f(t)|^{\frac{p}{p-2}}dt\right)^{\frac{p-2}{p}}
\left(\int^{T}|g(t)|^{p/2}\right)^{\frac{2}{p}}
\end{align}
where $n^{-1}+m^{-1}=1$ and $n=p/(p-2)$ and $m=p/2$. Hence
\begin{align}
&\left(\int_{t_{\epsilon}}^{T}K\mathlarger{\mathcal{E}}\bigg\llbracket|u(s)|^{2}
\bigg\rrbracket ds\right)^{p/2}\le
\left(\left(\int_{t_{\epsilon}}^{T}K^{\frac{p}{p-2}}\right)^{\frac{p-2}{p}}
\left(\int_{t_{\epsilon}}^{T}\mathlarger{\mathcal{E}}\bigg
\llbracket|u(s)|^{p}\bigg\rrbracket ds\right)^{\frac{2}{p}}
\right)^{\frac{p}{2}}\nonumber\\&\le\left(\int_{t_{\epsilon}}^{T}K^{\frac{p}{p-2}}
\right)^{\frac{p-2}{p}}\left(\int_{t_{\epsilon}}^{T}
\mathlarger{\mathcal{E}}\bigg\llbracket|u(s)|^{p}\bigg\rrbracket ds\right)= (K^{\frac{p}{p-2}}|t-t_{\epsilon}|)^{\frac{p-2}{2}}
\left(\int_{t_{\epsilon}}^{T}\mathlarger{\mathcal{E}}\bigg\llbracket|u(s)|^{p}\bigg\rrbracket ds\right)
\end{align}
and so
\begin{equation}
\sup_{t_{\epsilon}\le t\le T}\|\overline{u(t)}\|^{p}_{\mathcal{L}_{p}}-
\mathlarger{\mathlarger{\mathcal{E}}}\bigg\llbracket\sup_{t_{\epsilon}\le t\le T}|\overline{u(t)}|^{p}\bigg\rrbracket \le \alpha |u_{\epsilon}|^{p}+
\beta Q_{p}(K^{\frac{p}{p-2}}|t-t_{\epsilon}|)^{\frac{p-2}{2}}
\int_{t_{\epsilon}}^{T}\mathlarger{\mathlarger{\mathcal{E}}}\bigg\llbracket|u(s)|^{p}\bigg\rrbracket ds
\end{equation}
The Gronwall Lemma then gives the estimate
\begin{align}
&\sup_{t_{\epsilon}\le t\le T}\|\overline{u(t)}\|^{p}_{\mathcal{L}_{p}}=\mathlarger{\mathlarger{\mathcal{E}}}\bigg
\llbracket\sup_{t_{\epsilon}\le t\le T}|\overline{u(t)}|^{p}\bigg\rrbracket\le
\alpha |u_{\epsilon}|^{p}\exp\left(\beta
Q_{p}(K^{\frac{p}{p-2}}|t-t_{\epsilon}|)^{\frac{p-2}{2}}
\int_{t_{\epsilon}}^{T}ds\right)\nonumber\\&=\alpha |u_{\epsilon}|^{p}
\exp\left(\beta Q_{p}(K^{\frac{p}{p-2}}|t-t_{\epsilon}|)^{\frac{p-2}{2}}
|T-t_{\epsilon}|\right)
\end{align}
giving (7.3). At what was the blowup time $t=t_{*}=\pi/2k^{1/2}$ the estimate is finite so that
\begin{eqnarray}
\sup_{t_{\epsilon}\le t\le t_{*}}\|\overline{u(t)}\|^{p}_{\mathcal{L}_{p}}
\equiv\mathlarger{\mathlarger{\mathcal{E}}}\bigg\llbracket\sup_{t_{\epsilon}\le t\le t_{*}}|\widehat{u}(t)|^{p}\bigg\rrbracket\le
=\alpha |u_{\epsilon}|^{p}\exp\left(\beta Q_{p}(K^{\frac{p}{p-2}}
|\pi/2k^{1/2}-t_{\epsilon}|)^{\frac{p-2}{2}}
|T-t_{\epsilon}|\right)
\end{eqnarray}
\end{proof}
An alternative estimate or bound is derived via the polynomial growth condition.
\begin{cor}
\begin{align}
\sup_{t_{\epsilon}\le t\le T}\|\overline{u(t)}\|^{p}_{\mathcal{L}_{p}}&\equiv
\mathlarger{\mathlarger{\mathcal{E}}}(\sup_{t_{\epsilon}\le t\le T}|\overline{u(t)}|^{p})< \alpha |u_{\epsilon}|^{p}+
\beta Q_{p}\mathlarger{\mathlarger{\mathcal{E}}}\bigg\llbracket K\int_{t_{\epsilon}}^{T}u(s)^{2} ds\bigg\rrbracket^{p/2}
\nonumber\\&
<\alpha |u_{\epsilon}|^{p}+
\beta Q_{p}\mathlarger{\mathlarger{\mathcal{E}}}\bigg\llbracket K\int_{t_{\epsilon}}^{T}u(s)^{2} ds\bigg\rrbracket\equiv\alpha |u_{\epsilon}|^{p}+
\beta Q_{p}\mathlarger{\mathlarger{\mathcal{E}}}\bigg\llbracket K\int_{t_{\epsilon}}^{T}|u(s)^{p}|^{2/p} ds\bigg\rrbracket
\end{align}
Now if for any non-negative functions $(f,g)$ the following holds
\begin{equation}
f(t)\le C+\int_{0}^{t}g(s)|f(s)|^{\gamma}ds
\end{equation}
then there is a bound or estimate of the form
\begin{equation}
f(t)\le \bigg(C^{1-\gamma}+(1-\gamma)\int_{0}^{t}g(s)ds\bigg)^{\frac{1}{1-\gamma}}
\end{equation}
For the strict inequality
\begin{equation}
\sup_{t_{\epsilon}\le t\le T}\|\overline{u(t)}\|^{p}_{\mathcal{L}_{p}}\equiv
\mathlarger{\mathlarger{\mathcal{E}}}\bigg\llbracket\sup_{t_{\epsilon}\le t\le T}|\overline{u(t)}|^{p}\bigg\rrbracket <\alpha |u_{\epsilon}|^{p}+
\beta Q_{p}\bigg(D\int_{t_{\epsilon}}^{T}|\mathlarger{\mathlarger{\mathcal{E}}}
\bigg\llbracket u(s)|^{p}|^{2/p}\bigg\rrbracket ds\bigg)
\end{equation}
the estimate is then
\begin{equation}
\sup_{t_{\epsilon}\le t\le T}\|\overline{u(t)}\|^{p}_{\mathcal{L}_{p}}\equiv
\mathlarger{\mathlarger{\mathcal{E}}}\bigg\llbracket\sup_{t_{\epsilon}\le t\le T}|\overline{u(t)}|^{p}\bigg\rrbracket< \bigg(\big(\alpha|u_{\epsilon}|^{p}\big)^{1-\tfrac{2}{p}}+\big(1-\tfrac{2}{p}\big)|T-t_{\epsilon}|\bigg)^{
\frac{1}{1-(2/p)}}<\infty
\end{equation}
which is always finite and bounded for all $p\ge 2$ and $T>t_{\epsilon}$
\end{cor}
An estimate can also be derived via the Ito Lemma for all $p\ge 1$
\begin{thm}
Let the following hold:
\begin{enumerate}
\item $\overline{u(t)}$ is a martingale for all $t>t_{\epsilon}$ with initial data
$\overline{u(t_{\epsilon})}=u(t_{\epsilon})=u_{\epsilon}$ so that $\mathlarger{\mathlarger{\mathcal{E}}}\llbracket \overline{u(t)}\rrbracket =
u_{\epsilon}$ for $t>t_{\epsilon}$.
\item Let $t\in[t_{\epsilon},T]\subset\mathbf{X}_{II}\cup\mathbf{X}_{III}$ then
there is a $K>0$ such that $(u(t))^{4}(u(t)-1)<K|u(t)|^{2}$ for $t\in[t_{\epsilon},T]$.
\end{enumerate}
Then the following estimates can be made
\begin{equation}
\sup_{t_{\epsilon}\le t\le T}\|\overline{u(t)}\|^{p}_{\mathcal{L}_{p}}\equiv
\mathlarger{\mathlarger{\mathcal{E}}}\bigg\llbracket\sup_{t_{\epsilon}\le t\le T}|\overline{u(t)}|^{p}\bigg\rrbracket\le
|u_{\epsilon}|^{p}\exp\left(\frac{1}{2}p(p-1)K|T-t_{\epsilon}\right)
\end{equation}
and
\begin{equation}
\sup_{t_{\epsilon}\le t\le T}\|\overline{u(t)}\|^{p}_{\mathcal{L}_{p}}\equiv
\mathlarger{\mathlarger{\mathcal{E}}}\bigg\llbracket\sup_{t_{\epsilon}\le t\le T}|\overline{u(t)}|^{p}\bigg\rrbracket\le
|u_{\epsilon}|^{p}\exp\left(\frac{1}{2}p(p-1)|t-t_{\epsilon}|
\left|\int_{t_{\epsilon}}^{T}|{\psi}(u(s))ds\right|^{2}\right)
\end{equation}
Hence the moments are finite and bounded over any interval $[t_{\epsilon},T]$. In particular, at what was previously the comoving blowup time $T=t_{*}=t_{\epsilon}+\epsilon$, the moments are now finite.
\end{thm}
\begin{proof}
Using the Ito Lemma, the expectation of any functional $\Phi(\overline{u(t)})$ is
\begin{equation}
\bigg\|\int_{t_{\epsilon}}^{t}
\mathlarger{\mathrm{D}}_{u}\mathlarger{\mathrm{D}}_{u}\Phi(u(s))d\bigg\langle\overline{u},\overline{u}\bigg\rangle(s)
\bigg\|_{\mathcal{L}_{1}}
\equiv\mathlarger{\mathlarger{\mathcal{E}}}\bigg\llbracket\int_{t_{\epsilon}}^{t}
\mathlarger{\mathrm{D}}_{u}\mathlarger{\mathrm{D}}_{u}\Phi(u(s))d\bigg\langle\overline{u},\overline{u}\bigg\rangle(s)\bigg\rrbracket
\end{equation}
Setting $\Phi(\overline{u(t)})=|\overline{u(t)}|^{p}$
\begin{align}
&\Phi(\overline{u(t)})=|u_{\epsilon}|^{p}+
\frac{1}{2}\int_{t_{\epsilon}}^{t}\mathlarger{\mathrm{D}}_{u}
\mathlarger{\mathrm{D}}_{u}|u(s)|^{p}
d\bigg\langle \overline{u},\overline{u}\bigg\rangle(s)\equiv |u(t)|^{p}\nonumber\\&+
\frac{1}{2}k\int_{t_{\epsilon}}^{t}\mathlarger{\mathrm{D}}_{u}
\mathlarger{\mathrm{D}}_{u}|u(s)|^{p}
|u(s)^{4}(u(s)-1)|ds
\end{align}
Now using $\kappa|u(s)^{4}(u(s)-1)|\le K|u(t)|^{2}$.
\begin{equation}
\|\overline{u(t)}\|^{p}_{\mathcal{L}_{p}}\equiv
\mathlarger{\mathlarger{\mathcal{E}}}\bigg\llbracket|u(t)|^{p}\bigg\rrbracket \le |u_{\epsilon}|^{p}+\frac{1}{2}p(p-1)K\int_{t_{\epsilon}}^{t}|u(s)|^{p-2}u(s))|^{2}ds
\end{equation}
The Gronwall Lemma then gives the first estimate as
\begin{equation}
\|\overline{u(t)}\|^{p}_{\mathcal{L}_{p}}\equiv
\mathlarger{\mathlarger{\mathcal{E}}}\bigg\llbracket|\overline{u(t)}|^{p}\bigg\rrbracket \le|u_{\epsilon}|^{p}\exp(\frac{1}{2}p(p-1)|t-t_{\epsilon}|
\end{equation}
or
\begin{equation}
\|\overline{u(t)}\|^{p}_{\mathcal{L}_{p}}\equiv
\mathlarger{\mathcal{E}}\bigg\llbracket\sup_{t_{\epsilon}\le t\le T}|\overline{u(t)}|^{p}\bigg\rrbracket\le
|u_{\epsilon}|^{p}\exp(\frac{1}{2}p(p-1)|t-t_{\epsilon}|
\end{equation}
The second estimate is derived via the Holder integral inequality
\begin{align}
\|\overline{u(t)}\|^{p}_{\mathcal{L}_{p}}\equiv
&\mathlarger{\mathlarger{\mathcal{E}}}\bigg\llbracket|\overline{u(t)}|^{p}\bigg\rrbracket\le |u_{\epsilon}|^{p}+
\frac{1}{2}p(p-1)\int_{t_{\epsilon}}^{t}\mathlarger{\mathlarger{\mathcal{E}}}
\bigg\llbracket|u(s)|^{p-2}|{\psi}(u(s))|^{2}\bigg\rrbracket ds \nonumber\\&\le |u_{\epsilon}|^{p}+
\frac{1}{2}p(p-1)\left(\int_{t_{\epsilon}}^{t}\mathlarger{\mathlarger{\mathcal{E}}}\bigg\llbracket|u(s)|^{p}\bigg\rrbracket ds\right)^{\frac{p-2}{p}}
\left(\int_{t_{\epsilon}}^{t}\mathlarger{\mathlarger{\mathcal{E}}}\bigg\llbracket|(u(t))^{4}(u(t)-1)\bigg\rrbracket ds\right)^{\frac{2}{p}}
\end{align}
Since the integrals are monotone increasing in t. Note that for $p=2$ the
Ito isometry is recovered and the equality holds so that
\begin{equation}
\|\overline{u(t)}\|^{2}_{\mathcal{L}_{2}}\equiv
\mathlarger{\mathlarger{\mathcal{E}}}\bigg\llbracket|\overline{u(t)}|^{2}\bigg\rrbracket= |u_{\epsilon}|^{2}+
k\int_{t_{\epsilon}}^{t}\mathlarger{\mathlarger{\mathcal{E}}}\bigg\llbracket|u^{4}(t)(u(t)-1)ds\bigg\rrbracket
\end{equation}
Now since the integrals are monotone increasing in t.
\begin{align}
\|\overline{u(t)}\|^{p}_{\mathcal{L}_{p}}\equiv
&\mathlarger{\mathlarger{\mathcal{E}}}\bigg\llbracket|\overline{u(t)}|^{p}\bigg\rrbracket\le |u_{\epsilon}|^{p}+
\frac{1}{2}p(p-1)\left(\int_{t_{\epsilon}}^{t}\mathlarger{\mathlarger{\mathcal{E}}}\bigg\llbracket|u(s)|^{p}\bigg\rrbracket ds\right)^{\frac{p-2}{p}}
\left(\int_{t_{\epsilon}}^{t}\mathlarger{\mathlarger{\mathcal{E}}}\bigg\llbracket|{\psi}(u(s))|^{2}\bigg\rrbracket ds\right)^{\frac{2}{p}}
\nonumber\\&\le |u_{\epsilon}|^{p}+
\frac{1}{2}p(p-1)\left(\int_{t_{\epsilon}}^{t}\mathlarger{\mathlarger{\mathcal{E}}}\bigg\llbracket|u(s)|^{p}\bigg\rrbracket ds\right)
\left(\int_{t_{\epsilon}}^{t}\mathlarger{\mathlarger{\mathcal{E}}}\bigg\llbracket|{\psi}(u(s))|^{2}\bigg\rrbracket]ds\right)
\end{align}
The Gronwall Lemma then gives the second estimate
\begin{equation}
\|\overline{u(t)}\|^{p}_{\mathcal{L}_{p}}\equiv
\mathlarger{\mathlarger{\mathcal{E}}}\bigg\llbracket|\overline{u(t)}|^{p}\bigg\rrbracket\le |\psi_{\epsilon}|^{p}\exp\left(\frac{1}{2}p(p-1)|t-t_{\epsilon}|
\left|\int_{t_{\epsilon}}^{t}|{\psi}(u(s))\right|^{2}ds\right)
\end{equation}
so that the moments are bounded if $\int_{t_{\epsilon}}^{t}|{\psi}(u(s))|^{2}ds<\infty$.
\end{proof}
\begin{cor}
For $p=1$, the moments reduce to
$\mathlarger{\mathlarger{\mathcal{E}}}\big\llbracket \overline{u(t)}\bigg\rrbracket=|u_{\epsilon}|$, which is the condition for a martingale.
\end{cor}
\begin{thm}
The solution is unique.
\end{thm}
\begin{proof}
Let $\overline{u(t)}$ and $\overline{v(t)}$ be two solutions for the same initial
data $u(t_{\epsilon})=v_{\epsilon}$. Then
\begin{align}
&\mathlarger{\mathlarger{\mathcal{E}}}\bigg\llbracket\sup_{t\le T}|\overline{u(t)}|^{p}\bigg\rrbracket\le
\alpha_{p}|u_{\epsilon}|^{p})+\beta_{p}K_{p}\mathlarger{\mathlarger{\mathcal{E}}}\left\llbracket
\left|\int_{t_{\epsilon}}^{T}\psi(u(s)|^{2}ds\right|^{p/2}\right\rrbracket\nonumber\\&
\le\alpha_{p}|u_{\epsilon}|^{p})+\beta_{p}D_{p}\mathlarger{\mathlarger{\mathcal{E}}}
\left\llbracket \left|\int_{t_{\epsilon}}^{T}Du^{p}(s)\right|^{p/2}\right\rrbracket ds
\end{align}
\begin{align}
&\mathlarger{\mathlarger{\mathcal{E}}}\bigg\llbracket\sup_{t\le T}|\overline{v(t)}|^{p}\bigg\rrbracket
\alpha_{p}|v_{\epsilon}|^{p})+\beta_{p}D_{p}\mathlarger{\mathlarger{\mathcal{E}}}\left\llbracket
\left|\int_{t_{\epsilon}}^{T}\psi(v(s)|^{2}ds\right|^{p/2}\right\rrbracket\nonumber\\&
\le\alpha_{p}|v_{\epsilon}|^{p})+\beta_{p}K_{p}\mathlarger{\mathlarger{\mathcal{E}}}\left\llbracket
\left|\int_{t_{\epsilon}}^{T}Kv^{p}(s)\right|^{p/2}ds\right\rrbracket
\end{align}
Hence
\begin{equation}
\mathlarger{\mathlarger{\mathcal{E}}}\bigg\llbracket\sup_{t\le T}|\overline{u}(t)|^{p}-\overline{v}(t)|^{p}\bigg\rrbracket
=\beta_{p}C_{p}D^{p/2}\mathlarger{\mathlarger{\mathcal{E}}}\left\llbracket\int_{t_{\epsilon}}^{T}|u(s)-v(s))|ds\right\rrbracket
\end{equation}
From the Gronwall inequality, it follows that $\mathlarger{\mathlarger{\mathcal{E}}}\llbracket \sup_{t\le T}|
\overline{\psi(t)}|^{p}-\psi(t)|^{p}\rrbracket = 0 $ so that $\overline{\psi(t)}=\overline{\psi(t)}$
\end{proof}
\begin{thm}
Let $0\le p<\alpha<\beta<\infty$ then given the $\mathcal{L}_{p}$ norm $\|\overline{u(t)}\|_{\mathcal{L}_{p}}=\mathlarger{\mathlarger{\mathcal{E}}}\big\llbracket\overline{u}(t)\rrbracket$, if $\overline{u(t)}\in\mathcal{L}_{p}\cap\mathcal{L}_{\beta}$ then $\overline{u(t)}\in\mathcal{L}_{\alpha}$, for all $t\in\mathbf{X}_{II}\cup\mathbf{X}_{III}$ or $t>t_{\epsilon}$. The following bound or estimate then holds
\begin{equation}
\log\|\overline{u(t)}\|_{\mathcal{L}_{\alpha}}\le\left| \frac{\frac{1}{\alpha}-\frac{1}{\beta}}{  \frac{1}{p}-\frac{1}{\beta}}\right|\log\|\overline{u(t)}\|_{\mathcal{L}_{p}}+\left|\frac{\frac{1}{\alpha}
-\frac{1}{\beta}}{ \frac{1}{p}-\frac{1}{\beta}}\right|
\log\|\overline{u(t)}\|_{\mathcal{L}_{\beta}}
\end{equation}
or
\begin{equation}
\frac{\log\mathcal{E}\bigg\llbracket\overline{u(t)}|^{\alpha}\bigg\rrbracket\bigg)^{1/\alpha}}
{\log\left(\mathcal{E}\left\llbracket\overline{u(t)}|^{p}\right\rrbracket\right)^{1/p}}\equiv
\frac{\log\big\|\overline{u(t)}\big\|_{\mathcal{L}_{\alpha}}}{\log\big\|\overline{u(t)}\big
\|_{\mathcal{L}_{p}}}\le\left| \frac{\frac{1}{\alpha}-\frac{1}{\beta}}{  \frac{1}{p}-\frac{1}{\beta}}\right|+\left|\frac{\frac{1}{\alpha}-\frac{1}{\beta}}{  \frac{1}{p}-\frac{1}{\beta}}\right|\frac{\log\|\overline{u(t)}\|_{\mathcal{L}_{\beta}}}{\log\|
\overline{u(t)}\|_{\mathcal{L}_{p}}}
\end{equation}
which is
\begin{align}
&\frac{\log\left(\left\llbracket\overline{u(t)}|^{\alpha}\right\rrbracket\right)^{1/\alpha}}
{\log\left(\mathlarger{\mathlarger{\mathcal{E}}}\left\llbracket\overline{u(t)}|^{p}\right\rrbracket\right)^{1/p}}\equiv
\frac{\log\left(\mathlarger{\mathlarger{\mathcal{E}}}\left\llbracket
\left|u_{\epsilon}+\kappa^{1/2}\int_{t_{\epsilon}}^{t}\psi(u(s))
d\mathlarger{\mathlarger{\mathscr{B}}}(s)\right|^{\alpha}
\right\rrbracket\right)^{1/\alpha}}{\log\left(\mathlarger{\mathlarger{\mathcal{E}}}\left\llbracket\left|u_{\epsilon}+k^{1/2}\int_{t_{\epsilon}}^{t}\psi(u(s))
d\mathlarger{\mathlarger{\mathscr{B}}}(t)\right|^{p}
\right\rrbracket\right)^{1/p}}\nonumber\\&\equiv
\frac{\log\left\|u_{\epsilon}+\kappa^{1/2}\int_{t_{\epsilon}}^{t}\psi(u(s))
d\mathlarger{\mathlarger{\mathscr{B}}}(s)\right\|_{\mathcal{L}_{\alpha}}}{\log
\left\|u_{\epsilon}+\kappa^{1/2}\int_{t_{\epsilon}}^{t}\psi(u(s))
d\mathlarger{\mathlarger{\mathscr{B}}}(s)\right\|_{\mathcal{L}_{p}}
}\le\left| \frac{\frac{1}{\alpha}-\frac{1}{\beta}}{ \frac{1}{p}-\frac{1}{\beta}}\right|+\left|\frac{\frac{1}{\alpha}-\frac{1}{\beta}}{  \frac{1}{p}-\frac{1}{\beta}}\right|\frac{\log\left\|u_{\epsilon}+\kappa^{1/2}
\int_{t_{\epsilon}}^{t}\psi(u(s))
d\mathlarger{\mathlarger{\mathscr{B}}}(t)\right\|_{\mathcal{L}_{\beta}}}{\log
\left\|u_{\epsilon}+\kappa^{1/2}\int_{t_{\epsilon}}^{t}\psi(u(s))
d\mathlarger{\mathlarger{\mathscr{B}}}(s)\right\|_{\mathcal{L}_{p}}
}\end{align}
\end{thm}
\begin{proof}
Since$t\frac{1}{q}<\tfrac{1}{r}<\tfrac{1}{p}$, there is a unique $\xi$ such that
\begin{equation}
\frac{1}{r}=\frac{\xi}{p}+\frac{1-\xi}{p}
\end{equation}
Solving for $\xi$ and $(1-\xi)$ gives
\begin{align}
\xi=\frac{\frac{1}{r}-\frac{1}{\beta}}{\frac{1}{p}-\frac{1}{\beta}},~~~~
1-\xi=\frac{\frac{1}{p}-\frac{1}{\alpha}}{\frac{1}{p}-\frac{1}{\beta}}
\end{align}
Next, one establishes that $\log\|\overline{u(t)}\|_{\mathcal{L}_{r}}\le
\log\xi\|\overline{u(t)}\|_{\mathcal{L}_{p}}+(1-\xi)\log\|\overline{u(t)}\|_{\mathcal{L}_{\beta}}$
Since $1=\alpha\xi/p+\alpha(1-\xi)/q$, applying the Holder inequality gives
\begin{align}
&\|\overline{u(t)}\|_{\mathcal{L}_{\alpha}}=\|(\overline{u(t)})^{\xi}(\overline{u(t)})^{1-\xi}\|_{\mathcal{L}_{\alpha}}
=\|(\overline{u(t)})^{\alpha\xi}(\overline{u(t)})^{r(1-\xi)}\|_{\mathcal{L}_{\alpha}}^{1/\alpha}\nonumber\\&
\le\left(\|\overline{u(t)}\|_{\mathcal{L}_{p/\alpha \xi}}\|(\overline{u(t)})^{\alpha(1-\xi)}\|_{\mathcal{L}_{\beta/\alpha(1-\xi)}}\right)^{1/r}\le
\left(\|\overline{u(t)}\|^{\alpha\xi}_{\mathcal{L}_{p}}\|(\overline{u(t)})^{r(1-\xi)}\|_{\mathcal{L}_{\beta}}\right)^{1/r}\nonumber\\&\le
\|\overline{u(t)}\|^{\alpha\xi}_{\mathcal{L}_{p}}\|(\overline{u(t)})^{(1-\xi)}\|_{\mathcal{L}_{\beta}}
\end{align}
Taking the $\log$
\begin{align}
&\log\|\overline{u(r)}\|_{\mathcal{L}_{r}}\le\log(\|\overline{u(t)}\big\|^{\xi}_{\mathcal{L}_{p}}\|(\overline{u(t)})^{(1-\xi)}\big\|_{\mathcal{L}_{q}})
=\log(\big\|\overline{u(t)}\|^{\xi}_{\mathcal{L}_{p}}+\log\|(\overline{u(t)})^{(1-\xi)}\|_{\mathcal{L}_{\beta}}\nonumber\\&
=\xi\log(\big\|\overline{u(t)}\|_{\mathcal{L}_{p}}+(1-\xi)\log\|\overline{u(t)}\|_{\mathcal{L}_{\beta}}\nonumber\\
\end{align}
The result then follows using (7.39).
\end{proof}
\begin{cor}
Given the estimate for the growth of the moments
\begin{equation}
\mathlarger{\mathlarger{\mathcal{E}}}\bigg\llbracket|\overline{u}(t)|^{p}\bigg\rrbracket
=(\|\overline{u(t)}\|^{p}_{\mathcal{L}_{p}} \le \alpha|u_{\epsilon}|^{p}\exp(\beta Q D^{\frac{p}{p-2}}|T-t_{\epsilon}|^{p/2})
\end{equation}
then
\begin{equation}
\left(\mathlarger{\mathlarger{\mathcal{E}}}\bigg\llbracket|\overline{u}(t)|^{p}\bigg\rrbracket\right)^{1/p}
=(\|\overline{u(t)}\|_{\mathcal{L}_{p}} \le \alpha^{1/p}|u_{\epsilon}|\exp(\frac{\beta}{p} Q K^{\frac{p}{p-2}}|T-t_{\epsilon}|^{p/2})
\end{equation}
Taking the $\log$ gives
\begin{equation}
\log\|\overline{u(t)}\|_{\mathcal{L}_{p}}=\log(\alpha^{1/p}|u_{\epsilon}|)+\frac{\beta}{p}QD^{{p}{p-2}}|T-t_{\epsilon}|^{p/2}
\end{equation}
so that (7.37)gives the inequality
\begin{align}
&\frac{\log(\alpha^{1/\alpha}|u_{\epsilon}|)+\frac{\beta}{p}QK^{{\alpha}{\alpha-2}}|T-t_{\epsilon}|^{\alpha/2}}
{\log(\alpha^{1/p}|u_{\epsilon}|)+\frac{\beta}{p}QD^{{p}{p-2}}|T-t_{\epsilon}|^{p/2}}\le
\left|\frac{\frac{1}{\alpha}-\frac{1}{\beta}}{\frac{1}{p}-\frac{1}{q}}\right|+
\left|\frac{\frac{1}{p}-\frac{1}{\alpha}}{\frac{1}{p}-\frac{1}{\beta}}\right|
\frac{\log(\alpha^{1/\beta}|u_{\epsilon}|)+\frac{\beta}{p}QK^{{\beta}{\beta-2}}|T-t_{\epsilon}|^{\beta/2}}
{\log(\alpha^{1/p}|u_{\epsilon}|)+\frac{\beta}{p}QK^{{p}{p-2}}|T-t_{\epsilon}|^{p/2}}
\end{align}
\end{cor}
\subsection{Holder and Kolmogorov continuity indicating singularity smoothing}
Given the moments estimate $\mathlarger{\mathcal{E}}\big\llbracket|\overline{u(t)}|^{p}\big\rrbracket$, it can be shown that $\overline{u(t)}$ satisfies the Kolmogorov continuity criterion [60,61] on any interval $[t_{\epsilon},T]\subset\mathbf{R}^{+}$. The deterministic solution $u(t)$ of $du(t)=\psi(u(t)dt$ on the other hand, is not continuous over any $[0,T]$ when $t_{*}\in[0,T]$ since $u(t_{*})=\infty$, when the density blows up at $t_{*}=t_{\epsilon}+{\epsilon}=\pi/2\kappa^{1/2}$. However, since the stochasticity smooths out or 'dissipates' the singularity and restores global regularity, we should have $\mathlarger{\mathcal{E}}\llbracket \overline{u(t_{*}}-u(t)|^{p}\rrbracket\equiv \|u(t)\|^{p}_{\mathcal{L}_{1}}\le \infty$. In particular, we wish to show that
\begin{equation}
\mathlarger{\mathlarger{\mathcal{E}}}\big\llbracket| \overline{u(t_{*}})-u(t)|^{p}\big\rrbracket\equiv\big\|u(t)\big\|^{p}_{\mathcal{L}_{p}}\le K|t_{*}-t|^{\alpha+1}=K|\epsilon|^{\alpha+1}
\end{equation}
for any $(t_{*},t)\in[t_{\epsilon},T^]\subset\mathbf{R}^{+}$. First, we state the following definitions and theorems.
\begin{defn}
H$\ddot{o}$lder continuity. Let $(s,t)\in[0,T]\subset\mathbf{R}^{+}$ and let $f:\mathbf{R}^{+}\rightarrow\mathbf{R}^{+}$. Then if $\exists~ K>0$ such that
\begin{equation}
|\mathlarger{f}(t)-f(s)|\le K|t-s|^{\alpha}
\end{equation}
the function $f(t)$ is H$\ddot{o}$lder continuous with respect exponent $\alpha>0$. For $p>1$, this extends to
\begin{equation}
|f(t)-f(s)|^{p}\le \bar{K}|t-s|^{p\alpha}
\end{equation}
If $f(t_{*})=\infty$ for some $t_{*}\in[0,T]$ then there is no H$\ddot{o}$lder continuity over $[0,T]$.
\end{defn}
\begin{defn}
A stochastic process $\overline{Y(t)}$ is a modification of a stochastic process $\overline{X(t)}$ if for all $t>0$, one has
$\mathlarger{\mathrm{I\!P}}(\overline{Y(t)}=\overline{X(t)})=1$. Then $\overline{X}(t)$ and $\overline{Y}(t)$ have the same underlying distribution $(\Omega,\mathscr{F},\mathlarger{\mathrm{I\!P}})$.
\end{defn}
\begin{thm}
Let $(\alpha,p,K)>0$. A stochastic process $\overline{X(t)}$ defined with respect to $(\Omega,\mathfrak{F},\mathlarger{\mathrm{I\!P}})$, with $(s,t)\in[0,T]$, has Kolmogorov continuity if $\exists~K$ such that
\begin{equation}
\big\|\overline{X(t)}-\overline{X(s)}\big\|^{p}_{\mathcal{L}_{p}}=
\mathlarger{\mathcal{E}}\bigg\llbracket\bigg| \overline{X(t)}-\overline{X(s)}\bigg|^{p}\bigg\rrbracket\le C|t-s|^{\alpha+1}
\end{equation}
Then there is a modification $\overline{Y(t)}$ that is continuous with $\gamma$-H$\ddot{o}$lder paths for all $\gamma\in[0,\tfrac{\alpha}{p}]$. There also exist a random variable $\overline{Z}_{\gamma}$ such that $|\overline{X(t)}-\overline{X(s)}|\le \overline{Z}_{\gamma}|t-s|^{\gamma}$
\end{thm}
\begin{thm}
Garsia-Rodemich-Rumsey (GRR) Theorem. Let $f:[0,T]\rightarrow\mathbf{R}^{+}$ $p>1$ with $\alpha > p^{-1}$, then $\exists~ C(p,\alpha)>0$ such that for all $(t,s)\in [0,T]\subset\mathbf{R}^{+}$
\begin{equation}
|f(t)-f(s)|^{p}\le C(\alpha,p)|t-s|^{\alpha p-1}\int_{0}^{T}\int_{0}^{T}
\frac{|f(x)-f(y)|^{p}dxdy}{|x-y|^{\alpha+1}}
\end{equation}
or
\begin{equation}
|f(t)-f(s)|\le \left(C(\alpha,p)|t-s|^{\alpha p-1}\int_{0}^{T}\int_{0}^{T}
\frac{|f(x)-f(y)|^{p}dxdy}{|x-y|^{\alpha+1}}\right)^{1/p}
\end{equation}
\end{thm}
The following proposition extends the GRR Theorem to stochastic processes
\begin{prop}
Let $\overline{X(t)}$ be a stochastic process on an interval $[0,T]$. Let $p>1$ with $\alpha > p^{-1}$, then $\exists~ C(p,\alpha)>0$ such that for all $(t,s)\in [0,T]\subset\mathbf{R}^{+}$
\begin{align}
\mathlarger{\mathlarger{\mathcal{E}}}&\bigg\llbracket|\overline{X(t)}-\overline{X(s)}|^{p}\bigg\rrbracket \le C(\alpha,p)|t-s|^{\alpha p-1}\int_{0}^{T}\int_{0}^{T}
\frac{\mathlarger{\mathlarger{\mathcal{E}}}\bigg\llbracket\bigg|\overline{X(x)}-\overline{X(y)}\bigg|^{p}
\bigg\rrbracket dxdy}{|x-y|^{\alpha+1}}\nonumber\\&
\le C(\alpha,p)|t-s|^{\alpha p-1}\int_{0}^{T}\int_{0}^{T}
\frac{\mathlarger{\mathlarger{\mathcal{E}}}\bigg\llbracket\bigg|\overline{X(x)}\bigg|^{p}
\bigg\rrbracket dxdy}{|x-y|^{\alpha+1}}-C(\alpha,p)|t-s|^{\alpha p-1}\int_{0}^{T}\int_{0}^{T}
\frac{\mathlarger{\mathlarger{\mathcal{E}}}\bigg\llbracket\bigg|\overline{X(y)}\bigg|^{p}
\bigg\rrbracket dxdy}{|x-y|^{\alpha+1}}
\end{align}
or equivalently
\begin{align}
\big\|\overline{X(t)}&-\overline{X(s)}|^{p}\big\|^{p}_{\mathcal{L}_{p}} \le C(\alpha,p)|t-s|^{\alpha p-1}\int_{0}^{T}\int_{0}^{T}
\frac{\big\|\overline{X(x)}-\overline{X(y)}
\big\| dxdy}{|x-y|^{\alpha+1}}\nonumber\\&
\le C(\alpha,p)|t-s|^{\alpha p-1}\int_{0}^{T}\int_{0}^{T}
\frac{\big\|\overline{X(x)}
\big\|^{p}_{\mathcal{L}_{p}}dxdy}{|x-y|^{\alpha+1}}-C(\alpha,p)|t-s|^{\alpha p-1}\int_{0}^{T}\int_{0}^{T}\frac{\big\|\overline{X(y)}\big\|^{p}_{\mathcal{L}_{p}} dxdy}{|x-y|^{\alpha+1}}
\end{align}
\end{prop}
The main theorem for Kolmogorov continuity of $u(t)$  on $[t_{\epsilon},T]$ is now established
\begin{thm}
If the following hold:
\begin{enumerate}
\item $|u(t_{*})-u(s)|=\infty$ for $t=t_{*}=\pi/2\kappa^{1/2}$ where $u(t)$ is a solution of $du(t)=\kappa^{1/2}(u(t))^{2}(u(t)-1)^{1/2}dt$.
\item $\overline{u(t)}=u_{\epsilon}+\kappa^{1/2}\int_{t_{\epsilon}}^{s}(u(s))^{2}(u(s)-1)^{1/2}ds$ is the solution of the SDE $d\overline{u(t)}=\psi(u(t))d\mathlarger{\mathscr{B}}(t)$.
\item The moments estimate
\begin{equation}
\mathlarger{\mathlarger{\mathcal{E}}}\bigg\llbracket|\overline{u(t)}|^{p}
\bigg\rrbracket\le |u_{\epsilon}|^{p}\exp\left(\tfrac{1}{2}p(p-1)C|T-t_{\epsilon}|\right)
\end{equation}
\item Let $(\alpha,p,K)>0$
\end{enumerate}
Then the stochastic process $\overline{u(t)}$ defined with respect to $(s,t)\in[t_{\epsilon},T]\subset\mathbf{X}_{II}\cup\mathbf{X}_{III}$, has Kolmogorov continuity and $\exists~K$ such that
\begin{equation}
\big\|\overline{u(t)}-\overline{u(s)}\big\|^{p}_{\mathcal{L}_{p}}=
\mathlarger{\mathlarger{\mathcal{E}}}\bigg\llbracket\bigg| \overline{u(t)}-\overline{u(s)}\bigg|^{p}\bigg\rrbracket\le K|t-s|^{\alpha+1}<\infty
\end{equation}
where $K=\tfrac{1}{2}|u_{\epsilon}|^{p}Cp(p-1)$. In particular, at $t=t_{*}$ and $s=t_{\epsilon}$
\begin{equation}
\big\|\overline{u(t_{*})}-u_{\epsilon})\big\|^{p}_{\mathcal{L}_{p}}=
\mathlarger{\mathlarger{\mathcal{E}}}\bigg\llbracket\bigg| \overline{u(t)}-u_{\epsilon}\bigg|^{p}\bigg\rrbracket\le K|t_{*}-t_{\epsilon}|^{\alpha+1}
\equiv K|\epsilon|^{\alpha+1}<\infty
\end{equation}
\end{thm}
\begin{proof}
\begin{align}
&\big\|\overline{u(t)}-\overline{u(s)}\big\|^{p}_{\mathcal{L}_{p}}=
\mathlarger{\mathlarger{\mathcal{E}}}\bigg\llbracket\bigg| \overline{u(t)}-\overline{u(s)}\bigg|^{p}\bigg\rrbracket
\le\mathlarger{\mathlarger{\mathcal{E}}}\bigg\llbracket\bigg| \overline{u(t)}\bigg\rrbracket^{p}-\mathlarger{\mathlarger{\mathcal{E}}}\bigg\llbracket\overline{u(s)}\bigg|^{p}\bigg\rrbracket
\nonumber\\&\le |u_{\epsilon}|^{p}\exp\left(\tfrac{1}{2}p(p-1)C|t-t_{\epsilon}|\right)
- |u_{\epsilon}|^{p}\exp\left(\tfrac{1}{2}p(p-1)C|s-t_{\epsilon}|\right)\nonumber\\&
\le\underbrace{ |u_{\epsilon}|^{p}\left(1+\tfrac{1}{2}Cp(p-1)|t-t_{\epsilon}|\right)-
|u_{\epsilon}|^{p}\left(1+\tfrac{1}{2}Cp(p-1)|s-t_{\epsilon}|\right)}_{expanding~\exp~to ~first~order}\nonumber\\&
\equiv \frac{1}{2}|u_{\epsilon}|^{p}Cp(p-1)\big[|t-t_{\epsilon}|-|s-t_{\epsilon}|\big]=
\frac{1}{2}|u_{\epsilon}|^{p}Cp(p-1)|t-s|\nonumber\\&< \frac{1}{2}|u_{\epsilon}|^{p}Cp(p-1)|t-s|^{\alpha+1}\equiv K|t-s|^{\alpha+1}
\end{align}
where the exponential is expanded to 1st order and the constant C can be rescaled to any required value so that the inequality still holds. Hence Kolmogorov continuity  is established.
\end{proof}
\begin{thm}
Let the conditions and result of Theorem (7.14) hold. Then a second and third estimate for the KC criterion can be derived from the GRR Theorem and are
\begin{align}
&\big\|\overline{u(t)}-\overline{u(s)}\big\|^{p}_{\mathcal{L}_{p}}=
\mathlarger{\mathlarger{\mathcal{E}}}\bigg\llbracket\bigg| \overline{u(t)}-\overline{u(s)}\bigg|^{p}\bigg\rrbracket
\le\mathlarger{\mathlarger{\mathcal{E}}}\bigg\llbracket\bigg| \overline{u(t)}\bigg\rrbracket^{p}-\mathlarger{\mathlarger{\mathcal{E}}}\bigg\llbracket\overline{u(s)}\bigg|^{p}\bigg\rrbracket\le \bar{K}|t-s|^{\delta+1}
\end{align}
and
\begin{align}
&\big\|\overline{u(t)}-\overline{u(s)}\big\|^{p}_{\mathcal{L}_{p}}=
\mathlarger{\mathlarger{\mathcal{E}}}\bigg\llbracket\bigg| \overline{u(t)}-\overline{u(s)}\bigg|^{p}\bigg\rrbracket
\le\mathlarger{\mathlarger{\mathcal{E}}}\bigg\llbracket\bigg| \overline{u(t)}\bigg\rrbracket^{p}-\mathlarger{\mathlarger{\mathcal{E}}}\bigg\llbracket\overline{u(s)}\bigg|^{p}\bigg\rrbracket\le \bar{D}|t-s|^{\delta+1}
\end{align}
where $\bar{D}=C(\alpha,p)K$, with $K$ given by (-),and $\delta=\alpha p+1$.
\end{thm}
\begin{proof}
\begin{align}
\mathlarger{\mathlarger{\mathcal{E}}}\bigg\llbracket|\overline{u(t)}&^-\overline{u(s)}|^{p}\bigg\rrbracket \le C(\alpha,p)|t-s|^{\alpha p-1}\int_{0}^{T}\int_{t_{\epsilon}}^{T}
\frac{\mathlarger{\mathlarger{\mathcal{E}}}\big\llbracket\big|\overline{u(x)}-\overline{u(y)}\big|^{p}
\big\rrbracket dxdy}{|x-y|^{\alpha+1}}\nonumber\\&
\le C(\alpha,p)|t-s|^{\alpha p-1}\int_{t_{\epsilon}}^{T}\int_{t_{\epsilon}}^{T}
\frac{\mathlarger{\mathlarger{\mathcal{E}}}\big\llbracket\big|\overline{u(x)}\rrbracket-
\big\llbracket\overline{u(y)}\big|^{p}
\big\rrbracket dxdy}{|x-y|^{\alpha+1}}\nonumber\\&
\le C(\alpha,p)|t-s|^{\alpha p-1}K\int_{t_{\epsilon}}^{T}\int_{t_{\epsilon}}^{T}
\frac{|x-y|^{\alpha+1}dxdy}{|x-y|^{\alpha+1}}
\nonumber\\&
= C(\alpha,p)|t-s|^{\alpha p-1}K\int_{t_{\epsilon}}^{T}\int_{t_{\epsilon}}^{T}dxdy\nonumber\\&
= C(\alpha,p)|t-s|^{\alpha p-1}K|T-t_{\epsilon}|^{2}\equiv \bar{K}|t-s|^{\alpha p+2}\nonumber\\&\equiv
\bar{K}|t-s|^{\alpha p+1}+1\equiv \bar{K}|t-s|^{\delta+1}
\end{align}
where (7.56) has been used. The full result for the moments and the GRR Theorem can also be used to establish the second estimate, by setting $y=t_{\epsilon}$.
\begin{align}
\mathlarger{\mathlarger{\mathcal{E}}}\bigg\llbracket|\overline{u(t)}&^-\overline{u(s)}|^{p}\bigg\rrbracket \le C(\alpha,p)|t-s|^{\alpha p-1}\int_{0}^{T}\int_{t_{\epsilon}}^{T}
\frac{\mathlarger{\mathlarger{{\mathcal{E}}}}\big\llbracket\big|\overline{u(x)}-\overline{u(y)}\big|^{p}
\big\rrbracket dxdy}{|x-y|^{\alpha+1}}\nonumber\\&
\le C(\alpha,p)|t-s|^{\alpha p-1}\int_{t_{\epsilon}}^{T}\int_{t_{\epsilon}}^{T}
\frac{\mathlarger{\mathlarger{\mathcal{E}}}\big\llbracket\big|\overline{u(x)}\rrbracket-
\big\llbracket\overline{u(y)}\big|^{p}
\big\rrbracket dxdy}{|x-y|^{\alpha+1}}\nonumber\\&
\le C(\alpha,p)|t-s|^{\alpha p-1}\int_{t_{\epsilon}}^{T}\int_{t_{\epsilon}}^{T}
\frac{\big(|u_{\epsilon}|^{p}\exp(\tfrac{1}{2}Cp(p-1)|x-t_{\epsilon}|-
|u_{\epsilon}|^{p}\exp(\tfrac{1}{2}Cp(p-1)|y-t_{\epsilon}|\big)dxdy}
{|x-y|^{\alpha+1}}\nonumber\\&
\le C(\alpha,p)|t-s|^{\alpha p-1}\int_{t_{\epsilon}}^{T}\int_{t_{\epsilon}}^{T}
\frac{\big(|u_{\epsilon}|^{p}\exp\big(\tfrac{1}{2}Cp(p-1)|x-t_{\epsilon}|- 1\big)dxdy}
{|x-t_{\epsilon}|^{\alpha+1}}\nonumber\\&=
C(\alpha,p)|t-s|^{\alpha p-1}|T-t_{\epsilon}||u_{\epsilon}|^{p}(T-t_{\epsilon})^{\alpha}(\frac{1}{2}Cp(p-1)|T-t_{\epsilon}|^{\alpha})
\Gamma[-\alpha,-\tfrac{1}{2}Cp(p-1)
|T-t_{\epsilon}|]\nonumber\\&
=\underbrace{\frac{1}{2}|t-s|^{\alpha}|t-s|^{\alpha p-1}|t-s|^{\alpha}|t-s|Cp(p-1)|u_{\epsilon}|^{p}}_{setting~ t=T~ and~ s=t_{\epsilon}}\nonumber\\&=
\frac{1}{2}|t-s|^{2\alpha+\alpha p}Cp(p-1)|u_{\epsilon}|^{p}\Gamma[-\alpha,-\tfrac{1}{2}Cp(p-1)
|t-s|]\nonumber\\&\equiv\mathlarger{{K}}|t-s|^{(2\alpha+\alpha p-1)+1}\equiv{K}|t-s|^{\delta+1}
\end{align}
and where the digamma function is used.
\end{proof}
\subsection{Boundedness of the stochastically averaged Kretschmann scalar invariant}
Given the bounded estimates for the moments $|\overline{u(t)}|^{p}$, it is
possible (and necessary)to then estimate the stochastically averaged Kretschmann scalar
invariant and show that it is also finite and bounded for any finite $t>t_{\epsilon}$.
\begin{thm}
The Kretschmann scalar invariant is $\mathbf{K}=\mathbf{Ric}_{\mu\nu\gamma\delta}
\mathbf{Ric}^{\mu\nu\gamma\delta}$ and for the interior FRW metric interior to the collapsing fluid ball or star, it is given by (3.62) in terms of $R(t)$ or (3.102) in terms of $|u(t)$. The stochastically averaged Kretschmann scalar invariant is finite and bounded for all $t>t_{\epsilon}$ or $t\in\mathbf{X}_{II}\cup\mathbf{X}_{III}$ if
$\mathlarger{\mathlarger{\mathcal{E}}}(|\overline{u(t)}|^{p})\infty$ for all $t>t_{\epsilon}$.
\end{thm}
\begin{proof}
The Kretschmann scalar invariant in terms of $u(t)$ is given by (3.102) as a polynomial
\begin{align}
&\mathbf{K}(u(t))=k^{2}[|u(t)|^{6}-4|u(t)|^{5}+|u(t)|^{4}+12|u(t)|^{3}-
12|u(t)|^{2}+3]\nonumber\\&
=k^{2}[a_{6}|u(t)|^{6}+a_{5}|u(t)|^{5}+a_{4}|u(t)|^{4}+a_{3}|u(t)|^{3}
+a_{2}|u(t)|^{2}+b]
\end{align}
with $\lim_{u(t)\uparrow\infty}\mathbf{K}(u(t))=
\lim_{t\uparrow t_{*}}\mathbf{K}(u(t))=\infty$. If $\overline{u(t)}=
u_{\epsilon}+\int_{0}^{t}\psi(u(s))d\mathlarger{\mathscr{B}}(s)$ then the stochastic Kretschmann scalar is
\begin{equation}
\bm{\mathrm{K}}(\overline{u(t)})=\kappa^{2}[a_{6}|\overline{u(t)}|^{6}+a_{5}|\overline{u(t)}|^{5}+a_{4}
|\overline{u(t)}|^{4}+a_{3}|\overline{u(t)}|^{3}
+a_{2}|\overline{u(t)}|^{2}+b]
\end{equation}
Taking the stochastic average or expectation
\begin{align}
&\mathlarger{\mathlarger{\mathcal{E}}}\big\llbracket\bm{\mathrm{K}}(\overline{u(t)})\big\rrbracket
=\kappa^{2}[a_{6}\mathbf{K}(|\overline{u(t)}|^{6})
+a_{5}\mathlarger{\mathlarger{\mathcal{E}}}\big\llbracket|\overline{u(t)}|^{5}
\big\rrbracket+a_{4}\mathlarger{\mathlarger{\mathcal{E}}}\big\llbracket|\overline{(t)}|^{4}\big\rrbracket+
a_{3}\mathlarger{\mathlarger{\mathcal{E}}}\big\llbracket|\overline{u(t)}|^{3}\big\rrbracket
+a_{2}\mathlarger{\mathlarger{\mathcal{E}}}\big\llbracket |\overline{u(t)}|^{2}\big\rrbracket+b]
\end{align}
Now since
\begin{align}
&\mathlarger{\mathlarger{\mathcal{E}}}\big\llbracket|\overline{u(t)}|^{p}\big\rrbracket\le |u_{\epsilon}|^{p}\exp(\frac{1}{2}p(p-1)C|t-t_{\epsilon}|)\equiv |u_{\epsilon}|^{p}\exp(\frac{1}{2}\alpha_{p}C|t-t_{\epsilon}|)
\end{align}
The bounded estimate becomes
\begin{align}
&\mathlarger{\mathlarger{\mathcal{E}}}\big\llbracket\bm{\mathrm{K}}(\bm{u}(t))\big\rrbracket\le k^{2}[a_{6}|u_{\epsilon}|^{6}
\exp(\alpha_{6}C|t-t_{\epsilon}|+a_{5}|u_{\epsilon}|^{6}
\exp(\alpha_{5}C|t-t_{\epsilon}|\nonumber\\&
a_{4}|u_{\epsilon}|^{4}\exp(\alpha_{4}C|t-t_{\epsilon}|+a_{3}|u_{\epsilon}|^{3}
\exp(\alpha_{3}C|t-t_{\epsilon}|+a_{2}|u_{\epsilon}|^{2}
\exp(\alpha_{2}C|t-t_{\epsilon}|+b]
\end{align}
Hence $\mathlarger{\mathlarger{\mathcal{E}}}\big\llbracket\bm{\mathrm{K}}(\overline{u(t)})\big\rrbracket<\infty$ for all $t>t_{\epsilon}$ and in particular $\mathlarger{\mathlarger{\mathcal{E}}}\big\llbracket\bm{\mathrm{K}}(\overline{u(t_{*})})
\big\rrbracket<\infty$.
\end{proof}
\section{Criticality properties and utilisation of Feller's Test for blowup}
In this section, the suppression of a density function singularity is analysed using exit probabilities and the Feller Test for blowup [61,66]. First, given the generator $\mathlarger{\mathrm{I\!H}}$ of the density function diffusion $\overline{u(t)}$ it is also possible to construct additional diffusions, which are also martingales, from appropriate
functionals $\Phi(\overline{u(t)})$.                                                         \begin{prop}                                                                                                                                                                                                                                                                                                                                                                                                                               Consider the generic nonlinear SDE
\begin{equation}
d\overline{u(t)}=\phi(u(t))dt+\psi(u(t))d\mathlarger{\mathlarger{\mathscr{B}}}(t)
\end{equation}
where $\phi(u(t))$ is an arbitrary functional $\phi:[0,t]\times[1,\infty)
\rightarrow\mathbf{R}^{+}$. The Lispchitz and Holder conditions are
\begin{align}
&|{\phi}(u(t))-\phi(\widehat{u}(t))|\bigwedge |\psi(u(t))-
{\psi}(\widehat{u}(t)|\le L|u(t)-\widehat{u}(t)|^{2}
\\&|\phi(u(t))|^{2}\bigwedge |{\psi}(u(t)|^{2}\le K(1+|u(t)|^{2})
\end{align}
Then define
\begin{equation}
d\overline{u(t)}=k^{1/2}u^{4}(t)(u(t)-1)^{1/2}
d\mathlarger{\mathlarger{\mathscr{B}}}(t)=
\lim_{f(u(t))\rightarrow 0}\phi(u(t))dt+
{\psi}(u(t))d\mathlarger{\mathlarger{\mathscr{B}}}(t)
\end{equation}
\end{prop}
\begin{lem}
Consider the generic nonlinear SDE of the form
\begin{equation}
d\overline{u(t)}=\phi(u(t))dt+{\psi}(u(t))d\mathlarger{\mathlarger{\mathscr{B}}}(t)
\end{equation}
with $t>t_{\epsilon}$. The generator $\mathlarger{\mathrm{I\!H}}$ of the diffusion(8.1)is the linear differential operator
\begin{equation}
\mathlarger{\mathrm{I\!H}}=\frac{1}{2}{\psi}(u(t))^{2}
\mathlarger{\mathlarger{\mathrm{D}}}_{u}\mathlarger{\mathlarger{\mathrm{D}}}_{u}
+{\phi}(u(t))\mathlarger{\mathlarger{\mathrm{D}}}_{u}
\end{equation}
Then the following hold:
\begin{enumerate}
\item Functionals $\Phi(\overline{u(t)})$ satisfy the parabolic PDE
\begin{align}
&\frac{\partial}{\partial t}{\Phi}(\overline{u},t)=\mathlarger{\mathrm{I\!H}}\Phi
(\overline{u},t)=\frac{1}{2}{\psi}(u(t))^{2}
\mathlarger{\mathlarger{\mathrm{D}}}_{u}\mathlarger{\mathlarger{\mathrm{D}}}_{u}\Phi(\overline{u(t)})+\phi(u(t))
\mathlarger{\mathlarger{\mathrm{D}}}_{u} \Phi(\overline{u(t)})
\end{align}
\item If $\phi(u(t))=0$ then this is a nonlinear heat equation of the form
\begin{align}
&\frac{\partial}{\partial t}{\Phi}(\overline{u},t)=\mathlarger{\mathrm{I\!H}}|(\overline{u},t)
=\frac{1}{2}{\psi}(\overline{u}(t))^{2}
\mathlarger{\mathlarger{\mathrm{D}}}_{u}\mathlarger{\mathlarger{\mathrm{D}}}_{u}
\Phi(\overline{u},t)=\frac{1}{2}
\kappa u^{4}(t)(u(t)-1)\mathlarger{\mathlarger{\mathrm{D}}}_{u}
\mathlarger{\mathlarger{\mathrm{D}}}_{u}\Phi(\overline{u},t)
\end{align}
\item If $\Phi(u(t))$ is a harmonic functional of $\mathlarger{\mathrm{I\!H}}$ for
all $t\in\mathbf{R}^{+}$ then
\begin{align}
\mathlarger{\mathrm{I\!H}}\Phi(\overline{u(t)})=\frac{1}{2}(\psi(u(t))^{2}
\mathlarger{\mathlarger{\mathrm{D}}}_{u}\mathlarger{\mathlarger{\mathrm{D}}}_{u}
\Phi(u(t))+|\phi(u(t)|)\mathlarger{\mathlarger{\mathrm{D}}}_{u}
\Phi(\overline{u(t)})=0
\end{align}
There is then a solution or 'scale function' of the form
\begin{equation}
\Phi(\overline{u(t)})=\beta\int_{u_{\epsilon}}^{u(t)}
\exp[-\theta(\bar{u}(s))d\bar{u}(s)
\end{equation}
where
\begin{equation}
\theta(u(t))=2\int_{u_{\epsilon}}^{u(t)}
\phi(u(s)){\psi}[u(s)]^{-2}du(s)
\end{equation}
\item If $\Phi(\overline{u(t)}])$ is a harmonic function of
$\mathlarger{\mathrm{I\!H}}=\tfrac{1}{2}|(u(t))^{4}(u(t)-1)\mathlarger{\mathlarger{\mathrm{D}}}_{u}
\mathlarger{\mathlarger{\mathrm{D}}}_{u}
$ for all $t\in\mathbf{X}_{II}\cup\mathbf{X}_{III}$ then
\begin{equation}
\mathlarger{\mathrm{I\!H}}\Phi(u(t))=\frac{1}{2}|{\psi}(u(t))|^{2}
\mathlarger{\mathlarger{\mathrm{D}}}_{u}\mathlarger{\mathlarger{\mathrm{D}}}_{u}\Phi(u(t))=0
\end{equation}
and there is a solution or 'scale function' of the form
\begin{equation}
\Phi(u(t),u(t_{\epsilon}))=\beta\int_{u(0)}^{u(t)}
\exp[-\theta[v(t)]dv(t)=\beta|u(t)-u_{\epsilon}|
\end{equation}
when $\phi(u(t))=0$.
\end{enumerate}
\end{lem}
\begin{proof}
To prove (1), use Ito's Lemma again for a (stochastic) functional $\Phi(\overline{u(t)})$ so that
\begin{align}
d\Phi(\overline{u(t)})&=\frac{\partial}{\partial t}\Phi(\overline{u(t)})dt
+\mathlarger{\mathlarger{\mathrm{D}}}_{u}\Phi(u(t))d\overline{u(t)}+\frac{1}{2}
\mathlarger{\mathlarger{\mathrm{D}}}_{u}\mathlarger{\mathlarger{\mathrm{D}}}_{u}
\Phi(u(t))d\bigg\langle\overline{u},\overline{u}\bigg\rangle(t)\nonumber\\&
\equiv\frac{\partial}{\partial t}\Phi(\overline{u(t)})dt + \mathlarger{\mathlarger{\mathrm{D}}}_{u} \Phi(u(t))\psi(u(t))dt+\mathlarger{\mathlarger{\mathrm{D}}}_{u}\Phi(u(t)){\psi}^{2}
(u(t))d\mathlarger{\mathlarger{\mathscr{B}}}(t)+
\mathlarger{\mathlarger{\mathrm{D}}}_{u}\Phi(u(t))|\psi(u(t)|^{2})dt
\end{align}
Taking the expectation then gives
\begin{align}
&\mathlarger{\mathlarger{\mathcal{E}}}\left\llbracket d\Phi(\overline{u(t)})\right\rrbracket=\mathlarger{\mathlarger{\mathcal{E}}}
\left\llbracket\frac{\partial}{\partial t}\Phi(\overline{u(t)}\right\rrbracket dt\equiv\frac{\partial}{\partial t}\Phi(\overline{u(t)})dt+\mathlarger{\mathlarger{\mathrm{D}}}_{u}\Phi(u(t))
f(\psi(t))dt+\mathlarger{\mathlarger{\mathrm{D}}}_{u}\Phi(u(t))|\psi(u(t)|^{2})dt
\end{align}
To prove statements (3) and (4), the derivatives of the generic scale factor (8.13)are
\begin{align}
&{\mathlarger{\mathrm{D}}}_{u}\Phi(u(t),u(t_{\epsilon})=\exp(-\theta(u(t)))
\\&{\mathlarger{\mathrm{D}}}_{u}{\mathlarger{\mathrm{D}}}_{u}\Phi(u(t))=-
\mathlarger{\mathlarger{\mathrm{D}}}_{u}
\theta(t)\exp(-\theta(t))=-2\phi(u(t)){\psi}(u(t))\exp(-\theta(u(t))
\end{align}
Substituting, the harmonic equation (8.12) becomes
\begin{equation}
{\mathlarger{\mathrm{I\!H}}}\Phi(u(t)=\phi(u(t)\exp(-\theta(u(t))-
\phi(u(t)){\psi}^{2}(u(t)){\psi}^{-2}\exp(-\theta(u(t))=0
\end{equation}
For a driftless SDE we have $ \phi(u(t))=0$ so that the generator of the pure diffusion is $
\mathlarger{\mathrm{I\!H}}=\frac{1}{2}\kappa^{1/2}(u(t))^{4}(u(t)-1)
\mathlarger{\mathlarger{\mathrm{D}}}_{u}\mathlarger{\mathlarger{\mathrm{D}}}_{u}$
This gives $|\Phi(u(t))=\beta\int_{u_{t_{\epsilon}}}^{u(t)}du(t) =\beta|u(t)-
u(t_{\epsilon})| $. Then
\begin{align}
&\mathlarger{\mathrm{I\!H}}\psi(u(t))=\frac{1}{2}\kappa^{1/2}(u(t))^{4}(u(t)-1)
\mathlarger{\mathlarger{\mathrm{D}}}_{u}\psi(u(t))\nonumber\\&=\frac{1}{2}\kappa^{1/2}\beta(u(t))^{4}(u(t)-1)
\mathlarger{\mathlarger{\mathrm{D}}}_{u}|u(t)-u_{\epsilon}|=0
\end{align}
Since $\mathlarger{\mathlarger{\mathrm{D}}}_{u}\mathlarger{\mathlarger{\mathrm{D}}}_{u}u(t)=0$ and $\lambda(u(t))=0$, so that
$\mathlarger{\mathrm{I\!H}}\Phi(u(t))=0$ is a harmonic equation.
\end{proof}
\begin{defn}
The diffusion is recurrent if $\Phi(u(0),\infty)=\infty $ and transient
if $\Phi(u(0),\infty)<\infty $.
\end{defn}
A Dynkin theorem for the density function diffusion is as follows:
\begin{thm}
Given the generator $\mathlarger{\mathrm{I\!H}}=\tfrac{1}{2}|\psi(u(t)|^{2}){\mathlarger{\mathrm{D}}}_{u}
{\mathlarger{\mathrm{D}}}_{u}=
\tfrac{1}{2}=\tfrac{1}{2}\kappa u^{4}(t)(\psi(t)-1){\mathlarger{\mathrm{D}}}_{u}{\mathlarger{\mathrm{D}}}_{u}
$ for the density diffusion SDE $d\widehat{u}(t)=\tfrac{1}{2}\kappa u^{4}(t)(u(t)-1)d\mathlarger{\mathlarger{\mathscr{B}}}(t)$,
defined for all $t>t_{\epsilon}$ or $t\in\mathbf{X}^{+}_{II}\bigcup\mathbf{X}^{+}_{III}$, and given a $C^{2}$ continuous functional $\Phi(\overline{u(t)})$, then $\exists$ a martingale $\overline{m(t)}$ given by
\begin{align}
&\overline{m(t)}=|\Phi(\overline{u(t)})-\Phi(u_{\epsilon})|-\int_{t_{\epsilon}}^{t}
\mathlarger{\mathrm{I\!H}}\Phi(u(s))ds\nonumber\\&=|\Phi(\overline{u(t)})
-\Phi(u_{\epsilon})|-\frac{1}{2}k\int_{t_{\epsilon}}^{t}|u^{4}(s)(u(s)-1)|
\mathlarger{\mathrm{D}}_{u}\mathlarger{\mathrm{D}}_{u}|\Phi(u(s))|ds
\end{align}
with expectation
\begin{eqnarray}
{\mathcal{E}}\bigg\llbracket\overline{m(t)}\bigg\rrbracket=
\mathlarger{\mathlarger{\mathcal{E}}}\bigg\llbracket \Phi(\overline{u(t)})-\Phi(u_{\epsilon})|\bigg\rrbracket
-\frac{1}{2}\kappa\int_{t_{\epsilon}}^{t}\mathlarger{\mathlarger{\mathcal{E}}}\bigg\llbracket u^{4}(s)(u(s)-1){\mathlarger{\mathrm{D}}}_{u}{\mathlarger{\mathrm{D}}}_{u}
\Phi(u(s))\bigg\rrbracket ds = 0
\end{eqnarray}
\end{thm}
\begin{proof}
The result follows from the Ito Lemma
\begin{align}
\Phi(\overline{u(t)})=\Phi(u(t_{\epsilon})+\int_{t_{i}^{t}}\mathlarger{\mathrm{D}}_{u}
\Phi(u(t){\psi}(u(t))d\mathlarger{\mathlarger{\mathscr{B}}}(t)\nonumber
+\frac{1}{2}\int_{t_{i}^{T}}\mathlarger{\mathrm{D}}_{u}{\mathrm{D}}_{u}
\Phi(u(t))u^{4}(t)(u(t)-1)dt
\end{align}
Taking the expectation
\begin{align}
& \big\|\overline{u(t)}\big\|_{\mathcal{L}_{1}}=
\mathlarger{\mathlarger{\mathcal{E}}}\bigg\llbracket \Phi(u(t))\bigg\rrbracket=\Phi(u(t_{\epsilon})+\frac{1}{2}k
\mathlarger{\mathlarger{\mathcal{E}}}\bigg\llbracket \kappa\int_{t_{i}}^{T}(\mathlarger{\mathrm{D}}_{u}\mathlarger{\mathrm{D}}_{u}|
\Phi(u(s))u^{4}(t)(u(t)-1)ds\bigg\rrbracket\nonumber\\&=\Phi(u(t_{\epsilon})+\frac{1}{2}k\mathlarger{\mathlarger{\mathcal{E}}}\bigg\llbracket
\int_{t_{i}}^{T}(\mathlarger{\mathrm{D}}_{u}\mathlarger{\mathrm{D}}_{u}\Phi(u(s))u^{4}(s)(u(s)-1)^{4}ds\bigg\rrbracket\nonumber\\&=
\Phi(u(t_{\epsilon})+\mathlarger{\mathlarger{\mathcal{E}}}\bigg\llbracket\int_{t_{i}}^{T}
\mathlarger{\mathrm{I\!H}}\Phi(u(s))ds\bigg\rrbracket
\end{align}
Hence (5.21) follows.
\end{proof}
As a corollary and a consistency check, we can immediately reproduce the moments estimates from previous theorems.
\begin{cor}
Let$ \Phi(\overline{u(t)})=|u(t)|^{2}$ and $\Phi(\overline{u(t)})=|\overline{u(t)}|^{p}$
for all $p>2$ and $t\in\mathbf{X}_{I}\cup\mathbf{X}_{II}$. Then if the polynomial growth conditions hold then the estimates of Theorem (7.4) follow such that
\begin{equation}
\mathlarger{\mathlarger{\mathcal{E}}}(\Phi(\overline{u(t)}))
=\mathlarger{\mathlarger{\mathcal{E}}}\bigg\llbracket|\overline{u(t)}|^{2}\bigg\rrbracket <\infty
\end{equation}
and $\mathlarger{\mathlarger{\mathcal{E}}}(\Phi(\overline{u(t)}))
=\mathlarger{\mathlarger{\mathcal{E}}}(|\overline{u(t)}|^{p})<\infty $
\end{cor}
\begin{proof}
\begin{align}
\mathlarger{\mathlarger{\mathcal{E}}}(\Phi(\overline{u(t)}))&
=\mathlarger{\mathlarger{\mathcal{E}}}\bigg\llbracket|\overline{u(t)}|^{p}\bigg\rrbracket=
u_{\epsilon}^{p}+\frac{1}{2}k\mathlarger{\mathlarger{\mathcal{E}}}\bigg\llbracket\int_{t_{i}}^{t}
\mathlarger{\mathrm{D}}_{u}\mathlarger{\mathrm{D}}_{u}|\Phi(u(s))|\psi(u(t))|^{2}ds\bigg\rrbracket\nonumber\\&=
u(t_{\epsilon}^{2}+\frac{1}{2}k\mathlarger{\mathlarger{\mathcal{E}}}\bigg\llbracket\int_{t_{i}}^{T}
\mathlarger{\mathrm{D}}_{u}\mathlarger{\mathrm{D}}_{u}|u(s)|^{p})u^{4}(s)(u(s)-1)^{4}ds\bigg\rrbracket\nonumber\\&
\le u_{\epsilon}^{2}+\frac{1}{2}kp(p-1)\underbrace{\mathlarger{\mathlarger{\mathcal{E}}}\bigg\llbracket\int_{t_{i}}^{t}
|u(s)|^{p-2}u^{4}(s)(u(s)-1)^{4}ds}_{using~linear~growth~estimate}\bigg\rrbracket\nonumber\\&
\le u_{\epsilon}^{2p-2}+\frac{1}{2}k{K}\mathlarger{\mathlarger{\mathcal{E}}}\bigg\llbracket\int_{t_{i}}^{t}
u^{2p-2}(s)ds\bigg\rrbracket\le u_{\epsilon}^{2p}+\frac{1}{2}k K\mathlarger{\mathlarger{\mathcal{E}}}
\bigg\llbracket\int_{t_{i}}^{t}u^{2p}(s)ds\bigg\rrbracket
\end{align}
and the remainder of the proof follows a similar argument to Thm (7.4).
\end{proof}
\subsection{Exit probabilities from finite cells and exit blowup probabilities for semi-infinite cells}
The exit time of the matter density function diffusion from a finite cell $[u_{a},u_{b}]\subset[u_{\epsilon},\infty]\subset\bar{\mathbf{R}}^{+}$ is $
\tau_{\mathcal{C}}^{exit}=\inf\lbrace t>t_{\epsilon}:|\widehat{u}(t)|\notin\mathcal{C}\rbrace $ and is the first time the diffusion exist $\mathcal{C}$. Then $u(\tau_{\mathcal{C}}^{exit}=u_{a}$ or $u(\tau_{\mathcal{C}}^{exit}=u_{b}$ and $\tau_{\mathcal{C}}^{exit}$ is a stopping time. As before, the hitting times at $u_{a}$ and $u_{b}$ are $\tau_{a}=\inf\lbrace t>t_{\epsilon}:\widehat{u}(t)=u_{a}\rbrace$ and $\tau_{b}=\inf\lbrace t>t_{\epsilon}:\widehat{u}(t)=u_{b}\rbrace$. Then$\tau_{\mathcal{C}}^{exit}=\min(\tau_{a},\tau_{b})=\tau_{a}\wedge\tau_{b}$
\begin{defn}
For a semi-infinite cell $
\mathcal{C}^{\infty}=[u_{a},\infty]=\lim_{u_{b}\rightarrow\infty}[u_{a}, u_{b}]=\mathcal{C} $ The finite "hitting time" for the density function diffusion $|\widehat{u}(t)|$ to hit infinity or blowup, from some initial data $u(t_{\epsilon})\equiv u_{\epsilon}$
$\mathcal{T}=\inf\big\lbrace t>t_{\epsilon}:|\widehat{u}(t)|=\infty\big\rbrace $
The hitting time to the boundaries of the cell are $
\tau_{\alpha}=\inf\big\lbrace t>t_{\epsilon}:|\widehat{u}(t)|=u_{a}\big\rbrace $
and $ \tau_{\beta}=\inf\big\lbrace t>t_{\epsilon}:|\widehat{u}(t)|=u_{b}\big\rbrace
$ with an expected value $ \mathscr{S}\big\lbrace\mathcal{T}_{(a)}\big\rbrace
=\mathscr{S}\lbrace\inf\lbrace t>t_{\epsilon}:|\widehat{u}
(t)|=u_{a}\big\rbrace$. The blowup time or "hitting time to explosion" is the limit as $u_{b}\rightarrow\infty$
\begin{equation}
\mathcal{T}=\lim_{u_{\beta}\rightarrow\infty}u_{(\alpha)}=\inf\big\lbrace t>t_{\epsilon}:|\widehat{u}(t)|=u_{b}\rbrace
\end{equation}
with an expected value $\mathlarger{\mathlarger{\mathcal{E}}}(\mathcal{T})=\mathlarger{\mathlarger{\mathcal{E}}}(\inf\big\lbrace t>0:|\widehat{u}(t)|=\infty\big\rbrace)$. There is an associated probability:$\mathlarger{\mathrm{I\!P}}[\bm{T}<\infty]$,
\end{defn}
\begin{defn}
If $\mathcal{T}$ is a blowup hitting time then
\begin{enumerate}
\item If $\mathlarger{\mathrm{I\!P}}[\mathcal{T}<\infty]=1$ the blowup occurs a.s.
\item If $\mathlarger{\mathrm{I\!P}}[\mathcal{T}<\infty]=0$ the blowup never occurs a.s.
\item If $\mathlarger{\mathrm{I\!P}}[\mathcal{T}<\infty]\in[0,1)$, the blowup will occur with some probability
\item $\mathlarger{\mathrm{I\!P}}[\mathcal{T}<\infty]=1~~<=>
    ~~\mathlarger{\mathrm{I\!P}}[\mathcal{T}=\infty]=0$
\item $\mathlarger{\mathrm{I\!P}}[\mathcal{T}<\infty]=0~~<=>~~
    \mathlarger{\mathrm{I\!P}}[T=\infty]=1$
\end{enumerate}
\end{defn}
The exit probabilities from a finite cell $\mathcal{C}=[u_{a},u_{b}]$ and a semi-infinite cell $\mathcal{C}=[u_{a},\infty]$ can be determined from the harmonic equation and its scale factor solution. The probability that the process reaches $u_{b}$ before it reaches $u_{a}$ is $\mathlarger{\mathrm{I\!P}}[\tau_{b}<\tau_{a}]$ and the probability that the process reaches $u_{a}$ before $u_{b}$ is $\mathlarger{\mathrm{I\!P}}[\tau_{a}<
\tau_{b}]$.
\begin{prop}
For a semi-infinite interval $\mathcal{C}^{\infty}$ the probability that $\tau_{b}=\mathcal{T}<\tau_{a}$ should now be zero so that $\mathlarger{\mathrm{I\!P}}[T<\tau_{a}<\infty]=0$. If this is so then the diffusion $\overline{u(t)}$ does not blowup or become singular for any finite hitting time $\mathcal{T}$.
\end{prop}
\begin{defn}
The diffusion $\widehat{u}(t)$ is transient if for all $t\in\mathbf{X}_{II}
\cup\mathbf{X}_{III}=[t_{\epsilon},\infty]$ one has $
\lim_{t\rightarrow\infty}\mathlarger{\mathrm{I\!P}}[|\overline{u(t)}|=\infty]=1 $
otherwise it is recurrent. Hence, if $\overline{u(t)}$ is recurrent then it does not blowup for any finite t.
\end{defn}
\subsection{Feller's Test for blowup and singular or non-singular behavior}
Given an SDE, the definitive "acid test" for blowup of SDEs is usually Feller's Test. [61,66] Feller's Test gives the necessary and sufficient conditions for explosion or non-explosion of a SDE or diffusion given only the coefficients of the SDE. First, the preliminary theorem for recurrence and transience is given is as follows [66].
\begin{thm}
Let $d\overline{u(t)}=\phi(u(t))dt+\psi(u(t))d\mathlarger{\mathlarger{\mathscr{B}}}(t)(t)$ be a generic nonlinear SDE with generator $\mathlarger{\mathrm{I\!H}}$ given by
\begin{equation}
\mathlarger{\mathrm{I\!H}}=\phi(u(t))\mathlarger{\mathrm{D}}_{u}
+\frac{1}{2}|\psi(u(t)|^{2}\mathlarger{\mathrm{D}}_{u}\mathlarger{\mathrm{D}}_{u}
\end{equation}
The pure diffusion with zero drift then arises in the limit
\begin{equation}
d\overline{u(t)}=\lim_{\phi(u(t))\uparrow 0}\phi(u(t))dt+\psi(u(t))d\mathlarger{\mathlarger{\mathscr{B}}}(t)
=\psi(u(t))d\mathlarger{\mathlarger{\mathscr{B}}}(t)
\end{equation}
Define
\begin{equation}
\theta(u(t))=\int_{u_{\epsilon}}^{u(t)}\phi(v(t))|\psi(v(t)|^{-2}dv(t)
\end{equation}
and $\Lambda(\pm u(t))=\exp(\pm \theta(u(t))$, with harmonic condition $\mathrm{I\!H}\Lambda(-\theta(u(t))$. Let $\mathcal{C}
=[u_{\alpha},u_{\beta}]\subset[1,\infty)$ with $u_{\epsilon}\in\mathcal{C}$ where $1<u_{\alpha},u_{\beta}\le \infty$. The probability that $\overline{u(t)}$ is contained within the domain or interval $\mathcal{C}$ is then $\mathlarger{\mathlarger{\mathrm{I\!P}}}(\overline{u(t)}\in\Omega)$ Then:
\begin{enumerate}
\item The diffusion corresponding to the generator is recurrent if both of the following hold
\begin{align}
&\int_{u_{\alpha}}^{u_{\epsilon}}\Lambda(-\widehat{u}(t))d\widehat{u}(t)
\equiv\int_{u_{\alpha}}^{u_{\epsilon}}\exp(-\theta(\widehat{u}(t))d\widehat{u}(t)\equiv \int_{u_{\alpha}}^{u_{\epsilon}}\exp\left(2\int_{u_{\epsilon}}^{u(t)}\phi(v(t))|\psi(v(t))|^{-2}dv(t)
\right)du(t)=\infty
\end{align}
\begin{align}
&\int_{u_{\epsilon}}^{u_{\beta}}\Lambda(-\widehat{u}(t))d\widehat{u}(t)
\equiv\int_{u_{\epsilon}}^{u_{\beta}}\exp(-\theta(\widehat{u}(t))d\widehat{u}(t)\equiv \int_{u_{\epsilon}}^{u_{\beta}}\exp\left(2\int_{u_{\epsilon}}^{u(t)}\phi(v(t))|\psi(v(t))|^{-2}dv(t)
\right)du(t)=\infty
\end{align}
\item The diffusion corresponding to the generator (8.27) is transient if the following holds
\begin{align}
&\int_{u_{\alpha}}^{u_{\epsilon}}\Lambda(-\widehat{u}(t))d\widehat{u}(t)
\equiv\int_{u_{\alpha}}^{u_{\epsilon}}\exp(-\theta(\widehat{u}(t))d\widehat{u}(t)\equiv \int_{u_{\alpha}}^{u_{\epsilon}}\exp\left(2\int_{u_{\epsilon}}^{u(t)}\phi(v(t))|\psi(v(t))|^{-2}dv(t)
\right)du(t)=\infty(<\infty)
\end{align}
\begin{align}
&\int_{u_{\epsilon}}^{u_{\beta}}\Lambda(-\widehat{u}(t))d\widehat{u}(t)
\equiv\int_{u_{\epsilon}}^{u_{\beta}}\exp(-\theta(\widehat{u}(t))d\widehat{u}(t)\equiv \int_{u_{\epsilon}}^{u_{\beta}}\exp\left(2\int_{u_{\epsilon}}^{u(t)}\phi(v(t))|\psi(v(t))|^{-2}dv(t)
\right)du(t)<\infty(=\infty)
\end{align}
\item The diffusion corresponding to the generator is also transient if the following holds
\begin{align}
&\int_{u_{\alpha}}^{u_{\epsilon}}\Lambda(-\widehat{u}(t))d\widehat{u}(t)
\equiv\int_{u_{\alpha}}^{u_{\epsilon}}\exp(-\theta(\widehat{u}(t))d\widehat{u}(t)\equiv \int_{u_{\alpha}}^{u_{\epsilon}}\exp\left(-2\int_{u_{\epsilon}}^{u(t)}\phi(v(t))|\psi(v(t))|^{-2}dv(t)
\right)du(t)<\infty)
\end{align}
\begin{align}
&\int_{u_{\epsilon}}^{u_{\beta}}\Lambda(-\widehat{u}(t))d\widehat{u}(t)
\equiv\int_{u_{\epsilon}}^{u_{\beta}}\exp(-\theta(\widehat{u}(t))d\widehat{u}(t)\equiv \int_{u_{\epsilon}}^{u_{\beta}}\exp\left(-2\int_{u_{\epsilon}}^{u(t)}\phi(v(t))|\psi(v(t))|^{-2}dv(t)
\right)du(t)<\infty
\end{align}
and the exit probabilities from $\mathcal{C}$ are then
\begin{align}
\mathlarger{\mathlarger{\mathrm{I\!P}}}(\overline{u(t)}=u_{\alpha})= \frac{\int_{u(t)}^{u_{\beta}}\Lambda(-u(t))du(t)}{\int_{u_{\alpha}}^{u_{\beta}}\Lambda(-u(t))du(t)}=
\frac{\int_{u(t)}^{u_{\beta}}\exp\left(-2\int_{u_{\epsilon}}^{u(t)}\phi(v(t))|\psi(v(t))|^{-2}dv(t)
\right)du(t)}{\int_{u_{\alpha}}^{u_{\beta}}\exp\left(-2\int_{u_{\epsilon}}^{u(t)}\phi(v(t))|\psi(v(t))|^{-2}dv(t)
\right)du(t)}
\end{align}
\begin{align}
\mathlarger{\mathlarger{\mathrm{I\!P}}}(\overline{u(t)}=u_{\beta})= \frac{\int_{u_{\alpha}}^{u(t)}\Lambda(-u(t))du(t)}{
\int_{u_{\alpha}}^{u_{\beta}}\Lambda(-u(t))du(t)}=
\frac{\int_{u_{\alpha}}^{u(t)}\exp\left(-2\int_{u_{\epsilon}}^{u(t)}\phi(v(t))|\psi(v(t))|^{-2}dv(t)
\right)du(t)}{\int_{u_{\alpha}}^{u_{\beta}}\exp\left(-2\int_{u_{\epsilon}}^{u(t)}\phi(v(t))|\psi(v(t))|^{-2}dv(t)
\right)du(t)}
\end{align}
\item From (8.37), the probability of a singularity or blowup is then given by the limit in which $u_{\beta}\rightarrow \infty$ so that
\begin{align}
&\lim_{u_{\beta}\uparrow\infty}{\mathlarger{\mathrm{I\!P}}}(\overline{u(t)}
=u_{\beta})={\mathlarger{\mathrm{I\!P}}}(\overline{u(t)}=\infty)=\lim_{u_{\beta}\uparrow\infty}\frac{\int_{u_{\alpha}}^{u(t)}\Lambda(-u(t))du(t)}{
\int_{u_{\alpha}}^{u_{\beta}}\Lambda(-u(t))du(t)}\nonumber\\&=\lim_{u_{\beta}\uparrow\infty}\frac{\int_{u_{\alpha}}^{u(t)}\exp\bigg(-2\int_{u_{\epsilon}}^{u(t)}\phi(v(t))
|\psi(v(t))|^{-2}dv(t)\bigg)du(t)}{\int_{u_{\alpha}}^{u_{\beta}}\exp\bigg(-2\int_{u_{\epsilon}}^{u(t)}\phi(v(t))|\psi(v(t))|^{-2}dv(t)
\bigg)du(t)}
\end{align}
\end{enumerate}
\end{thm}
The proof is given in Pinsky [66].
\begin{cor}
For a driftless diffusion with $\phi(u(t))\rightarrow 0$ it follows that the probability of blowup or singularity formation is zero for any coefficient $\psi(u(t))$ since
\begin{align}
&\lim_{u_{\beta}\uparrow\infty}\mathlarger{\mathlarger{\mathrm{I\!P}}}\big(\overline{u(t)}
=u_{\beta}\big)=\mathlarger{\mathlarger{\mathrm{I\!P}}}\big(\overline{u(t)}=\infty\big)=\lim_{\phi(u(t)\uparrow 0}\lim_{u_{\beta}\uparrow\infty}\frac{\int_{u_{\alpha}}^{u(t)}\Lambda(-u(t))du(t)}{
\int_{u_{\alpha}}^{u_{\beta}}\Lambda(-u(t))du(t)}\nonumber\\&=\lim_{\phi(u(t)\uparrow 0}\lim_{u_{\beta}\uparrow\infty}\frac{\int_{u_{\alpha}}^{u(t)}\exp\bigg(-2\int_{u_{\epsilon}}^{u(t)}\phi(v(t))
|\psi(v(t))|^{-2}dv(t)\bigg)du(t)}{\int_{u_{\alpha}}^{u_{\beta}}\exp\bigg(-2\int_{u_{\epsilon}}^{u(t)}\phi(v(t))|\psi(v(t))|^{-2}dv(t)
\bigg)du(t)}\\&
=\lim_{u_{\beta}\uparrow\infty}\frac{\int_{u_{\alpha}}^{u(t)}du(t)}{\int_{u_{\alpha}}^{u_{\beta}}du(t)}
=\lim_{u_{\beta}\uparrow\infty}\bigg|\frac{u(t)-u_{\alpha}}{u_{\beta}-u_{\alpha}}\bigg|=0
\end{align}
\end{cor}
Feller's Test is now stated formally for any general(nonlinear)SDE describing a diffusion
$\overline{u(t)}$.
\begin{thm}
Given a generic nonlinear SDE for $\overline{u(t)}$ for all $t\ge t_{\epsilon}$
\begin{equation}
d\overline{u(t)}=\phi(u(t))dt + \psi(u(t))d\mathlarger{\mathlarger{\mathscr{B}}}(t)
\end{equation}
for the diffusion process $\overline{u(t)}$ defined for all $t>t_{\epsilon}$ and where one can have $ \phi(u(t))=\psi(u(t))$. We then consider the specific SDE with
\begin{equation}
d\overline{u(t)}=\lim_{\phi\rightarrow 0}\phi(u(t))dt + \psi(u(t))
d\mathlarger{\mathlarger{\mathscr{B}}}(t)=\kappa^{1/2}u^{2}(t)(u(t)-1)^{1/2}
d\mathlarger{\mathlarger{\mathscr{B}}}(t)
\end{equation}
for all $t\in\mathbf{X}_{I}\cup\mathbf{X}_{II}$ and where $f(u(t))$ is an arbitrary functional to be reduced to zero. The generator of the full diffusion
$\mathlarger{\mathrm{I\!H}}=\phi(u)\mathlarger{\mathrm{D}}_{u}
+\frac{1}{2}|\psi(u(t))|^{2}\mathlarger{\mathrm{D}}_{u}\mathlarger{\mathrm{D}}_{u} $. The underlying deterministic ODE is $du(t)/dt=\psi(u(t))$ and has a solution,
for some initial data $u_{\epsilon}=u(t_{\epsilon})$, which explodes at $t=t_{*}=t_{\epsilon}+\epsilon$ then $u(t*)=\infty$ for some $t_{*}>t_{\epsilon}=t_{\epsilon}+\epsilon$. Then for
$u(t)\in[-\infty, \infty)$. The stochastic SDE (8.41) does not explode for any $t>t_{\epsilon}$ if at least one of the following hold:
\begin{align}
&|\bm{\mathcal{FEL}}(|u_{\epsilon}|,\infty)|
=\bigg|\int_{|u_{\epsilon}|}^{\infty}\bigg(\frac{
\int_{|u_{\epsilon}|}^{|u|}d\hat{u}(t)[\psi(\hat{u}(t))]^{-2}
\exp[+\psi(\hat{u}(t)]}{\exp[+\theta(u(t))]}\bigg)du(t)\bigg|=\infty\\&
|\bm{\mathcal{FEL}}(-\infty,|u_{\epsilon}|) |=\bigg|\int_{-\infty}^{|u_{\epsilon}|}\bigg(\frac{\int_{|u_{\epsilon}|}^{|u(t)|}
d\hat{u}(t)[\psi(\hat{u}(t))]^{-2}\exp[+\theta(\hat{u}(t)]}
{\exp[+\theta(u(t))]}\bigg)du(t) \bigg|=\infty
\end{align}
\begin{enumerate}
\item For $u\in[u_{\epsilon},\infty)$ only (5.43) need be considered. Then the SDE
    (5.41) does not blow up for any $t\in\mathbf{R}^{+}$ iff
\begin{equation}
|\bm{\mathcal{FEL}}(|u_{\epsilon}|,\infty)|
=\bigg|\int_{|u_{\epsilon}|}^{\infty}
\bigg(\frac{\int_{|u_{\epsilon}|}^{|u(t)|}d\hat{u}(t)
[\psi(\hat{u}(t))]^{-2}\exp[+\theta(\hat{u}(t)]}
{\exp[+\theta(u(t))]}\bigg)du(t)\bigg|=\infty
\end{equation}
\item The stochastic SDE does blow up for any $t\in\mathbf{X}_{II}\bigcup\mathbf{X}_{III}$ if
\begin{equation}
|\bm{\mathcal{FEL}}(|u_{\epsilon}|,\infty)|
=\bigg|\int_{|u_{\epsilon}|}^{\infty}\bigg(\frac{
\int_{|u_{\epsilon}|}^{|u(t)|}d\hat{u}(t)
[(\psi(\hat{u}t)))]^{-2}\exp[+\theta(\hat{u}(t)]}
{\exp[+\theta(u(t))]}\bigg)du(t) \bigg|<\infty
\end{equation}
\end{enumerate}
and where
\begin{equation}
\theta(u(t))=\int^{|u(t)|}2 \phi[\hat{u}(t)]
[\psi(\hat{u}(t))]^{-2} d\hat{u}(t)
\end{equation}
Setting
\begin{equation}
\Lambda(\pm u(t))=\exp[\pm \theta(u(t))]=\exp\left[\pm\int^{|u(t)|}2
\phi(\hat{u}(t))(\psi(\hat{u}(t))]^{-2} d\hat{u}(t)\right]
\end{equation}
these can be written even more succinctly as
\begin{align}
&|\bm{\mathcal{FEL}}(|u_{\epsilon}|,\infty)|
=\underbrace{\bigg|\int_{|u_{\epsilon}|}^{\infty}\bigg(\Lambda[-u(t)]
\int_{|u_{\epsilon}|}^{|u(t)|}d\hat{u}(t)
[\psi(\hat{u}(t))]^{-2}\Lambda[+u(t)]
\bigg)du(t)\bigg|=\infty}_{no~blowup}\\&
|\bm{\mathcal{FEL}}(|u_{\epsilon}|,\infty)|=
\underbrace{\bigg|\int_{|U_{\epsilon}|}^{\infty}\bigg(\Lambda[-u(t)]
\int_{|u_{\epsilon}|}^{|u(t)|}d\hat{u}(t)
[{\psi}(\hat{u}(t))]^{-2}\Lambda[+u(t)]
\bigg)du(t)\bigg|<\infty}_{blowup}
\end{align}
\end{thm}
\raggedbottom
Rigorous technical proofs of Feller's test for explosion of SDEs relies on a detailed analysis of exit times from an interval and the application of the Feynman-Kac formula [61,66].
\begin{exam}
As an example of application of Feller's Test consider the Riccati SDE given by
\begin{equation}
d\overline{u(t)}=(\alpha+|u(t)|^{2})d\mathlarger{\mathlarger{\mathscr{B}}}(t)
\end{equation}
with $\alpha >0$, $t\in\mathbf{R}^{+}$ and $\overline{u(t)}\in[0,\infty)$. The underlying ODE $du(t)=(\alpha+|u(t)|^{2})dt$ blows up at finite $t$. If $u(0)=0$ is the initial data, then the solution is $u(t)=\sqrt{\alpha}\tan|t|$ so that $u(\pi/2)=\infty$. However, the driftless Ito diffusion (8.51) describes a true martingale and will not explode. Applying the Feller Test (8.53)
\begin{align}
&\bm{\mathcal{FEL}}(0,\infty)=\int_{0}^{\infty}\bigg|\int_{0}^{u(t)} \frac{dv(t)}{(\alpha+|v(t)|^{2})^{2}}\bigg|du(t)\nonumber\\&
=\int_{0}^{\infty}\bigg|\frac{\tan^{-1}(u(t)/\sqrt{\alpha})+ \frac{\sqrt{\alpha}u(t)}{(\alpha+|u(t)|^{2})}}{2\alpha^{2/3}}
\bigg|du(t)=\bigg| \frac{u(t)\tan^{-1}(u(t)/\sqrt{\alpha})}{2\alpha^{3/2}}\bigg|_{0}^{\infty}=\infty
\end{align}
so that the criterion for non-explosion is satisfied.
\end{exam}
\begin{exam}
The SDE
\begin{equation}
d\overline{u(t)}=(\alpha+|u(t)|^{2})dt+d\mathlarger{\mathlarger{\mathscr{B}}}(t)
\end{equation}
with additive noise, does explode. Clearly, this Ito diffusion with drift is not a martingale. Applying the Feller Test with $\phi(u(t))=(\alpha+|u(t)|^{2})$ and $\psi(u(t))=1$ then (8.48) is
\begin{align}
&\Lambda(\pm u(t))=\exp\bigg(\pm 2\int_{0}^{u(t)}2\phi(u(t))(\psi(u(t))^{-2}du(t)\bigg)\nonumber\\&=\exp\bigg(\pm 2\int_{0}^{u(t)}(\alpha+|v(t)|^{2}dv(t)\bigg)=\exp(\pm 2\alpha u(t) \pm \tfrac{2}{3}|u(t)|^{3})
\end{align}
Then
\begin{align}
\bm{\mathcal{FEL}}(0,\infty)&=\int_{0}^{\infty}\bigg|3\exp(-2\alpha u(t)-\frac{2}{3}|u(t)|^{3})\int_{0}^{u(t)}\exp(+2\alpha u(t)+t\frac{2}{3}u(t))du(t)\bigg|du(t)\nonumber\\&
=\int_{0}^{\infty}\bigg|\exp(-2\alpha u(t)-\tfrac{2}{3}|u(t)|^{3})   \exp\bigg(\frac{\tfrac{2}{3}(3\alpha+1)u(t)}{6\alpha+2}\bigg)\bigg|du(t)\nonumber\\&
=\frac{9\exp(-\tfrac{1}{3}(6\alpha-2(3\alpha+1)+2)u(t))}{(6\alpha+2) (6\alpha-2(3\alpha+1)+2)}<\infty
\end{align}
so the SDE explodes.
\end{exam}
{\allowdisplaybreaks
The Feller Test is now applied to the SDE
\begin{equation}
d\overline{u(t)}=\psi(u(t)
d\mathlarger{\mathlarger{\mathscr{B}}}(t)=\kappa^{1/2}(u(t))^{2}(u(t)-1)^{1/2}
d\mathlarger{\mathscr{B}}(t)\nonumber
\end{equation}
\begin{thm}
The density function diffusion $d\overline{u(t)}=\psi(u(t)d\mathlarger{\mathlarger{\mathscr{B}}}(t)
=\kappa^{1/2}(u(t))^{2}(u(t)-1)^{1/2}d\mathlarger{\mathlarger{\mathscr{B}}}(t)$ with initial data $u_\epsilon=u(t_{\epsilon})$ has no blowups or singularities for any $t\in\mathbf{X}_{II}\cup\mathbf{X}_{III}$ since $\bm{\mathcal{FEL}}(u_{\epsilon},\infty)|=\infty$
\end{thm}
\begin{proof}
Using (8.49)
\begin{align}
|\bm{\mathcal{FEL}}(|u_{\epsilon}|,\infty)|
&=\bigg|\int_{|u_{\epsilon}|}^{\infty}\bigg(\Lambda(-u(t))
\int_{|u_{\epsilon}|}^{|u(t)|}d\hat{u}(t)
[\psi(\hat{u}(t))]^{-2}\Lambda(+u(t))
\bigg)du(t)\bigg|\nonumber\\&
=\bigg|\int_{|u_{\epsilon}|}^{\infty}\bigg(\Lambda(-u(t))
\int_{|u_{\epsilon}|}^{|u(t)|}\frac{d\hat{u}(t)}{(\hat{u}(t))^{4}(\hat{u}(t)-1)}
\Lambda(+u(t))\bigg)du(t)\bigg|\nonumber\\&
=\bigg|\int_{|u_{\epsilon}|}^{\infty}\bigg(
\int_{|u_{\epsilon}|}^{|u(t)|}\frac{d\hat{u}(t)}{(\hat{u}(t))^{4}(\hat{u}(t)-1)}
\bigg)du(t)\bigg|\nonumber\\&
=\lim_{u(t)\uparrow\infty}\bigg|\int_{|u_{\epsilon}|}^{u(t)}\bigg(\int_{|u_{\epsilon}|}^{|u(t)|}\frac{d\hat{u}(t)}{(\hat{u}(t))^{4}(\hat{u}(t)-1)}
\bigg)du(t)\bigg|
\end{align}
Since $\Lambda(\pm u(t))=1$ when $\phi(u(t))=0$ for a driftless diffusion. Using a series expansion for $n\ge 0$
\begin{equation}
\frac{1}{(u(t))^{4}(u(t)-1)}=\sum_{n=-\infty}^{\infty}\left[\tfrac{1}{6}(-10^{n}(24+26n+9n^{2}+n^{3})\right]
(u(t)-1)^{n}
\end{equation}
the following estimate can be made
\begin{align}
&\bm{\mathcal{FEL}}(u_{\epsilon},\infty)=\lim_{u(t)\rightarrow\infty}\sum_{n=-\infty}^{\infty}\left[\tfrac{1}{6}(-10^{n}(24+26n+9n^{2}+n^{3})\right]
\times\left|\int_{u_{\epsilon}}^{u(t)}\left(\int_{u_{\epsilon}}^{\hat{u}(t)}(\hat{u}(t)-1)^{n}d\hat{u}(t)\right) du(t)\right|\nonumber\\&=\lim_{u(t)\rightarrow\infty}\sum_{n=-\infty}^{\infty}\left[\tfrac{1}{6}(-1)^{n}(24+26n+9n^{2}+n^{3})\right]
\times\left|\int_{u_{\epsilon}}^{u(t)}\left(\frac{(u(t)-1)^{n+1}-(u_{\epsilon}-1)^{n+1}}{(n+1)}
\right) du(t)\right|\nonumber\\&
=\lim_{u(t)\rightarrow\infty}\sum_{n=-\infty}^{\infty}
\bigg[\tfrac{1}{6}(-1)^{n}(24+26n+9n^{2}+n^{3})\bigg]\nonumber\\&\times\bigg(\frac{(u(t)-1)^{n+2}}{n^{2}+3n+2}-
\frac{u_{\epsilon}-1)^{n+2}}{n^{2}+3n+2}-\frac{u(t)(u_{\epsilon}-1)^{n+1}}{n+1}
+\frac{u_{\epsilon}(u_{\epsilon}-1)^{n+1}}{n+1}\bigg)\nonumber\\&
=\lim_{u(t)\rightarrow\infty}\sum_{n=-\infty}^{\infty}
\bigg[\tfrac{1}{6}(-1)^{n}(24+26n+9n^{2}+n^{3})\bigg]
\times\bigg(\frac{(u(t)-1)^{n+2}}{n^{2}+3n+2}-\frac{u(t)(u_{\epsilon}-1)^{n+1}}{n+1}+C\bigg)\nonumber\\&
=\lim_{u(t)\rightarrow\infty}\sum_{n=-\infty}^{\infty}
\bigg[\tfrac{1}{6}(-1)^{n}(24+26n+9n^{2}+n^{3})\bigg]\times\bigg(\alpha(u(t)-1)^{n+2}-\beta u(t)+C\bigg)=\infty
\end{align}
Hence, the stochastic ODE or diffusion does not explode for any $t\in\mathbf{X}_{II}\cup\mathbf{X}_{III}$.
\end{proof}}
Rather than computing (8.49) explicitly for a SDE, it is possible, and more tractable to reduce down the expression for $\bm{\mathcal{FEL}}(u_{\epsilon},\infty))$ in terms of a limit and apply the basic L'Hopital rule.
\begin{lem}
Given a nonlinear SDE,
\begin{equation}
d\overline{u(t)}=\phi(u(t))dt + \psi(u(t))d\mathlarger{\mathlarger{\mathscr{B}}}(t)
\end{equation}
with $u(t)\in [u_{\epsilon},\infty)$ for all $t>t_{\epsilon}$ then the general Feller Tests
can be expressed as
\begin{equation}
|\bm{\mathcal{FEL}}(u_{\epsilon},\infty)|=\lim_{\hat(t)\rightarrow -\infty}
\underbrace{\bigg|\int_{u_{\epsilon}}^{u(t)}d\hat{u}(t)[\phi^{-1}(\hat{u}(t)]\bigg|=
\infty}_{no~blowup}
\end{equation}
\begin{equation}
|\bm{\mathcal{FEL}}(u_{\epsilon}\infty)|=\lim_{u(t)\rightarrow
-\infty}\underbrace{\bigg|\int_{u_{\epsilon}}^{u(t)} d\hat{u}(t)[\phi^{-1}(\hat{u}(t)]
\bigg|<\infty}_{blowup}
\end{equation}
Hence, all driftless Ito SDEs or diffusions of the form
\begin{equation}
d\overline{u(t)}=\psi(u(t))d\mathlarger{\mathlarger{\mathscr{B}}}(t)
=\lim_{\phi(t)\rightarrow 0}[\phi(u(t))dt]+\psi(u(t))
d\mathlarger{\mathlarger{\mathscr{B}}}(t)
\end{equation}
 never explode for any $t\in\mathbf{R}^{+}$ since $|\bm{\mathcal{FEL}}(u_{\epsilon}\infty)|=\infty$ when $\phi(u(t))=0$. Hence
\begin{equation}
|\bm{\mathcal{FEL}}(u_{\epsilon},\infty)|=\lim_{\phi(u)\rightarrow 0}\bigg|\lim_{u(s)\rightarrow-\infty}\bigg|\int_{u_{\epsilon}}^{u(t)}d\hat{u}(t)[\phi^{-1}
(\hat{u}(t)]\bigg|=\infty
\end{equation}
Hence, the SDE $d\overline{u(t)}=\psi(u(t))d\mathlarger{\mathlarger{\mathscr{B}}}(t)=k|u^{4}(t)
(u(t)-1)|d\mathlarger{\mathlarger{\mathscr{B}}}(t)$ cannot explode or become singular for any $t>t_{\epsilon}$.
\end{lem}
\begin{proof}
Concentrating on equation (8.61) only, it is useful to write it in the form
\begin{align}
&|\bm{\mathcal{FEL}}(u_{\epsilon},\infty)|=\int_{u_{\epsilon}}^{\infty}
\bigg(\frac{U(u(t))}{V(u(t))}\bigg)du(t)
=\bigg|\int_{u_{\epsilon}}^{\infty}\bigg(\frac{\int_{u_{\epsilon}}^{|u(t)|}
d\hat{u}(t)[\psi(\bar{u}(t))]^{-2}\exp[+\psi(\hat{u}(t)]}{\exp[+\theta(u(t))]}\bigg)
du(t)\bigg|\nonumber\\&=\lim_{u\rightarrow\infty}\int_{u_{\epsilon}}^{u}
\bigg|\bigg(\frac{\int_{u_{\epsilon}}^{|u(t)|}d\hat{u}(t)[\psi(\hat{u}(t))]^{-2}
\exp[+\theta(\hat{u}(t)]}{\exp[+\theta(u)(t))]}\bigg)du(t)\bigg|\nonumber\\&=
\int_{u_{\epsilon}}^{u}\bigg|\lim_{u\rightarrow\infty}\bigg(\frac{\int_{u_{\epsilon}}^{|u(t)|}
d\hat{u}(t)[\psi(\hat{u}(t))]^{-2}\exp[+\theta(\hat{u}(t)]}{\exp[+\theta(u)(t))]}\bigg)du(t)\bigg|
\end{align}
where
\begin{align}
&U(\hat(t))=\int_{u_{\epsilon}}^{|u(t)|}
d\hat{u}(t)[\psi(\hat{u}(t))]^{-2}
\exp[+\theta(\hat u(t)]\\&V(\hat(t))=\exp(+\phi(u(t)))
\end{align}
Applying the L'Hopital rule to the integrand in (5.66)
\begin{equation}
\lim_{u(t)\rightarrow\infty}\frac{U(u(t))}{V(u(t))}=
\lim_{u(t)\rightarrow \infty}\frac{\frac{\delta}{\delta u(t)}U(u(t))}
{\frac{\delta}{\delta u(t)}V(u(t))}\equiv \lim_{u(t)\rightarrow\infty}
\frac{\mathlarger{\mathrm{D}}_{u}
{U}(u(t))}{\mathlarger{\mathrm{D}}_{u}{V}(u(t))}
\end{equation}
where again $\mathlarger{\mathrm{D}}_{u}=\delta/\delta u(t)$. Then the convergence properties can be deduced as
\begin{align}
&\lim_{u(t)\rightarrow \infty}\bigg\lvert\int_{u_{\epsilon}}^{u(t)|}\bigg(
\frac{U(\hat{u}(t))}{V(\hat{u}(t))}\bigg)\bigg\rvert d\hat{u}(t)
\sim \bigg\lvert\int_{u_{\epsilon}}^{u(t)}\lim_{u(t)\rightarrow \infty}\bigg(
\frac{U(\hat{u}(t))}{V(\hat{u}(t))}du(t)\bigg)\bigg\rvert
\nonumber\\& =\bigg\lvert\int_{u_{\epsilon}}^{u(t)}\lim_{u(t)\rightarrow\infty}
\bigg(\frac{\mathlarger{\mathrm{D}}_{u}{U}(\hat{u}(t))}
{\mathlarger{\mathrm{D}}_{u}{V}
(\hat{u}(t))}\bigg)d\hat{u}(t)\bigg\rvert
\end{align}
so that
\begin{equation}
\bigg\lvert\int_{u_{\epsilon}}^{u(t)}\lim_{u(t)\rightarrow \infty}
\bigg(\frac{\mathlarger{\mathrm{D}}_{u}U(\hat{u}(t))}
{\mathlarger{\mathrm{D}}_{u}V(\hat{u}(t))}
\bigg)d\hat{u}(t)\bigg\rvert=\bigg|\int_{u_{\epsilon}}^{u(t)}\lim_{u(t)
\rightarrow \infty}\bigg(\frac{\psi(\hat{u}(t))]^{-2}\exp(+\theta(\hat{u}(t))}
{\mathlarger{\mathrm{D}}_{u}{\theta(\hat{u}(t))\exp(+\theta(\hat{u}(t)))}}
\bigg)d\hat{u}(t)\bigg|
\end{equation}
and the exponential terms cancel leaving
\begin{align}
&\bm{\mathcal{FEL}}(u_{\epsilon},\infty)=\bigg\lvert\int_{u_{\epsilon}}^{u(t)}
\lim_{u(t)\rightarrow \infty}\bigg(\frac{\mathlarger{\mathrm{D}}_{u}
{U}(\hat{u}(t))}{\mathlarger{\mathrm{D}}_{u}{V}(\hat{u}(t))}\bigg)
d\hat{u}(t)\bigg\rvert=\bigg|\int_{u_{\epsilon}}^{u(t)}\lim_{u(t)\rightarrow \infty}\bigg(\frac{[\psi(\hat{u}(t))]^{-2}}{\mathlarger{\mathrm{D}}_{u}
\theta(\hat{u}(t)}\bigg)d\hat{u}(t)\bigg|
\end{align}
But the derivative of (5.47) is
{\allowdisplaybreaks
\begin{equation}
\mathlarger{\mathrm{D}}_{u}\theta(u(t))=\mathlarger{\mathrm{D}}_{u}\bigg|\int_{u_{\epsilon}}^{u}\phi(u(t))
[\psi(\hat{u}(t)]^{-2}d\hat{u}(t)\bigg|=\phi(u(t))[\psi(u(t)]^{-2}\bigg|
\end{equation}
so that
\begin{align}
\bm{\mathcal{FEL}}(u_{\epsilon},\infty)&
=\bigg\lvert\int_{u_{\epsilon}}^{u(t)}\lim_{u(t)\rightarrow\infty}
\bigg(\frac{U(\hat{u}(t))}{V(\hat{u}(t))}\bigg)
d\hat{u}(t)\bigg\rvert
=\bigg\lvert\int_{u_{\epsilon}}^{u(t)}\lim_{u(t)\rightarrow \infty}
\bigg(\frac{\mathlarger{\mathrm{D}}_{u}U
(\hat{u}(t))}{\mathlarger{\mathrm{D}}_{u}V(\hat{u}(t))}\bigg)d\hat{u}(t)\bigg\rvert\nonumber\\&
=\bigg\lvert\int_{u_{\epsilon}}^{u(t)}\lim_{u(t)\rightarrow \infty}\bigg(
\frac{(\psi(u(t))^{-2}}{\phi(u(t))(\psi(u(t))^{-2}}\bigg)
d\hat{u}(t)\bigg)\equiv\bigg|\int_{u_{\epsilon}}^{\hat{u}(t)}\lim_{u(t)\rightarrow \infty}\phi^{-1}(u(t))d\hat{\hat{u}}(t)\bigg|\nonumber\\&=\lim_{u(t)\rightarrow \infty}\bigg|\int_{u_{\epsilon}}^{\hat{u}(t)}\phi^{-1}(u(t))d\hat{u}(t)
\bigg|
\end{align}
Then
\begin{align}
&|\bm{\mathcal{FEL}}(u_{\epsilon},\infty)|=\lim_{\phi(u)\rightarrow 0}\bigg|
\lim_{u(s)\rightarrow \infty}\bigg|\int_{u_{\epsilon}}^{u(t))}d\hat{u}(t)[\phi^{-1}(\hat{u}(t)]\bigg|
\bigg|=\lim_{\phi(u(t))\rightarrow 0}\bigg|
\int_{u_{\epsilon}}^{\infty}d\bar{u}(t)
[\phi^{-1}(\hat{u}(t)]\bigg|=\infty
\end{align}}
Hence, pure nonlinear driftless diffusions of the form $d\overline{u(t)}=\psi(u(t))d\mathlarger{\mathlarger{\mathscr{B}}}(t)$ never blow up since
$\phi(u(t))=0$ for these SDEs.
\end{proof}
\section{Non-explosion criteria via Lyaponov functions and exponents}
Rigorous no-blowup criteria in terms of Lyuponov functionals (LFs) are established and applied to the diffusion $\overline{u}(t)$. If $\overline{u}(t)$ is indeed a martingale, then consistency should demand that (LF) criteria also predict no blowup for this diffusion for all $t\in\mathbf{X}_{II}\bigcup\mathbf{X}_{III}$. Lyapunov functionals and exponents are constructed for the generator $\mathrm{I\!H}$.
\subsection{Existence of Lyapunov functionals in relation to the non-explosion of solutions}
{\allowdisplaybreaks
Given the generator $\mathlarger{\mathrm{I\!H}}$ for the diffusion $\overline{u(t)}$, it is possible to define a stochastic Lyapunov functional $V(\overline{u(t)})$, related to boundedness and non-explosion of the density function diffusion [67,68]. In particular, Khasminksii's no blow-up theorem can be utilised [67]. Given the generator, it is necessary to impose certain conditions only on a functional $V(u(t))$ for the underlying deterministic ODE.
\begin{thm}
Let $\overline{u(t)}$ be the density function diffusion satisfying the nonlinear stochastic differential equation $d\overline{u(t)}=\psi(u(t))
d\mathlarger{\mathlarger{\mathscr{B}}}(t)(t)\equiv \kappa^{1/2}(u(t))^{2}(u(t)-1)^{1/2}d\mathlarger{\mathlarger{\mathscr{B}}}(t)$ for all $t\in \mathbf{X}_{b}\cup\mathbf{X}_{c}$. The generator
of the diffusion is $\mathlarger{\mathlarger{\mathrm{I\!H}}}=\tfrac{1}{2}|\psi(u(t))^{2}
\mathlarger{\mathrm{D}}_{u}\mathlarger{\mathrm{D}}_{u}$. Let $V:\mathbf{X}_{II}\cup\mathbf{X}_{II}\times\mathbf{X}^{+} \rightarrow\mathbf{R}^{+}$
be a $C^{2}$-functional of the deterministic $u(t)|$. The corresponding stochastic functional $V(\overline{u(t)})$ for the diffusion $\overline{u(t)}$, for all $t\in\mathbf{X}_{II}\cup\mathbf{X}_{III}$ has an Ito expansion
\begin{equation}
V(\overline{u(t)})={V}(u_{\epsilon})+\int_{t_{\epsilon}}^{t}\mathlarger{\mathrm{D}}_{u}
{V}\psi(u(s))(u(s)) d\mathlarger{\mathlarger{\mathscr{B}}}(s)
+\frac{1}{2}\int_{t_{\epsilon}}^{t}\mathlarger{\mathrm{D}}_{u}\mathlarger{\mathrm{D}}_{u}
{V}(u(s))|
|\psi(u(s))|^{2}ds
\end{equation}
with expectation
\begin{align}
&\mathlarger{\mathlarger{\mathcal{E}}}\bigg\llbracket V(\overline{u}(t)))\bigg\rrbracket=V(u_{\epsilon})+
\frac{1}{2}\int_{t_{\epsilon}}^{t}\mathlarger{\mathlarger{\mathcal{E}}}\bigg\llbracket
\mathlarger{\mathrm{D}}_{u}\mathlarger{\mathrm{D}}_{u}V(u(s))|\psi(u(s))|^{2}ds
\bigg\rrbracket \nonumber\\&\equiv{V}(u_{\epsilon})+\frac{1}{2}\int_{t_{\epsilon}}^{t}
\mathlarger{\mathlarger{\mathcal{E}}}\bigg\llbracket \mathlarger{\mathrm{I\!H}}V(u(s))ds\bigg\rrbracket
\end{align}
or equivalently
\begin{align}
&\bigg\| V(\overline{u}(t)))\bigg\|_{\mathcal{L}_{1}}=V(u_{\epsilon})+
\frac{1}{2}\int_{t_{\epsilon}}^{t}\bigg\|
\mathlarger{\mathrm{D}}_{u}\mathlarger{\mathrm{D}}_{u}
V(u(s))|\psi(u(s))|^{2}ds\bigg\|_{\mathcal{L}_{1}} \nonumber\\&\equiv{V}(u_{\epsilon})+\frac{1}{2}\int_{t_{\epsilon}}^{t}
\bigg\|\mathlarger{\mathrm{I\!H}}V(u(s))ds\bigg\|_{\mathcal{L}_{1}}
\end{align}
so that $|\mathlarger{\mathlarger{\mathcal{E}}}\llbracket V(\overline{u(t)}))-{V}(u_{\epsilon}\rrbracket|=
{V}(u_{\epsilon})+\frac{1}{2}\mathlarger{\mathlarger{\mathcal{E}}}\llbracket \int_{t_{\epsilon}}^{t}\mathlarger{\mathrm{I\!H}}V(u(s))\rrbracket ds$ which is the Dynkin Theorem. Then $V(u(t))$ is a stochastic Lyuponov functional (SLF) if the following hold:
\begin{enumerate}
\item $\mathlarger{\mathlarger{\mathcal{E}}}V\llbracket\overline{u(t)})\rrbracket\ge 0$ and $V(u(t))\ge 0$.
\item $\lim_{\overline{u(t)}\rightarrow\infty}\|\mathlarger{\mathlarger{\mathcal{E}}}\llbracket V(\overline{u(t)}\rrbracket \|
=\infty$ so that $\mathlarger{\mathlarger{\mathcal{E}}}\llbracket(V(\overline{u(t)})\rrbracket$ is infinite only if
$\overline{u(t)}$ is infinite or blows up for some finite t. The satisfies
$V(\overline{u(t)})<\infty$ for $\overline{u(t)}<\infty$.
\item Given the generator $\mathlarger{\mathrm{I\!H}}$  for the diffusion
$\exists$ constants $(\beta_{1},\beta_{2})\ge 0$ such that
\begin{equation}
\mathlarger{\mathrm{I\!H}}V(u(t))\le \beta_{1}V(u(t))+\beta_{2}
\end{equation}
\end{enumerate}
If a functional $V(u(t))$ satisfying these criteria can be found then the diffusion does not blowup for any $t>t_{\epsilon}$ and there is zero probability of a blowup in the diffusion at any finite $t\in\mathbf{X}_{II}\cup\mathbf{X}_{III}$ such that
$\mathlarger{\mathrm{I\!P}}[\overline{\psi(t)}=\infty]=0$, in agreement with previous theorems. Suppose $\exists$ a Lyapunov function $V(u(t))$ satisfying these conditions. Then the following hold:
\begin{enumerate}
\item There is a bound or estimate of the form
\begin{equation}
\bigg\|V(\overline{u(t)}\bigg\|_{\mathcal{L}_{1}}\equiv\mathlarger{\mathlarger{\mathcal{E}}}
\bigg\llbracket V(\overline{u(t)}\bigg\rrbracket
\le V(u_{\epsilon})+\beta_{2}|t-t_{\epsilon}|\exp(\beta_{1}|t-t_{\epsilon}|)
\end{equation}
or more precisely
\begin{equation}
\mathlarger{\mathlarger{\mathcal{E}}}\bigg\llbracket\sup_{t\le T}V(\overline{u(t)}\bigg\rrbracket\le V(u_{\epsilon})+
\beta_{2}|t-t_{\epsilon}|)\exp(\beta_{1}|T-t_{\epsilon}|)
\end{equation}
\item The probability that $V(\overline{u(t)})=\infty$ is zero so that
\begin{equation}
{\mathlarger{\mathrm{I\!P}}}(V(\overline{u(t)})=\infty]=0
\end{equation}
\end{enumerate}
It follows that the SDE does not blowup or develop a singularity  if conditions (1),(2) and (3) are satisfied, since then $V(\overline{u(t)})=\infty $ iff $\widehat{u}(t)=\infty$ for some finite $t$; and so the density function diffusion $\widehat{u}(t)$ is finite and bounded for all $t\in\mathbf{X}_{II}\cup\mathbf{X}_{III}\equiv[t_{\epsilon},\infty)$.
\end{thm}
\begin{proof}
For $n,m\in\mathbf{Z}$, first define the hitting times
\begin{align}
&T^{n}=\inf\lbrace t>t_{\epsilon}:V(\overline{u(t)})\ge n\rbrace
\\&
T^{m}=\inf\lbrace t>t_{\epsilon}:\overline{u(t)}\ge m\rbrace
\end{align}
and so finite blowup times are $T^{\infty}=\lim_{n\uparrow\infty}T^{n}$ and $T^{\infty}=\lim_{n\uparrow\infty}T^{n}$. Then given $d\overline{u(t)}=\psi(u(t)]d\mathlarger{\mathlarger{\mathscr{B}}}(t)\equiv \kappa^{1/2}u^{2}(t)(u(t)-1)^{1/2}d\mathlarger{\mathlarger{\mathscr{B}}}(t)$,express $V(\overline{u(t)})$ using Ito's formula for $t>t_{\epsilon}$ or $t\in\mathbf{X}_{II}\cup\mathbf{X}_{III}$
\begin{align}
V(\overline{u(t)})&=V(u_{\epsilon})+\int_{t_{\epsilon}}^{t}
\mathlarger{\mathrm{D}}_{u}
V(u(s))\psi(u(\tau)]d\mathlarger{\mathlarger{\mathscr{B}}}(t)+\frac{1}{2}\int_{t_{\epsilon}}^{t}
\mathlarger{\mathrm{D}}_{u}\mathlarger{\mathrm{D}}_{u}V[u(\tau)]\psi^{2}(u(s))ds\nonumber\\&\equiv V(u_{\epsilon})+\frac{1}{2}\kappa^{1/2}\int_{t_{\epsilon}}^{t}
\mathlarger{\mathrm{D}}_{u}\mathlarger{\mathrm{D}}_{u} V[u(s)]u^{2}(s)
(u(s)-1)^{1/2}d\mathlarger{\mathlarger{\mathscr{B}}}(s)\nonumber\\&+\frac{1}{2}\kappa
\int_{t_{\epsilon}}^{t}u^{4}(s)(u^{2}(s)-1)\mathlarger{\mathrm{D}}_{u}
V(u(s))ds\nonumber\\&\equiv V(u_{\epsilon}) +\int_{t_{\epsilon}}^{t}{\psi}_{u}
V(u(s))ds+\frac{1}{2}\kappa^{1/2}\int_{t_{\epsilon}}^{t}
\mathlarger{\mathrm{D}}_{u}V(u(s))u^{2}(s)(u(s)-1)^{1/2}
d\mathlarger{\mathlarger{\mathscr{B}}}(s)
\end{align}
Now let $t\bigwedge T^{n}\wedge T^{m}=\min\lbrace t, T^{n}, T^{m}\rbrace$ so that the Ito integral formula for $V(\overline{u(t)\wedge T^{n}\wedge T^{m})})$ is now
\begin{eqnarray}
V(\overline{\overline{u(t\wedge T^{n}}\wedge T^{m})})\equiv V(\psi_{\epsilon}) +
\int_{t_{\epsilon}}^{t\wedge\mathfrak{J}^{n}\wedge T^{m}}
\mathlarger{\mathrm{I\!H}}V(u(s))
ds\nonumber\\+\frac{1}{2}\kappa^{1/2}\int_{t_{\epsilon}}^{t\wedge T^{n}\wedge T^{m}}\mathlarger{\mathrm{D}}_{u} V(u(s))u^{2}(s)(u(s)-1)^{1/2}d\mathlarger{\mathlarger{\mathscr{B}}}(s)
\end{eqnarray}
Taking the expectation $\mathlarger{\mathlarger{\mathcal{E}}}\big\llbracket...\big\rrbracket$ and since
\begin{equation}
\frac{1}{2}\kappa^{1/2}\mathlarger{\mathlarger{\mathcal{E}}}\bigg\llbracket \int_{t_{\epsilon}}^{t\wedge \mathfrak{J}^{n}\wedge T^{m}}(\mathlarger{\mathrm{D}}_{u} V(u(s))(u(s)^{2}(u(s)-1)^{1/2}
d\mathlarger{\mathlarger{\mathscr{B}}}(s)\bigg\rrbracket=0
\end{equation}
and using condition (2) such that if $\mathlarger{\mathrm{I\!H}} V(u(t)))\le \beta_{1}V(u(t))+\beta_{2}$, this reduces to
\begin{align}
\mathlarger{\mathlarger{\mathcal{E}}}\bigg\llbracket V(\overline{\widehat{u}(t\wedge T^{n}\wedge T^{m})})\bigg\rrbracket &\equiv V(u_{\epsilon}) + \int_{t_{\epsilon}}^{t\wedge T^{n}\wedge T^{m}}\mathlarger{\mathlarger{\mathcal{E}}}\bigg\llbracket \mathlarger{\mathrm{I\!H}}V(u(s)])\bigg\rrbracket \mathcal{C}
[s<T^{n}\wedge T^{m}]ds\nonumber\\&
\le V[u_{\epsilon}] + \int_{t_{\epsilon}}^{t\wedge {T}^{n}\wedge T^{m}}\mathlarger{\mathlarger{\mathcal{E}}}\bigg\llbracket \beta_{1}V[u(s)]+\beta_{2}\bigg\rrbracket\mathcal{C}[s < {T}^{n}\wedge T^{m}]ds\nonumber\\& \le {V}(u_{\epsilon})+\beta_{2}(t-t_{\epsilon})+\beta_{1}
\int_{t_{\epsilon}}^{t\wedge T^{n} \wedge T^{m}}\mathlarger{\mathlarger{\mathcal{E}}}\bigg\llbracket V(u(s))\bigg\rrbracket
\mathcal{C}(s<T^{n}<T^{m})ds\nonumber\\&\le V(u_{\epsilon})+\beta_{2}(t-t_{\epsilon})+\beta_{1}\int_{t_{\epsilon}}^{t}
\mathlarger{\mathlarger{\mathcal{E}}}\bigg\llbracket V(u(s))\bigg\rrbracket ds
\end{align}
or equivalently
\begin{align}
\bigg\| V(\overline{\widehat{u}(t\wedge T^{n}\wedge T^{m})})\bigg\|_{\mathcal{L}_{1}} &\equiv V(u_{\epsilon}) +\bigg\| \int_{t_{\epsilon}}^{t\wedge T^{n}\wedge T^{m}}\mathlarger{\mathrm{I\!H}}V(u(s)]\|_{\mathcal{L}_{1}} \mathcal{C}
[s<T^{n}\wedge T^{m}]ds\nonumber\\&
\le V[u_{\epsilon}] + \bigg\| \int_{t_{\epsilon}}^{t\wedge {T}^{n}\wedge T^{m}} \beta_{1}V[u(s)]+\beta_{2}\bigg\|_{\mathcal{L}_{1}}\mathcal{C}[s < {T}^{n}\wedge T^{m}]ds\nonumber\\& \le {V}(u_{\epsilon})+\beta_{2}(t-t_{\epsilon})+\beta_{1}
\bigg\|\int_{t_{\epsilon}}^{t\wedge T^{n} \wedge T^{m}} V(u(s))\bigg\|_{\mathcal{L}_{1}}
\mathcal{C}(s<T^{n}<T^{m})ds\nonumber\\&\le V(u_{\epsilon})+\beta_{2}(t-t_{\epsilon})+\bigg\|\beta_{1}\int_{t_{\epsilon}}^{t}
\mathlarger{\mathrm{I\!H}}(u(s))\bigg\|_{\mathcal{L}_{1}} ds
\end{align}
and where $\mathlarger{\mathcal{C}}[t<T^{n}\wedge T^{m}]=1$ for $t<T^{n}\wedge T^{m}$ and zero otherwise. From the Gronwall inequality it follows that
\begin{equation}
\mathlarger{\mathlarger{\mathcal{E}}}(V(\overline{u(t\wedge T^{n}\wedge T^{m})}\le (V(t_{\epsilon})+\beta_{2}(t-t_{\epsilon})\exp(\beta_{1}t)
\end{equation}
Letting $m\rightarrow\infty$ then $T^{\infty}=\inf\lbrace\|\overline{u(t)}|
=\infty\rbrace$ and $\overline{u(t)}$ cannot blow up before time $T^{n}$
\begin{equation}
\mathlarger{\mathlarger{\mathcal{E}}}\bigg\llbracket V(\overline{u(t\wedge T^{n})})\bigg\rrbracket\le (V(t_{\epsilon})+
\beta_{2}(t-t_{\epsilon})\exp(\beta_{1}t)
\end{equation}
Now using the basic Markov inequality
\begin{equation}
\mathlarger{\mathrm{I\!P}}[t\ge T^{n}]\equiv
\mathlarger{\mathrm{I\!P}}[V[\overline{u(t)}]\ge n]
\le\frac{1}{n}\mathlarger{\mathlarger{\mathcal{E}}}\bigg
\llbracket[V(\overline{u(t)}]\bigg\rrbracket
\end{equation}
gives
\begin{equation}
\mathlarger{\mathrm{I\!P}}[t\ge T^{n}]\le \frac{1}{n}\mathlarger{\mathlarger{\mathcal{E}}}\bigg\llbracket(V[u(t\wedge T^{n})]\bigg\rrbracket
\le \frac{1}{n}[V(u_{\epsilon})+\beta_{2}(t-t_{\epsilon}))\exp(\beta_{1}t)
\end{equation}
Taking the equality
\begin{equation}
\mathlarger{\mathrm{I\!P}}[t= T^{n}]\le \frac{1}{n}\mathlarger{\mathlarger{\mathcal{E}}}\bigg\llbracket(V[u(t\wedge T^{n})\bigg\rrbracket
\le \frac{1}{n}[V(u_{\epsilon})+\beta_{2}(t-t_{\epsilon}))\exp(\beta_{1}t)
\end{equation}
Taking the limit as $n\rightarrow\infty$ gives
\begin{equation}
\lim_{n\rightarrow\infty}\mathlarger{\mathrm{I\!P}}[t=T^{n}]=
\mathlarger{\mathrm{I\!P}}[V[\overline{u(t)}]
=\infty]=\lim_{n\rightarrow\infty}\frac{1}{n}\left(V(u_{\epsilon})+Q_{2}(t-t_{\epsilon}))\exp(Q_{1}t\right)
=0
\end{equation}
so there is zero probability that $V(\overline{u(t)})=\infty$. But $V(\overline{u}(t))
=\infty$ iff $|\overline{u(t)}|=\infty$ for some finite $t$ so there is no
blowup or density singularity for all $t\in\mathbf{X}^{+}_{II}\cup\mathbf{X}^{+}_{III}$.
\end{proof}
\subsection{Lyuponov characteristic exponent and relation to moments estimates}
\begin{defn}
The $p^{th}$ moments have boundedness for all $t\ge t_{\epsilon}$ by the estimate
$\mathlarger{\mathcal{E}}\llbracket|\overline{u(t)}|^{p}\rrbracket\equiv
\|\overline{u(t)}\|^{p}_{\mathcal{L}_{1}}\le |u_{\epsilon}|^{p}\exp(\frac{1}{2}Cp(p-1)|t-t_{\epsilon}|)$. Hence, the $p^{th}$ moments grow exponentially with exponent $\tfrac{1}{2}Cp(p-1)$. The Lyupanov exponent is defined as [44]
\begin{equation}
\mathlarger{\mathscr{L}}_{p}=\lim_{t\uparrow\infty}
\log[\mathcal{E}\llbracket\|\overline{u(t)}|^{p}\rrbracket
\le \frac{1}{t}\log|u_{\epsilon}|^{p}+
\frac{1}{t}\frac{1}{2}Cp(p-1)\frac{|t-t_{\epsilon}|}{t}\sim \frac{1}{t}\log|u_{\epsilon}|+\frac{1}{2}Cp(p-1)
\end{equation}
so that
\begin{equation}
\mathlarger{\mathscr{L}}_{p}\le \frac{1}{2}Cp(p-1)=D_{p}
\end{equation}
with $\mathlarger{\mathscr{L}}_{1}\le 0$
\end{defn}
\begin{thm}
$V(u(t))=\log(1+u(t))$ is a $p=1$ Lyapunov functional for the SDE $d\overline{u(t)}=\kappa^{1/2}(u(t))^{2}(u(t)-1)^{1/2}d\mathlarger{\mathlarger{\mathscr{B}}}(t)$. Hence, this SDE does not explode for any finite $t>t_{\epsilon}$.
\end{thm}
\begin{proof}
The functional $V(u(t))=\log(u(t))$ satisfies the criteria of Thm(9.1) since $\log(u(t))\ge 0$ for all $u(t)\in[0,\infty)$ and $\lim_{u(t)\uparrow\infty}\log(u(t))=\infty$. For the generator $\mathlarger{\mathrm{I\!H}}$ there $\exists (\beta_{1},\beta_{2})>0$ such that
\begin{equation}
\mathlarger{\mathrm{I\!H}}\log(1+u(t))\le \beta_{1}\log(u(t))+\beta_{2}
\end{equation}
which is
\begin{align}
=\frac{1}{2}k(u(t))^{2}(u(t)-1)\mathlarger{\mathrm{D}}_{u}\mathlarger{\mathrm{D}}_{u}
\log(1+u(t))=-\frac{1}{2}(u(t))^{2}(u(t)-1)\le \beta_{1}\log(u(t))+\beta_{2}
\end{align}
as required.
\end{proof}
The general Lyuponov exponent to all orders can also be established using the Ito Lemma
\begin{thm}
Let the following hold
\begin{enumerate}
\item The diffusion $\overline{u(t)}$ is a solution of the SDE
$d\overline{u(t)}=\psi(u(t))d\mathscr{B}(t)=\kappa^{1/2}(u(t))^{2}(u(t)-1)^{1/2}d\mathscr{B}(t)$ for all $t>t_{\epsilon}$ or $t\in\mathbf{X}_{II}\cup\mathbf{X}_{III}$.
\item  The polynomial growth condition such that for any finite $[t_{\epsilon},T]\subset\mathbf{R}^{+}$, $\exists~C>0$ such that
\begin{align}
&|\psi(u(t)|^{2}\equiv |(u(t))^{4}(u(t)-1)\le C[1+|u(t)|^{p}]\nonumber\\&\le \frac{C(1+|u(t)|^{p})^{2}}{|u(t)|^{p-2}}\le C[1+|u(t)|^{p}]^{2}
\end{align}
\end{enumerate}
Then the LE for the solutions of the SDE are bounded for all $t>t_{\epsilon}$ as
\begin{equation}
\mathlarger{\mathscr{LE}}_{p}=\lim_{t\uparrow\infty}\sup\frac{1}{t}\log\big(|\overline{u(t)}|^{p}\big)\le \frac{1}{2}Cp(p-1)\equiv D_{p}
\end{equation}
or
\begin{equation}
\mathlarger{\mathscr{LE}}_{p}=\lim_{t\uparrow\infty}\sup\frac{1}{t}
\log\big(\mathlarger{\mathlarger{\mathcal{E}}}\llbracket |\overline{u(t)}|^{p}\rrbracket \big)\le \frac{1}{2}Cp(p-1)\equiv D_{p}
\end{equation}
Then $\mathscr{LE}_{p}=0$ and $\mathlarger{\mathscr{LE}}_{p}$ corresponds to the $p^{th}$-moment bound of $\overline{u(t)}$.
\end{thm}
\begin{proof}
For any functional $\Phi(\overline{u(t)})$ of $\overline{u(t)}$ existing for $t>t_{\epsilon}$, the Ito Lemma gives
\begin{align}
\Phi(\overline{u(t)})=\Phi(u_{\epsilon})+\int_{t_{\epsilon}}^{t}
\mathlarger{\mathrm{D}}_{u}(\phi(\overline{u(s)})\psi(u(s))d\mathscr{B}(s)
+\frac{1}{2}\int_{t_{\epsilon}}^{t}
\mathlarger{\mathrm{D}}_{u}\mathlarger{\mathrm{D}}_{u}(\Phi(\overline{u(t)})|\psi(s)|^{2}ds
\end{align}
For $\Phi(u(t))=\log[1+|u(t)|^{p}]$, this becomes
\begin{align}
\big\|\Phi(\overline{u(t)})\big\|_{1}=&\big\|\log[1+|u_{\epsilon}|^{p}]\big\|
+\left\|\int_{t_{\epsilon}}^{t}\mathlarger{\mathrm{D}}_{u}(\Phi(\log[1+|u(t)|^{p}])\psi(u(s))d\mathlarger{\mathscr{B}}(s)\right\|_{1}
\nonumber\\&+\frac{1}{2}\left\|\int_{t_{\epsilon}}^{t}\mathlarger{\mathrm{D}}_{u}
\mathlarger{\mathrm{D}}_{u} (\Phi(\log[1+|u(t)|^{p}])|\psi(s)|^{2}ds\right\|_{1}\nonumber\\&
=\big\|log[1+|u_{\epsilon}|^{p}]\big\|_{1}
+p\left\|\int_{t_{\epsilon}}^{t}\frac{\psi(u(t))d\mathlarger{\mathlarger{\mathscr{B}}}(s)}{(|u(s)|^{1-p}+u(s))}
\right\|\nonumber\\&-\frac{1}{2}p\left\|\int_{t_{\epsilon}}^{t}
\frac{|u(t)|^{p-2}(|(u(t)|^{p}+(1-p)}{(1|+u(t)|^{p})^{2}}ds\right\|_{1}
\nonumber\\&=
\big\|\log[1+|u_{\epsilon}|^{p}]\big\|_{1}
+p\left\|\int_{t_{\epsilon}}^{t}\frac{\psi(u(t))d\mathlarger{\mathlarger{\mathscr{B}}}(s)}{(|u(s)|^{1-p}+u(s))}\right\|_{1}
\nonumber\\&-\frac{1}{2}p\int_{t_{\epsilon}}^{t}\frac{|u(t)|^{p-2}|u(t)|^{p}|\psi(s)|^{2}}
{(1+|u(t)|^{p})^{2}}ds+\frac{1}{2}p(p-1)\int_{t_{\epsilon}}^{t}\frac{|u(t)|^{p-2}|\psi(s)|^{2}}
{(1+|u(t)|^{p})^{2}}ds\nonumber\\&=
\big\|\log[1+|u_{\epsilon}|^{p}]\big\|
+p\left\|\int_{t_{\epsilon}}^{t}\frac{\psi(u(t))d\mathlarger{\mathscr{B}}(s)}{(|u(s)|^{1-p}+u(s))}\right\|_{1}
\nonumber\\&-\frac{1}{2}Cp\left\|\int_{t_{\epsilon}}^{t}|u(s)|^{p}ds+\frac{1}{2}p(p-1)\int_{t_{\epsilon}}^{t}Cds\right\|_{1}
\nonumber\\&=
\big\|\log[1+|u_{\epsilon}|^{p}]\big\|_{1}
+p\left\|\int_{t_{\epsilon}}^{t}\frac{\psi(u(t))d\mathlarger{\mathscr{B}}(s)}{(|u(s)|^{1-p}+u(s))}\right\|
\nonumber\\&-\frac{1}{2}Cp\left\|\int_{t_{\epsilon}}^{t}|u(s)|^{p}ds\right\|_{1}+
\frac{1}{2}Cp(p-1|t-t_{\epsilon}|\nonumber\\&\equiv
\big\|\log[1+|u_{\epsilon}|^{p}]\big\|+\big\|\mathlarger{\mathscr{R}}(t_{\epsilon},t)\big\|
\nonumber\\&-\frac{1}{2}Cp\left\|\int_{t_{\epsilon}}^{t}|u(s)|^{p}ds\right\|_{1}+
\frac{1}{2}Cp(p-1|t-t_{\epsilon}|\nonumber\\&\equiv
\big\|\log[1+|u_{\epsilon}|^{p}]\big\|
-\frac{1}{2}Cp\left\|\int_{t_{\epsilon}}^{t}|u(s)|^{p}ds\right\|_{1}+
\frac{1}{2}Cp(p-1|t-t_{\epsilon}|
\end{align}
where (9.25) has been used. Now since
\begin{align}
&\log[1+|\overline{u(t)}|^{p}]=
\log[1+|u_{\epsilon}|^{p}]+\mathlarger{\mathscr{R}}(t_{\epsilon},t)\nonumber\\&
-\frac{1}{2}Cp\int_{t_{\epsilon}}^{t}|u(s)|^{p}ds+
\frac{1}{2}Cp(p-1|t-t_{\epsilon}|
\end{align}
one can choose an $n\in\mathbf{Z}$ such that let $n\ge t_{\epsilon}$. Then using the exponential inequality of Thm 6.11
\begin{equation}
\mathlarger{\mathrm{I\!P}}\left[\sup_{t_{\epsilon}\le t\le n}\left|\mathlarger{\mathscr{R}}(t_{\epsilon},t)-\frac{1}{2}
Cp\int_{t_{\epsilon}}^{t}|u(s)|^{p}ds\right| > 2\log n\right]\le\frac{1}{n^{2}}
\end{equation}
Using the Borel-Cantelli Lemma, for almost all $\omega\in\Omega$,$\exists~n_{o}=n_{o}(\omega)\ge t_{\epsilon}+1$ such that
\begin{equation}
\sup_{t_{\epsilon}\le t\le n}\left|\mathlarger{\mathlarger{\mathscr{R}}}(t_{\epsilon},t)-
\frac{1}{2}Cp\int_{t_{\epsilon}}^{t}|u(s)|^{p}ds\right|\le 2\log n
\end{equation}
if $n\ge n_{o}$, or
\begin{equation}
\mathlarger{\mathlarger{\mathscr{R}}}(t_{\epsilon},t)\le 2\log n+
\frac{1}{2}Cp\int_{t_{\epsilon}}^{t}|u(s)|^{p}ds
\end{equation}
for all $t_{\epsilon}\le t\le n$ and for $n\ge n_{o}$. Inequality (1.25) then becomes
\begin{equation}
\log[1+|\overline{u(t)|^{p}}\le \log[1+|u_{\epsilon}|]+\frac{1}{2}Cp(p-1)|t-t_{\epsilon}|+2\log n
\end{equation}
for all $t_{\epsilon}\le t\le n$ and for $n\ge n_{o}$. Hence
\begin{equation}
\frac{1}{t}\log[1+|\overline{u(t)}|^{p}]\le \frac{1}{n-1}\bigg[\log[1+|u_{\epsilon}|]+\frac{1}{2}Cp(p-1)|n-t_{\epsilon}|+2\log n\bigg]
\end{equation}
so that
\begin{align}
&\lim\sup_{t\uparrow\infty}\frac{1}{t}\log[1+|\overline{u(t)}|^{p}]\le \lim\sup_{n\uparrow\infty}\frac{1}{n-1}\log\bigg[1+|u_{\epsilon}|]+\frac{1}{2}Cp(p-1)|n-t_{\epsilon}|+2\log n\bigg]\nonumber\\&=\frac{1}{2}Cp(p-1)\equiv D_{p}
\end{align}
which completes the proof.
\end{proof}
\begin{cor}
The $p^{th}$-order moments corresponding to this LE are
\begin{equation}
\mathlarger{\mathcal{E}}\llbracket|\overline{u(t)}|^{p}\rrbracket\equiv
\|\overline{u(t)}\|^{p}_{\mathcal{L}_{1}}\le |u_{\epsilon}|^{p}\exp(\frac{1}{2}Cp(p-1)|t-t_{\epsilon}|)
\end{equation}
which agrees exactly with the previous estimates.
\end{cor}
\begin{cor}
It follows from Thm (7.7) that
\begin{equation}
\frac{\log[1+\mathlarger{\mathcal{E}}\llbracket |\overline{u(t)}|^{r}\rrbracket}{
\log[1+\mathlarger{\mathcal{E}}\llbracket |\overline{u(t)}|^{p}\rrbracket}\le
\frac{|\tfrac{1}{\alpha}-\tfrac{1}{\beta}|}{|\tfrac{1}{p}-\tfrac{1}{\beta}|}
+\frac{|\tfrac{1}{\alpha}-\tfrac{1}{\beta}|}{|\tfrac{1}{p}-\tfrac{1}{\beta}|}
\frac{\log[1+\mathlarger{\mathcal{E}}\llbracket |\overline{u(t)}|^{\beta}\rrbracket}{\log[1+\mathlarger{\mathcal{E}}\llbracket |\overline{u(t)}|^{p}\rrbracket}
\end{equation}
so the LEs obey the inequality
\begin{equation}
\frac{\mathlarger{\mathscr{L}}_{\alpha}}{\mathlarger{\mathscr{L}}_{p}}\le
\frac{|\tfrac{1}{\alpha}-\tfrac{1}{\beta}|}{|\tfrac{1}{p}-\tfrac{1}{\beta}|}
+\frac{|\tfrac{1}{\alpha}-\tfrac{1}{\beta}|}{|\tfrac{1}{p}-\tfrac{1}{\beta}|}
\frac{\mathlarger{\mathscr{L}}_{\beta}}{\mathlarger{\mathscr{L}}_{p}}
\end{equation}
\end{cor}
\section{A Kolmogorov-Fokker-Planck (KFP)description}
Generally, most SDEs are impossible to solve in closed form and have to be dealt with approximately using methods that include eigenfunction decompositions, WKB expansions, and variational methods within the Langevin or Fokker–Planck formalisms [69,70,71,72]. Knowing which method to utilise is generally not obvious and practical application can become unwieldy. Methods adapted from quantum mechanics and quantum and statistical field theory can provide a framework to produce perturbative approximations to the moments of SDEs [70,72,73,74,75]. In particular, any stochastic system can be expressed in terms of a functional path integral for which asymptotic methods can be systematically applied [74,75,76]. Often of interest are the moments of $\overline{u(t)}$ or the probability density function $\mathcal{P}(u(t),t)=\mathcal{P}(u,t)$ that $\overline{u(t)}$ has value $u(t)$ at time $t$. Path integral methods provide a convenient tool to compute quantities such as moments and transition probabilities perturbatively.

Due to the presence of the noise or random perturbation terms of nonlinear SDEs it is usually more appropriate and tractable to consider the extrema and behavior of the probability density function (PDF)or distribution function $\mathcal{P}(u(t))\equiv \mathcal{P}(u(t),t)$. The $\mathcal{P}(u(t),t)$ would be solutions of a Kolmogorov forward equation or the Kolmogorov-Fokker-Planck (KPF) equation, which is a linear or nonlinear parabolic PDE.
The KFP equation or Kolmogorov's second equation, is essentially a parabolic-type PDE describing the evolution of a probability density or transition function. The Fokker–Planck equation has a central role in statistical physics and in the study of fluctuations in physical and biological systems. It has found many diverse and useful applications within physics, chemistry and biology, including quantum optics, population dynamics, finance, electrical circuits and laser technology [20,70,71,72].

Such equations are often impossible to solve analytically although in
the infinite-time relaxation limit, the equilibrium or stationary solution can usually be found for nonlinear equations.[61,71] Here, it is shown that a KFP equation can reproduce the exact form of the moments estimate from Section 7.
\begin{defn}
The transition probability density function from $u(s)$ at time $s$ to $u(t)$ at time $t$ is $\mathcal{P}(u(t),t|u(s),s)$ and obeys the Chapman-Kolmogorov equation
\begin{equation}
\mathcal{P}(u(T),T|u(t_{\epsilon}),t_{\epsilon})=\int_{t_{\epsilon}}^{t}\mathcal{P}(u(t),t|u(s),s)
\mathcal{P}(u(s),s|u(t_{\epsilon}),t_{\epsilon})du(s)
\end{equation}
so that
\begin{align}
&\mathlarger{\mathlarger{\mathcal{P}}}(u(T),T|u(t_{\epsilon}),t_{\epsilon})
=\prod_{i=1}^{N}\int_{t_{\epsilon}}^{t}\mathcal{P}(u(t),t|u_{i},t_{i}),t_{i})
\mathcal{P}(u_{i}(t_{i}),t_{i}|u(t_{\epsilon}),t_{\epsilon})du_{i}(t_{i})\nonumber
\\&=\int\mathlarger{\mathrm{I\!D}}u(t)\prod_{i=1}^{N}\mathlarger{\mathlarger{\mathcal{P}}}(u(t),t|u_{i},t_{i}),t_{i})
\mathcal{P}(u_{i}(t_{i}),t_{i}|u(t_{\epsilon}),t_{\epsilon})
\end{align}
The process is a Markov chain and has no "memory" if future evolution and probabilities do not depend on the past.
\end{defn}
\begin{prop}
The transition probability as a criterion for no-blowup at some time $T$ is
\begin{equation}
\mathlarger{\mathlarger{\mathcal{P}}}(\infty,T|u(t_{\epsilon}),t_{\epsilon})=0
\end{equation}
This can be interpreted as a non-exploding random walk.
\end{prop}
\begin{thm}
Let $u(t)$ be a generic diffusion process satisfying a nonlinear SDE
\begin{equation}
d\overline{u(t)}=\phi(u(t)]dt + \psi(u(t))d\mathlarger{\mathlarger{\mathscr{B}}}(t)
\end{equation}
with some initial data $u(0),\mathlarger{\mathlarger{\mathscr{B}}}(0)=0.$ with $u(t)\in[u(0),\infty)$ and $t\in[0,\infty)$. A pure drift-free diffusion is described by
\begin{equation}
d\overline{u(t)}=\psi(u(t))d\mathlarger{\mathlarger{\mathscr{B}}}(t)=\lim_{\Lambda\rightarrow 0} \phi(u(t))dt +\psi[u(t)]d\mathlarger{\mathlarger{\mathscr{B}}}(t)
\end{equation}
The associated Fokker-Planck-Kolmogorov equation for the SDE is the nonlinear 'heat equation'
\begin{equation}
\frac{\partial}{\partial t}\mathcal{P}(u(t),t)=\mathlarger{\mathrm{D}}_{u}(\phi(u(t))
\mathcal{P}(u(t),t))+\frac{1}{2}
\mathlarger{\mathrm{D}}_{u}\mathlarger{\mathrm{D}}_{u}[|\psi(u(t)|^{2}
\mathcal{P}(u(t),t)]
\end{equation}
If $\phi(u(t))=0$ then this is a nonlinear diffusion or heat equation for $\mathcal{P}(u(t),t)$.
\begin{equation}
\frac{\partial}{\partial t}\mathcal{P}(u(t),t)=+\frac{1}{2}
\mathlarger{\mathrm{D}}_{u}\mathlarger{\mathrm{D}}_{u}[\psi^{2}(u(t))\mathcal{P}(u(t),t)]
\end{equation}
Given the density function diffusion $\overline{u(t)}$ satisfying the nonlinear SDE
\begin{equation}
d\overline{u(t)}=\psi(u(t))d\mathlarger{\mathlarger{\mathscr{B}}}(t)=\kappa^{1/2}u^{2}(t)(u(t)-1)^{1/2}
d\mathlarger{\mathlarger{\mathscr{B}}}(t)
\end{equation}
for all $u(t)\in[u_{\epsilon},\infty]$ and for all $
t\in\mathbf{X}_{II}\bigcup\mathbf{X}_{III}$ or $ t>t_{\epsilon}$, the
associated nonlinear KFP equation is then of a nonlinear heat equation form
\begin{align}
&\frac{\partial}{\partial{t}}\mathcal{P}(u(t),t)=+\frac{1}{2}\kappa\mathlarger{\mathrm{D}}_{u}\mathlarger{\mathrm{D}}_{u}
\big[\psi(u(t))^{2}\mathcal{P}(u(t),t)\big]\equiv\frac{1}{2}\mathlarger{\mathrm{D}}_{u}
\mathlarger{\mathrm{D}}_{u}\big[u^{4}(t)(u^{2}(t)-1)\mathcal{P}(u(t),t)\big]
\end{align}
\end{thm}
\begin{proof}
To prove (10.9), let $\Phi(u(t))$ be some $C^{2}$-differentiable functional of $u(t)$ and let $\Phi(\overline{u(t)})$ be that for $\Phi(u(t))$. The Ito Lemma gives
\begin{equation}
d\Phi(\overline{u(t)})=\frac{\partial\Phi(u(t))}{\partial t}dt+\frac{1}{2}\big[
\mathlarger{\mathrm{D}}_{u}\mathlarger{\mathrm{D}}_{u}\Phi(u(t))\big](u(t))^{4}(u(t)-1)dt
\end{equation}
so that
\begin{equation}
\Phi(\overline{u(t)})=\Phi(u_{\epsilon})+\int_{t_{\epsilon}}^{T}\frac{\partial\Phi(u(t))}{\partial t}dt+\frac{1}{2}\kappa\int_{t_{\epsilon}}^{T}\big[\mathlarger{\mathrm{D}}_{u}\mathlarger{\mathrm{D}}_{u}
\Phi(u(t))\big](u(t))^{4}(u(t)-1)dt
\end{equation}
The expectation is
\begin{align}
&\mathlarger{\mathlarger{\mathcal{E}}}\bigg\llbracket|\Phi(\overline{u(t)})-\Phi(u_{\epsilon})|\bigg\rrbracket=\mathlarger{\mathlarger{\mathcal{E}}}\left\llbracket\int_{t_{\epsilon}}^{T}
\frac{\partial\Phi(u(t))}{\partial t}dt\right\rrbracket+\frac{1}{2}\kappa\mathlarger{\mathlarger{\mathcal{E}}}
\left\llbracket\int_{t_{\epsilon}}^{T}\big[\mathlarger{\mathrm{D}}_{u}\mathlarger{\mathrm{D}}_{u}
\Phi(u(t))\big](u(t))^{4}(u(t)-1)dt\right\rrbracket
\end{align}
or equivalently
\begin{align}
&\bigg\||\Phi(\overline{u(t)})-\Phi(u_{\epsilon})|\bigg\|_{\mathcal{L}_{1}}=
\left\|\int_{t_{\epsilon}}^{T}\frac{\partial\Phi(u(t))}{\partial t}dt\right\|_{\mathcal{L}_{1}}+\frac{1}{2}\kappa\left\|\int_{t_{\epsilon}}^{T}
\big[\mathlarger{\mathrm{D}}_{u}\mathlarger{\mathrm{D}}_{u}\Phi(u(t))\big](u(t))^{4}(u(t)-1)dt
\right\|_{\mathcal{L}_{1}}
\end{align}
In terms of the probability density function this becomes
\begin{align}
&\mathlarger{\mathlarger{\mathcal{E}}}\bigg\llbracket|\Phi(\overline{u(t)})-\Phi(u_{\epsilon})\bigg\rrbracket\equiv
\bigg\||\Phi(\overline{u(t)})-\Phi(u_{\epsilon})|\bigg\|_{\mathcal{L}_{1}}\nonumber\\&
=\int_{\mathbf{U}}\bigg(\int_{t_{\epsilon}}^{T}\frac{\partial\Phi(u(t))}{\partial t}dt+\frac{1}{2}\kappa\int_{t_{\epsilon}}^{T}\big[\mathlarger{\mathrm{D}}_{u}
\mathlarger{\mathrm{D}}_{u}
\Phi(u(t))\big](u(t))^{4}(u(t)-1)dt\bigg)
\mathcal{P}(u(t),t|u_{\epsilon},t_{\epsilon})du(t)
\nonumber\\&
=\int_{\mathbf{U}}\bigg(\int_{t_{\epsilon}}^{T}\frac{\partial\Phi(u(t))}{\partial t}dt\bigg)\mathcal{P}(u(t),t|u_{\epsilon},t_{\epsilon})du(t)\nonumber\\&+\int_{\mathbf{U}}\bigg(\frac{1}{2}\kappa\int_{t_{\epsilon}}^{T}
\big[\mathlarger{\mathrm{D}}_{u}\mathlarger{\mathrm{D}}_{u}\Phi(u(t))\big](u(t))^{4}(u(t)-1)dt\bigg)
\mathcal{P}(u(t),t|u_{\epsilon},t_{\epsilon})du(t)
\end{align}
where $\mathbf{U}=[u_{\epsilon},u(T)]$. This can be considered as a sum of two integrals so that
\begin{equation}
\mathlarger{\mathlarger{\mathcal{E}}}\bigg\llbracket\Phi(\overline{u(t)})-\Phi(u_{\epsilon})\bigg\rrbracket\equiv
\bigg\||\Phi(\overline{u(t)})-\Phi(u_{\epsilon})\bigg\|_{\mathcal{L}_{1}}
=\bm{I}_{1}(u(T),T))+\bm{I}_{2}(u(T),T)
\end{equation}
where
\begin{align}
&\bm{I}_{1}(u(T),T)=\int_{\mathbf{U}}\bigg(\int_{t_{\epsilon}}^{T}\frac{\partial\Phi(u(t))}{\partial t}dt\bigg)\mathcal{P}(u(t),t|u_{\epsilon},t_{\epsilon})du(t)
\\&
\bm{I}_{2}(u(T),T)=\int_{\mathbf{U}}\bigg(\frac{1}{2}\kappa\int_{t_{\epsilon}}^{T}
\big[\mathlarger{\mathrm{D}}_{u}\mathlarger{\mathrm{D}}_{u}\Phi(u(t))\big](u(t))^{4}(u(t)-1)dt\bigg)
\mathcal{P}(u(t),t|u_{\epsilon},t_{\epsilon})du(t)
\end{align}
Setting $\mathcal{P}(u(t),t|u_{\epsilon},t_{\epsilon})=\mathcal{P}$ for notational convenience
\begin{align}
&\bm{I}_{1}(u(T),T)=\int_{\mathbf{U}}\bigg(\int_{t_{\epsilon}}^{T}\frac{\partial\Phi(u(t))}{\partial t}\mathcal{P}dt\bigg)du(t)\\&
\bm{I}_{2}(u(T),T)=\int_{t_{\epsilon}}^{T}\bigg(\frac{1}{2}
\kappa\int_{\mathbf{U}}\mathcal{P}\big[\mathlarger{\mathrm{D}}_{u}\mathlarger{\mathrm{D}}_{u}
\Phi(u(t))\big](u(t))^{4}(u(t)-1)du(t)\bigg)
dt
\end{align}
Now evaluate the inner integrals by parts. First
\begin{equation}
\int_{t_{\epsilon}}^{T}\frac{\partial\Phi(u(t))}{\partial t}\mathcal{P}dt=\big[\mathcal{P}\Phi(u(t))\big]_{t_{\epsilon}}^{t}-
\int_{t_{\epsilon}}^{t}\frac{\partial\mathcal{P}}{\partial t}\Phi(u(t))dt=
\int_{t_{\epsilon}}^{t}\frac{\partial\mathcal{P}}{\partial t}\Phi(u(t))dt
\end{equation}
with the boundary conditions that $\Phi(u(T))=0$ and $\Phi(u_{\epsilon})=0$. The inner integral of (10.19) is evaluated by integrating by parts twice. Then
\begin{align}
&\frac{1}{2}\kappa\int_{\mathbf{U}}\mathcal{P}
\big[\mathlarger{\mathrm{D}}_{u}\mathlarger{\mathrm{D}}_{u}\Phi(u(t))\big](u(t))^{4}(u(t)-1)du(t)=\frac{1}{2}\kappa\big[
(u(t))^{4}(u(t)-1)\mathlarger{\mathrm{D}}_{u}\Phi(u(t))\big]_{\mathbf{U}}\nonumber\\&-\frac{1}{2}\kappa
\int_{\mathbf{U}} \mathlarger{\mathrm{D}}_{u}[\mathcal{P}(u(t))^{4}(u(t)-1)]\mathlarger{\mathrm{D}}_{u}\Phi(u(t))du(t)=
-\frac{1}{2}\kappa\int_{\mathbf{U}} \mathlarger{\mathrm{D}}_{u}[\mathcal{P}(u(t))^{4}(u(t)-1)]\mathlarger{\mathrm{D}}_{u}\Phi(u(t))du(t)
\end{align}
Integrating by parts again
\begin{align}
&-\frac{1}{2}\kappa\int_{\mathbf{U}} \mathlarger{\mathrm{D}}_{u}[\mathcal{P}(u(t))^{4}(u(t)-1)]
\mathlarger{\mathrm{D}}_{u}\Phi(u(t))du(t)\nonumber\\&
=-\frac{1}{2}\mathlarger{\mathrm{D}}_{u}[(u(t))^{4}(u(t)-1)]_{\mathbf{U}}
+\frac{1}{2}\kappa\int_{\mathbf{U}}\mathlarger{\mathrm{D}}_{u}
\mathlarger{\mathrm{D}}_{u}[(u(t))^{4}(u(t)-1)\mathcal{P}]\Phi(u(t))du(t)\nonumber\\&
=+\frac{1}{2}\kappa\int_{\mathbf{U}}\mathlarger{\mathrm{D}}_{u}
\mathlarger{\mathrm{D}}_{u}[(u(t))^{4}(u(t)-1)\mathcal{P}]
\Phi(u(t))du(t)
\end{align}
Hence, the integrals $\bm{I}_{1}(u(T),T)$ and $\bm{I}_{2}(u(T),T)$ are
\begin{align}
&\bm{I}_{1}(u(T),T)=\int_{\mathbf{U}}\int_{t_{\epsilon}}^{t}\frac{\partial\mathcal{P}}{\partial t}\Phi(u(t))dt du(t)\\&
\bm{I}_{2}(u(T),T)=\frac{1}{2}\int_{t_{\epsilon}}^{T}\int_{\mathbf{U}}\kappa
\mathlarger{\mathrm{D}}_{u}\mathlarger{\mathrm{D}}_{u}
[(u(t))^{4}(u(t)-1)\mathcal{P}]\Phi(u(t))du(t)
\end{align}
Then (10.15) becomes
\begin{align}
&\mathlarger{\bm{\mathcal{E}}}\bigg\llbracket|\Phi(\overline{u(t)})-\Phi(u_{\epsilon})\bigg\rrbracket\equiv
\bigg\||\phi(\overline{u(t)})-\Phi(u_{\epsilon})|\bigg\|_{\mathcal{L}_{1}}\nonumber\\&
=\int_{\mathbf{U}}\int_{t_{\epsilon}}^{t}\bigg(\frac{\partial\mathcal{P}}{\partial t}+\frac{1}{2}\kappa\mathlarger{\mathrm{D}}_{u}\mathlarger{\mathrm{D}}_{u}[(u(t))^{4}(u(t)-1)\mathcal{P}]\bigg)
\Phi(u(t))du(t)dt
\end{align}
Setting $\mathlarger{\mathlarger{\mathcal{E}}}\bigg\llbracket|\Phi(\overline{u(t)})-\Phi(u_{\epsilon})
\bigg\rrbracket\equiv
\bigg\||\Phi(\overline{u(t)})-\Phi(u_{\epsilon})|\bigg\|_{\mathcal{L}_{1}}=0$ then gives the KFP equation (10.9) as required.
\end{proof}
\begin{rem}
The KFP equation is the special case of the Kramers-Moyal equation truncated at $n=2$. The Kramers-Moyal equation is
\begin{equation}
\partial_{t}\mathcal{P}(u(t),t)=\sum_{n=1}^{\infty}\frac{(-1)^{n}}{n!}
(\mathlarger{\mathrm{D}}_{u})^{n}[C_{n}(u(t)\mathcal{P}(u(t),t)]
\end{equation}
where $C_{2}(u(t))=\psi(u(t))^{2}$. The Puwala Theorem states that the expansion can only truncated at n=1 or n=2 terms, otherwise it must be continued to infinity.
\end{rem}
\begin{cor}
If $\psi(u(t))=\beta=const.$ then $d\overline{u}(t)=\beta d\mathscr{B}(t)$ is standard Brownian motion and the nonlinear KFP equation reduces to a linear heat equation or diffusion equation
\begin{equation}
\frac{\partial}{\partial t}\mathcal{P}(u(t),t)=+\frac{1}{2}\beta
\mathlarger{\mathrm{D}}_{u}\mathlarger{\mathrm{D}}_{u}\mathcal{P}(u(t),t)
\end{equation}
which is exactly solvable. The solution is then the well-known heat kernel.
\end{cor}
The general solution of a nonlinear Fokker-Planck equation often cannot be found but a stationary or equilibrium solution is often possible in the 'infinite-time relaxation limit' $t\rightarrow\infty$ such that $\partial_{t}\mathcal{P}(u(t),t)=0$. This gives a stationary probability density $\mathcal{P}(u)$, for a stationary solution.
\begin{lem}
The stationary or equilibrium FPK equation is
\begin{equation}
\mathlarger{\mathrm{D}}_{u}(\phi(u)\mathcal{P}(u)+\frac{1}{2}
\mathlarger{\mathrm{D}}_{u}\mathlarger{\mathrm{D}}_{u}\big[|\psi(u))|^{2}
\mathcal{P}(u)\big]=0
\end{equation}
For any underlying SDE of the form $d\overline{u(t)}=\phi(u(t))dt+\psi(u(t))d\mathlarger{ \mathscr{B}}(t)$
The stationary solution is then given via the integrating factor method as
\begin{equation}
\mathcal{P}(u)=C \psi(u))^{-2}\exp\bigg(2\int_{u(0)}^{u}\phi(u)\psi^{-2}(u))\bigg)
\end{equation}
where $C$ is a normalisation constant. If $\phi(u)=0$ then the unique solution is $
\mathcal{P}(u)= \frac{C}{(\psi(u))^{2}} $. Now given such a probability density $\mathcal{P}(\widehat{u})$, the probability that $\overline{u(t)}$ resides within some finite interval $\mathbf{Q}$ is
\begin{equation}
\mathlarger{\mathrm{I\!P}}(\overline{u}\in\mathbf{Q})
=\int_{\mathbf{Q}}\bm{\mathcal{P}}(u(t))du(t)
\end{equation}
\end{lem}
\subsection{The estimates of the moments reproduced from the KFP equation}
The following theorem establishes the moments via the KFP equation and exactly reproduces the estimate of Theorem 7.4
\begin{lem}
Given the SDE $d\overline{u(t)}=k^{\frac{1}{2}}\psi^{4}(t)(\psi(t)-1)d\mathlarger{\mathlarger{\mathscr{B}}}(t) $ for $t>t_{\epsilon}$, the corresponding FPK equation is the nonlinear parabolic PDE
\begin{equation}
\frac{\partial}{\partial{t}}\mathcal{P}(u(t),t)=\frac{1}{2}\mathlarger{\mathrm{D}}_{u}
\mathlarger{\mathrm{D}}_{u}
u^{2}(t)(u^{2}(t)-1)^{1/2}\mathcal{P}(u(t),t))
\end{equation}
Given the probability density $\mathcal{P}(u(t),t)$, the moments for $p\ge 0$ are defined as
\begin{equation}
\mathlarger{\mathfrak{M}}_{p}(t)\equiv\mathlarger{\mathlarger{\mathcal{E}}}\bigg\llbracket|\overline{u(t)}|^{p}\bigg\rrbracket=
\int_{\mathbf{R}^{+}}|u(t)|^{p}\mathcal{P}(u(t),t)du(t)
\end{equation}
Then if $\exists C>0$ such that $ |u^{4}(t)(u(t)-1)|<C|u(t)|^{2}$
on any $[t_{\epsilon},T]$, or at least on $[t_{\epsilon},t_{*}]$ then the following estimate holds
\begin{equation}
\mathlarger{\mathfrak{M}}_{p}(t)<\exp(\frac{1}{2}\beta^{2}Cp(p-1)|t-t_{\epsilon}|)
\end{equation}
agreeing exactly with estimate of Theorem 7.4.
\end{lem}
\begin{proof}
Multiply the KFP equation by $|u(t)|^{p}$.
\begin{equation}
\frac{\partial}{\partial t}u(t)|^{p}\mathcal{P}(u(t),t)=\tfrac{1}{2}k|u(t)|^{p}
\mathlarger{\mathrm{D}}_{u}\mathlarger{\mathrm{D}}_{u}\big[
u^{4}(t)(u^{2}(t)-1)\mathcal{P}(u(t),t)\big]
\end{equation}
then integrate over $[u_{\epsilon},u(t)]$
\begin{align}
&\frac{\partial}{\partial t}\bigg\Arrowvert\overline{u(t)}\bigg\Arrowvert^{p}_{\mathcal{L}_{p}}
\equiv\frac{\partial}{\partial t}\mathlarger{\mathcal{E}}\bigg\llbracket\bigg|\overline{u(t)}\bigg|^{p}\bigg\rrbracket=
\frac{\partial}{\partial t}\left|\int_{\mathbf{U}}|u(t)|^{p}\mathcal{P}(u(t),t)
du(t)\right|\nonumber\\&=\tfrac{1}{2}\kappa\int_{\mathbf{U}}|u(t)|^{p}
\mathlarger{\mathrm{D}}_{u}\mathlarger{\mathrm{D}}_{u}
\big[u^{4}(t)(u^{2}(t)-1)\mathcal{P}u(t),t))\big]du(t)
\end{align}
Integrating the rhs by parts and taking the boundary term to vanish gives
\begin{align}
&\frac{\partial}{\partial t}\bigg\Arrowvert\overline{u(t)}\bigg\Arrowvert^{p}_{\mathcal{L}_{p}}
\equiv\frac{\partial}{\partial t}\mathlarger{\mathlarger{\mathcal{E}}}\bigg\llbracket\bigg|\overline{u(t)}\bigg|^{p}\bigg\rrbracket=
\frac{\partial}{\partial t}\left|\int_{\mathbf{U}}|u(t)|^{p}\mathcal{P}(u(t),t)
du(t)\right|\nonumber\\&=-\tfrac{1}{2}\kappa\int_{\mathbf{U}}|u(t)|^{p-1}
\mathlarger{\mathrm{D}}_{u}\big[u^{4}(t)(u^{2}(t)-1)
\mathcal{P}(u(t),t))\big]du(t)
\end{align}
Integrating by parts again
\begin{align}
&\frac{\partial}{\partial t}\bigg\Arrowvert\overline{u(t)}\bigg\Arrowvert^{p}_{\mathcal{L}_{p}}
\equiv\frac{\partial}{\partial t}\mathlarger{\mathlarger{\mathcal{E}}}\bigg\llbracket\bigg|\overline{u(t)}\bigg|^{p}\bigg\rrbracket=\frac{\partial}{\partial t}\left|\int_{\mathbf{U}}|u(t)|^{p}\mathcal{P}(u(t),t)du(t)\right|\nonumber\\&
=\tfrac{1}{2}kp(p-1)\int_{\mathbf{U}}|u(t)|^{p-2}
[u^{4}(t)(u^{2}(t)-1)]\mathcal{P}(u(t),t))du(t)
\end{align}
Let $Q>0$ be a constant such that
\begin{align}
&\frac{\partial}{\partial t}\bigg\Arrowvert\overline{u(t)}\bigg\Arrowvert^{p}_{\mathcal{L}_{p}}
\equiv\frac{\partial}{\partial t}\mathlarger{\mathcal{E}}\bigg\llbracket\bigg|\overline{u(t)}\bigg|^{p}\bigg\rrbracket
=\frac{\partial}{\partial t}\left|\int_{\mathbf{U}}|u(t)|^{p}\mathcal{P}(u(t),t)
du(t)\right|\nonumber\\&<\tfrac{1}{2}kQp(p-1)\int_{\mathbf{U}}|u(t)|^{p-2}
[u^{4}(t)(u^{2}(t)-1)]\mathcal{P}(u(t),t))du(t)
\end{align}
Now using the linear growth estimate $k u^{4}(t)(u(t)-1)<C|u(t)|^{2}$ on any finite subinterval $[t_{\epsilon},T]$
\begin{align}
&\frac{\partial}{\partial t}\bigg\Arrowvert\overline{u(t)}\bigg\Arrowvert^{p}_{\mathcal{L}_{p}}
\equiv\frac{\partial}{\partial t}\mathlarger{\mathlarger{\mathcal{E}}}\bigg\llbracket\bigg|\overline{u(t)}\bigg|^{p}\bigg\rrbracket\nonumber\\&=
\frac{\partial}{\partial t}\int_{\mathbf{U}}|u(t)|^{p}\mathcal{P}(u(t),t)
du(t)<\tfrac{1}{2}kQCp(p-1)\int_{\mathbf{U}}|u(t)|^{p}
\mathcal{P}(u(t),t))du(t)
\end{align}
which is
\begin{equation}
\frac{\partial}{\partial t}\bigg\Arrowvert\overline{u(t)}\bigg\Arrowvert^{p}_{\mathcal{L}_{p}}
\le\tfrac{1}{2}kQCp(p-1)\bigg\Arrowvert\overline{u(t)}\bigg\Arrowvert^{p}_{\mathcal{L}_{p}}
\end{equation}
or equivalently
\begin{equation}\frac{\partial}{\partial t}\mathlarger{\mathlarger{\mathcal{E}}}\bigg\llbracket \big|\overline{u(t)}\big|^{p}\bigg\rrbracket
\le\tfrac{1}{2}kQCp(p-1)\mathlarger{\mathlarger{\mathcal{E}}}\bigg\llbracket\big|\overline{u(t)}\big|^{p}\bigg\rrbracket
\end{equation}
Setting $\bar{C}=kQC$ gives the estimate
\begin{align}
&\mathlarger{\mathlarger{\mathcal{E}}}\bigg\llbracket |\sup_{t\le T}\overline{u(t)}|^{p}\bigg\rrbracket\le\mathlarger{\mathfrak{M}}_{p}(t_{\epsilon}(\exp\big(\tfrac{1}{2}\bar{C}p(p-1)|T-t_{\epsilon}|\big)\nonumber\\&
\equiv |u(t_{\epsilon}|^{p}(\exp\big(\tfrac{1}{2}\bar{C}p(p-1)|T-t_{\epsilon}|\big)
\end{align}
\end{proof}
\begin{rem}
Note that this estimate derived from the KFP equation agrees exactly with the estimate derived via the Ito Lemma.
\end{rem}
\subsection{Path integral representations and an Onsanger-Machlup Lagrangian}
The transition PDF $\mathcal{P}(u(t),t|u_{\epsilon},t_{\epsilon})$ can be represented as a functional-integral or Feynman-type path integral representation [73-79]. This is very well known in quantum theory and for the heat kernel or Greens function solution of the linear heat or diffusion theory. However, the KFP equation is a nonlinear PDE of parabolic type. The path-integral representation of $\mathcal{P}(u(T),T|u_{\epsilon},t_{\epsilon})$  can be constructed using the results of Dekker [76] who established the general form of the Onsanger-Machlup Lagrangian for any nonlinear KFP equation. It can then be established that
$\mathcal{P}(\infty,T|u_{\epsilon},t_{\epsilon})=0$, so that there is zero probability of any diffusive path blowing up or becoming singular. The main result is as follows [76]
\begin{thm}
Let $d\overline{u(t)}=\psi(u(t))d\mathlarger{\mathlarger{\mathscr{B}}}(t)$ be the driftless diffusion starting at some $t_{\epsilon}$ with $u(t_{\epsilon})=u_{\epsilon}$. The corresponding KFP equation for all $t\in\mathbf{X}_{II}\cup\mathbf{X}_{II}$ is the nonlinear diffusion or heat equation
\begin{equation}
\frac{\partial}{\partial t}\mathcal{P}(u(t),t)=\frac{1}{2}
\mathlarger{\mathrm{D}}_{u}\mathlarger{\mathrm{D}}_{u}[|\psi(u(t))|^{2}\mathcal{P}(u(t),t)
\end{equation}
The (Markovian) transition PDF $\mathcal{P}(u(T),T)|u_{\epsilon},t_{\epsilon})$ then has a functional-integral or path-integral representation
\begin{align}
&\mathcal{P}(u(T),T|u_{\epsilon},t_{\epsilon})
=\int\mathlarger{\mathrm{I\!D}}u(t)\exp\left(-\int_{t_{\epsilon}}^{T}
\mathcal{L}\left(u(t),\frac{du(t)}{dt}\right)dt\right)\nonumber\\&=
\int\mathlarger{\mathrm{I\!D}}u(t)\exp\left(-\int_{t_{\epsilon}}^{T}\left[
\frac{1}{2\beta_{1}(u(t))}\left(\left|\frac{du(t)}{dt}\right|-\beta_{2}(u(t))\right)^{2}+\beta_{3}(u(t))
\right]dt\right)
\end{align}
where $\int\mathlarger{\mathrm{D}}u(t)$ is over all paths from $u_{\epsilon}$ to $\overline{u(T)}$ and where
\begin{align}
&\bm{\beta}_{1}(u(t))=|\psi(u(t)|^{2}\\&
\beta_{2}(u(t))=-\tfrac{1}{4}\mathlarger{\mathrm{D}}_{u}|\psi(u(t)|^{2}\\&
\beta_{3}(u(t))=\frac{1}{2}\psi(u(t))\mathlarger{\mathrm{D}}_{u}\left(\frac{\beta_{2}(u(t))}{\psi(u(t))}\right)
\end{align}
\end{thm}
\begin{rem}
The Lagrangian
\begin{align}
\mathcal{L}\left(u(t),\frac{du(t)}{dt}\right)=\frac{1}{2\beta_{1}(u(t))}\left(\left|\frac{du(t)}{dt}
\right|-\beta_{2}(u(t))\right)^{2}+\beta_{3}(u(t))
\end{align}
is the Onsanger-Machlup Lagrangian for the diffusion [77,78,79].
\end{rem}
\begin{cor}
Using
\begin{equation}
\int^{T}dt=\int^{u(T)}\frac{dt}{du(t)}du(t)=\int^{u(T)}\psi^{-1}(u(t))du(t)
\end{equation}
The path-integral representation becomes
\begin{align}
\mathcal{P}&(u(T),T|u_{\epsilon},t_{\epsilon})
=\int\mathlarger{\mathrm{I\!D}}u(t)\exp\left(-\int_{u_{\epsilon}}^{u(T)}
\psi^{-1}(u(t))\mathcal{L}\left(u(t),\frac{du(t)}{dt}\right)dt\right)\nonumber\\&=
\int\mathlarger{\mathrm{I\!D}}u(t)\exp\left(-\int_{t_{\epsilon}}^{T}\left[
\frac{1}{2\psi(u(t))\beta_{1}(u(t))}\left(\left|\frac{du(t)}{dt}\right|-\beta_{2}(u(t))\right)^{2}+
\frac{\beta_{3}(u(t))}{\psi(u(t))}
\right]du(t)\right)\nonumber\\&=
\int\mathlarger{\mathrm{I\!D}}u(t)\exp\left(-\int_{t_{\epsilon}}^{T}\left[
\frac{1}{2\psi(u(t))\beta_{1}(u(t))}\bigg(\psi(u(t))-\beta_{2}(u(t))\bigg)^{2}+
\frac{\beta_{3}(u(t))}{\psi(u(t))}
\right]du(t)\right)\nonumber\\&
=\int\mathlarger{\mathrm{I\!D}}u(t)\exp\bigg(-\bigg(\frac{1}{2}\int_{u_{\epsilon}}^{u(T)}                            \frac{du(t)\bigg(\psi(u(t))+\frac{1}{4}
\mathlarger{\mathrm{I\!D}}|\psi(u(t))|^{2}\bigg)^{2}}{|\psi(u(t))|^{3}}
\nonumber\\&+\frac{1}{2}\int_{u_{\epsilon}}^{u(T)} \mathlarger{\mathrm{I\!D}}\bigg(-\frac{\tfrac{1}{4}
\mathlarger{\mathrm{I\!D}}|\psi(u(t)|^{2}}{\psi(u(t)}\bigg)du(t)
\bigg)\bigg)
\nonumber\\&\equiv
\int\mathlarger{\mathrm{I\!D}}u(t)\exp\left(-\left(\frac{1}{2}\int_{u_{\epsilon}}^{u(T)}                            \frac{du(t)\bigg(\psi(u(t))+\frac{1}{4}\mathlarger{\mathrm{D}}_{u}|\psi(u(t))|^{2}\bigg)^{2}}{|\psi(u(t))|^{3}}
-\frac{1}{8}\bigg(\frac{\mathlarger{\mathrm{D}}_{u}|\psi(u(t)|^{2}}{\psi(u(t)}\bigg)
\right)\right)
\end{align}
since $\int\mathlarger{\mathrm{D}}_{u} F(u(t))du(t)=F(u(t))$ for any functional $F(u(t))$.
\end{cor}
\begin{cor}
For the diffusion $d\overline{u(t)}=\kappa^{1/2}(u(t))^{2}(u(t)-1)^{1/2}$ with KFP equation
\begin{equation}
\frac{\partial}{\partial t}\mathcal{P}(u(t),t)=\frac{1}{2}k
\mathlarger{\mathrm{D}}_{u}\mathlarger{\mathrm{D}}_{u}\bigg[(u(t))^{4}(u(t)-1)
\mathcal{P}(u(t),t)\bigg]
\end{equation}
the path integral becomes
\begin{align}
&\mathcal{P}(u(T),T|u_{\epsilon},t_{\epsilon})\nonumber\\&
=\int\mathlarger{\mathrm{I\!D}}u(t)\exp\bigg(-\frac{1}{2}\bigg[\frac{1}{\kappa^{1/2}}
\bigg(\int_{u_{\epsilon}}^{u(T)}\frac{(\kappa^{1/2}(u(t))^{2}(u(t)-1)^{1/2}
+\frac{1}{4}(u(t))^{3}(5u(t)-4))^{2}}{(u(t))^{6}(u(t)-1)^{3/2}}du(t)\bigg)\nonumber\\&
-\frac{1}{8}\bigg(\frac{\kappa^{2}u(t)(5u(t)-4)}{(u(t)-1)^{1/2}}\bigg)
\bigg]\bigg)
\end{align}
\end{cor}
\begin{thm}
No diffuse path blows up for any finite $T$ so that
\begin{align}
&\mathcal{P}(\infty,T|u_{\epsilon},t_{\epsilon})=\lim_{u(T)\uparrow\infty}
\mathcal{P}(u(T),T|u_{\epsilon},t_{\epsilon})=0
\end{align}
and there is continuity of sample paths.
\end{thm}
\begin{proof}
Let $u(T)\rightarrow\infty$ for fixed $T$.
\begin{align}
&\mathcal{P}(\infty,T|u_{\epsilon},t_{\epsilon})=\lim_{u(T)\uparrow\infty}
\mathcal{P}(u(T),T|u_{\epsilon},t_{\epsilon})=0
\\&=\lim_{u(T)\uparrow\infty}
\int\mathlarger{\mathrm{I\!D}}u(t)\exp\bigg(-\frac{1}{2}\bigg[\frac{1}{\kappa^{1/2}}
\bigg(\int_{u_{\epsilon}}^{u(T)}\frac{(\kappa^{1/2}(u(t))^{2}(u(t)-1)^{1/2}
+\frac{1}{4}(u(t))^{3}(5u(t)-4))^{2}}{(u(t))^{6}(u(t)-1)^{3/2}}du(t)\bigg)
\nonumber\\&-\frac{1}{8}\bigg(\frac{\kappa^{2}u(t)(5u(t)-4)}{(u(t)-1)^{1/2}}\bigg)
\bigg]\bigg)
\end{align}
where the path integral is now over all paths on $[u_{\epsilon},\infty)$. The integral can be estimated
\end{proof}
\begin{align}
&\mathcal{P}(\infty,T|u_{\epsilon},t_{\epsilon})=\lim_{u(T)\uparrow\infty}
\mathcal{P}(u(T),T|u_{\epsilon},t_{\epsilon})
\\&=\lim_{u(T)\uparrow\infty}\int\mathlarger{\mathrm{I\!D}}u(t)\exp\bigg(-\frac{1}{2}\bigg[k^{1/2}\bigg(
\frac{525(u(t))^{7}+700(u(t))^{6}+1000(u(t))^{5}+1138(u(t))^{4}}{1848u(t)\sqrt(u(t)-1)}\bigg)
\nonumber\\&+k^{1/2}\bigg(\frac{2276(u(t))^{3}+9104(u(t))^{2}
-1696u(t)-1848}{1848u(t)\sqrt(u(t)-1)}\bigg)+k^{1/2}\bigg(\frac{5u(t)}{2}+2\log(1-u(t))\nonumber\\&+\frac{1}{2}\log(u(t))+\tan^{-1}((u(t)-1)^{1/2}\bigg)
-\frac{1}{8}\bigg(\frac{ k^{2}u(t)(5u(t)-4)}{(u(t)-1)^{1/2}}\bigg)
\bigg]\bigg)\nonumber\\&=\lim_{u(T)\uparrow\infty}\int\mathlarger{\mathrm{I\!D}}u(t)\exp\bigg(-\frac{1}{2}\bigg(S_{1}(u(T))+S_{2}(u(T))
+S_{3}(u(T)+S_{4}(u(T)))\bigg)\bigg)
\end{align}
The asymptotic behaviour of the sum of terms or actions can be estimated and one finds that
\begin{equation}
\lim_{u(T)\uparrow\infty}(S_{1}(u(T))+S_{2}(u(T))+S_{3}(u(T)+S_{4}(u(T)))=\infty
\end{equation}
so that
$\exp(-S_{1}(u(T))-S_{2}(u(T))-S_{3}(u(T))\rightarrow 0$. Hence (10.53) immediately follows and the transition probability to a singular state or blowup is exactly zero.
\section{Transform to geometric Brownian motion: exponential martingale, moments and finite bounds}
Although the properties of the formal stochastic Ito integral $\overline{u(t)}=
u_{\epsilon}+\int_{t_{\epsilon}}^{t}(u(t))^{2}(u(t)-1)^{1/2}
d\mathlarger{\mathlarger{\mathscr{B}}}(t)$ can be studied, an explicit solution cannot be found. However, by a 'change of measure' the SDE can be transformed into a linear form (geometric Brownian motion) that can be explicitly solved [60,61,80]
\begin{prop}
Let $Y(t)$ be a smooth and continuous $C^{2}$-differentiable function defined for all $t>t_{\epsilon}$ or $t\in\mathbf{X}_{II}\bigcup\mathbf{X}_{III}$, and finite for all $t>t_{\epsilon}$. Then one can formulate an SDE for a 'geometric diffusion' for a new density function $Y(t)$
\begin{equation}
d\overline{Y(t)}=Y(t)d\overline{u(t)}=Y(t)\kappa^{1/2}|u^{2}(t)(u(t)-1)^{1/2}|
d\mathlarger{\mathlarger{\mathscr{B}}}(t)
\end{equation}
We can now exactly solve and study the SDE $d\overline{Y(t)}=Y(t)d\overline{u(t)}(t)$. This is tantamount to a change of measure [61] such that $ \int d\mathlarger{\mathlarger{\mathscr{B}}}(t)\longrightarrow \int~d\overline{u(t)})$. The formal solution is now the stochastic integral
\begin{equation}
\overline{Y(t)}=Y_{\epsilon}+\int_{t_{\epsilon}}^{t}Y(s)d\overline{u(s)}=
Y_{\epsilon}+\kappa^{1/2}\int_{t_{\epsilon}}^{t}Y(s)|u^{2}(s)(u(s)-1)^{1/2}|
d\mathlarger{\mathlarger{\mathscr{B}}}(t)
\end{equation}
\end{prop}
\begin{lem}
Let $\mathbf{R}^{+}=\mathbf{X}_{I}\bigcup\mathbf{X}_{II}\mathbf{X}_{III}$, where $\mathbf{X}_{I}\cup\mathbf{X}_{II}$ is the usual collapse interval for the original deterministic description. The new hybrid ODE-SDE over this partition is
\begin{equation}
d\overline{Y(t)}=\mathlarger{\mathcal{C}}(\mathbf{X})_{I})Y(t)dt
+\mathlarger{\mathcal{C}}(\mathbf{X})_{I}\cup\mathbf{X}_{II})Y(t)d\overline{u}(t)
\end{equation}
The unique strong solution of the SDE exists for $t>t_{\epsilon}$ and is the Doleans-Dade stochastic exponential [61]
\begin{equation}
\overline{Y(t)}={Y}_{\epsilon}+\int_{t_{\epsilon}}^{t}Y(s)d\overline{u(s)}
=\exp\left(\overline{u(t)}-u_{\epsilon}-\frac{1}{2}\bigg\langle\overline{u},\overline{u}
\bigg\rangle(t)\right)
\end{equation}
or
\begin{equation}
\overline{Y(t)}=Y_{\epsilon}\exp\bigg(\int_{t_{\epsilon}}^{t}
u^{2}(s)(u(s)-1)^{1/2}d\mathlarger{\mathlarger{\mathscr{B}}}(t) -
\frac{1}{2}\int_{t_{\epsilon}}^{t}d\bigg\langle\overline{u},\overline{u}\bigg\rangle(t)\bigg)
\end{equation}
which describes a stochastic exponential with expectation
\begin{equation}
\mathlarger{\mathlarger{\mathcal{E}}}(\overline{Y}(t))=Y_{\epsilon}\mathlarger{\mathlarger{\mathcal{E}}}
\left(\exp\bigg(\kappa^{1/2}\int_{t_{\epsilon}}^{t}u^{2}(s)(u(s)-1)^{1/2}d\mathlarger{\mathlarger{\mathscr{B}}}(t)
-\frac{1}{2}\int_{t_{\epsilon}}^{t}d\bigg\langle\overline{u},\overline{u}\bigg\rangle(t)\bigg)
\right)
\end{equation}
\end{lem}
\begin{proof}
Let $\mathcal{F}(Y(t))$ be a $C^{2}$-differentiable functional then Ito' Lemma gives
\begin{align}
d\mathcal{F}(\overline{Y(t)})&=\mathcal{F}(Y(t))dY(t)
+\frac{1}{2}(\mathlarger{\mathlarger{\mathrm{D}}}_{u}\mathlarger{\mathlarger{\mathrm{D}}}_{u}\mathcal{F}(Y(t))d\bigg\langle\overline{u},\overline{u}\bigg\rangle(t)
\nonumber\\&=\kappa^{1}{2}\mathcal{F}(Y)|u^{2}(t)(u(t)-1)^{1/2}|
d\mathlarger{\mathlarger{\mathscr{B}}}(t)+\frac{1}{2}\mathlarger{\mathlarger{\mathrm{D}}}_{u}
\mathlarger{\mathlarger{\mathrm{D}}}_{u}\mathcal{F}(Y(t))|Y(t)|^{2}
d\bigg\langle\overline{u},\overline{u}\bigg\rangle(t)\nonumber\\&\equiv
=\kappa^{1}{2}\mathcal{F}(Y)u^{2}(t)(u(t)-1)^{1/2}d\mathlarger{\mathlarger{\mathscr{B}}}(t)+\frac{1}{2}
\kappa\mathlarger{\mathlarger{\mathrm{D}}}_{u}\mathlarger{\mathlarger{\mathrm{D}}}_{u}
\mathcal{F}(Y(t))|Y(t)|^{2}u^{4}(t)(u(t)-1)^{1/2}dt
\end{align}
Let $\mathcal{F}(\overline{Y(t)})=\log(\overline{Y(t)})$ then
$\log(Y_{\epsilon})>0$ and $\lim_{S\rightarrow\infty}
\log(\widehat{Y}(t))=\infty $.
Then
\begin{align}
&d\log(\overline{Y(t)})=\kappa^{1}{2}\frac{1}{Y(t)} u^{4}(t)(u(t)-1)^{1/2}
d\mathlarger{\mathlarger{\mathscr{B}}}(t)+\frac{1}{2}\frac{1}{|Y(t)|^{2}}|S(t)|^{2}
\kappa u^{4}(t)(u(t)-1)^{1/2}dt
\nonumber\\&=\kappa^{1}{2}u^{4}(t)(u(t)-1)^{1/2}
d\mathlarger{\mathlarger{\mathscr{B}}}(t)+\frac{1}{2}\kappa u^{4}(t)(u(t)-1)dt
\end{align}
Integrating
\begin{align}
&\log(\overline{Y(t)})=\log (Y_{\epsilon})+\kappa^{1}{2}
\int_{t_{i}}^{t}u^{4}(s)(u(s)-1)^{1/2}d\mathlarger{\mathlarger{\mathscr{B}}}(s)
+k\int_{t_{i}}^{t}u^{4}(s)(u(s)-1)ds
\end{align}
so that
\begin{align}
&\overline{Y(t)}=Y_{\epsilon}\exp\bigg(\kappa^{1}{2}
\int_{t_{i}}^{t}(u(s))^{2}(u(s)-1)^{1/2}d\mathlarger{\mathlarger{\mathscr{B}}}(s)+\kappa u^{4}(s)(u(s)-1)^{1/2}du\bigg)
\nonumber\\& \equiv Y_{\epsilon}\exp\bigg(\kappa^{1}{2}
\int_{t_{i}}^{t}(u(s))^{2}(u(s)-1)^{1/2} d\mathlarger{\mathlarger{\mathscr{B}}}(s)+
\int_{t_{i}}^{t}d\bigg\langle\overline{u},\overline{u}\bigg\rangle(s)\bigg)
\bigg)
\end{align}
\end{proof}
\begin{cor}
Novikov's Thms states that the expectation of the Dolean-Dades exponential is unity so so that $\mathlarger{\mathlarger{\mathcal{E}}}(\overline{Y(t)})=Y_{\epsilon}$.
\end{cor}
If $\overline{u(t)}$ is free from blowups for all $t>t_{\epsilon}$, which is the case since it is a martingale, then $\mathlarger{\mathlarger{\mathcal{E}}}(\widehat{Y(t)})$ will also
not blow up for all finite $t>t_{\epsilon}$.
\begin{thm}
Given the density function diffusion $\widehat{u(t)}$ which is a square-integrable martingale with $\big\langle\overline{u},\overline{u}\bigg\rangle(t)<\infty$ then $\overline{Y(t)}$ can be expressed as the bounded convergent sum
\begin{align}
&\overline{Y(t)}=Y_{\epsilon}\sum_{n=0}^{\infty}\frac{1}{n}(\frac{1}{2}
\bigg\langle\overline{u},\overline{u}\bigg\rangle(t))^{n/2}H_{n}
\left(\frac{|\overline{u}(t)|}{\sqrt{2\bigg\langle\overline{u},\overline{u}\bigg\rangle(t)}}\right)\nonumber\\&
=Y_{\epsilon}\sum_{n=0}^{\infty}\Lambda_{n}=S_{\epsilon}\sum_{n=0}^{n}\int_{t_{\epsilon}}^{t}\Lambda_{n-1}(s)d\overline{u(s)}
\end{align}
where $\Lambda \Lambda_{0}=1$ and $H_{n}(...)$ are the Hermite polynomials.
\end{thm}
\begin{proof}
Defining the iterated stochastic integrals
\begin{equation}
\Lambda_{n}(t)=\int_{t_{\epsilon}}^{t}\Lambda_{n-1}(s)d\overline{u(s)}
\end{equation}
with $\lambda_{0}(t)=1$, then (-) can be expressed as
\begin{equation}
\overline{Y(t)}=Y_{\epsilon}\exp\bigg(|\widehat{u}(t)|-\frac{1}{2}\bigg\langle\overline{u},\overline{u}
\bigg\rangle(t)\bigg)=Y_{\epsilon}\sum_{n=1}^{\infty}
\int_{t_{\epsilon}}^{t}\Lambda_{n-1}(s)ds
=\sum_{n=0}^{\infty}\Lambda_{n}
\end{equation}
Let $\xi\in\mathbf{R}^{+}$ and replace $|\overline{u(t)}|$ with $\xi|\overline{u(t)}|$ so that
\begin{equation}
\exp\bigg(\xi|\widehat{u}(t)|-\frac{1}{2}\xi^{2}\big[\overline{u},\overline{u}\big](t)\bigg)=\sum_{n=0}^{\infty}\xi^{n}\Lambda_{n}
\end{equation}
Setting $x=|\overline{u}(t)|/\sqrt{2\bigg\langle\overline{u},\overline{u}\rangle(t)}$ and $\theta=\xi\sqrt{\frac{1}{2}\bigg\langle\overline{u},\overline{u}\bigg\rangle (t)}$ then (11.16) becomes
\begin{equation}
\exp(2\bar{x}(t)\theta-\theta^{2})=\sum_{n=0}^{\infty}\xi^{n}\Lambda_{n}
\equiv\sum_{n=0}^{\infty}\frac{H_{n}(\bar{x(t)})}{n!}|\theta|^{n}
\end{equation}
where $H_{n}(...)$ are the standard Hermite polynomials. Hence
\begin{equation}
\exp\bigg(\xi|\overline{u(t)}|-\frac{1}{2}\xi^{2}\bigg\langle\overline{u},\overline{u}\bigg\rangle(t)\bigg)
=\sum_{n=0}^{\infty}\frac{\xi^{n}}{n!}(\frac{1}{2}\bigg\langle\overline{u},\overline{u}\bigg\rangle(t))^{n/2}H_{n}
\left(\frac{|\overline{u}(t)|}{\sqrt{2\bigg\langle\overline{u},\overline{u}\bigg\rangle(t)}}\right)
\equiv\sum_{n=0}^{\infty}\xi^{n}\Lambda_{n}
\end{equation}
\end{proof}
The Novikov Theorem [61] establishes that $\mathcal{S}\llbracket
\widehat{\mathcal{Y}}(t)\rrbracket =1$ for any SDE $d\widehat{Y}(t)=Y(t)\psi((u(t))d\mathlarger{\mathlarger{\mathscr{B}}}(t)(t)$, where $\overline{u(t)}$ is a martingale. The Novikov condition also requires that $\frac{1}{2}\mathlarger{\mathlarger{\mathcal{E}}}(\bigg\langle\overline{u},\overline{u}])<\infty$ then $\overline{Y(t)} $ is uniformly integrable and is a martingale. This is important when, for example, one wants to utilise change-of-probability techniques such as the Girsanov theorem. Though the Novikov condition has been relaxed, the constant 1/2 within it cannot be replaced by any smaller one. In some problems [81], however, it is also important to know whether, in addition, the moments are finite(for some or all p>1) and under which conditions.
\subsection{Estimates for moments}
Here, we compute the moments $\mathfrak{M}_{p}(t)\equiv \mathlarger{\mathlarger{\mathcal{E}}}|\widehat{Y(t)|}^{p}$ for $p\ge 1$ for all solutions of SDEs describing stochastic exponential growth, via Ito's Lemma. $\mathlarger{\mathfrak{M}}_{p}(t)=\mathlarger{\mathlarger{\mathcal{E}}}(|\widehat{Y(t)|}^{p}$ of the DD exponential. The following theorem then establishes that the moments $ \mathlarger{\mathlarger{\mathcal{E}}}(|\widehat{Y(t)}|^{p})$ are always finite and bounded for $t>t_{\epsilon}$
\begin{thm}
Given the diffusion $\overline{X(t)}$ for all $t>t_{\epsilon}$ or
$t\in\mathbf{X}_{II}\bigcup\mathbf{X}_{III}^{\infty}$ such that $d\overline{X(t)}=X(t)(u(t))^{4}(u(t)-1)d\mathlarger{\mathlarger{\mathscr{B}}}(t)$ then
\begin{enumerate}
\item All moments $\mathfrak{M}_{p}(t)=\mathlarger{\mathlarger{\mathcal{E}}}|\overline{X(t)}|^{p}$ are finite and bounded for all $t>t_{\epsilon}$ such that
\begin{equation}
\mathfrak{M}_{p}(t)\equiv \mathlarger{\mathlarger{\mathcal{E}}}(|\overline{X(t)}|^{p})=
|X_{\epsilon}|^{p}\exp\left(\tfrac{1}{2}
kp(p-1)\int_{t_{\epsilon}}^{t}\mathlarger{\mathlarger{\mathcal{E}}}u^{4}(s)(u(s)-1)ds
\right)<\infty
\end{equation}
so that the moments are finite if $\widehat{u}(t)$ is square integrable. Hence, the moments are finite and bounded and the stochastic exponential of the matter density does not blow up.
\item The moments satisfy a linear differential equation so that
\begin{equation}
\frac{d\mathfrak{M}_{p}(t)}{dt}=\frac{1}{2}kp(p-1)|u^{4}(t)(u(t)-1)|\mathfrak{M}_{p}(t)
\end{equation}
\end{enumerate}
\end{thm}
\begin{proof}
Let $F(Y(t)=|\overline{Y(t)}|^{p}$ then Ito's Lemma gives
\begin{align}
&d\mathcal{F}(\overline{Y(t)})=d|Y(t)|^{p}
\nonumber\\&=p|Y(t)|^{p-1}Y(t)u^{4}(t)(u(t)-1)^{1/2}d\mathlarger{\mathlarger{\mathscr{B}}}(t)
+\frac{1}{2}\kappa p(p-1)|Y(t)|^{p-2}|Y(t)|^{2}
u^{4}(t)(u(t)-1)dt\nonumber\\&=p|Y(t)|^{p}\kappa^{1/2}u^{2}(t)(u(t)-1)^{1/2} d\mathlarger{\mathlarger{\mathscr{B}}}(t)+\frac{1}{2}\kappa p(p-1)|Y(t)|^{p}
\kappa u^{4}(t)(u(t)-1)dt
\end{align}
Integrating
\begin{align}
&|\overline{Y(t)}|^{p}=|Y_{\epsilon}|^{p}+p\kappa\int_{t_{i}}^{t}|Y(t)|^{p}u^{2}(t)(u(t)-1)^{1/2} d\mathlarger{\mathlarger{\mathscr{B}}}(t)\nonumber\\&+\frac{1}{2}\kappa p(p-1)   \int_{t_{i}}^{t}|Y(s)|^{p}u^{4}(s)(u(s)-1)ds
\end{align}
The expectation is then
\begin{equation}
\mathlarger{\mathlarger{\mathcal{E}}}(|\overline{Y(t)}|^{p})=|Y_{\epsilon}|^{p}+
\frac{1}{2}\kappa p(p-1)\int_{t_{i}}^{t}\mathlarger{\mathlarger{\mathcal{E}}}(|Y(u)|^{p})u^{4}(s)(u(s)-1)ds
\end{equation}
The Gronwall Lemma then gives
\begin{equation}
\mathfrak{M}_{p}(t)=\mathlarger{\mathlarger{\mathcal{E}}}(u(t)|^{p})
\exp\left(\frac{1}{2}\kappa p(p-1)\int_{t_{i}}^{t}
u^{4}(s)(u(s)-1)ds\right)
\end{equation}
which can be interpreted as a solution of the ODE
\begin{equation}
\frac{d\mathfrak{M}(t)}{dt}=\frac{1}{2}|\kappa p(p-1)u^{4}(t)(u(t)-1)\mathfrak{M}_{p}(t)
\end{equation}
\end{proof}
These results can also be reproduced via the corresponding Kolmogorov-Fokker-Planck (KFP) equation
\begin{lem}
Given the SDE $d\overline{Y}(t)=\kappa^{1/2}Y(t)u^{2}(t)(u(t)-1)^{1/2} d\mathlarger{\mathlarger{\mathscr{B}}}(t)$ for $ t>t_{\epsilon}$ and $Y(t_{\epsilon})=Y_{\epsilon}$, the corresponding FPK equation for the evolution of a probability density $\mathcal{P}(Y(t),t)$ is
\begin{equation}
\frac{\partial}{\partial t}\mathcal{P}(Y(t),t)=\frac{1}{2}\kappa^{1/2}
\mathlarger{\mathlarger{\mathrm{D}}}_{Y}\mathlarger{\mathlarger{\mathrm{D}}}_{Y}
u^{4}(t)(u(t)-1)(|Y(t)|^{2}\mathcal{P}(Y(t),t)
\end{equation}
where $\mathlarger{\mathlarger{\mathrm{D}}}_{Y}=d/dY(t)$. If
\begin{equation}
\mathfrak{M}_{p}(t)=\int_{\mathbf{R}^{+}}|Y(t)|^{p}
\mathcal{P}(Y(t),t))
\end{equation}
then
\begin{equation}
\frac{d\mathfrak{M}_{p}(t)}{dt}=\frac{1}{2}\kappa p(p-1)u^{4}(t)(u(t)-1)\mathfrak{M}_{p}(t)
\end{equation}
and
\begin{equation}
\mathlarger{\mathfrak{M}}_{p}(t)\equiv\mathlarger{\mathlarger{\mathcal{E}}}(|\widehat{Y}(t)|^{p})=
|Y_{\epsilon}|^{p}\exp\left(\tfrac{1}{2}
\kappa p(p-1)\int_{t_{\epsilon}}^{t}\mathlarger{\mathlarger{\mathcal{E}}}u^{4}(s)(u(s)-1)ds
\right)<\infty
\end{equation}
\end{lem}
\begin{proof}
Multiply the FPK equation by $|Y(t)|^{p}$ and integrate so that
\begin{align}
&\frac{\partial}{\partial t}\mathlarger{\mathfrak{M}}_{p}(t)\equiv
\frac{\partial}{\partial t}\bigg\Arrowvert\overline{Y(t)}\bigg\Arrowvert^{p}_{\mathcal{L}_{p}}
\equiv\frac{\partial}{\partial t}\mathlarger{\mathlarger{\mathcal{E}}}\bigg\llbracket\bigg|
\overline{Y(t)}\bigg|^{p}\bigg\rrbracket=
\frac{\partial}{\partial t}\int_{t_{\epsilon}}^{t}|Y(t)|^{p}\mathcal{P}(Y(t),t)dY(t)\nonumber\\&
=\frac{1}{2}\kappa^{1/2}u^{4}(t)(u(t)-1)
Y(t)|^{p}\int_{t_{\epsilon}}^{t}\mathlarger{\mathlarger{\mathrm{D}}}_{u}
\mathlarger{\mathlarger{\mathrm{D}}}_{u}(|Y(t)|^{2}
\mathcal{P}(Y(t),t)dY(t)
\end{align}
and assuming that $\mathcal{P}(Y(t),t)\rightarrow 0$
as $Y(t)\rightarrow\infty$. Integrating by parts once on the rhs
\begin{align}
&\frac{\partial}{\partial t}\mathlarger{\mathfrak{M}}_{p}(t)\equiv
 \frac{\partial}{\partial t}\bigg\Arrowvert\overline{Y(t)}\bigg\Arrowvert^{p}_{\mathcal{L}_{p}}
\equiv\frac{\partial}{\partial t}\mathlarger{\mathlarger{\mathcal{E}}}\bigg\llbracket\bigg|\overline{Y(t)}\bigg|^{p}
\bigg\rrbracket=
\frac{\partial}{\partial t}\int_{t_{\epsilon}}^{t}|Y(t)|^{p}\mathcal{P}(Y(t),t)dY(t)\nonumber\\&
=\frac{1}{2}\kappa^{1/2}pu^{4}(t)(u(t)-1)
Y(t)|^{p-1}\int_{t_{\epsilon}}^{t}\mathlarger{\mathlarger{\mathrm{D}}}_{u}(|Y(t)|^{2}
\mathcal{P}(Y(t),t)dY(t)
\end{align}
Integrating the rhs by parts again
\begin{align}
&\frac{\partial}{\partial t}\mathlarger{\mathfrak{M}}_{p}(t)\equiv\frac{\partial}{\partial t}\bigg\Arrowvert\overline{Y(t)}\bigg\Arrowvert^{p}_{\mathcal{L}_{p}}
\equiv\frac{\partial}{\partial t}\mathlarger{\mathlarger{\mathbf{E}}}\bigg\llbracket\bigg|\overline{Y(t)}\bigg|^{p}\bigg\rrbracket=
\frac{\partial}{\partial t}\int_{t_{\epsilon}}^{t}|Y(t)|^{p}\mathcal{P}(Y(t),t)dY(t)\nonumber\\&
=\frac{1}{2}k^{1/2}p(p-1)u^{4}(t)(u(t)-1)
Y(t)|^{p-2}\int_{t_{\epsilon}}^{t}|Y(t)|^{2}\mathcal{P}(Y(t),t)dY(t)
\end{align}
which is exactly (11.l9)
\begin{equation}
\frac{d\mathfrak{M}_{p}(t)}{dt}=\frac{1}{2}kp(p-1)u^{4}(t)(u(t)-1)\mathfrak{M}_{p}(t)
\end{equation}
Hence, the moments estimate follows as before.
\end{proof}
\section{The Stratanovich interpretation}
The previous sections established the no-blowup criteria for the density function diffusion within the Ito interpretation. The diffusion is a martingale and that the suprema of the martingale is bounded for all $t\in\mathbf{X}^{+}_{II}\bigcup\mathbf{X}^{+}_{III}$. Consequently, there is zero probability of a blowup or singularity. The problem is now considered from within the Stratanovich interpretation. One key advantage of this interpretation is that the usual rules of calculas apply, rather than those of the Ito calculas, and Stratanovich  SDEs can be explicitely solved [58,61]
\begin{defn}
Given the Ito SDE is $d\overline{u(t)}=\psi(u(t))d\mathlarger{\mathlarger{\mathscr{B}}}(t)(t)$ for initial data $u(t_{\epsilon}=u_{\epsilon}$ and $\mathlarger{\mathlarger{\mathscr{B}}}(t_{\epsilon}=0$, then the Ito integral solution is formally the Riemann-Stieljes sum
\begin{equation}
\overline{u(t)}=u_{\epsilon}+\int_{t_{\epsilon}}^{T}\psi(u(s))
d\mathlarger{\mathlarger{\mathscr{B}}}(s)
=\lim_{n\uparrow\infty}\sum_{i=1}^{n}\psi(u(t_{i}^{n})[\mathlarger{\mathlarger{\mathscr{B}}}(T_{i+1}^{n})-
\mathlarger{\mathlarger{\mathscr{B}}}(T_{i}^{n})],~
t_{i}^{n}=T_{i}^{n}
\end{equation}
with respect to the partition of $[T_{i}^{n},T_{i+1}^{n}]$. The Stratanovich SDE is
\begin{equation}
d\overrightarrow{u(t)}=\psi(u(t))\partial\mathlarger{\mathlarger{\mathscr{B}}}(t)
\end{equation}
where an overline right arrow now denotes stochastic quantities within the Stratanovich interpretation and $\partial\mathlarger{\mathlarger{\mathscr{B}}}(t)$ represents the differential. The Stratanovich integral is then formally
\begin{equation}
\overrightarrow{u(t)}=u_{\epsilon}+\int_{t_{\epsilon}}^{T}\psi(u(s))\partial
\mathlarger{\mathlarger{\mathscr{B}}}(t)(s)
=\lim_{n\uparrow\infty}\sum_{i=1}^{n}\psi(u(t_{i}^{n})
[\mathlarger{\mathlarger{\mathscr{B}}}(T_{i+1}^{n})-\mathlarger{\mathlarger{\mathscr{B}}}(T_{i}^{n})],~
t_{i}^{n}=\tfrac{1}{2}(T_{i}^{n}+T_{i+1}^{n})
\end{equation}
These are related by
\begin{equation}
d\overrightarrow{u(t)}=d\overline{u(t)}+\psi(u(t))\mathlarger{\mathrm{D}}_{u}\psi(u(t))dt
\end{equation}
or
\begin{equation}
\overrightarrow{u(t)}=\overline{u(t)}+\psi(u(t))\mathlarger{\mathrm{D}}_{u}\psi(u(t))
\end{equation}
The Stratanovich SDE can be solved using ordinary calculas such that
\begin{equation}
\int_{t_{\epsilon}}^{\overrightarrow{u(t)}}\frac{du(t)}{\psi(u(t))}=\int_{t_{\epsilon}}^{T}\partial \mathlarger{\mathlarger{\mathscr{B}}}(s)=\mathlarger{\mathlarger{\mathscr{B}}}(t)-\mathlarger{\mathlarger{\mathscr{B}}}(t_{\epsilon}=
\mathlarger{\mathlarger{\mathscr{B}}}(t)
\end{equation}
\end{defn}
\begin{prop}
Let the following hold as in Proposition 24.
\begin{itemize}
\item The spherical star collapses from $R(0)=1$ to $R(t_{*})=0$ within the finite comoving proper time $t_{*}=\pi/2 k^{1/2}$. This is equivalent to a blowup in the density function $u(t)=(\rho(t)/\rho(0))^{1/3}$ such that $u(t_{*})=\infty$ from initial data $u(0)=R^{-1}(0)=1$.
\item Let $t_{\epsilon}$ be a proper time infinitesimally close to the blowup proper time $t_{*}$ such that $t_{*}=t_{\epsilon}+\epsilon$ with $|\epsilon|\ll 1$. Define the same partition of $\widehat{\mathbf{R}}^{+}$ such that
    \begin{align}
\widehat{\mathbf{R}}^{+}=\mathbf{X}_{I}\bigcup\bigg(\mathbf{X}^{+}_{II}\bigcup\mathbf{X}^{+}_{III}\bigg)=
[0,t_{\epsilon}]\bigcup\bigg([t_{\epsilon},\infty]
=[0,t_{\epsilon}]\bigcup[t_{\epsilon},t_{*}]~
\bigcup~[t_{*},\infty]\bigg)
\end{align}
\item For all $t\in\mathbf{X}^{+}_{II}\bigcup\mathbf{X}^{+}_{III}$, or $t>t_{\epsilon}$
    'switch on' a Gaussian (white) stochastic (white noise) control
    $\mathlarger{\mathlarger{\mathscr{B}}}(t)$ with expectation $\mathlarger{\mathlarger{\mathcal{E}}}\big\lbrace\mathlarger{\mathlarger{\mathscr{W}}}(t)\big\rbrace=0$ and 2-point correlation $\mathlarger{\mathlarger{\mathcal{E}}}(\mathlarger{\mathlarger{\mathscr{W}}}(t)
    \mathlarger{\mathlarger{\mathscr{W}}}(s))=
    \alpha\delta(t-s)$ for $t,s\in\mathbf{X}_{II}\bigcup\mathbf{X}_{III}$
\end{itemize}
Then one can 'engineer' the following system of equations over the entire partition
$\mathbf{X}^{+}=\mathbf{X}_{I}\bigcup\mathbf{X}_{II}\bigcup\mathbf{X}_{III}$ such that
\begin{align}
&\frac{d\overrightarrow{u(t)}}{dt}=\psi(u(t))=k^{1/2}(u(t))^{2} (u(t)-1)^{1/2}~if~t\in \mathbf{X}^{+}_{I}=[0,t_{\epsilon}]\nonumber\\&
\frac{d}{dt}\overrightarrow{u(t)}=\psi(u(t))\mathlarger{\mathlarger{\mathscr{B}}}(t)=k^{1/2}(u(t))^{2}
(u(t)-1)^{1/2}\mathlarger{\mathlarger{\mathscr{W}}} if~t\in\mathbf{X}^{+}_{II}
\bigcup\mathbf{X}^{+}_{III}
\end{align}
or
\begin{align}
&d\overrightarrow{u(t)}=\psi(u(t))=\kappa^{1/2}(u(t))^{2} (u(t)-1)^{1/2}dt~if~t\in \mathbf{X}^{+}_{I}=[0,t_{\epsilon}]\nonumber\\&
d\widehat{u}(t)= \sigma[u]\otimes\mathlarger{\mathlarger{\mathscr{B}}}(t)dt=\kappa^{1/2}(u(t))^{2} (u(t)-1)^{1/2}d\mathlarger{\mathlarger{\mathscr{B}}}(t)~if~t\in\mathbf{X}_{II}
\bigcup\mathbf{X}_{III}
\end{align}
since $d\mathlarger{\mathlarger{\mathscr{B}}}(t)=\mathlarger{\mathlarger{\mathscr{W}}}(t)dt$, where $\mathlarger{\mathlarger{\mathscr{B}}}(t)$ is the standard Weiner process or Brownian motion which exists for all $t\in\mathbf{X}_{II}\bigcup\mathbf{X}_{III}$. Using 'indicator functions' this can be expressed on $ \mathbf{R}^{+}=\mathbf{X}_{II}\bigcup\mathbf{X}_{III}$ as a 'hybrid' DE-SDE as before but now within the Stratanovich  interpretation
\begin{eqnarray}
\dot{\widehat{u}}(t)=\mathcal{C}(\mathbf{X}_{I})\kappa^{1/2}(u(t))^{2}(u(t)-1)^{1/2}+
\mathbf{U}(\mathbf{X}_{II}\cup\mathbf{X}_{III})\kappa^{1/2}(u(t))^{2}(u(t)-1)^{1/2}
\mathlarger{\mathlarger{\mathscr{B}}}(t)
\end{eqnarray}
where $\mathcal{C}(\mathbf{X}_{I})=1$ if $t\le t_{\epsilon}$ and zero otherwise, while $\mathcal{C}(\mathbf{X}_{II}\cup\mathbf{X}_{III})=1$ if $t\ge t_{\epsilon}$ or equivalently
\begin{eqnarray}
d\overrightarrow{u(t)}=\mathcal{C}(\mathbf{X}_{I})k^{1/2}u(t))^{2}(u(t)-1)^{1/2}
\mathcal{C}(\mathbf{X}_{II}\cup\mathbf{X}_{III}k^{1/2}u(t))^{2}(u(t)-1)^{1/2}
\partial\mathlarger{\mathlarger{\mathscr{B}}}(t)
\end{eqnarray}
The Stratanovich integral solution over $[0,\infty)=\mathbf{X}^{+}_{I}\bigcup(\mathbf{X}^{+}_{II}\bigcup\mathbf{X}^{+}_{III})$
is now formally
\begin{align}
&\overrightarrow{u(t)}=\mathcal{C}(\mathbf{X}_{I})
\int_{0}^{t}\kappa^{1/2}(u(\tau))^{2}(u(\tau)-1)^{1/2}d\tau+
\mathcal{C}(\mathbf{X}_{II}\cup\mathbf{X}_{III})\int_{t_{\epsilon}}^{t}\kappa^{1/2}(u(\tau))^{2}
(u(\tau)-1)^{1/2}\partial\mathlarger{\mathlarger{\mathscr{B}}}(s)
\end{align}
or
\begin{eqnarray}
\overrightarrow{u(t)}=u_{\epsilon}+\mathlarger{\mathcal{C}}(t\le t_{\epsilon})\int_{t_{\epsilon}}^{t}k^{1/2}(u(\tau))^{2}(u(\tau)-1)^{1/2} \partial \mathlarger{\mathlarger{\mathscr{B}}}(s)
\end{eqnarray}
since $u_{\epsilon}=k^{1/2}\int_{t_{\epsilon}}^{t_{*}}(u(\tau))^{2}(u(\tau)-1)^{1/2}d\tau$. This is also equivalent to:
\begin{eqnarray}
\overrightarrow{u(t)}=u_{\epsilon}+u_{*}+
\mathcal{C}(\mathbf{X}_{II}\cup\mathbf{X}_{III})\int_{t_{*}}^{t}k^{1/2}(u(\tau))^{2}(u(\tau)-1)^{1/2} \partial\mathlarger{\mathlarger{\mathscr{B}}}(s)
\end{eqnarray}
where $\overrightarrow{u(t_{*})}$ may be finite.
\end{prop}
If $\mathlarger{\mathlarger{\mathscr{W}}}(t;\zeta)$ is a (Gaussian) non-white noise or perturbation with correlation $\zeta$' and switched on' at $t=t_{\epsilon}$ and defined for all $t>t_{\epsilon}$, then the regular and perturbed equations for the matter density function on $\mathbf{X}_{I}$ and $\mathbf{X}_{II}\bigcup\mathbf{X}_{III}$ are
\begin{align}
&\frac{d\widehat{u}(t)}{dt}=\kappa^{1/2}(u(t))^{2}(u(t)-1)^{1/2}];~~t\in\mathbf{X}_{I}
\\& \frac{d\widehat{u}(t)}{dt}=\kappa^{1/2}(u(t))^{2}(u(t)-1)^{1/2}
\otimes \mathlarger{\mathlarger{\mathscr{W}}}(t;\beta);~~t\in\mathbf{X}_{II}\cup\mathbf{X}_{III}
\end{align}
where $\mathlarger{\mathlarger{\mathcal{E}}}\lbrace\mathlarger{\mathlarger{\mathscr{W}}}(t;\zeta)\rbrace=0$ and with a regulated 2-point function
\begin{equation}
\mathlarger{\mathlarger{\mathcal{E}}}\bigg\lbrace\mathlarger{\mathlarger{\mathscr{W}}}(t;\zeta)\mathlarger{\mathscr{W}}(t';\zeta)
\bigg\rbrace=C(|t-t'|:\zeta)\nonumber
\end{equation}
such that $C(0;\zeta)<\infty$. In the limit as $\zeta\rightarrow 0$ then $\mathlarger{\mathlarger{\mathscr{W}}}(t;\zeta)\rightarrow\mathlarger{\mathlarger{\mathscr{W}}}$ and $d\mathlarger{\mathlarger{\mathscr{B}}}(t)=\lim_{\rightarrow\zeta}\mathlarger{\mathlarger{\mathscr{W}}}(t;\zeta)dt$. Then (12.16) becomes a Stratanovich SDE with the same coefficients.
\begin{align}
&d\overline{u(t)}=k^{1/2}[(u(t))^{2}u(t)-1)^{1/2}]
\otimes \lim_{\epsilon\rightarrow
0}\widehat{\mathcal{F}}(t;\zeta)dt\nonumber\\&=k^{1/2}[(u(t))^{2}(u(t)-1)^{1/2}]
\partial\mathlarger{\mathlarger{\mathscr{B}}}(t)
\end{align}
with solution
\begin{equation}
{\widehat{u}}(t)=u_{\epsilon}+k^{1/2}\int_{t_{\epsilon}}^{t}
[(u(\tau))^{2}(u(\tau)-1)^{1/2}]\partial\mathlarger{\mathlarger{\mathscr{B}}}(t)
\end{equation}
for the initial data $u_{\epsilon}$ on $\mathbf{X}^{+}_{II}\bigcup\mathbf{X}^{+}_{III}$. The density diffusion within this interpretation is denoted $\overrightarrow{u(t)}$. The Riemann sum approximation for the Stratanovich stochastic integral is
\begin{align}
\overline{u(t)}=u_{\epsilon}+\frac{1}{2}k^{1/2}\sum_{i=1}^{n-1}
\bigg[(u(t_{i+1}^{n}))^{2}(u(t_{i+1}^{n})-1)^{1/2}-
(u(t_{i}^{n}))^{2}(u(t_{i}^{n})-1)^{1/2}\bigg]\bigg[\mathlarger{\mathlarger{\mathscr{B}}}(t_{i+1})-
\mathlarger{\mathlarger{\mathscr{B}}}(t_{i}^{n})]
\end{align}
in the limit that the partitions $l\lbrace t_{i}^{n}\rbrace$ tend to zero.
\begin{rem}
Stratanovich integrals are not martingales. From Appendix A, the Stratanovich and Ito integrals are related by
\begin{align}
\overrightarrow{u(t)}=&\int_{t_{\epsilon}}^{t}k^{1/2}(u(\tau))^{2}(u(\tau)-1)^{1/2} \partial\mathlarger{\mathlarger{\mathscr{B}}}(t)
\nonumber\\&=\frac{1}{2}k\int_{t_{\epsilon}}^{t}(u(\tau))^{2}(u(\tau)-1)^{1/2}
(\mathlarger{\mathrm{D}}_{u}(u(\tau))^{2}(u(\tau)-1)^{1/2})d\tau+
k^{1/2}\int_{t_{\epsilon}}^{t}(u(\tau))^{2}(u(\tau)-1)^{1/2}\partial
\mathlarger{\mathlarger{\mathscr{B}}}(\tau)\nonumber\\&=
\frac{1}{2}k\int_{t_{\epsilon}}^{t}(u(\tau))^{3}\bigg(\frac{5}{4}u(\tau)-1\bigg)d\tau +
k^{1/2}\int_{t_{\epsilon}}^{t}(u(\tau))^{2}(u(\tau)-1)^{1/2}\partial
\mathlarger{\mathlarger{\mathscr{B}}}(\tau)
\end{align}
\end{rem}
\subsection{Exact solution and hitting time}
The following gives the exact solution  and the blowup 'hitting comoving proper time' for the stochastic integral solution
\begin{thm}
If for all $t\in{\mathbf{X}}_{II}\bigcup{\mathbf{X}}_{III}$, the following Stratanovich SDE holds for the density function diffusion $\overrightarrow{u(t)}$
\begin{equation}
d\overrightarrow {u_{s}(t)}=k^{1/2}(u(t))^{2}(u(t)-1)^{1/2}d\mathlarger{\mathlarger{\mathscr{B}}}(t)
\end{equation}
with initial data $\lbrace t=t_{\epsilon},u(t_{\epsilon})=u_{\epsilon}, \mathlarger{\mathlarger{\mathscr{B}}}(t_{\epsilon})=0\rbrace$, then the density function diffusion blows up at the (random) hitting comoving proper time (HPT)
\begin{equation}
T_{\sigma}=\inf\lbrace t>t_{\epsilon}:\mathlarger{\mathlarger{\mathscr{B}}}(t)=k^{-1/2}|
\tfrac{\pi}{2}+(u_{\epsilon}-1)^{1/2}u_{\epsilon}^{-1}+\tan^{-1}
((u_{\epsilon}-1)^{1/2})\bigg|\equiv L(u_{\epsilon})\rbrace
\end{equation}
\end{thm}
\begin{proof}
Within the Stratanovich interpretation, the SDE (12.21) can be solved as in ordinary calculas
\begin{equation}
\int_{u_{\epsilon}}^{\widehat{u}_{s}(t)}du(t)[  (u(t))^{-2}(u(t)-1)^{-1/2}]=\int_{t_{\epsilon}}^{t}\partial\mathlarger{\mathlarger{\mathscr{B}}}(t)
\end{equation}
The integral on the lhs can be done so that
\begin{equation}
\left[(u(t)-1)^{1/2}(u(t))^{-1}+\tan^{-1}(u(t)-1)^{1/2})
\right]_{u_{\epsilon}}^{\widehat{u}(t)}=
\mathlarger{\mathlarger{\mathscr{B}}}(t)
\end{equation}
which is
\begin{align}
&\left[(u(t)-1)^{1/2}(u(t))^{-1}+\tan^{-1}((u(t)-1)^{1/2}
\right]-\bigg[(u_{\epsilon}-1)^{1/2}u_{\epsilon}^{-1}+\tan^{-1}
((u_{\epsilon}(t)-1)^{1/2}\bigg]=\mathlarger{\mathlarger{\mathscr{B}}}(t)
\end{align}
The density function diffusion blows up when $\overrightarrow{u(t)}=\infty$ so that
$\tan^{-1}(\infty)=\pi/2$. And $(u(t)-1)^{1/2})(u(t))^{-1}$ vanishes as $u(t)\rightarrow\infty$. This then occurs when the Wiener process $\mathlarger{\mathscr{B}}(t)$ hits level $L(u_{\epsilon})$ with associated hitting time ${T}$.
\end{proof}
\begin{rem}
The analysis of blowups in the solution within $\mathbf{X}^{+}_{II}\bigcup\mathbf{X}^{+}_{III}$ , for initial data $u_{\epsilon}$ is then reduced to the problem of a Brownian motion hitting the finite level $L(u_{\epsilon})$ and the associated probabilities and expectations for this event to occur.
\end{rem}
Feller's test can be applied to the corresponding Ito SDE.
\begin{lem}
Given the Stratanovich SDE $d\widehat{u}(t)=k^{1/2}(u(t))^{2}(u(t)-1))^{1/2}
d\mathlarger{\mathlarger{\mathscr{B}}}(t)$ for all $t\in\mathbf{X}_{II}\bigcup\mathbf{X}_{III}$. then Feller's Test gives $\bm{\mathcal{FEL}}(u_{\epsilon},\infty)<\infty $, so that the SDE blows up for a finite hitting time.
\end{lem}
\begin{proof}
The Feller Test applies to an Ito diffusion. The Stratanovich SDE is equivalent to an Ito SDE
\begin{align}
&d\overrightarrow{u(t)}=\kappa(u(t))^{2}(u(t)-1)^{1/2}D_{u}(u(t))^{2}
(u(t)-1)^{1/2}dt+\nonumber\\&
\kappa^{1/2}(u(t))^{2}(u(t)-1)^{1/2}d\mathlarger{\mathlarger{\mathscr{B}}}(t)
\end{align}
so that
\begin{equation}
\bm{\mathcal{FEL}}[u(t)]
=k(u(t))^{2}(u(t)-1)^{1/2}\mathlarger{\mathlarger{\mathrm{D}}}_{u}(u(t))^{2}
(u(t)-1)^{1/2}
\end{equation}\
and the associated diffusion generator is
\begin{equation}
\mathlarger{\mathrm{I\!H}}=(u(t))^{4}(u(t)-1)\mathlarger{\mathlarger{\mathrm{D}}}_{u}+
\bigg(k(u(t))^{2}(u(t)-1)^{1/2}\mathlarger{\mathlarger{\mathrm{D}}}_{u}\mathlarger{\mathlarger{\mathrm{D}}}_{u}(u(t)-1)^{1/2}\bigg)
\mathlarger{\mathlarger{\mathrm{D}}}_{u}\mathlarger{\mathlarger{\mathrm{D}}}_{u}
\end{equation}
Since this is a diffusion with 'drift' within the Ito interpretation, and not a pure diffusion, the blowup may not be removed. Feller's Test is then
\begin{equation}
\bm{\mathcal{FEL}}(u_{\epsilon},\infty)=
\lim_{u(t)\rightarrow\infty}=\bigg|\frac{1}{k}
\int_{u_{\epsilon}}^{\bar{u}(t)}\frac{d\bar{u}(t)}{(\bar{u}(t))^{2}
(\bar{u}(t)-1)^{1/2}\mathlarger{\mathlarger{\nabla}}(\bar{u}(t))^{2}
(u(t)-1)^{1/2}}\bigg|
\end{equation}
This reduces to
\begin{equation}
\bm{\mathcal{FEL}}(u_{\epsilon},\infty)=\lim_{u(t)\rightarrow\infty}=\bigg| \frac{1}{2k}\int_{u_{\epsilon}}^{u(t)}\frac{d\bar{u}(t)}{\bar{u}(t))^{2}
(5\bar{u}(t)-4)}
\bigg|
\end{equation}
giving
\begin{equation}
\bm{\mathcal{FEL}}(u_{\epsilon},\infty)=\lim_{u(t)\rightarrow\infty}\left|\left[\frac{25(u(t))^{2}\ln \left(5 - \frac{4}{u(t)}\right)+20 u(t)+8}{64 u(t))^{2}}\right]_{u_{\epsilon}}^{u(t)}\right|
\end{equation}
or equivalently
\begin{equation}
\bm{\mathcal{FEL}}(u_{\epsilon},\infty)=\lim_{u(t)\rightarrow\infty}
\left|\frac{1}{2k}\left[ \frac{\ln \left(5 - \frac{4}{u(t)}\right)}{32}+\frac{5}{16 u(t)}+ \frac{1}{8(u(t))^{2}}\right]_{u_{\epsilon}}^{u(t)}\right|
\end{equation}
Since $u(0)=1$ and $u_{\epsilon}\gg u(0)$, the log term is always positive and taking the limit gives a finite result.
\begin{equation}
\bm{\mathcal{FEL}}(u_{\epsilon},\infty)=\left|\frac{1}{2k}\left(
\frac{\ln(5)}{32}- \frac{ln\left(5-\frac{4}{u_{\epsilon}}\right)}{32} - \frac{5} {6 u_{\epsilon}}-\frac{1} {8u_{\epsilon}^{3}}\right)\right|<\infty
\end{equation}
Hence $\bm{\mathcal{FEL}}(u_{\epsilon},\infty)<\infty$ and so the SDE blows up.
\end{proof}
Since the SDE blows up as suggested by the Feller test, then one expects a unit probability for blowup of the density function diffusion $\widehat{u}_{s}(t)$ so that $\mathlarger{\mathrm{I\!P}}[T<\infty]=1$. However, it does not necessarily follow that the expected value of the hitting comoving proper time $\mathlarger{\mathlarger{\mathcal{E}}}\lbrace T\rbrace$ is actually finite.
\subsection{Unit probability and infinite hitting-time expectation for blowup: null recurrence}
The main theorem for this section therefore gives the probability and expectation for such
a blowup to occur in the diffusion $\overrightarrow{u(t)}$ for any $t\in\mathbf{X}_{II}\bigcup\mathbf{X}_{III}$
\begin{thm}
For $t\in (t_{\epsilon},\infty)$, the Stratanovich SDE is
\begin{equation}
d\overrightarrow{u(t)}=\kappa^{1/2}(u(t)^{2}(u(t)-1)^{1/2}d\mathlarger{\mathlarger{\mathscr{B}}}(t)
\end{equation}
For $t>t_{\epsilon}$, the diffusion solution $\overrightarrow{u(t)}$ blows up for the random 'hitting comoving proper times'
\begin{equation}
T=\inf\lbrace t>t_{\epsilon}:\mathlarger{\mathlarger{\mathscr{B}}}(t)=k^{-1/2} |
\frac{\pi}{2}+(u_{\epsilon}-1)^{1/2}u_{\epsilon}^{-1}+\tan^{-1}
((u_{\epsilon}-1)^{1/2})|\equiv L(u_{\epsilon})\rbrace
\end{equation}
or $\bm{T}=\inf\lbrace t>t_{\epsilon}:\mathlarger{\mathlarger{\mathscr{B}}}(t)
\equiv Q(u_{\epsilon},k)\rbrace $ so that $|\widehat{u}(T)|=\infty$. The probabilities for this are absolute or occur almost surely so that
\begin{equation}
\mathlarger{\mathrm{I\!P}}[T_{Q} <\infty]=1
\end{equation}
or equivalently
\begin{equation}
\lim_{T\uparrow\infty}\mathlarger{\mathrm{I\!P}}[\sup_{t_{\epsilon}\le t\le T}\mathlarger{\mathlarger{\mathscr{B}}}(t)=Q]=1
\end{equation}
But the stochastic expectation of the random comoving hitting proper time when this occurs is infinite so that
\begin{equation}
\mathlarger{\mathlarger{\mathcal{E}}}\big\lbrace T\big\rbrace =\infty
\end{equation}
and is therefore null recurrent.
\end{thm}
\begin{proof}
Let $\mathlarger{\mathscr{B}}(t)$ be the Brownian process beginning at $t=t_{\epsilon}$ so that $\mathscr{B}(t_{\epsilon})=0$. Then
\begin{equation}
\mathlarger{\mathrm{I\!P}}[T<t]
=2\mathlarger{\mathrm{I\!P}}[\mathlarger{\mathlarger{\mathscr{B}}}(t)(t)>Q(u_{\epsilon})]
\end{equation}
If $\mathlarger{\mathlarger{\mathscr{B}}}(t)>h$ then ${T}<t$ and $(\mathlarger{\mathlarger{\mathscr{B}}}(t)(T+t)-\mathlarger{\mathlarger{\mathscr{B}}}(t))$ is also a Brownian motion. By symmetry $ \mathlarger{\mathrm{I\!P}}[(\mathlarger{\mathlarger{\mathscr{B}}}(t)(t)-
\mathlarger{\mathlarger{\mathscr{B}}}(t))>0| T<t]
=\frac{1}{2}$. Hence
\begin{align}
&\mathlarger{\mathrm{I\!P}}[\widehat{\mathscr{B}}(t)>h]=\mathlarger{\mathrm{I\!P}}[{T}<t)
\bigcap(\mathlarger{\mathlarger{\mathscr{B}}}(t)-\mathlarger{\mathlarger{\mathscr{B}}}(T))>0)]\nonumber\\&
=\mathlarger{\mathrm{I\!P}}[{T}<t]\otimes\mathlarger{\mathrm{I\!P}}
[(\mathlarger{\mathlarger{\mathscr{B}}}(t)-\mathlarger{\mathlarger{\mathscr{B}}}(T))>0|
{T}<t]=\frac{1}{2}\mathlarger{\mathlarger{\mathcal{P}}}[{T}<t]
\end{align}
But $\mathlarger{\mathlarger{\mathscr{B}}}(t)$ is Gaussian so that
\begin{equation}
\mathlarger{\mathrm{I\!P}}[T<t]=2\mathlarger{\mathrm{I\!P}}
[\mathlarger{\mathlarger{\mathscr{B}}}(t)>w]=\frac{2}{(2\pi)^{1/2}t}
\int_{\Sigma}^{\infty}\exp[-\frac{y^{2}}{2t}]dy
\end{equation}
Making the substitution $y\rightarrow y/t^{1/2}$
\begin{equation}
\mathlarger{\mathrm{I\!P}}[T<t]=2\mathlarger{\mathlarger{\mathrm{I\!P}}}[w(t)>Q(u_{\epsilon},k)]=\frac{2}{(2\pi)^{1/2}}
\int_{h/t^{1/2}}^{\infty}
\exp[-\frac{y^{2}}{2}]dy
\end{equation}
Then
\begin{equation}
\mathlarger{\mathrm{I\!P}}[T<\infty]=\lim_{t\rightarrow\infty}\frac{2}{(2\pi)^{1/2}}\int_{h/t^{1/2}
}^{\infty}\exp[-\frac{y^{2}}{2}]dy~=~1
\end{equation}
Since $\mathlarger{\mathlarger{\mathscr{B}}}(t)$ is also a martingale, the Doob inequality can be applied. And since the exponential function is monotone increasing it follows that
\begin{equation}
\mathlarger{\mathrm{I\!P}}(\sup_{t_{\epsilon}\le t\le T}\mathlarger{\mathlarger{\mathscr{B}}}(t)\ge Q)\equiv\mathlarger{\mathrm{I\!P}}(\sup_{t_{\epsilon}\le t\le T}\exp(\xi\mathlarger{\mathlarger{\mathscr{B}}}(t))\ge\exp( \xi Q))
\end{equation}
The Doob inequality gives
\begin{align}
&\mathlarger{\mathrm{I\!P}}(\sup_{t_{\epsilon}\le t\le T}\mathlarger{\mathlarger{\mathscr{B}}}(t)\ge Q)\equiv\mathlarger{\mathlarger{\mathcal{P}}}(\sup_{t_{\epsilon}\le t\le T}\exp(\xi\mathlarger{\mathlarger{\mathscr{B}}}(t))\ge\exp( \xi\lambda))\nonumber\\&
=\frac{\mathlarger{\mathlarger{\mathcal{E}}}\lbrace(\exp(\xi\mathlarger{\mathlarger{\mathscr{B}}}(t))\rbrace}{\exp(\xi Q)}=\frac{\exp(\tfrac{1}{2}\xi^{2}T)}{\exp( \xi\lambda))}
=\exp(t\frac{1}{2}\xi^{2}Q-\xi Q)
\end{align}
The parameter $\xi$ can be chosen so that the rhs is minimised
\begin{equation}
\frac{d}{d\xi}\exp\left(\frac{1}{2}\xi^{2}TQ-\xi Q\right)
=(\xi T-Q)\exp\left(\frac{1}{2}\xi^{2}TQ-\xi Q\right)=0
\end{equation}
Then $\xi=Q/T$. Equation (12.45) becomes
\begin{equation}
\mathlarger{\mathlarger{\mathrm{I\!P}}}(\sup_{t_{\epsilon}\le t\le T}
\mathlarger{\mathlarger{\mathscr{B}}}(t)\ge Q)\le \exp\left(-\frac{Q^{2}}{2T}\right)
\end{equation}
so that
\begin{equation}
\lim_{T\uparrow\infty}\mathlarger{\mathlarger{\mathrm{I\!P}}}(\sup_{t_{\epsilon}\le t\le T}\mathlarger{\mathlarger{\mathscr{B}}}(t)=Q)\le \lim_{T\uparrow\infty}\exp\left(-\frac{Q^{2}}{2T}\right)=1
\end{equation}
Hence, the probability of a blowup $|\overrightarrow{u(T)}|=\infty$ occurring at some random comoving proper time occurs almost surely. However, the expected value is infinite in that $\mathlarger{\mathlarger{\mathcal{E}}}({T})=\infty$. To prove this, note that the integral is essentially over a probability density function $\mathcal{P}(t)$ for the distribution of random blowup times on $\mathbf{X}_{II}\bigcup\mathbf{X}_{III}=[t_{\epsilon},\infty)$. $ \mathlarger{\mathlarger{\mathrm{I\!P}}}[T<\infty]=\int_{t_{\epsilon}}^{\infty}\mathcal{P}(t)dt$
so that $\mathlarger{\mathlarger{\mathcal{E}}}(T)=\int_{t_{\epsilon}}^{\infty}t
\mathcal{P}(t)dt $ and where
\begin{equation}
\mathlarger{\mathlarger{\mathrm{I\!P}}}(t)=\frac{1}{\sqrt{2\pi}}\bigg(Q(u_{\epsilon},k)\bigg)t^{-1/2}\exp [-\frac{(Q((u_{\epsilon},k))^{2}}{2t}]
\end{equation}
Then
\begin{equation}
\mathlarger{\mathlarger{\mathcal{E}}}(T)=\mathlarger{\mathlarger{\mathcal{P}}}(t)=
\frac{1}{\sqrt{2\pi}}\int_{t_{\epsilon}}^{\infty}dt(Q(u_{\epsilon},k))t^{-1/2}
\exp\left[-\frac{(Q(u_{\epsilon},k))^{2}}{2t}\right]
\end{equation}
The exponential can be expanded so that
\begin{equation}
\mathlarger{\mathlarger{\mathcal{E}}}(T)\sim
\frac{1}{\sqrt{2\pi}}(Q(u_{\epsilon},k))\int_{t_{\epsilon}}^{\infty}t^{-1/2}
\bigg[1-\frac{Q(_{\epsilon})^{2}}{2t}+...\bigg]dt
\end{equation}
For large $t$, the terms beyond order one vanish rapidly so that
\begin{equation}
\mathlarger{\mathlarger{\mathcal{E}}}({T})\sim
\frac{1}{\sqrt{2\pi}}(|Q(u_{\epsilon},k)
|\int_{t_{\epsilon}}^{\infty}t^{-1/2}dt=\infty
\end{equation}
The integral can also be done explicitly and gives
\begin{align}
&\mathlarger{\mathlarger{\mathcal{E}}}\lbrace T \rbrace=\frac{1}{\sqrt{2\pi}}
\left(Q(\psi_{\epsilon},k)\right)\int_{t_{\epsilon}}^{\infty}t^{-1/2}
\exp\left[-\frac{Q(u_{\epsilon},k))^{2}}{2t}\right]dt\nonumber\\&
=\bigg[2t^{1/2}\exp\bigg(-\frac{Q(u_{\epsilon})}{2t}\bigg)+(2\pi)^{1/2}
Q(u_{\epsilon},K)erf\bigg(\frac{Q(\psi_{\epsilon})}{(2t)^{1/2}}\bigg)\bigg]_{t_{\epsilon}}^{\infty}=\infty
\end{align}
as required. A second proof also follows from the Doleans-Dade exponential martingale. As before, the blowup hitting time is when the Brownian motion $
\mathlarger{\mathlarger{\mathscr{B}}}(t)$ hits level $Q$ give by (-) so that $T_{L}=\inf\lbrace t>t_{\epsilon}:\mathlarger{\mathlarger{\mathscr{B}}}(t)=Q\rbrace$. Since $\mathlarger{\mathlarger{\mathscr{B}}}(t)$ is a martingale then so is $\mathscr{M}(t)=\exp(\xi\mathlarger{\mathscr{B}}(t)-\tfrac{1}{2}\xi^{2}t)$. By the Novikov Theorem,$\mathlarger{\mathlarger{\mathcal{E}}}\lbrace \mathscr{M}(t)\rbrace =1$. By the stopping theorem $\mathscr{M}(t\wedge T_{L})$ is also a martingale. For $t>T_{Q}$ one has $\mathlarger{\mathlarger{\mathscr{B}}}(t\wedge T_{Q})=Q$ then
\begin{equation}
\lim_{t\uparrow\infty}\exp\left(\xi\mathlarger{\mathlarger{\mathscr{B}}}(t)(t\wedge T_{Q}-\frac{1}{2}\xi^{2}(t\wedge T_{Q})\right)=\lim_{t\uparrow\infty}\exp(\xi Q-\tfrac{1}{2}\xi^{2}T_{Q})=0
\end{equation}
For $T_{L}=\infty$ then $\mathlarger{\mathlarger{\mathscr{B}}}(t)(t\wedge T_{L})<L$ so that now
\begin{equation}
\lim_{t\uparrow\infty}\exp\left(\xi\mathlarger{\mathlarger{\mathscr{B}}}(t)(t\wedge T_{Q}-\frac{1}{2}\xi^{2}(t\wedge T_{Q})\right)\le \lim_{t\uparrow\infty}\exp(\xi Q-\tfrac{1}{2}\xi^{2}t)=0
\end{equation}
Both cases are expressed as
\begin{equation}
\lim_{t\uparrow\infty}\exp\left(\xi\mathlarger{\mathlarger{\mathscr{B}}}(t\wedge T_{L})-\frac{1}{2}\xi^{2}(t\wedge T_{Q})\right)=\mathlarger{\mathcal{C}}(T_{L}<\infty)\exp(\xi Q-\tfrac{1}{2}\xi^{2}t)
\end{equation}
Taking the expectation and using the Novikov Thereom
\begin{equation}
1=\mathlarger{\mathcal{C}}(T_{Q}<\infty)\exp(\xi L-\tfrac{1}{2}\xi^{2}t)
\end{equation}
Then $\mathlarger{\mathlarger{\mathcal{E}}}(T_{Q}<\infty)=\mathlarger{\mathlarger{\mathcal{P}}}(T_{Q}<\infty)=1$. Equation (9.53) can be written as
\begin{equation}
\mathlarger{\mathlarger{\mathcal{E}}}\big\lbrace\exp(-\tfrac{1}{2}\xi^{2}T_{Q})\rbrace=\exp(-\xi Q)
\end{equation}
Setting $\xi=\tfrac{1}{2}\xi^{2}$ gives the Laplace transform of $T_{Q}$.
\begin{equation}
\mathlarger{\mathlarger{\mathcal{E}}}\big\lbrace\exp(-\lambda T_{Q})\rbrace=\exp(-Q\sqrt{2\lambda})
\end{equation}
The derivative is
\begin{equation}
\frac{d}{d\lambda}\mathlarger{\mathlarger{\mathcal{E}}}\big\lbrace\exp(-\lambda T_{Q})\rbrace=
\mathlarger{\mathlarger{\mathcal{E}}}\big\lbrace-T_{Q}\exp(-\lambda T_{Q})\big\rbrace=\frac{1}{2\sqrt{2\lambda}}\exp(-\sqrt{2\lambda}T_{Q})
\end{equation}
\end{proof}
so that $\mathlarger{\mathlarger{\mathcal{E}}}\lbrace T_{Q}\rbrace=\infty$ for $Q\rightarrow 0$.
\section{Conclusion}
In this paper, an attempt has been made to demonstrate how methods and tools from stochastic analysis can be tentatively and rigorously incorporated into a scenario within general relativity; in order to describe and account for the possible effects of randomness, fluctuations or noise at a classical level. In particular, randomness and noise can suppress blowups and singularities which are unavoidable within the deterministic classical theory.
\clearpage
\appendix
\renewcommand{\theequation}{\Alph{section}.\arabic{equation}}
\section{Basic Definitions}
For convenience, the following well-established definitions from stochastic analysis are briefly stated without proof. Details and proofs can be found in a number of standard works and texts [56-61].
\begin{defn}
let $(\Omega,\mathfrak{F},\mathrm{I\|P})$ be a probability space.
Within the probability triplet, $(\Omega,\mathfrak{F})$ is a measurable space where $\mathfrak{F}$ is the $\sigma$-algebra (or Borel field) that should be interpreted as being comprised of all reasonable subsets of the state space $\Omega$. Then $\mathrm{I\|P}$ is a function such that $\mathrm{I\|P}:\mathfrak{F}\rightarrow [0,1]$, so that for all $A\in\mathfrak{F}$, there is an associated probability $\mathrm{I\|P}$. The measure is a probability measure when $\mathrm{I\|P}(A)=1$.

These are fairly standard (and abstract) definitions within probability theory, stochastic functional analysis and ergodic theory.
\begin{defn}
Let $t\in\mathbf{Q}=[0,T]$ and let $(\Omega,\mathfrak{F},\mathrm{I\|P})$ be a probability space. Let $\overline{X(x;\zeta)}$ be a random scalar function that depends only on $t\in\mathbf{Q}$ and also $\omega\in\Omega$. Given any pair $(t,\omega)$ there is a mapping $\psi:\mathbf{Q}\times\Omega\rightarrow\mathbf{R}$ such that
\begin{equation}
X:(\omega,x)\longrightarrow\overline{X(t;\omega)}\nonumber
\end{equation}
so that $\overline{X(t,\zeta)}$ is a random variable or field on $\mathbf{R}$ with respect to the probability space $(\Omega,\mathfrak{F},\mathrm{I\|P})$. The stochastic field is then essentially a family of random variables $\lbrace X(t;\omega)\rbrace$ defined with respect to $(\Omega,\mathfrak{F},\mathrm{I\|P})$.
\end{defn}
\begin{defn}
Given a random or stochastic process $X(t;\zeta)$, then if
\begin{equation}
\int_{\Omega}\|\overline{X(t;\zeta)}\|d\mathrm{I\|P}(\omega)<\infty
\end{equation}
the expectation of $\overline{X(t;\zeta)}$ is
\begin{equation}
\mathlarger{\mathlarger{\mathcal{E}}}(\overline{X(t;\zeta)})
=\int_{\Omega}X(t;\omega)d\mathrm{I\|P}(\omega)
\end{equation}
The expectation will be denoted $\mathlarger{\mathlarger{\mathcal{E}}}(\overline{X}(t))$.
\end{defn}
\begin{defn}
Given the probability space $(\Omega,\mathcal{F},\bm{\mathcal{P}})$, the a filtration $\mathbb{F}$ is a family of increasing $\sigma$-fields on $(\Omega,\mathfrak{F})$ with $\mathfrak{F}_{t}\subset\mathcal{F}$, for all $t\in\mathbf{Q}$. Hence,
\begin{equation}
\mathbb{F}=\lbrace \mathfrak{F}_{o},\mathfrak{F}_{1},...,\mathfrak{F}_{t},...,\mathfrak{F}_{T}\rbrace
\end{equation}
If $\mathcal{F}_{o}\subset\mathcal{F}_{1}\subset...\mathcal{F}_{t}\subset...,\mathcal{F}_{T}$ then $(\Omega,\mathcal{F},\mathbb{F},\mathrm{I\|P})$ is a filtered probability space.
\end{defn}
\begin{defn}
A Weiner process or Brownian motion $\lbrace \mathscr{B}(t)\rbrace$ is a stochastic process with following properties
\begin{enumerate}
\item $\widehat{\mathscr{B}}(t)$ are continuous functions of $t$. (Continuity of sample paths.)
\item Any $|\widehat{B}(t)-\widehat{\mathscr{B}}(s)|$ for $t>s$ is independent of the past and is essentially Markovian.
\item Increments $|\widehat{\mathscr{B}}(t)-\widehat{\mathscr{B}}(s)| $ have a normal or Gaussian distribution.
\item Any sample path $\widehat{\mathscr{B}}(t)$ for all $t\in\mathbf{Q}$ is nowhere differentiable.
\end{enumerate}
The expectation is $\mathcal{E}(\mathscr{W}(t))=0$ and the differential is $d\widehat{\mathscr{B}}(t)=\mathscr{W}(t)dt$, where $\widehat{\xi}(t)$ is a white noise with 2-point function $\mathcal{E}(\mathscr{W}(t)\mathscr{W}(s))=\delta(t-s)$.
\end{defn}
The Ito and and Stratanovich stochastic integrals are defined as follows
\begin{defn}
Let $\overline{X(t)}$ be a continuous function of t. If $\mathbf{Q}=[0,T]$ is partitioned so that $t_{i=0}=t_{0}$ then $t_{0}<t_{1}<t_{2}<...<t_{i}$ and $t_{*}=t_{m}$. The Riemann-Steiltjes sum over $\mathbf{Q}$ is
\begin{equation}
\overline{X(t)}=\int_{t_{\epsilon}}^{t}\psi(s) d\mathscr{B}(s)=\sum_{i=0}^{n-1}\psi
(t_{i}^{n})[\widehat{\mathscr{B}}(t_{i+1}^{n})-
\mathscr{B}(t_{i}^{n})]
\end{equation}
in the limit that partitions $\lbrace t_{i}^{n}\rbrace\rightarrow 0$. For each $i=1$ to $(n-1)$. This interpretation essentially always takes the minimum value of the pair $\lbrace {X}(t_{i+1}^{n}),{X}(t_{i}^{n}))\rbrace$. The alternative Stratanovich interpretation always takes the averaged value$ \frac{1}{2}|{X}((t_{i+1}^{n}))+{X}(t_{i}^{n})|$ so that the Stratanovich stochastic integral is
\begin{equation}
\overline{X(t)}= \int_{t_{\epsilon}}^{t}{X}[v] d\mathscr{B}(u)=
\sum_{i=0}^{n-1}\frac{1}{2}[{X}(t_{i+1}^{n})-{X}(t_{i}^{n})][\mathscr{B}(t_{i+1}^{n})-
\mathscr{B}(t_{i}^{n})]
\end{equation}
\end{defn}
The Ito interpretation requires the Ito calculas and $d\mathscr{B}(t)d\mathscr{B}(t)=dt$ and $(dt^{2})=0$. However, within the Stratanovich interpretation the rules of ordinary calculas apply. Two important properties of Ito integrals are zero mean an Ito isometry such that
\begin{equation}
\mathlarger{\mathcal{E}}\left\llbracket\int_{0}^{T}\psi(t) d\mathscr{B}(t)\right\rrbracket=0
\end{equation}
and
\begin{align}
&\mathlarger{\mathcal{E}}\bigg\llbracket\int_{0}^{T}\psi(t) d\mathscr{B}(t)\bigg\rrbracket^{2}=\int_{0}^{T}\mathcal{E}\bigg\llbracket\psi(t)\bigg\rrbracket^{2}dt
\\&
\mathlarger{\mathcal{E}}\bigg\llbracket\int_{0}^{T}\psi(t)d\mathscr{B}(t)\int_{0}^{T}
\lambda(t)d\mathscr{B}(t)\bigg\rrbracket^{2}=\int_{0}^{T}\mathcal{E}\bigg\llbracket
\lambda(t)\psi(t)\bigg\rrbracket dt
\end{align}
\begin{defn}
A martingale is essentially any stochastic process $X:\mathbf{R}^{+}\times\Omega\rightarrow\mathbf{R}^{+}$ or $\overline{X(t,\zeta)}$, with $\zeta\in\Omega$ defined with respect to an underlying probability triplet $(\Omega,\mathcal{F},\mathlarger{\mathrm{I\!P}})$, whose expected value $\mathcal{E}(X(t';\zeta))\equiv\int_{\Omega}d\mathlarger{\mathrm{I\!P}}(\omega)X(t;\zeta)$ for some future time $t'>t$ is the same as the present value $t$. The proper time parameter $t$ is continuous with $t\in\mathbf{R}^{+}$ or some subset. There is also an increasing family of filtrations $\lbrace\bm{\mathscr{F}}_{t}\rbrace\subset\bm{\mathcal{F}}$ of $\sigma$-algebras. Then for any $t'<t$ one has a martingale if
\begin{equation}
\bigg(\int_{\Omega}\overline{X(t,\zeta)}d\mathlarger{\mathrm{I\!P}}(\zeta)\bigg\|\mathcal{F}_{t'}\bigg)
\equiv\mathcal{E}\bigg\llbracket(\overline{X(t)}\|\bm{\mathcal{F}}_{t'})\bigg\rrbracket =\overline{X(t')}
\end{equation}
If the process is a submartingale it describes expected growth on average so that
\begin{equation}
\bigg(\int_{\Omega}\overline{X(t,\zeta)}d\mathlarger{\mathrm{I\!P}}(\zeta)\|\mathcal{F}_{t'}\bigg)
\equiv\mathcal{E}\bigg\llbracket(\overline{X(t)}\|\bm{\mathcal{F}}_{t'})\bigg\rrbracket >\overline{X(t')}
\end{equation}
The process is a supermartingale if it diminishes on average
\begin{equation}
\bigg(\int_{\Omega}\overline{X(t,\zeta)}
d\mathlarger{\mathrm{I\!P}}(\zeta)\|\mathcal{F}_{t'}\bigg)\equiv
\mathcal{E}\bigg\llbracket(\overline{X(t)}\|\bm{\mathcal{F}}_{t'})\bigg\rrbracket <\overline{X(t')}
\end{equation}
\end{defn}
\begin{rem}
Ito integrals are generally martingales whose suprema are bounded from above. This makes them useful tools for studying the growth,bounds and blowups in random processes or diffusions.
\end{rem}
Ito integrals are solutions of a stochastic differential equations. Let $(\Omega,\mathcal{F},\mathlarger{\mathrm{I\!P}})$ be a probability space and let $\mathscr{B}(t)$ be a one-dimensional Brownian
motion or Weiner process. Let $t_{o}\let\le T$ or $t\in\mathbf{Q}$.
Let $\psi:\mathbf{Q}\rightarrow\mathbf{R}$ be a function $X(t)$
and define linear or nonlinear coefficients $\phi(X(t))$ and $\psi(X(t))$, where
$\mathcal{X}:\mathcal{Q}\rightarrow\mathbf{R}\rightarrow\mathbf{R}$ and $Y:\mathcal{Q}\rightarrow\mathbf{R}\rightarrow\mathbf{R}$. Then given the SDE
for all $t\in\mathbf{Q}$
\begin{equation}
d\overline{X(t)}=\phi(X(t))dt+\psi(X(t))d\mathscr{B}(t)
\end{equation}
with initial data $X(t_{o}))=X_{o}$ and $\mathscr{B}(t_{o})=0$, this initial-value
problem is equivalent to the stochastic Ito integral for all $t\in\mathbb{Q}$
\begin{equation}
\overline{X(t)}=X_{o}+\int_{t_{o}}^{t}\phi(X(s)ds+\int_{t_{o}}^{t}
\psi(u(s))d\mathscr{B}(s)
\end{equation}
where the stochastic process $\overline{X(t)}$ is continuous and
$\mathcal{F}$-adapted for all $t\in \mathbf{Q}$. For a driftless  diffusion
$\phi(X(t))=0$ and
\begin{equation}
d\overline{X}(t)=\psi(X(t))d\mathscr{B}(t)
\end{equation}
with solution given by
\begin{equation}
\overline{X(t)}=X_{o}+\int_{t_{o}}^{t}\psi(X(s))d\mathlarger{\mathscr{B}}(s)
\end{equation}
\end{defn}
\begin{defn}
The quadratic variation of a stochastic process $\overline{X(t)}$ for all $t>t_{\epsilon}$ is
\begin{equation}
\bigg\langle\overline{X}(t),\overline{X}\bigg\rangle(t)=\lim\sum_{i=1}^{n}|
\overline{X}(t_{i+1}^{n})-\overline{X}(t_{i}^{n})|^{2}=|\psi(X(t))|^{2}t
\end{equation}
and so the differential is
\begin{equation}
d[\overline{X},\overline{X}](t)=\lim\sum_{i=1}^{n}|
\bm{\psi}(t_{i+1}^{n})-\bm{\psi}(t_{i}^{n})|^{2}=|\psi(X(t))|^{2}dt
\end{equation}
The quadratic variation of a pure Brownian motion $\widehat{W}(t)$ for all $t>t_{\epsilon}$ is
\begin{equation}
[\mathlarger{\mathscr{B}}(t),\mathlarger{\mathscr{B}}(t)](t)=\lim\sum_{i=1}^{n}|
\mathlarger{\mathscr{B}}(t_{i+1}^{n})-\mathlarger{\mathscr{B}}(t_{i}^{n})|^{2}=t
\end{equation}
\end{defn}
\begin{defn}
Given the stochastic Ito integral $\overline{X}(t)=X_{\epsilon}+
\int_{t_{\epsilon}}^{t}\psi[X(s)]ds$ defined for all $t>t_{\epsilon}$, then if $\overline{X(t)}$ is a square-integrable martingale then the quadratic variation of the diffusion $\overline{X(t)}$ is bounded so that
\begin{equation}
\bigg\langle \overline{X(t)},\overline{X(t)}\bigg\rangle(t)=\lim_{\delta t_{i}^{n}\downarrow 0,n\downarrow \infty
}\sum_{i=1}^{n}(\widehat{X}(t_{i+1}^{n})-\widehat{X}(t_{i}^{n})]^{2})=
\int_{t_{\epsilon}}^{t}(\psi(X(s))|^{2}ds < \infty
\end{equation}
In full
\begin{align}
&\bigg\langle\overline{X(t)},\overline{X(t)}\bigg\rangle(t)=
\bigg[\bigg\|\int_{t_{\epsilon}}^{t}\psi(X(s))d\mathlarger{\mathscr{B}}(s)
\bigg\|,\bigg\|\int_{t_{\epsilon}}^{t}\psi(X(s))
d\mathlarger{\mathscr{B}}(s)
\bigg\|\bigg](t)\nonumber\\&=\int_{t_{o}}^{t}|\psi(X(s))|^{2}ds
< \infty
\end{align}
\end{defn}
\begin{defn}
Given constants $c_{p},C_{p}$, the Buckholder-Gundy inequality for a martingale $\widehat{\psi}(t)$ is
\begin{equation}
c_{p}\mathcal{E}\bigg\llbracket\bigg\langle\overline{X},\overline{X}\bigg\rangle(T)^{p/2}\big\rrbracket
\le \mathcal{E}(\sup_{t\le T}|\psi(t)|)^{p})\le C_{p}\mathcal{E}\bigg\llbracket \bigg\langle\overline{X},\overline{X}\bigg\rangle(T)^{p/2}\bigg\rrbracket
\end{equation}
for $p\ge 1$ and for $t\in[0,T)$. For p=1, this reduces to the Davis inequality.
\begin{equation}
c\mathcal{E}(\bigg\langle\overline{X},\overline{X}\bigg\rangle(T)\le \mathcal{E}(\sup_{t\le T}|\psi(t)|))\le
C\mathcal{E}(\bigg\langle\overline{X},\overline{X}\bigg\rangle(T))
\end{equation}
\end{defn}
The final result give is the Ito formula for an Ito process
\begin{thm}
Given the following generic SDE for a process $\widehat{\psi}(t)$ all $t\in\mathbf{Q}=[0,T]$
\begin{equation}
d\widehat{X}(t)=\phi(X(t))dt+\psi(X(t))d\mathlarger{\mathscr{B}}(t)
\end{equation}
let $f(\widehat{\psi}(t))$ be a continuously differentiable $C^{2}$-function of $\widehat{\psi}(t)$. Then the stochastic differential exists and is
\begin{align}
d\Phi(\overline{X(t)})&=\mathrm{D}_{X}\Phi(\widehat{X}(t))d\widehat{\psi}(t)+
\frac{1}{2}\mathrm{D}_{X}\mathrm{D}_{X}\Phi(X(t))\bigg\langle\overline{X},\overline{X}\bigg\rangle(t)
\nonumber\\&=\mathrm{D}_{X}\Phi(\widehat{X}(t))d\overline{X(t)}+
\frac{1}{2}\mathrm{D}_{X}\mathrm{D}_{X}\Phi(X(t))\psi(X(t))^{2}dt\nonumber\\&
=(\mathrm{D}_{X}\Phi(X(t)))\Phi(X(t))dt+\mathrm{D}_{X}\Psi(X(t))
\mathrm{D}_{X}(X(t))d\mathscr{B}(t)+\frac{1}{2}\mathrm{D}_{X}\mathrm{D}_{X}\Phi(X(t))
\psi(X(t))^{2}dt
\end{align}
As an integral
\begin{equation}
\Phi(\widehat{X}(t))=\Phi(\widehat{X}(0))+\int_{t_{o}}^{t}
\mathrm{D}_{X}\Phi(\widehat{X}(u))d\widehat{X}(t)+\frac{1}{2}\int_{t_{o}}^{t}
\mathrm{D}_{X}\mathrm{D}_{X}\Phi(X(u))\mathcal{Y}(\psi(t))^{2}dt
\end{equation}
For a driftless diffusion with $\phi(X(t))=0$, then
\begin{equation}
\Phi(\widehat{X}(t))=\Phi(\widehat{X}(0))+\int_{t_{o}}^{t}
(\mathrm{D}_{X}\Phi(\widehat{X}(u))\mathcal{Y}(X(t))d\mathscr{B}(t)
+\frac{1}{2}\int_{t_{o}}^{t}\mathrm{D}_{X}\mathrm{D}_{X}\Phi(X(s))\mathcal{Y}(X(s))^{2}ds
\end{equation}
\end{thm}
\section{Borel-Cantelli Lemma}
Let $\lbrace \overline{X}_{1},...,\overline{X}_{n}...\rbrace $ be a sequence of independent random events with respect to a probability space. Then if the sum of probabilities is finite such that
\begin{equation}
\sum_{n=1}^{\infty}\mathlarger{\mathrm{I\!P}}(\overline{X}_{n})<\infty
\end{equation}
then
\begin{equation}
\mathlarger{\mathrm{I\!P}}(\lim_{n\uparrow\infty}\sup \overline{X}_{n})=\mathlarger{\mathrm{I\!P}}\bigg(\bigcap_{n=1}^{\infty}
\bigcup_{k\ge n}^{\infty}\overline{X}_{n}\bigg)=0
\end{equation}
If the sum of probabilities is finite then the set of all outcomes that are repeated occurs with probability zero.
\section{Gronwall's Lemma}
The Gronwall Lemma is a fundamental estimation for bounds on non-negative functions of one real variable, usually time.It is particulary useful in evolution or growth problems where one analyses the bounds or estimates on growth of a function that evolves in time. Also, if the lemma holds for $t\in[0,T]$ or $t\in\mathbf{R}^{+}$ then there will be no blowups or singularities in the function of interest $\phi(t)$ within those intervals.
\begin{lem}
Let $\Lambda:[0,T]\rightarrow\mathbf{R}^{+}$ or $\mathbf{R}^{+}\rightarrow\mathbf{R}^{+}$ and let $f:\mathbf{R}^{+}\rightarrow\mathbf{R}^{+}$ be functions such that
\begin{equation}
X(t)\le B +C\int_{0}^{t}f(u)\Lambda(u)du
\end{equation}
for $t\in[0,T]$ or $t\in\mathbf{R}^{+}$ then $X(t)\le C+\int_{0}^{t}f(u)X(u)
du$ and
\begin{equation}
X(t)~~\le~~B\exp\left(C\int_{0}^{t} f(u)du\right)
\end{equation}
\end{lem}
This is well-known and standard result. The result also holds when the inequality is reversed. As a corollary, it will also immediately apply to expectations of stochastic quantities.
\begin{cor}
Let $\overline{X(t)}$ be a stochastic quantity with expectation $\mathlarger{\mathcal{E}}(|\bm{\Lambda}(t)|^{\gamma})\ge 0 $ for $t\in\mathbf{R}^{+}$ and $p\ge 1$ then if
\begin{equation}
(|\overline{X}(t)|)^{p}\le B+C\int_{0}^{t}f(s)\mathlarger{\mathcal{E}}|X(s)|^{p}ds
\end{equation}
one has
\begin{equation}
\mathlarger{\mathcal{E}}(|\overline{X}(t)|^{p})~~\le~~B\exp\left(C\int_{0}^{t}f(u)du\right)
\end{equation}
And if
\begin{equation}
\mathlarger{\mathcal{E}}(|\widehat{\psi}(t)|)^{\gamma}\ge A\pm\int_{0}^{t}f(\tau)\mathlarger{\mathcal{E}}|\psi(\tau)|^{\gamma}d\tau
\end{equation}
one has
\begin{equation}
\mathlarger{\mathcal{E}}(|\psi(t)|^{\gamma})~~\ge~~A\exp(\pm\int_{0}^{t}f(\tau)d\tau)
\end{equation}
\end{cor}
}
\end{document}